\documentclass{fundam}

\renewenvironment{proof}[1][{\bf Proof:}]{\begin{trivlist}
\item[\hskip \labelsep {\bfseries #1}]}{\end{trivlist}}

\newcommand*{\vv}[1]{\overrightarrow{#1}}

\newcommand{\cauA}[2]{#1\circ#2}

\newcommand{\inp}[5]{\SigmaB_{#2\in#3}#1?\Seq{#4_#2}{#5_#2}}

\newcommand{\oup}[5]{\bigoplus_{#2\in#3}#1!\Seq{#4_#2}{#5_#2}}
\newcommand{\oupp}[5]{\bigoplus_{#2\in#3}#1!\Seq{#4_#2}{#5}}
\newcommand{\oupTx}[5]{\textstyle{\bigoplus}_{#2\in#3}#1!\Seq{#4_#2}{#5_#2}}

 \newcommand{\RG}{\ensuremath{{\sf G}}}
\newcommand{\GlSyB}{\mbox{$\mathlarger{\mathlarger{\boxplus}}$}}
 \newcommand{\PiB}{\mathlarger{\mathlarger\Pi}}
  \newcommand{\SigmaB}{\mathlarger{\mathlarger\Sigma}}

     \newcommand{\bms}{\vartheta}
   \newcommand{\ES}{\mathcal{P\!E}}
   \newcommand{\ESA}{\mathcal{P\!E}}

     \newcommand{\EGGA}{\mathcal{T\!E}}

 \newcommand{\te}{\eta}
  \newcommand{\ee}{\epsilon}

  \newcommand{\sep}{~|~}
  \newcommand{\eh}[1]{|#1|}

 \newcommand{\mylabel}[1]{\label{#1}}

\newcommand{\comoccA}{\ensuremath{\delta}}

\newcommand{\Commpair}[2]{\ensuremath{(#1,#2)}}

\newcommand{\pastpref}[2]{\ensuremath{\concat{{#1}}{#2}}}

\newcommand{\act}[1]{\ensuremath{\sf act}(#1)}

\newcommand{\concat}[2]{\ensuremath{#1\,{\cdot}\,#2}}

 \newcommand{\inpP}[5]{\SigmaB_{#2\in#3}#1?\Seq{#4_#2}{#5}}
\newcommand{\oupP}[5]{\bigoplus_{#2\in#3}#1!\Seq{#4_#2}{#5}}

\newcommand{\refToFigure}[1]{Figure~\ref{#1}}
\newcommand{\refToSection}[1]{Section~\ref{#1}}
\newcommand{\refToExample}[1]{Example~\ref{#1}}
\newcommand{\refToLemma}[1]{Lemma~\ref{#1}}
\newcommand{\refToTheorem}[1]{Theorem~\ref{#1}}

\newcommand{\refToDef}[1]{Definition~\ref{#1}}

\newcommand{\refToProp}[1]{Proposition~\ref{#1}}
\newcommand{\refToFact}[1]{Fact~\ref{#1}}

\newcommand{\set}[1]{\{#1\}}
\newcommand{\kf}[1]{\ensuremath{\mathsf{#1}\xspace}}

\newcommand{\cSinferrule}[3][]{
  \mprset{fraction={===},
  fractionaboveskip=0.2ex,
  fractionbelowskip=1.30ex}
  \inferrule[#1]{#2}{~#3}
}

\newcommand{\PP}{\ensuremath{P}}
\newcommand{\la}{\ell}

\newcommand{\M}{\la}
\newcommand{\Messages}{\ensuremath{{\sf Lab}}}

\newcommand{\pp}{{\sf p}}
\newcommand{\q}{{\sf q}}
\newcommand{\pr}{{\sf r}}
\newcommand{\ps}{{\sf s}}
\newcommand{\pt}{{\sf t}}
\newcommand{\Participants}{\ensuremath{{\sf Part}}}
\newcommand{\sendL}[2]{#1!#2}
\newcommand{\rcvL}[2]{#1?#2}
\newcommand{\dagL}[2]{#1\!\dagger\! #2}

\newcommand{\pc}{~|~}
\newcommand{\Seq}[2]{#1;#2}

\newcommand{\inact}{\ensuremath{\mathbf{0}}}
\newcommand{\val}{v}
\newcommand{\Q}{\ensuremath{Q}}
\newcommand{\R}{\ensuremath{R}}

\newcommand{\del}[1]{\ensuremath{\delta}}
\newcommand{\Nt}{\ensuremath{{\sf N}}}

\newcommand{\parN}{\mathrel{\|}}
\newcommand{\pP}[2] {#1[\![\,#2\,]\!]}

\newcommand{\stackred}[1]{\xrightarrow{#1}}

 \newcommand{\G}{\ensuremath{{\sf G}}}

 \newcommand{\End}{\kf{End}}

\newcommand{\play}[1]{\ensuremath{{\sf play}(#1)}}
\newcommand{\GlSy}{\GlSyB}

 \newcommand{\parG}{\mathrel{\|}}
  
\newcommand{\proj}[2]{#1 \!  \upharpoonright  \! #2\,}
\newcommand{\projs}[2]{#1 \! \Rsh   \! #2\,}
\newcommand{\projAP}[2]{#1\,@\,#2\,}
\newcommand{\projtau}[2]{#1\,@\,#2\,}
\newcommand{\subt}{\leq}

\newcommand{\derN}[2]{\vdash #1 :#2}

\newcommand{\rulename}[1]{\ensuremath{[\textsc{#1}]}}

\newcommand{\weight}{\ensuremath{{\sf depth}}}

\newcommand{\swap}[3]{#1{\rhd}_{\!#3}#2}

\newcommand{\gr}{\ensuremath{\,\#\,}}

\newcommand{\grr}{\ensuremath{\,\#\,}}

\newcommand{\impl}{\ensuremath{\Rightarrow}}
\newcommand{\procev}{\eta}
\newcommand{\Procev}{\ES}
\newcommand{\preEv}{\chi}

\newcommand{\netevA}{\rho}

\newcommand{\globevA}{\delta}
\newcommand{\precP}{\ensuremath{\leq}}
\newcommand{\precN}{\ensuremath{\prec^\os}}
\newcommand{\precE}{\ensuremath{\prec^\ee}}
\newcommand{\precNL}[1]{\ensuremath{\prec^{#1}}}

   \newcommand{\GEa}{\mathcal{T\!E}}
\newcommand{\DE}{\mathcal{D\!E}}
\newcommand{\DEA}{\mathcal{N\!E}}

 \newcommand{\GEA}{\mathcal{N\!E}}

\newcommand{\dualA}[2]{\ensuremath{#1\Join #2}}

\newcommand{\dualprecsim}[2]{\ensuremath{#1~_{\precapprox}\!\Join_{\,\succapprox} #2}}
\newcommand{\locev}[2]{\ensuremath{#1:: #2}}
\newcommand{\locevA}[3]{\ensuremath{#1:: #3}}
\newcommand{\locevAS}[3]{\ensuremath{#1:: #3}}
\newcommand{\actseq}{\zeta}

\newcommand{\comseqA}{ \tau } %{\mathit{s}}

\newcommand{\procActs}{\vv{\pi;}}
\newcommand{\procActsS}{\vv{\pi}}
\newcommand{\procActsP}{\vv{\pi';}}
\newcommand{\procActsPS}{\vv{\pi'}}

\newcommand{\eqclass}[1]{\ensuremath{[#1]_{\sim}}}

\newcommand{\ESP}[1]{\ensuremath{\mathcal{S^P}(#1)}}

\newcommand{\ESN}[1]{\ensuremath{\mathcal{S}^{\mathcal{N}}(#1)}}
\newcommand{\ESNA}[1]{\ensuremath{\mathcal{S}^{\mathcal{N}}(#1)}}
\newcommand{\ESet}{\ensuremath{\mathcal{X}}}
\newcommand{\GSet}{\ensuremath{\mathcal{Y}}}

\newcommand{\ESGA}[1]{\ensuremath{\mathcal{S}^{\mathcal{T}}(#1)}}
\newcommand{\emptyseq}{\ensuremath{\epsilon}}
\newcommand{\eqdef}{\ensuremath{{=_{\sf def}}}}

\newcommand{\Conf}[1]{\ensuremath{\mathcal{C}(#1)}}
\newcommand{\CD}[1]{\ensuremath{\mathcal{D}(#1)}}

\newcommand{\io}[1]{\ensuremath{{\sf i/o}}(#1)}
\newcommand{\last}[1]{\ensuremath{{\sf last}}(#1)}

\newcommand{\necA}[1]{\ensuremath{{\sf nec}(#1)}}
\newcommand{\gecA}[1]{\ensuremath{{\sf tec}(#1)}}
\newcommand{\postPA}[2]{\ensuremath{{#2}\,\blacklozenge\,{(#1)}}}
\newcommand{\prePA}[2]{\ensuremath{{#2}\,\lozenge\,{(#1)}}}
\newcommand{\postA}[2]{\ensuremath{{#2}\,\blacklozenge\,{#1}}}
\newcommand{\preA}[2]{\ensuremath{{#2}\,\lozenge\,{#1}}}

\newcommand{\postGA}[2]{\ensuremath{{#2}\bullet{#1}}}

\newcommand{\preGA}[2]{\ensuremath{{#2}\circ{#1}}}
\newcommand{\pair}[2]{(#1,#2)}

\newcommand{\coDefGr}{::=^{coind}}

\newcommand{\GP}{\G}

\newcommand{\asCom}{\beta}
\newcommand{\FPaths}[1]{{\sf Tr^+}(#1)}
\newcommand{\IPaths}[1]{{\sf Tr^+}(#1)}

\newcommand{\confAs}[2]{#1\parN#2}
\newcommand{\Msg}{\mathcal{M}}
\newcommand{\addMsg}[2]{#1\cdot #2}

\newcommand{\nr}[1]{{\sf nr}(#1)}

\newcommand{\agtO}[6]{\GlSyB_{#3\in #4}#1#2!#5_#3;#6_#3}
\newcommand{\agtOP}[6]{\GlSyB_{#3\in #4}#1#2!#5_#3;#6}
\newcommand{\agtI}[4]{#1#2?#3;#4}
\newcommand{\agtOS}[4]{#1#2!#3;#4}
\newcommand{\agtIS}[3]{#1#2?#3}
\newcommand{\agtSOS}[3]{#1#2!#3}

\newcommand{\asty}[2]{#1 \parN #2}   %%% Ilaria 27/09/21

\newcommand{\BB}{\mathcal B}

\newcommand{\derSI}[2]{\vdash^{b} #2}
\newcommand{\derPI}[2]{\vdash^{b} #2}

\newcommand{\projnet}[2]{{\sc proj}_{#1}(#2)}
\newcommand{\projnetfun}[1]{{\sc proj}_{#1}}

\newcommand{\mq}[3]{\langle#1,#2,#3\rangle}

\newcommand{\CommAs}[3]{#1#3!#2}
\newcommand{\CommAsI}[3]{#1#3?#2}
\newcommand{\Comm}[3]{#1#3#2}%{#1#3!?#2}

\newcommand{\os}{\omega}

\newcommand{\ct}[3]{#1\propto^{\,#2}\!\!#3}
\newcommand{\oks}[2]{{\sf req}(#1,#2)}

\newcommand{\eqA}[2]{[#1,#2]_\sim}

\newcommand{\point}{{\sf ev}}
\newcommand{\osq}[1]{\ensuremath{{\sf otr}(#1)}}

\newcommand{\pro}[2]{#1\,@\,#2}

\newcommand{\at}[2]{#1[#2]}

\newcommand{\range}[3]{#1[#2\,...\,#3]}

\newcommand{\cardin}[1]{\!\!\pc\! #1\!\!\pc\!}
\newcommand{\mult}[2]{\ensuremath{{\sf m}(#1, #2)}}

\newcommand{\traceS}{trace }
\newcommand{\tracesS}{traces }
\newcommand{\trace}{trace}
\newcommand{\traces}{traces}

\newcommand{\mapWh}[2]{#1\triangleright#2}
\newcommand{\mapBl}[2]{#1\blacktriangleright#2}
\newcommand{\prefMsg}[2]{#2\cdot #1}
\newcommand{\prR}[3]{(#1,#2)\in\downarrow_#3}
\newcommand{\iR}[3]{(#1,#2)\in #3}
\newcommand{\RR}{{\mathcal R}}
\newcommand{\filt}[2]{#1\,\lceil_\os \, #2}

\newcommand{\filtP}[3]{#1\,\lceil_{#2} \, #3}

\newcommand{\idepth}{\ensuremath{{\sf idepth}}}
\newcommand{\maxP}{\ensuremath{{\sf maxPath}}}

\newcommand{\AsOut}{Ext-Out}
\newcommand{\AsIn}{Ext-In}

\newcommand{\qt}{asynchronous type}
\newcommand{\qts}{asynchronous types}
\newcommand{\Qts}{Asynchronous types}

\newcommand{\sgt}{global type}
\newcommand{\sgts}{global types}
\newcommand{\agt}{asynchronous type}
\newcommand{\agts}{asynchronous types}

\newcommand{\oi}{input/output}

\newcommand{\balanced}{balanced}
\newcommand{\balancing}{balancing}

\newenvironment{lemmaa}[2]{\begin{trivlist}
\item[\hskip \labelsep {\bfseries Lemma #1}] {#2} \it}{\end{trivlist}}

\usepackage{xcolor}
\usepackage{prooftree}
\usepackage{mathpartir}
\usepackage{afterpage}
\usepackage{amsmath}
\usepackage{xspace}
\usepackage{mathtools,amssymb}
\usepackage{hyperref}
\usepackage{relsize}
\usepackage[curve,matrix,arrow,cmtip]{xy}
\usepackage{tikz}
\usepackage{enumerate} %per avere diversi contatori

\def\sm{\smallskip}

\usepackage[dvips=false,pdftex=true,paperwidth=192mm,paperheight=262mm,margin=1.8cm,bottom=1.8cm,top=2.8cm]{geometry}

\begin{document}

%%%%%%%  parameters to be filled in by copy-editor  %%%%%%%%%%

\setcounter{page}{1}
\publyear{24}
\papernumber{2188}
\volume{192}
\issue{1}

\finalVersionForARXIV
%\finalVersionForIOS
%%%%%%%%%%%%%%%%%%%%%%%%%%%%%%%%%%

\title{Global Types and Event Structure  Semantics \\
                for Asynchronous Multiparty Sessions}

\author{Ilaria Castellani\thanks{This research has been supported by the ANR17-CE25-0014-01
                              CISC project.}
\\
INRIA, Universit\'e C\^ote d'Azur, France\\
ilaria.castellani@inria.fr
\and Mariangiola Dezani-Ciancaglini\\
Dipartimento di Informatica, Universit\`a di Torino, Italy\\
 dezani@di.unito.it
\and Paola Giannini\thanks{Address for correspondence:  Computer Science Institute, DiSSTE, Pz. S.
                 Eusebio 5 - 13100 Vercelli, Italy.}\thanks{This original research has the financial support of the Universit\`{a}  del Piemonte Orientale.
                         This work was partially funded by the MIUR project ``T-LADIES'' (PRIN 2020TL3X8X).
                                     \newline \newline
                    \vspace*{-6mm}{\scriptsize{Received August 2023; \ accepted April 2024.}}}
\\
DiSSTE, Universit\`{a} del Piemonte Orientale,  Italy\\
paola.giannini@uniupo.it }

\maketitle

\runninghead{I. Castellani et al.}{Types and Semantics for Asynchronous Sessions}

\begin{abstract}
  We propose an interpretation of multiparty sessions with
  asynchronous communication as \emph{Flow Event Structures}.  We
  introduce a new notion of \emph{asynchronous type} for such
  sessions, ensuring the expected properties for multiparty sessions,
  including progress.
  Our asynchronous types, which reflect asynchrony more directly and
  more precisely than standard global types and are more permissive,
  are themselves interpreted as \emph{Prime Event Structures}.  The
  main result is that the Event Structure interpretation of a session
  is equivalent, when the session is typable, to the Event Structure
  interpretation of its asynchronous type, namely their domains of
  configurations are isomorphic.
  \end{abstract}


%\tableofcontents

 \section{Introduction}\label{intro}
 Session types describe interactions among a number of participants,
 which proceed according to a given protocol.  They extend classical
 data types by specifying, in addition to the type of exchanged data,
 also the interactive behaviour of participants, namely the sequence
 of their \oi\ actions towards other participants.  The aim of
 session types is to ensure safety properties for sessions, such as
 the \emph{absence of communication errors} (no type mismatch in
 exchanged data) and \emph{deadlock-freedom} (no standstill until
 every participant  is terminated).  Sometimes, a
 stronger property is targeted, called \emph{progress} (no
 participant waits forever).\medskip

 % \hspace*{5.08mm}
 Initially conceived for describing binary protocols in the
$\pi$-calculus~\cite{THK94,HVK98}, session types have been later
extended to multiparty protocols~\cite{CHY08,CHY16} and embedded into
a range of functional, concurrent, and object-oriented programming
languages~\cite{ABB0CDGGGH16}.
While binary sessions can be described by a single session type,
multiparty sessions require two kinds of types: a \emph{global type}
that describes the whole session protocol, and \emph{local types} that
describe the contributions of the individual participants to the
protocol. The key requirement in order to achieve the expected safety
properties is that all local types be obtained as projections from the
same global type.

%\\ \hspace*{5.08mm}
 Communication in sessions is always directed from a given sender to a
given receiver. It can be synchronous or asynchronous. In the first
case, sender and receiver need to synchronise in order to exchange a
message. In the second case, messages may be sent at any time, hence a
sender is never blocked. The sent messages are stored in a queue,
where they may be fetched by the intended receiver.  Asynchronous
communication is often favoured for multiparty sessions, since such
sessions may be used to model web services or distributed
applications, where the participants are spread over different sites.

Session types have been shown to bear a strong connection with models
of concurrency such as communicating automata~\cite{DY12}, as well as
with message-sequence charts~\cite{CHY16}, graphical
choreographies~\cite{LTY15,TuostoG18}, and various brands of linear
logics~\cite{CP10,TCP11,Wadler14,PCPT14,CPT16}.

In a companion paper~\cite{CDG22}, we investigated the relationship
between synchronous multiparty sessions and Event Structures
(ESs)~\cite{Win88}, a well-known model of concurrency which is
grounded on the notions of causality and conflict between events.  We
considered a simple calculus, where sessions are described as networks
of sequential processes~\cite{DGJPY15}, equipped with standard global
types~\cite{CHY08}. We proposed an interpretation of sessions as
\emph{Flow Event Structures} (FESs)~\cite{BC88a,BC94}, as well as an
interpretation of global types as \emph{Prime Event Structures}
(PESs)~\cite{Win80,NPW81}. We showed that for typed sessions these two
interpretations yield isomorphic domains of configurations.

In the present paper, we undertake a similar endeavour in the
asynchronous setting.  This involves devising a new notion of
asynchronous type for asynchronous networks.  We start by considering
a core session calculus as in the synchronous case, where processes
are only able to exchange labels, not values, hence local types
coincide with processes.
%and global types may be directly projected to processes.
Moreover, networks are now endowed with a queue and they act on this
queue by performing outputs or inputs: an output stores a message in
the queue, while an input fetches a message from the queue.
%Then
 The  present paper differs from~\cite{CDG22} not only
for the operational semantics, but also for the typing rules and more
essentially for the event structure semantics of sessions and
types.

\eject
To illustrate the difference between synchronous and asynchronous
sessions and motivate the introduction of
new types for the latter, let us discuss a simple example. Consider the
network:\smallskip

\centerline{ $\Nt=\pP{\pp}{\sendL{\q}\la;\rcvL{\q}{\la'}} \parN
  \pP{\q}{\sendL{\pp}{\la'};\rcvL{\pp}\la}$ }\smallskip

\noindent where each of the participants $\pp$ and $\q$ wishes to first send a
message to the other one and then receive a message from %she.
the other one.
\medskip\\
 \hspace*{5.08mm}
In a synchronous setting this network is stuck, because a network
communication arises from the synchronisation of an output with a
matching input, and here the output $\sendL{\q}{\la}$ of $\pp$ cannot
synchronise with the input $\rcvL{\pp}{\la}$ of $\q$, since the latter
is guarded by the output $\sendL{\pp}{\la'}$.  Similarly, the output
$\sendL{\pp}{\la'}$ of $\q$ cannot synchronise with the input
$\rcvL{\q}{\la'}$ of $\pp$.  Indeed, this network is not typable
because  any  global type for it
%would need to
should  have one of the two forms:\smallskip

\centerline{$\G_1 = \pp \rightarrow \q : \la ; \, \q \rightarrow \pp :
  \la' $ \qquad\qquad $\G_2 = \q \rightarrow \pp : \la' ;  \,\pp \rightarrow \q : \la $}

\noindent
However, neither of the $\G_i$ projects down to the correct
processes
for both $\pp$ and $\q$ in $\Nt$. For instance, $\G_1$ projects to the
correct process
$\concat{\sendL{\q}{\la}}{\rcvL{\q}{\la'}}$ for $\pp$,
but its projection on $\q$ is $\concat
{\rcvL{\pp}{\la}}{\sendL{\pp}{\la'}}$, which is not the correct
process for $\q$.
\medskip\\
 \hspace*{5.08mm}
In an asynchronous setting, on the other hand, this network is run in
parallel with a queue $\Msg$, which we indicate by
$\confAs{\Nt}{\Msg}$, and it can always move for whatever choice of
$\Msg$. Indeed, the moves of an asynchronous network are
not complete communications but rather ``communication halves'',
namely outputs or inputs.  For instance, if the queue is empty, then
$\confAs\Nt\emptyset$ can move by first performing the two outputs in
any order, and then the two inputs in any order. If instead the queue
contains a message from $\pp$ to $\q$ with
label $\la_1$,
% a label $\la_1$ different from $\la$,
followed by
a message from $\q$ to $\pp$ with
label $\la_2$,
% a label $\la_2$ different from $\la'$,
which we indicate by $\Msg =
\addMsg{\mq{\pp}{\la_1}{\q}}{\mq{\q}{\la_2}{\pp}}$, then, assuming
$\la_1 \neq \la$ and $\la_2 \neq \la'$, the network will be stuck
after performing the two outputs.
Indeed, the two inputs will not be able to occur, since the two
messages on top of the queue are not those expected by $\pp$ and
$\q$. Hence we look for a notion of type that accepts the network
$\confAs\Nt\emptyset$ but rejects the network
$\confAs{\Nt}{\addMsg{\mq{\pp}{\la_1}{\q}}{\mq{\q}{\la_2}{\pp}}}$.
\\
 \hspace*{5.08mm}
The idea for our new \emph{\qts} is quite simple: to split
communications into outputs and inputs, and to add a queue to the
type, thus mimicking very closely the behaviour of asynchronous
networks. Hence, our \qts\ have the form $\G \parN \Msg$. Clearly, we
must impose some well formedness conditions on such types, taking into
account also the content of the queue. Essentially, this amounts to
requiring that each input appearing in the type be justified by a
preceding output in the type or by a message in the queue, and vice
versa, that each output in the type or message in the
queue be matched by a corresponding input in the type.
\\
 \hspace*{5.08mm}
Having introduced this new notion of type, % With our new \qts,
it becomes now possible to type the network $\confAs\Nt\emptyset$ with
the \qt\ $\G \parN \emptyset$, where
$\G =
\Seq{\Seq{\Seq{\CommAs{\pp}{\la}{\q}}{\CommAs{\q}{\la'}{\pp}}}{\CommAsI{\pp}{\la}{\q}
  }}{\CommAsI{\q}{\la'}{\pp}}$, or with the other  \qts\ obtained from it by swapping the outputs and/or the
inputs.  Instead, the network
$\confAs{\Nt}{\addMsg{\mq{\pp}{\la_1}{\q}}{\mq{\q}{\la_2}{\pp}}}$ will
be rejected, because the \qt\
$\G \parN \addMsg{\mq{\pp}{\la_1}{\q}}{\mq{\q}{\la_2}{\pp}}$ is not
well formed, since its two inputs do not match the first two messages
in the queue.
\\
 \hspace*{5.08mm}
A different solution was proposed in~\cite{Mostrous2009} by means of
an asynchronous subtyping relation on local types which allows outputs
to be brought forward.  In our setting this boils down to a subtyping
relation on processes yielding both $\sendL{\q}\la;\rcvL{\q}{\la'}\leq
\rcvL{\q}{\la'};\sendL{\q}\la$ and
$\sendL{\pp}{\la'};\rcvL{\pp}\la\leq\rcvL{\pp}\la;\sendL{\pp}{\la'}$.
With the help of this subtyping, both $\G_1$ and $\G_2$ become types
for the network $\Nt\parG\emptyset$ above.  Unfortunately, however,
this subtyping turned out to be undecidable~\cite{BCZ17,LY17}.

\noindent To define our interpretations of asynchronous networks and types into
FESs and PESs respectively, we follow the same schema as for their
synchronous counterparts in our previous work~\cite{CDG22}.  In
particular, the events of the ESs are defined syntactically and they
record the ``history'' of the particular communication occurrence they
represent. More specifically, the events of the FES associated with a
network -- which we call \emph{network events} -- record the local
history of their communication.  By contrast, the events of the PES
associated with an \qt\ -- which we call \emph{type events} -- record
the global history of their communication, namely the whole sequence
of past communications that caused it, which is extracted from the
computation trace using a permutation equivalence.  However, while
in~\cite{CDG22} an event represented a communication between two
participants, here it represents an output or an input pertaining to a
single \mbox{participant}.  Hence, some care must be taken in defining the
flow relation between network events\footnote{In FESs, the flow
  relation represents a direct causality between events.}, and in
particular the ``cross-flow'' between an output event and the matching
input event, since input events may also be justified by a message in
the queue.  For type events, queues appear inside events and affect
the permutation equivalence.  Therefore, our ES semantics for the
asynchronous setting is far from being a trivial adaptation of that
given in~\cite{CDG22} for the synchronous setting.

\medskip
\noindent To sum up,  the contribution of this paper is twofold:\sm

1) We propose an original syntax for \qts, which,
in our view, models asynchronous communication in a more
natural  and precise  way than existing approaches.
Our type system is more permissive than the standard one~\cite{CHY16}
-- in particular, it allows outputs to take precedence over inputs as
in~\cite{Mostrous2009}, a characteristics of asynchronous
communication -- but it remains decidable. We show that our \qts\
ensure classical safety properties as well as progress.

2) We present an Event Structure semantics for asynchronous networks
and \qts. Networks are interpreted as FESs and \qts\ are interpreted
as PESs. Our main result here is an isomorphism between the
configuration domains of the FES of a typed network and the PES of its
\qt.

\medskip
The paper is organised as follows. Section~\ref{sec:calculus}
introduces our calculus for asynchronous multiparty sessions.
Section~\ref{sec:typesAs} introduces asynchronous types and the
associated type system, establishing its main properties.
Section~\ref{sec:eventStr} recaps some necessary background
on Event Structures.  In Section~\ref{sec:process-ES} we recall
our interpretation of processes as PESs, taken from our companion
paper~\cite{CDG22}. In Section~\ref{sec:netA-ES} we present our
interpretation of asynchronous networks as FESs.  In
Section~\ref{sec:eventsA} we define our interpretation of asynchronous
types as PESs.  Finally, in Section~\ref{sec:resultsA} we prove the
equivalence between the FES semantics of a network and the PES
semantics of its asynchronous type.  We conclude with a discussion on
related work and future research directions in
Section~\ref{sec:relatedA}.
% The Appendix contains the proofs of three technical lemmas and the
% glossary of symbols.
All results are given with full proofs. When not appearing in
the main text, the proofs may be found in the Appendix, which contains
also the glossary of symbols.

\section{A core calculus for multiparty sessions}\mylabel{sec:calculus}

We now formally introduce our calculus, where asynchronous multiparty
sessions are represented as networks of sequential processes with
queues.

\noindent We assume the following base sets: \emph{participants}, ranged over by
$\pp,\q,\pr$ and forming the set $\Participants$, and \emph{labels},
ranged over by $\la,\la',\dots$ and forming the set $\Messages$.

\medskip
 Let $\pi \in \{ \sendL{\pp}{\la}, \rcvL{\pp}{\la} \pc
 \pp\in \Participants, \la \in \Messages\}$ denote an \emph{atomic
   action}.  The action $\sendL{\pp}{\la}$ represents an output of
 label $\la$ to participant $\pp$, while the action $\rcvL{\pp}{\la}$
 represents an input of label $\la$ from participant $\pp$.

 \begin{definition}[Processes]\mylabel{p}
Processes are  defined by:\sm

  \centerline{$
\begin{array}{lll}
\PP & \coDefGr  &
\oup\pp{i}{I}{\la}{\PP}%
~~\mid~~
\inp\pp{i}{I}{\la}{\PP}%
~~\mid~~
\inact
\end{array}
$} \sm

\noindent
where $I$ is non-empty and $\la_h\not=\la_k$ for all $h,k\in I$,
$h\neq k$, i.e.
labels in choices are all different.
\end{definition}

The symbol $ \coDefGr$, in the definition above and in later
definitions, indicates that the productions should be interpreted
\emph{coinductively}.
%That is,
Namely, they define possibly infinite
processes.  However, we assume such processes to be \emph{regular},
that is, with finitely many distinct subprocesses. In this way, we
only obtain processes which are solutions of
finite sets of equations, see~\cite{Cour83}. So, when writing
processes, we shall use (mutually) recursive equations.

In the following,
trailing $\inact$'s will be omitted. Processes of the shape$\,\oup\pp{i}{I}{\la}{\!\PP}$ and
$\inp\pp{i}{I}{\la}{\!\PP}$ are called {\em output} and {\em input
  processes}, respectively. We will write $\pp!\la ; \PP $ or $\pp?\la; \PP$ for choices with
just one branch, and use the infix notation $\oplus$ and $+$ in the examples.

Processes can be seen as trees where internal nodes are decorated by
$\pp!$ or $\pp?$, leaves by $\inact$, and edges by labels $\la$.

In a full-fledged calculus, labels would carry values, namely
they would be of the form $\la(\val)$.  For
simplicity, we consider only simple labels $\la$ here.  This will allow us to
project global types directly to processes, without having to
explicitly introduce local types, see Section \ref{sec:typesAs}.

\medskip
In our calculus, \emph{asynchronous
communication} is modelled in the usual way, by
storing sent messages in a queue.
We define \emph{messages} to be triples $\mq\pp{\la}\q$, where $\pp$
is the sender and $\q$ is the receiver, and \emph{message queues} (or
simply \emph{queues}) to be possibly empty sequences of messages:\smallskip

\centerline{$\Msg::=\emptyset \mid \addMsg{\mq\pp{\la}\q}{\Msg} $}\sm

The order of messages in the queue is the order in which they will be
read. Since the only reading order that matters is that between
messages with the same sender and the same receiver, we consider
message queues modulo the structural equivalence given by:\smallskip

\centerline{$
\addMsg{\addMsg{\Msg}{\mq\pp{\la}\q}}{\addMsg{\mq\pr{\la'}\ps}{\Msg'}}\equiv
  \addMsg{\addMsg{\Msg}{\mq\pr{\la'}\ps}}{\addMsg{\mq\pp{\la}\q}{\Msg'}}
  ~~\text{if}~~\pp\not=\pr~~\text{or}~~\q\not=\ps
$}\sm

\smallskip\noindent
Note in particular that $\addMsg{\mq\pp{\la}\q}{\mq\q{\la'}\pp} \equiv
\addMsg{\mq\q{\la'}\pp}{\mq\pp{\la}\q}$. These two equivalent queues
represent a situation in which both participants $\pp$ and $\q$ have
sent a label to the other one, and neither of them has read the
label sent by the other one.
%This situation may indeed
%happen in a network with asynchronous communication.
Since the two sends occur in parallel, the order of the corresponding
messages in the queue should be irrelevant. This point will be further
illustrated by  Examples~\ref{sync-async-characteristic-example} and~\ref{htp}.

\medskip
 In the following we will always consider queues modulo structural equivalence.

Networks are comprised of located processes of the form
$\pP{\pp}{\PP}$ composed in parallel, each with a different
participant $\pp$, and by a message queue.

\eject
\begin{definition}[Networks]
{\em  Networks} are defined by:\sm

\centerline{$   \Nt \parallel \Msg$} \sm

\noindent where $ \Nt = \pP{\pp_1}{\PP_1} \parN
\cdots \parN \pP{\pp_n}{\PP_n}$ with
 $\pp_h \neq \pp_k $ for any $h \neq k$,  and $\Msg$ is a message queue.
\end{definition}

We assume the standard laws for parallel composition, stating that
$\parN$ is associative, commutative, and has neutral element
$\pP\pp\inact$ for any fresh $\pp$. These laws, together with the
structural equivalence on queues, give rise to the structural
congruence on networks, also denoted by the symbol $\equiv$.

\medskip
If $\PP\neq\inact$ we write $\pP{\pp}{\PP}\in\Nt$ as short for
$\Nt\equiv\pP{\pp}{\PP}\parN\Nt'$ for some $\Nt'$.  This abbreviation
is justified by the associativity and commutativity of parallel composition.

\begin{figure}[!ht]\small
 \centerline{$
\begin{array}{c}
\confAs{\pP{\pp}{\oup\q{i}{I}{\la}{\PP}}\parN\Nt}{\Msg} \stackred{\CommAs\pp{\la_k}\q}
  \confAs{\pP{\pp}{\PP_k}\parN\Nt}{\addMsg{\Msg}{\mq\pp{\la_k}\q}}\quad \text{ where  }\
   k \in I  \quad
   {~~~~~~\rulename{Send}} \\[3pt]
\confAs{\pP{\q}{\inp\pp{j}{J}{\la}{\Q}}\parN\Nt}{\addMsg{\mq\pp{\la_k}\q}{\Msg}}\stackred{\CommAsI\pp{\la_k}\q}
 \confAs{\pP{\q}{\Q_k}\parN\Nt}{\Msg}\quad  \text{ where  }\
 k \in J \quad
  {~~~~~~\rulename{Rcv}}
\end{array}
$}
\caption{LTS for networks.}\mylabel{fig:netredAs}\mylabel{fig:asynprocLTS}\vspace*{-1mm}
\end{figure}

To define the {\em operational semantics} of networks, we use  an
 LTS whose  transitions are decorated by
outputs or inputs.
%Therefore,
We define the set of {\em \oi\ communications} (communications for
short), ranged over by $\asCom$, $\asCom'$, to be $\{
\CommAs{\pp}{\M}{\q}, \CommAsI{\pp}{\M}{\q} \pc
\pp,\q\in \Participants, \M \in \Messages\}$, where
$\CommAs{\pp}{\M}{\q}$ represents the send of a label $\M$ from
participant $\pp$ to participant $\q$, and $\CommAsI{\pp}{\M}{\q}$ the
read by participant $\q$ of the label $\M$ sent by participant $\pp$.
To memorise this notation, it is helpful to view $\pp\q$ as the
channel from $\pp$ to $\q$ and the exclamation/question mark as the
mode (write/read) in which the channel is used. The LTS semantics of
networks, defined modulo $\equiv$,  is specified by the two
Rules $\rulename{Send}$ and $\rulename{Rcv}$ given in
\refToFigure{fig:asynprocLTS}.  Rule $\rulename{Send}$ allows a
participant $\pp$ with an internal choice (a sender) to send to a
participant $\q$ one of its possible labels $\la_k$ by adding it to
the queue. Symmetrically, Rule $\rulename{Rcv}$ allows a participant
$\q$ with an external choice (a receiver) to read the first label
$\la_k$ sent to it by participant $\pp$, provided this label is among
the $\la_j$'s specified in the choice.  Using structural equivalence,
the first message from $\pp$ to $\q$, if any, can always be moved to
the top of the queue.

A key role will be played by (possibly empty) sequences of
communications, defined as traces.
\begin{definition}[Traces]\label{tra}  (Finite) traces are defined by:

\centerline{$\comseqA::=\ee\mid\concat\beta\comseqA$}
 We use $\cardin{\comseqA}$ to denote the length of the trace $\comseqA$.
\end{definition}
When $\comseqA=\concat{\beta_1}{\concat\ldots{\beta_n}}$ ($n\geq 1)$
we write $\Nt\parN\Msg\stackred{\comseqA}\Nt'\parN\Msg'$ as short for\smallskip

\centerline{$\Nt\parN\Msg\stackred{\beta_1}\Nt_1\parN\Msg_1\cdots\stackred{\beta_n}\Nt_{n}\parN\Msg_{n}
  = \Nt'\parN\Msg'$}

\smallskip
Let us now consider the semantics of the network
$\confAs{\Nt}{\emptyset}$ discussed in the introduction.

\begin{example}
\mylabel{sync-async-characteristic-example}
Let $\Nt=\pP{\pp}{\sendL{\q}\la;\rcvL{\q}{\la'}} \parN
  \pP{\q}{\sendL{\pp}{\la'};\rcvL{\pp}\la}$.
Then $\confAs\Nt\emptyset$ can move by first
performing the two sends, in any order, and then the two reads, in any
order.
A possible execution of $\confAs\Nt\emptyset$, where
the use of $\equiv$ on the queue is explicitly shown, is:\smallskip

\centerline{$ \begin{array}{lcl}
\confAs\Nt\emptyset &\stackred{\CommAs\pp{\la}\q}&
  \confAs{\pP{\pp}{\rcvL{\q}{\la'}} \parN
  \pP{\q}{\sendL{\pp}{\la'};\rcvL{\pp}\la}}{\mq\pp{\la}\q}\\[2pt]
&\stackred{\CommAs\q{\la'}\pp}&
  \confAs{\pP{\pp}{\rcvL{\q}{\la'}} \parN
  \pP{\q}{\rcvL{\pp}\la}}{\addMsg{\mq\pp{\la}\q}{\mq\q{\la'}\pp}
}\\[2pt]
&\equiv&
  \confAs{\pP{\pp}{\rcvL{\q}{\la'}} \parN
  \pP{\q}{\rcvL{\pp}\la}}{\addMsg{\mq\q{\la'}\pp}{\mq\pp{\la}\q}
}\\[2pt]
&\stackred{\CommAsI\q{\la'}\pp}&
  \confAs{\pP{\pp}{\inact} \parN
  \pP{\q}{\rcvL{\pp}\la}}{\mq\pp{\la}\q}\\[2pt]
&\stackred{\CommAsI\pp{\la}\q}&
  \confAs{\pP{\pp}{\inact} \parN
  \pP{\q}{\inact}}{\emptyset}\\
\end{array}
$}
\end{example}

The following example illustrates a simple case in which both
participants need to do outputs before inputs.

\begin{example}\label{htp}
  Alice and Bob play heads and tails as follows: they each write their
  prediction on a piece of paper and then they exchange their papers.
  If the predictions are the same they start again, otherwise they
  flip the coin and the winner is the one whose prediction was
  correct.  The initial interaction in this scenario may be
  represented by the network $\pP\pp\PP\parN\pP\q\Q\parN\emptyset$
  where $\pp,\q$ incarnate Alice and Bob, $h, t$ stand for head and
  tail, and $\PP, \Q$ are defined as follows:\medskip

  \centerline{$\begin{array}{lcl}
                 \PP&=&\q! h;(\q?h; P+\q?t)\oplus \q! t;(\q?h+\q?t; P)\\
                 \Q&=&\pp! h;(\pp?h; Q+\pp?t)\oplus \pp! t;(\pp?h+\pp?t; Q)
                       \end{array}$}

% \noindent and $h,t$ stand for head, tail.
\end{example}

We now introduce the notion of player, which will be extensively used
in the rest of the paper. The player of a communication $\beta$ is the
participant who is active in $\beta$.

\begin{definition}[Players of communications and traces]
    \mylabel{def:partAct}
  We denote by \play{\beta} the
  {\em set of players of communication  %\oi\
    $\beta$}
defined by\medskip

\centerline{
 $\play{\CommAs{\pp}{\M}{\q}}=\set{\pp} \qquad
    \play{\CommAsI{\pp}{\M}{\q}}=\set{\q}$ }\medskip

The function $\sf play$ is extended to traces in the obvious way:

\medskip
\centerline{$\play\ee=\emptyset\qquad\play{\concat\beta\comseqA}=
 \play\beta\cup\play\comseqA$}
 \end{definition}

\noindent
 Notice that the notion of player is characteristic of asynchronous
 communications, where only one of the involved participants is
 active, namely the sender for an output communication and the
 receiver for an input communication.  Instead, in synchronous
 communications both participants (also called roles in the
 literature) are active.

\section{Asynchronous types}
\mylabel{sec:typesAs}

In this section we introduce our new types for asynchronous
communication. The underlying idea is quite simple: to split the
communication constructor of standard global types into an output
constructor and an input constructor.  This will allow us to type
networks in which all participants make all their outputs before their
inputs, like the networks of
Examples~\ref{sync-async-characteristic-example} and \ref{htp}, whose
asynchronous types will be presented in~\refToExample{ta}.
\eject

\begin{definition}[Global and \qts]
\mylabel{def:GlobalTypesAs}\hbox{}\hfill\vspace*{-5mm}
\begin{enumerate}
\item {\em Global types}  are defined by:\sm

\centerline{
$\begin{array}{lll}
      \G~~& \coDefGr &
    \agtO{\pp}{\q}i I{\la}{\G}
    ~\mid ~\agtI \pp\q \la \G
     ~\mid ~\End
  \end{array}
$
}
\noindent
where $I$ is non-empty and  $\pp\neq \q$ and $\la_h\not=\la_k$ for all $h,k\in I$,
$h\neq k$. %,  i.e. labels in choices are all different.
\item {\em \Qts} are pairs made of a
  global type and a queue, written $ \GP \parG\Msg$.
 \end{enumerate}
\end{definition}
\noindent
As for processes, $ \coDefGr$ indicates that global types are
coinductively defined \emph{regular} terms.  The global type
$\agtO{\pp}{\q}i I{\la}{\G}$ specifies that $\pp$ sends a label
$\la_k$ with $k\in I$ to $\q$ and then the interaction described by
the global type $\G_k$ takes place. Dually, the global type $\agtI
\pp\q \la \G $ specifies that $\q$ receives label $\la$ from $\pp$ and
then the interaction described by the global type $\G$ takes place. We
will omit trailing $\End$'s.

Global types can be naturally seen as trees where internal nodes are
decorated by $\pp\q!$ or $\pp\q?$, leaves by $\End$, and edges by
labels $\la$. The sequences of decorations of nodes and edges on the
path from the root to an edge of the tree are \traces, in the sense of
\refToDef{tra}.  We denote by $\FPaths{\G}$ the set of such \traces\
in the tree of $\G$.  By definition, $\FPaths{\End} = \emptyset$ and
each trace in $\FPaths{\G}$ is non-empty.

\medskip
In~\refToDef{def:partAct} we introduced the notion of player for
communications and traces. It is useful to extend this notion to global types,
by defining the set of {\em players of  a \sgt\ $\GP$},  $\play{\GP}$,
to be the union of the sets of players of all its traces, namely \sm

\centerline{$\play{\GP} = \bigcup_{\comseqA \in
    \FPaths{\G}} \play{\comseqA} $}\medskip

\noindent The regularity assumption ensures that the set of
players of a \sgt\  is finite.

\medskip
Asynchronous types will be used to type networks, see
\refToFigure{fig:typingAs}.  A standard guarantee they should ensure
is that each participant whose behaviour is not terminated can do some
action.  Moreover, since communications are split into outputs and
inputs in global types, we must make sure that each input is balanced
by an output in the type or by a message in the queue, and vice versa.
These requirements will be formulated as well-formedness
conditions on \agts.

The remainder of this section is divided in two subsections: the first
focusses on well-formedness of asynchronous types, and the second
presents the type system and shows that it enjoys the properties of
subject reduction and session fidelity and that moreover it ensures
progress.

\subsection{Well-formed asynchronous types}

We start by defining the projection of \sgts\ onto participants
(\refToFigure{fig:projAsP}).  We proceed by defining the boundedness
predicate for \sgts\ (\refToDef{def:depth-ila}) and the \balancing\
predicate for asynchronous types (\refToFigure{wfagtA}).

 \begin{figure}[!h]
 \vspace*{-2mm}
 \centerline{\small
 $
 \begin{array}{c}
  \proj\G{\pr} = \inact \text{ if }\pr\not\in  \play{\G}
%\participant\G
\\
\proj{(\agtO{\pp}{\q}i I{\la}{\G})}\pr=\begin{cases}
 \oupP\q{i}{I}{\M}{\proj{\G_i}\pp}    & \text{if }\pr=\pp, \\
 \proj{\G_1}\q     & \text{if }\pr=\q \text{ and }I=\set 1 %\eh{I}=1
 \\
  \procActs
\,\inpP\pp{i}{I}{\M}{\PP_i}     & \text{if }\pr=\q \text{ and }\eh{I}>1\text{ and }
\proj{\G_i}\q=
%\procActs
  \procActs    \,\Seq{\rcvL\pp{\M_i}}{\PP_i}, \\
  %\inact&\text{ if } \pr\not\in\participant{\agtO{\pp}{\q}i I{\la}{\G}}\\
 \proj{\RG_1}\pr  & \text{if } \pr\not\in\set{\pp,\q} \text{
   and } \pr\in   \play{\G_1}\\
   & \text{ and }
 \proj{\G_i}\pr=\proj{\G_1}\pr\text{ for all } i \in I
\end{cases}\\ \\
\proj{(\agtI \pp\q \la \G )}\pr=\begin{cases}
 \Seq{\rcvL{\pp}{\la}}{\proj{\G}\pr}   & \text{if }\pr=\q\\
   \proj{\G}\pr   &  \text{if }\pr\not=\q \text{ and }\pr\in  \play{\G}
\end{cases}% \\ \\
\end{array}
$
}
\caption{Projection of \sgts\  onto participants.
} \mylabel{fig:projAsP}\vspace*{-3mm}
\end{figure}

As mentioned earlier, the projection of \sgts\ on participants yields
processes. Its coinductive definition is given in
\refToFigure{fig:projAsP}, where we use $\vv\pi$ to denote any
sequence, possibly empty, of \oi\ actions separated by ``;''. We write
$\eh I$ for the cardinality of $I$.

\medskip
The projection of a \sgt\ on a participant which is not a player of
the type is the inactive process $\inact$.  In particular, the
projection
of $\End$ is $\inact$ on all participants.

\medskip

The projection of an output choice type on the sender yields an output
process sending one of its labels to the receiver and then acting
according to the projection of the
corresponding branch. \medskip\\
The projection of an output choice type on the receiver $\q$ has two
clauses: one for the case where the choice has a single branch, and
one for the case where the choice has more than one branch.  In the
first case, the projection is simply the projection of the
continuation of the single branch on $\q$. In the second case, the
projection is defined if the projection of the continuation of each
branch on $\q$ starts with the same sequence of actions $\vv\pi$,
followed by an input of the label sent by $\pp$ on that branch and
then by a possibly different process in each branch.  In fact,
participant $\q$ must receive the label chosen by participant $\pp$
before behaving differently in different branches. The projection on
$\q$ is then the initial sequence of actions $\vv\pi$ followed by an
external choice on the different sent labels.  The sequence
$\procActsS$ is allowed to contain another input of a (possibly equal)
label from $\pp$, for example:
\[
\centerline{
$\begin{array}{c}
\proj{(\CommAs\pp{\la_1}\q;\CommAs\pp\la\q;\CommAsI\pp\la\q;\CommAsI\pp{\la_1}\q;
\CommAsI\pp\la\q \ \GlSyB \ \CommAs\pp{\la_2}\q;\CommAs\pp{\la'}\q;\CommAsI\pp\la\q;\CommAsI\pp{\la_2}\q;
\CommAsI\pp{\la'}\q)}{\q} =  \\
\pp?\la; (\pp?\la_1 ;  \pp? \la + \pp?\la_2 ;  \pp? {\la'})
\end{array}
$ }
\]
In \refToExample{ex:projRcv} we will show why we need to
distinguish these two cases.

\medskip
The projection of an output choice type on the other participants is
defined only if it produces the same process
for all branches of the choice. \sm

The projection of an input type on the receiver is an input action
followed by the projection of the rest of the type. For the other
participants,
the projection is simply the projection of the rest of the type.\sm

Note that our projection adopts the strict requirement of~\cite{CHY08}
for participants not involved in a choice, namely it requires their
behaviours to be the same in all branches of a choice. A more
permissive projection (in line with~\cite{SY19}) for the same global
types is given in~\cite{DGD21a}. Our choice here is motivated by
simplicity, in order to focus on the event structure semantics.

To guarantee the property of progress, our types for networks must
ensure that each network player occurs in every computation, whether
finite or infinite. To this end, each type player should occur in
every path of the tree of the type. Now, projectability already
ensures that each player of a choice type occurs in all its
branches. This implies that if one branch of the choice gives rise to
an infinite path, either the player occurs at some finite depth in
this path, or this path crosses infinitely many branching points in
which the player occurs in all branches. In the latter case, since the
depth of the player increases when crossing each branching point,
there is no bound on the depth of the player over all paths of the
type. Hence, to ensure that all type players occur in all paths, it is
enough to require the existence of such bounds. This motivates the
following definitions of depth and boundedness.

We first define the {\em depth of a participant $\pp$ in a
  \sgt\ $\G$}, $\weight(\G,\pp)$, which uses the length function
$|~|$ of \refToDef{tra} as well as the functions $\play{\beta}$
and $\play{\comseqA}$
%defined previously.
of \refToDef{def:partAct} and the function $\play{\G}$ defined earlier
in this section.
%and $\play{\G}$
%, the function $\sf play$ given after
%\refToDef{def:partAct} and the new function $\ord$ given below.
Intuitively, $\weight(\G,\pp)$
is the limit, computed over all paths of the tree of $\G$, of the
depth of the first occurrence of the player $\pp$ in the path.

\begin{definition}[Depth]
Let the two functions $\weight (\comseqA,\pp)$ and $\weight(\G,\pp)$ be defined by:\medskip

\centerline{
$\begin{array}{l} \weight(\comseqA,\pp)=\begin{cases}
      n&\text{ if }\comseqA =
      \concat{\concat{\comseqA_1}{\asCom}}{\comseqA_2}\text { and
      }\cardin{\comseqA_1} = n-1\text { and }\pp \notin \play{\comseqA_1}\text { and }\pp \in \play{\asCom}\\
      0 & \text{otherwise }
\end{cases}\\
\weight(\G,\pp)=
   \sup  \{\weight(\comseqA,\pp)\ |\ \comseqA\in\IPaths{\G}\}
\end{array}$}
\end{definition}

\begin{definition} [Boundedness]\mylabel{def:depth-ila}
 We say that a \sgt\ $\GP$ is {\em bounded} if $\weight(\G',\pp)$ is
 finite for all subtrees $\G'$ of $\GP$ and for all players $\pp$.
 \end{definition}

 To show that $\GP$ is bounded  it is enough to check
 $\weight(\GP',\pp)$ for  all  subtrees  $\G'$ of $\GP$ and
 $\pp\in\play{\GP'}$,  since   for any
 other $\pp$ we have $\weight(\GP',\pp)=0$.

\sm
Note that the depth of a participant which is a player of $\G$
 does not necessarily decrease in the subtrees of $\G$. As a matter of fact,
 this depth can be finite in $\G$ but infinite in one of its subtrees, as
 shown by the following example.

\begin{example}\label{ex:inf}
Consider
 $\G= \agtOS \pr\q {\la}{\agtIS \pr\q \la;\G'} $  where\sm

 \centerline{ $\G'=
 {\agtOS \pp\q {\la_1}{\agtIS \pp\q \la_1;\agtSOS \pp\pr {\la_3} ;\agtIS \pp\pr{\la_3}}~\GlSyB~\agtOS \pp\q {\la_2}{\Seq{\agtIS \pp\q \la_2}{\G'}}}$}

 \sm Then we have:

 \centerline{
 $\weight(\G,\pr)=1\quad\quad \weight(\G,\pp)=  3
 \quad\quad\weight(\G,\q)=2$
 }
 whereas

 \centerline{
 $\weight(\G',\pr)=\infty\quad\quad \weight(\G',\pp)=1\quad\quad\weight(\G',\q)=2$
 }

 \medskip\noindent  since $\underbrace{\CommAs\pp{\la_2}\q\cdot\CommAsI\pp{\la_2}\q\cdots\CommAs\pp{\la_2}\q\cdot\CommAsI\pp{\la_2}\q}_n\cdot\CommAs\pp{\la_1}\q\cdot\CommAsI\pp{\la_1}\q\cdot\CommAs\pp{\la_3}\pr\cdot\CommAsI\pp{\la_3}\pr\in\IPaths{\G'}$  \ for all $n\geq 0$  and \linebreak
   $\sup\{4+2n\ |\ n\geq 0\} = \infty$.
 \end{example}

 However, the  finite   depth of a participant which is a player of $\G$ but not
 the player of its root communication decreases in the immediate
 subtrees of $\G$, as stated in the following lemma.

 \begin{lemma}\label{ddA}
 \begin{enumerate}
\item\label{ddA1} If $\G= \agtO{\pp}{\q}i I{\la}{\G}$ and $\pr\in\play{\G}$ and $\pr\not=\pp$, then $\weight(\G,\pr)>\weight(\G_i,\pr)$ for all $i\in I$.
\item\label{ddA2} If $\G= \agtI \pp\q \la {\G'}$ and  $\pr\in\play{\G}$ and $\pr\not=\q$, then $\weight(\G,\pr)>\weight(\G',\pr)$.
\end{enumerate}
\end{lemma}

%We show now that
The definition of projection given in \refToFigure{fig:projAsP} is
sound for bounded global types.
%The proof can be found in the Appendix.

\begin{lemma}\label{pb}
If $\GP$ is bounded, then $\proj\GP\pr$ is a partial function for all $\pr$.
\end{lemma}

\begin{figure}[!ht]\small
\vspace*{-2mm}
\begin{center}
    $\begin{array}{c}
            \cSinferrule{\derSI{\BB,\confAs{\agtO{\pp}{\q}i I{\la}{\G}}{\Msg}}\confAs{\G_i}{\addMsg\Msg{\mq\pp{\la_i}\q}}
              \text{ for all } i\in I\quad
              \text{if $\agtO{\pp}{\q}i I{\la}{\G}$ is cyclic then $\Msg=\emptyset$}}
              {\derSI\BB{\confAs{\agtO{\pp}{\q}i I{\la}{\G}}{\Msg}}}{~~\rulename{Out}}\\ \\

   \cSinferrule{\derSI{\BB,(\agtI \pp\q \la \G,\addMsg{\mq\pp{\la}\q}\Msg)}{\confAs{\G}{\Msg}}}{\derSI\BB{\confAs{\agtI \pp\q \la \G}{\addMsg{\mq\pp{\la}\q}\Msg}}}{~~\rulename{In}}    \qquad
        \cSinferrule{}{\derSI\BB{\confAs{\End}{\emptyset}}}{}~~{\rulename{\End}}
  \quad
     \end{array}$
    \end{center} \vspace*{-3mm}
  \caption{Balancing predicate.} \mylabel{wfagtA}
\end{figure}

% \bcomila
% Nella sintassi dei tipi (Def. 3.1), $\End$ \`e l'ultima produzione. Mi sembra che
% dovremmo seguire lo stesso ordine nella \refToFigure{wfagtA}.
% \ecomila

 To ensure the correspondence between outputs and inputs,
in~\refToFigure{wfagtA} we define
 the \emph{balancing} predicate $\derSI{}{}$ on asynchronous types, and we
say that $\asty{\G}{\Msg}$ is \emph{balanced} if
$\derSI{}{\asty{\G}{\Msg}}$.
The intuition is that every  initial  input should
come with a corresponding message in the queue (Rule $\rulename{In}$),
ensuring that the input can take place.  Then, each message in the
queue can be exchanged for a corresponding output that will prefix the
type (Rule $\rulename{Out}$): this output will then precede the
previously inserted input and thus ensure again that the input can
take place. In short,
%\oi\ matching
\balancing\ holds if the messages in the queue and the outputs in the
global type are matched by inputs in the global type and vice versa.
We say that a \sgt\ is {\em cyclic} if its tree contains itself as
proper subtree.  So the condition ``if the \sgt\ is cyclic then the
queue is empty'' in Rule $\rulename{Out}$ ensures that there is no
message left in the queue at the beginning of a new cycle, namely that
all messages put in the queue by cyclic \sgts\ have matching inputs in
the same cycle.  For instance we can apply Rule $\rulename{Out}$ to
the asynchronous type $\confAs{\G'}{\emptyset}$ of
\refToExample{ex:iwf}(\ref{ex:iwf3}), but not to the asynchronous type
$\confAs{\G}{\mq\pp{\la}\q}$ of
\refToExample{ex:iwf}(\ref{ex:iwf2}). Similarly, in
\refToExample{ex:iwf}(\ref{ex:iwf4}), Rule $\rulename{Out}$ can be
used for $\asty{\G_2}{\emptyset}\,$ but not
for $\confAs{\G_2}{\mq\pp{\la}\q}$. \medskip

The double line indicates that the rules are interpreted coinductively
\cite{pier02} (Chapter 21).  The condition in Rule $\rulename{Out}$
guarantees that we get only regular proof derivations, therefore the
judgement $\derPI{\BB}{\confAs{\GP}{\emptyset}}$ is decidable.  If we
derive $\derPI{\BB}{\confAs{\GP}{\emptyset}}$ we can ensure that in
$\GP\parG\emptyset$ all outputs are matched by corresponding inputs
and vice versa, see the Progress Theorem (\refToTheorem{pr}). The
progress property holds also for standard global
types~\cite{DY11,Coppo2016}.

\medskip
The next example illustrates the use of the \balancing\ predicate on
several asynchronous types.

\begin{example}\label{ex:iwf}
\begin{enumerate}
\item\label{ex:iwf1}  The  \agt\  $\confAs{\agtIS \q\pp {\la};\agtSOS \pp\q {\la'};\agtIS \pp\q
  {\la'}}{\mq\q{\la}\pp}$  is  \balanced,
 % \oi\ matching for the queue  $\mq\q{\la}\pp$
  as shown by the following derivation:

\centerline{\prooftree
\prooftree
\prooftree
\derSI\BB{\confAs{\End}{\emptyset}}
\justifies
\derSI\BB{\confAs{\agtIS \pp\q {\la'}}{\mq\pp{\la'}\q}}
\endprooftree
\justifies
\derSI\BB{\confAs{\agtSOS \pp\q {\la'};\agtIS \pp\q {\la'}}{\emptyset}}
\endprooftree
\justifies
\derSI\BB{\confAs{\agtIS \q\pp {\la};\agtSOS \pp\q {\la'};\agtIS \pp\q {\la'}}{\mq\q{\la}\pp}}
\endprooftree}\sm

\item\label{ex:iwf2}  Let $\G=\agtSOS \pp\q \la;\agtSOS \pp\q \la;\agtIS \pp\q \la;\G$.  Then   $\confAs{\G}{\emptyset}$
is not \balanced.
%\oi\ matching for the empty queue.
Indeed, we cannot complete the proof tree for
$\derSI\BB{\confAs{\G}{\emptyset}}$ because, since $\G$ is cyclic,
%since $\G$ is cyclic we cannot complete the proof tree for
% $\derSI\BB{\confAs{\G}{\emptyset}}$, so
we cannot apply Rule $\rulename{Out}$ to infer the premise
$\derSI\BB{\confAs{\G}{\mq\pp{\la}\q}}$ in the following
deduction:\sm

\centerline{\prooftree \prooftree \prooftree
  \derSI\BB{\confAs{\G}{\mq\pp{\la}\q}} \justifies
  \derSI\BB{\confAs{\agtIS \pp\q
      \la;\G}{\addMsg{\mq\pp{\la}\q}{\mq\pp{\la}\q}}}
\endprooftree
\justifies
\derSI\BB{\confAs{\agtSOS \pp\q \la;\agtIS \pp\q \la;\G}{\mq\pp{\la}\q}}
\endprooftree
\justifies
\derSI\BB{\confAs{\G}{\emptyset}}
\endprooftree}\sm

\item\label{ex:iwf3}  Let $\G'=(\agtOS \pp\q {\la_1}{\agtIS \pp\q
    \la_1;\G'}~\GlSyB~\agtOS \pp\q {\la_2}{\agtIS \pp\q
    \la_2})$. Then   $\confAs{\G'}{\emptyset}$
 is balanced,
%Then $\G'$ is \oi\ matching for the empty queue,
as we can see from the
infinite   (but regular)    proof tree that follows:

\centerline{
\prooftree
{
\prooftree
\vdots
\justifies
\prooftree
\derSI\BB{\confAs{\G'}{\emptyset}}
\justifies
\derSI{}{\confAs{{\agtIS \pp\q \la_1;\G'}}{\mq\pp{\la_1}\q}}
\endprooftree
\endprooftree
}
\quad
{
\prooftree
\derSI\BB{\confAs{\End}{\emptyset}}
\justifies
\derSI{}{\confAs{{\agtIS \pp\q \la_2}}{\mq\pp{\la_2}\q}}
\endprooftree
}
\justifies
\derSI\BB{\confAs{\G'}{\emptyset}}
\endprooftree}

\item\label{ex:iwf4} Let $\G_1=\agtSOS \pp\q \la;\agtSOS \pp\q
  \la;\agtIS \pp\q \la;\G_2$ and $\G_2=\agtSOS \pp\pr \la;\CommAsI
  \pp\la\pr;\G_2$.  Then   $\confAs{\G_1}{\emptyset}$
is not \balanced.
 % Then $\G_1$ is not \oi\ matching for the empty queue.
  Indeed, we cannot complete the proof tree for
  $\derSI\BB{\confAs{\G_1}{\emptyset}}$,  since $\G_2$ is cyclic, so we cannot
  apply Rule $\rulename{Out}$ to  infer
the premise   $\derSI\BB{\confAs{\G_2}{\mq\pp{\la}\q}}$  in the following
  deduction:\sm

  \centerline{\prooftree \prooftree \prooftree
    \derSI\BB{\confAs{\G_2}{\mq\pp{\la}\q}} \justifies \derSI\BB{\confAs{\agtIS \pp\q
      \la;\G_2}{\addMsg{\mq\pp{\la}\q}{\mq\pp{\la}\q}}}
\endprooftree
\justifies
\derSI\BB{\confAs{\agtSOS \pp\q \la;\agtIS \pp\q \la;\G_2}{\mq\pp{\la}\q}}
\endprooftree
\justifies
\derSI\BB{\confAs{\G_1}{\mq\pp{\la}\q}}
\endprooftree}\sm

Instead,
%$\G_2$ is \oi\ matching for the empty queue:\\
 $\asty{\G_2}{\emptyset}\,$ is balanced:

\centerline{\prooftree
\prooftree
\prooftree
\vdots
\justifies
\derSI\BB{\confAs{\G_2}{\emptyset}}
\endprooftree
\justifies
\derSI\BB{\confAs{\CommAsI \pp\la\pr;\G_2}{\mq\pp{\la}\pr}}
\endprooftree
\justifies
%\derSI\BB{\confAs{\G_2}{\mq\pp{\la}\q}}
\derSI\BB{\confAs{\G_2}{\emptyset}}
\endprooftree}
\end{enumerate}
\end{example}

\noindent  It is interesting to notice that balancing of asynchronous types does not imply
projectability of their global types. For example, the type $\G\parG\emptyset$ where
\[
\centerline{$\G=\CommAs
\pp{\la_1}\q;\CommAsI \pp{\la_1}\q;\CommAs \pr{\la_1}\q;\CommAsI
\pr{\la_1}\q ~\GlSyB~\CommAs \pp{\la_2}\q;\CommAsI
\pp{\la_2}\q;\CommAs \pr{\la_2}\q;\CommAsI \pr{\la_2}\q$}\]
 is balanced,
but $\G$ is not projectable on participant $\pr$  for any of the
projection definitions  in  the literature. Notably,
 type $\G$ prescribes that  $\pr$ should behave differently according
to the message exchanged between $\pp$ and $\q$, an unreasonable
requirement.

\vspace{1.6mm}
Projectability, boundedness and \balancing\
%\oi\ matching
are the three
properties that single out the \agts\ we want to use in our type system.
\begin{definition}[Well-formed   \agts]
\mylabel{def:agtA}
We say that the   \agt\
$\GP\parG\Msg$
is   {\em well formed}
if  it is balanced,
$\proj{\GP}{\pp}$ is defined for all $\pp$ and $\GP$ is bounded.
% and $\derPI{}{(\GP,\Msg)}$ is derivable.
 \end{definition}
 Clearly, it is sufficient to check that $\proj{\GP}{\pp}$ is defined
 for all $\pp\in\play{\GP}$, since %for any other $\pp$ we have
 $\proj{\GP}{\pp}=\inact$ otherwise.

\subsection{Type system}

In this subsection we present our type system.  For establishing its
expected properties, we then introduce an LTS for \agts\
(\refToFigure{ltsgtAs}) and show that well-formedness of \agts\ is
preserved by transitions (\refToLemma{prop:wfa}).

\begin{figure}[!h]\small
\vspace*{1.8mm}
 \centerline{
 $%\;\\
 \begin{array}{c}
 \cSinferrule{\PP_i\subt\Q_i \quad \text{for all }i\in I}{\oupTx\pp{i}{I}{\M}{\PP}\subt \oup\pp{i}{I}{\M}{\Q}}{\rulename{ $\subt$-out}}
 \quad
 \cSinferrule{\PP_i\subt\Q_i \quad \text{for all }i\in I}{\inp\pp{i}{I\cup J}{\M}{\PP}\subt \inp\pp{i}{I}{\M}{\Q}}{\rulename{ $\subt$-In}}\\
  \inact\subt\inact~\rulename{ $\subt$ -$\inact$}
\\ \\
\inferrule{%\proj\GP{\pp_i}=\PP'_i\qquad
\PP_i\subt  \proj\GP{\pp_i}  %\PP'_i
\quad \text{for all }i\in I\qquad\play\GP\subseteq \set{\pp_i\mid i\in I}}
{\derN{\confAs{\PiB_{i\in I}\pP{\pp_i}{\PP_i}}{\Msg}}\GP\parG\Msg} ~\rulename{Net}
\\
\end{array}
$}
\caption{Preorder on processes and network typing rule.} \mylabel{fig:typingAs}
\end{figure}
\noindent The typing rules are given in \refToFigure{fig:typingAs}. In the
unique typing rule for networks, Rule $\rulename{Net}$, we assume the
\agt\  to be well formed.

 \medskip
We first define a preorder on processes, $\PP\leq\Q$, %saying when a
meaning that {\em process $\PP$ can be used where we expect process
  $\Q$}.
%In particular,
More precisely, $\PP\leq\Q$ if either $\PP$ is equal to $\Q$,   or we
are in one of two situations:
%they are both output processes
either both $\PP$ and $\Q$ are output processes, sending the same labels to the same participant,
%receiving the same labels from the same participant
and after the send $\PP$ continues with a process that can be used
when we expect the corresponding one in $\Q$; or they are both input
processes receiving labels from the same participant, and $\PP$ may
receive more labels than $\Q$ (and thus have more behaviours) but
whenever it receives the same label as $\Q$ it continues with a
process that can be used when we expect the corresponding one in $\Q$.
The rules are interpreted coinductively, since the processes may
be infinite.
However, derivability is decidable, since
 processes have finitely many distinct subprocesses.

\medskip  Clearly, our preorder on processes plays the same role as the
 subtyping relation on local types in other  works.  In the original
 standard subtyping of~\cite{GH05} a better type has more outputs and
 less inputs, while in the subtyping of~\cite{DemangeonH11} a better
 type has less outputs and more inputs. The subtyping of~\cite{GH05}
 allows channel substitution, while the subtyping of
 \cite{DemangeonH11} allows process substitution, as observed
 in~\cite{Gay16}. This justifies our structural preorder on processes,
 which is akin to a restriction of the subtyping
 of~\cite{DemangeonH11}. The advantage of this restriction is a strong
 version of Session Fidelity, see \refToTheorem{sfA}. On the other
 hand, as shown in~\cite{BDLT21}, such a restriction does not
 change the class of networks that can be typed by standard global
 types (but may change the types assigned to them). The proof
 in~\cite{BDLT21} easily adapts to our asynchronous types.

A network $\Nt\parN\Msg$ is typed by the \agt\ $\GP\parG\Msg$ if
for every participant $\pp$ such that $\pP{\pp}{\PP} \in \Nt$ the
process $\PP$ behaves as specified by the projection of $\GP$ on
$\pp$.
In Rule $\rulename{Net}$, the condition
$\play{\GP}\subseteq\set{\pp_i\mid i\in I}$ ensures that all players
of $\GP$ appear in the network.  Moreover it permits additional
participants that do not appear in $\GP$, allowing the typing of
sessions containing $\pP{\pp}{\inact}$ for a fresh $\pp$ -- a
property required to guarantee invariance of types under structural
congruence of networks.

\begin{example}\label{ta}
  The network of \refToExample{sync-async-characteristic-example} can
  be typed by $\G\parN\emptyset$ for four possible choices of $\G$:
\[\centerline{$\begin{array}{lll}\agtSOS \pp\q {\la};\agtSOS \q\pp
      {\la'};\agtIS \pp\q {\la};\agtIS \q\pp {\la'}&\qquad&
      \agtSOS \pp\q {\la};\agtSOS \q\pp {\la'};\agtIS \q\pp {\la'};\agtIS \pp\q {\la}\\
      \agtSOS \q\pp {\la'};\agtSOS \pp\q {\la};\agtIS \pp\q
      {\la};\agtIS \q\pp {\la'}&& \agtSOS \q\pp {\la'};\agtSOS \pp\q
      {\la};\agtIS \q\pp {\la'};\agtIS \pp\q {\la}\end{array}$}
\]
\noindent since   each participant only needs to do the output before the
  input. Notice that this network cannot be typed with the standard
  global types of~\cite{CHY16}.

\medskip
  The network of \refToExample{htp} can be typed by $\G\parN\emptyset$
  where
\[  \centerline{$\begin{array}{lll}\G&=&\pp\q!h;(\q\pp!h;\pp\q?h;\q\pp?h;\G \boxplus\q\pp!t;\pp\q? h;\q\pp? t ) \boxplus\\
  &&\pp\q!t;(\q\pp!h;\pp\q? t ;\q\pp? h  \boxplus\q\pp!t;\pp\q?t;\q\pp?t;\G)\parG\emptyset \end{array}$
  }\]
  \noindent   Again, this   network cannot  be typed  with the standard global
 types of~\cite{CHY16}.
  On the other hand, both these networks  can be typed by adding
  the semantic subtyping as defined in~\cite{Mostrous2009} to the type
  system with the standard global types.    However,   the resulting system is less precise than ours, although
  it can   type more networks. As an example, consider the standard global type
\[
\centerline{
$\G' = \pp \rightarrow \q:h;(\q \rightarrow \pp:h;\G' \boxplus \q \rightarrow \pp:t) \boxplus
\pp \rightarrow \q:t;(\q \rightarrow \pp:h \boxplus \q \rightarrow \pp:t;\G')$}
\]
  \noindent and its projections $\proj{\G'}{\pp}=\PP'$  and $\proj{\G'}{\q}=\Q'$ defined by
\[   \centerline{$\begin{array}{lcl}
                 \PP'&=& \q! h;(\q?h; \PP'+\q?t)\oplus \q! t;(\q?h+\q?t; \PP')\\
                 \Q'&=&\pp? h;(\pp!h; \Q'\oplus\pp!t) + \pp? t;(\pp!h\oplus\pp!t; \Q')
                       \end{array}$}\]
 \noindent  Then  $\G'$ is a type for  the network   $\pP\pp{\PP}\parN\pP\q{\Q}\parN\emptyset$
                     of \refToExample{htp},  since $\Q$ can be
                     obtained from $\Q'$
                      using  asynchronous subtyping,
                     by  anticipating     the  outputs towards $\pp$.
                      Clearly, the global type $\G'$ also types   the network
                     $\pP\pp{\PP'}\parN\pP\q{\Q'}\parN\emptyset$. However,
                     this network does not follow the rules of the
                      game  given in \refToExample{htp}.
\end{example}

   \begin{figure}[!h] %%\small
   \vspace*{-2mm}
 \scalebox{0.85}{
 \centerline{$%\;\\
\begin{array}{c}
 \agtO{\pp}{\q}i I{\la}{\G}\parG\Msg
 \stackred{\CommAs\pp{\la_k}\q}\G_k\parG\addMsg\Msg{\mq\pp{\la_k}\q}\text{~~~~~~~~}
 ~~\text{where}  ~~~~k\in I{~~~\text{~~~~~~~~}\rulename{\AsOut}}
 \\ \\[-2pt]
\agtI \pp\q \la \G\parG\addMsg{\mq\pp{\la}\q}\Msg \stackred{\CommAsI\pp{\la}\q}\G\parG\Msg
 ~~~\rulename{\AsIn}
\\ \\[-2pt]
 \prooftree
 \G_i\parG\addMsg{\Msg}{\mq\pp{\la_i}\q}\stackred\asCom\G_i' \parG\addMsg{\Msg'}{\mq\pp{\la_i}\q}
\text{ for all }  i \in I \qquad\pp\not\in
\play{\asCom}
 \justifies
  \agtO{\pp}{\q}i I{\la}{\G}\parG\Msg \stackred \asCom\agtO{\pp}{\q}i I{\la}{\G'}\parG\Msg'
 \using ~~~\rulename{IComm-Out}
  \endprooftree\\ \\[-2pt]
\prooftree
 \G\parG\Msg\stackred\asCom\G'\parG\Msg'
\qquad\q\not\in
\play{\asCom}
 \justifies
 \agtI \pp\q \la \G\parG\addMsg{\mq\pp{\la}\q}{\Msg} \stackred\asCom\agtI \pp\q \la {\G'}\parG\addMsg{\mq\pp{\la}\q}{\Msg'}
 \using ~~~\rulename{IComm-In}
  \endprooftree\\ \\[-2pt]
\end{array}
$} }\vspace*{-4mm}
\caption{
LTS for    \agts.
}\mylabel{ltsgtAs}\vspace*{-4mm}
\end{figure}

Figure \ref{ltsgtAs} gives the LTS for \agts.  It shows that a
communication can be performed also under a choice or an input guard,
provided it is independent from it. Actually, we are interested
% in reducing only asynchronous types which are balanced,
only in the transitions of balanced asynchronous types, see
Figure~\ref{wfagtA}. This justifies the shape of the message queues in
Rules $\rulename{IComm-Out}$ and $\rulename{IComm-In}$. The conditions
$\pp\not\in \play{\asCom}$ and $\q\not\in \play{\asCom}$ in these
rules ensure that $\asCom$ does not depend on the enclosing
communications.

\medskip
We say that $\G\parN\Msg\stackred{\beta}\G'\parN\Msg'$ is a \emph{top
  transition} if it is derived using either Rule $\rulename{\AsOut}$
or Rule $\rulename{\AsIn}$.  Top transitions preserve well-formedness
of \agts:
\begin{lemma}\label{wftr}
If $\G\parN\Msg\stackred{\beta}\G'\parN\Msg'$ is a top transition and $\G\parG\Msg$ is well formed, then $\G'\parG\Msg'$ is well formed too.
\end{lemma}

The following lemma
% (proved in the Appendix)
detects the immediate transitions of an \agt\ from the projections of
its \sgt.

 \begin{lemma}\label{keysrA} Let $\G\parN\Msg$ be well formed.

 \vspace*{-3mm}
\begin{enumerate}
\item\label{keysrA1} If $\proj\G\pp=\oup\q{i}{I}{\M}{\PP}$, then
  $\G\parN\Msg\stackred{\CommAs\pp{\la_i}\q}\G_i\parN\addMsg\Msg{\mq\pp{\la_i}\q}$
  and $\proj{\G_i}\pp=\PP_i$ for all $i\in I$.
\item\label{keysrA2}
If $\proj\G\q=\inp\pp{i}{I}{\M}{\PP}$
and  $\Msg\equiv\addMsg{\mq\pp{\la}\q}{\Msg'}$   for some $\la$,
then  $I=\set{k}\,$ and $\la=\la_k$  and
$\G\parN\Msg\stackred{\CommAsI\pp{\la_k}\q}\G'\parN\Msg'$
and $\proj{\G'}\q=\PP_k$.
\end{enumerate}
\end{lemma}

The next lemma shows how to retrieve the projections of a \sgt\ from
the immediate transitions of the \qt\ obtained by combining the \sgt\
%putting the \sgt\ in parallel
with a compliant queue. It also shows that transitions of asynchronous
types preserve projectability of their global types. For this purpose,
one needs the two clauses in the projection of an output choice on the
receiver, see \refToFigure{fig:projAsP}, as illustrated by the
following example.

 \begin{example}\label{ex:projRcv}
   Let $\G=\Seq{\CommAs\pp{\la}\q}{\Seq{\CommAsI\pp{\la}\q}}{\G'}$,
   where $\G'=
   \Seq{\CommAs\q{\la_1}\pr}{\Seq{\CommAsI\q{\la_1}\pr}{\CommAsI\pp{\la}\q}}\,\GlSy\,
   \Seq{\CommAs\q{\la_2}\pr}{\Seq{\CommAsI\q{\la_2}\pr}{\CommAsI\pp{\la}\q}}
   $.  The \agt\ $\G\parG\mq{\pp}{\la}{\q}$ is well formed.  Assume we
   modify the definition of projection of an output choice on the
   receiver by removing its first clause and the restriction of the
   second to  $\eh{I}>1$.  Then $\proj{\G}{\q}$ is defined since
   $\proj{(\Seq{\CommAsI\pp{\la}\q}{\G'})}{\q}=\Seq{\rcvL{\pp}{\la}}{(\Seq{\sendL{\pr}{\la_1}}{\rcvL{\pp}{\la}}\,
     \oplus\, \Seq{\sendL{\pr}{\la_2}}{\rcvL{\pp}{\la}} )}$ has the
   required shape.  Applying Rule $\rulename{IComm-Out}$ we get
   $\G\parG\mq{\pp}{\la}{\q}\stackred{\CommAsI\pp{\la}\q}\Seq{\CommAs\pp{\la}\q}{\G'}\parG\emptyset$.
   The projection $\proj{(\Seq{\CommAs\pp{\la}\q}{\G'})}{\q}$ would
   not be defined since $\proj{{\G'}}{\q}=
   \Seq{\sendL{\pr}{\la_1}}{\rcvL{\pp}{\la}}\, \oplus\,
   \Seq{\sendL{\pr}{\la_2}}{\rcvL{\pp}{\la}}$ does not have the
   required shape.
%is not the shape required to define the projection $\proj{\Seq{\CommAs\pp{\la}\q}{\G_1}}{\q}$.
\end{example}

\begin{lemma}\label{keysrA34} Let $\G\parN\Msg$ be well formed.

\vspace*{-2mm}
\begin{enumerate}
\item\label{keysrA3} If $\G\parN\Msg\stackred{\CommAs\pp{\la}\q}\G'\parN\Msg'$,  then
  $\Msg'\equiv\addMsg{\Msg}{\mq\pp\la\q}$ and
  $\proj\G\pp=\oup\q{i}{I}{\M}{\PP}$
 and $\M=\M_k$
  and $\proj{\G'}\pp=\PP_k$ for some $k\in I$
  and  $\proj{\G}\pr\subt\proj{\G'}\pr$ for all   $\pr\not=\pp$.
\item\label{keysrA4} If $\G\parN\Msg\stackred{\CommAsI\pp{\la}\q}\G'\parN\Msg'$,  then  $\Msg\equiv\addMsg{\mq\pp{\la}\q}\Msg'$ and $\proj\G\q=\agtI{\pp}{\q}{\la}{\proj{\G'}\q}$
  and  $\proj\G\pr\subt\proj{\G'}\pr$ for all   $\pr\not=\q$.
  \end{enumerate}
  \end{lemma}

The LTS also preserves \balancing\ of asynchronous types.
%preserves well-formedness if also \balancing\ is maintained.

\begin{lemma}\mylabel{srgA}
If  $\derPI{}{\confAs{\G}{\Msg}}$ and $\G\parG\Msg\stackred{\beta}\G'\parG\Msg'$, then $\derPI{}{\confAs{\G'}{\Msg'}}$.
 \end{lemma}

 We are now able to show that transitions preserve well-formedness of \agts.

\begin{lemma}\label{prop:wfa}
  If $\G\parG\Msg$
  is a well formed      \agt\   and
  $\G\parG\Msg\stackred{\beta}\G'\parG\Msg'$, then $\G'\parG\Msg'$ is
  a well formed      \agt\  too.
\end{lemma}

\begin{proof}
Let $\beta=\CommAs\pp{\la}\q$. By \refToLemma{keysrA34}(\ref{keysrA3}) we have that
%$\proj\G\pp=\oup\q{i}{I}{\M}{\PP}$ and
 $\proj{\G'}{\pr}$ is defined for all $\pr\in\play\G$.
Similarly  for
$\beta=\CommAsI\pp{\la}\q$, using  Lemma~\ref{keysrA34}(\ref{keysrA4}).
%and~\ref{keysrA}(\ref{keysrA2}).
The proof that $\weight(\G'',\pr)$    is finite for all $\pr$ and   $\G''$
subtree of $\G'$ is easy by induction on the
  transition rules of \refToFigure{ltsgtAs}. \small
  \\
Finally, from \refToLemma{srgA} we have
  that $\confAs{\G'}{\Msg'}$ is \balanced.
%  \oi\ matching for the queue $\Msg'$.
\end{proof}

By virtue of  this lemma,
%\refToLemma{prop:wfa},
{\em we will henceforth only consider well-formed \agts}.

We end this section with the expected results of Subject Reduction,
Session Fidelity~\cite{CHY08,CHY16} and
Progress~\cite{DY11,Coppo2016}, which rely as usual on Inversion and
Canonical Form lemmas.

\begin{lemma}[Inversion]\mylabel{lemma:InvAsync}
If $\derN{\Nt\parN\Msg}{\G\parG\Msg}$, then $\PP\subt\proj\G{\pp}$ for all  $\pP{\pp}{\PP}\in\Nt$.
\end{lemma}

\begin{lemma}[Canonical Form]\mylabel{lemma:CanAsync}
If  $\derN{\Nt\parN\Msg}{\G\parG\Msg}$ and
$\pp\in\play\G$, then $\pP{\pp}{\PP}\in\Nt$ and $\PP\subt\proj\G{\pp}$.
\end{lemma}

\begin{theorem}[Subject Reduction]\mylabel{srA}~
\text{~}\\[-13pt]
%\centerline{
If $\derN{\Nt\parN\Msg}{\G\parG\Msg}$ and
$\Nt\parN\Msg\stackred{\beta}\Nt'\parN\Msg'$, then
%}
%\centerline{
$\G\parG\Msg\stackred\beta\G'\parG\Msg'$ and
  \mbox{$\derN{\Nt'\parN\Msg'}{\G'\parG\Msg'}$.} %\vspace*{-6mm}
\end{theorem}

\begin{proof}
  Let $\beta=\CommAs\pp{\la}\q$.  By Rule $\rulename{Send}$ of
  \refToFigure{fig:asynprocLTS},  $\pP{\pp}{\oup\q{i}{I}{\la}{\PP}}\in\Nt$
  and
  $\pP{\pp}{\PP_k}\in\Nt'$ and
  $\Msg'=\addMsg{\Msg}{\mq\pp{\la_k}\q}$ and $\la=\la_k$ for some
  $k\in I$.  Moreover $\pP\pr\R\in\Nt$ iff $\pP\pr\R\in\Nt'$ for all $\pr\not=\pp$. From \refToLemma{lemma:InvAsync} we get
\begin{enumerate}
\item \label{psra1} $\oup\q{i}{I}{\la}{\PP}\subt\proj\G{\pp}$, which
  implies $\proj\G{\pp}=\oup\q{i}{I}{\la}{\PP'}$ with
  $\PP_i\subt\PP'_i$ for all $i\in I $ from Rule $\rulename{ $\subt$
    -out}$ of \refToFigure{fig:typingAs} , and
\item \label{psra2}  $\R\subt\proj\G{\pr}$ for all  $\pr\not=\pp$
  such that  $\pP{\pr}{\R}\in\Nt$.
  \end{enumerate}
By Lemma~\ref{keysrA}(\ref{keysrA1}) $\G\parG\Msg\stackred{\Comm\pp{\la_k}\q}\G_k\parG\addMsg\Msg{\mq\pp{\la_k}\q}$ and $\proj{\G_k}\pp=\PP_k'$,  which implies \mbox{$\PP_k\subt\proj{\G_k}\pp$.}  By Lemma~\ref{keysrA34}(\ref{keysrA3}) $\proj{\G}\pr\subt\proj{\G_k}\pr$  for all $\pr\not=\pp$.
By transitivity of $\subt$ we have $\R\subt\proj{\G_k}\pr$  for all $\pr\not=\pp$. We can then choose $\G'=\G_k$.\vspace*{1.8mm}
\\
Let $\beta=\CommAsI\pp{\la}\q$.  By Rule $\rulename{Rcv}$ of
  \refToFigure{fig:asynprocLTS},   $\pP{\q}{\inp\pp{j}{J}{\la}{\Q}}\in\Nt$
  and
  $\pP{\q}{\Q_k}\in\Nt'$ and
  $\Msg=\addMsg{\mq\pp{\la_k}\q}{\Msg'}$ and $\la=\la_k$ for some
  $k\in J$.  Moreover $\pP\pr\R\in\Nt$ iff $\pP\pr\R\in\Nt'$ for all $\pr\not=\q$. From \refToLemma{lemma:InvAsync} we get
  \begin{enumerate}
\item \label{psraI1}  $\inp\pp{j}{J}{\la}{\Q}\subt\proj\G{\q}$, which implies
$\proj\G{\q}=\inp\pp{j}{I}{\la}{\Q'}$   with $I\subseteq J$ and $\Q_i\subt\Q'_i$  for all $i\in I$
from Rule $\rulename{ $\subt$ -in}$ of \refToFigure{fig:typingAs}, and
\item \label{psraI2}  $\R\subt\proj\G{\pr}$ for all  $\pr\not=\q$
  such that  $\pP{\pr}{\R}\in\Nt$.
  \end{enumerate}
By Lemma~\ref{keysrA}(\ref{keysrA2}), since $\Msg=\addMsg{\mq\pp{\la_k}\q}{\Msg'}$, we get
 $\G\parG\Msg\stackred{\CommAsI\pp{\la_k}\q}\G_k\parG\Msg'$ and  $I=\set k$
and $\proj{\G_k}\q=\Q_k'$,
% for some $k\in I$ and $\la_k=\la$,
 which implies $\Q_k\subt\proj{\G_k}\pp$.
By Lemma~\ref{keysrA34}(\ref{keysrA4}) $\proj{\G}\pr\subt\proj{\G_k}\pr$  for all $\pr\not=\q$.
By transitivity of $\subt$ we have $\R\subt\proj{\G_k}\pr$  for all $\pr\not=\q$. We can then choose $\G'=\G_k$.
\end{proof}%%\vspace*{-7mm}

\begin{theorem}[Session Fidelity]\mylabel{sfA}
If $\derN{\Nt\parN\Msg}{\G\parG\Msg}$ and $\G\parG\Msg\stackred\beta\G'\parG\Msg'$, then

\centerline{$\Nt\parN\Msg\stackred\beta\Nt'\parN\Msg'$ and $\derN{\Nt'\parN\Msg'}{\G'\parG\Msg'}$}
\end{theorem}\vspace*{-2mm}

\begin{proof}
  Let $\beta=\CommAs\pp{\la}\q$. By
  \refToLemma{keysrA34}(\ref{keysrA3})
  $\Msg'\equiv\addMsg{\Msg}{\mq\pp\la\q}$,
  $\proj\G\pp=\oup\pp{i}{I}{\M}{\PP}$, $\M=\M_k$,
  $\proj{\G'}\pp=\PP_k$ for some $k\in I$ and
  $\proj{\G}\pr\subt\proj{\G'}\pr$ for all
  $\pr\not=\pp$. %By Lemma~\ref{keysrA}(\ref{keysrA1}) .
  From \refToLemma{lemma:CanAsync} we get $\Nt\equiv
  \pP{\pp}{\PP}\parN\Nt''$ and
\begin{enumerate}
\item \label{sfcA1} $\PP=\oup\q{i}{I}{\la}{\PP'}$ with $\PP'_i\subt\PP_i$ for all $i\in I$, from Rule $\rulename{ $\subt$ -out}$ of \refToFigure{fig:typingAs}, and
\item \label{sfcA2} $\R\subt\proj\G{\pr}$ for all   $\pP{\pr}{\R}\in\Nt''$.
\end{enumerate}
We can then choose  $\Nt'=\pP{\pp}{\PP'_k}\parN\Nt''$.

\medskip
Let $\beta=\CommAsI\pp{\la}\q$. By
\refToLemma{keysrA34}(\ref{keysrA4})
$\Msg\equiv\addMsg{\mq\pp\la\q}{\Msg'}$,
$\proj\G{ \q }=\Seq{\rcvL\pp\la}{\PP}$,
$\proj{\G'}{ \q }=\PP$ and $\proj{\G}\pr\subt\proj{\G'}\pr$ for
all $\pr\not=\q$.  From \refToLemma{lemma:CanAsync} we get $\Nt\equiv
\pP{\q}{\Q}\parN\Nt''$ and
 \begin{enumerate}
\item \label{sfcA1a} $\Q=\Seq{\rcvL\pp\la}{\PP'}+\Q'$ with $\PP'\subt\PP$, from Rule $\rulename{ $\subt$ -in}$ of \refToFigure{fig:typingAs}, and
\item \label{sfcA2a} $\R\subt\proj\G{\pr}$ for all   $\pP{\pr}{\R}\in\Nt''$.
\end{enumerate}
We can then choose  $\Nt'=\pP{\q}{\PP'}\parN\Nt''$.
\end{proof}

We are now able to prove that in a typable network, every participant
whose process is not terminated may eventually perform an action, and
every message that is stored in the queue is eventually read. This property is generally referred to as progress~\cite{H2016}.

%%%%%%%%%%% Variante Ilaria prova progresso %%%%%%%%%%%%%%

\begin{theorem}[Progress]\mylabel{pr}
 A typable  network $\Nt\parN\Msg$
satisfies progress, namely:\vspace*{-2mm}
\begin{enumerate}
\item\mylabel{pr1} $\pP{\pp}{\PP}\in\Nt$ implies $\Nt\parN\Msg\stackred{\concat{\comseqA}\beta}\Nt'\parN\Msg'$ with $\play\beta=\set{\pp}$;
\item\mylabel{pr2} $\Msg\equiv\addMsg{\mq\pp\la\q}{\Msg_1}$ implies $\Nt\parN\Msg\stackred{\concat{\comseqA}{\CommAsI\pp\la\q}}\Nt'\parN\Msg'$.
\end{enumerate}
\end{theorem}\vspace*{-2mm}

\begin{proof}
By hypothesis $\derN{\Nt\parN\Msg}{\G\parG\Msg}$ for some $\G$.\medskip

(\ref{pr1}) If $\PP$ is an output
process, then it can  immediately  move.
 Let then $\PP$ be an input process.  From $\pP{\pp}{\PP}\in\Nt$
we get
$\pp \in \play{\G}$ and therefore $\weight(\G,\pp)>0$. Moreover, since
 $\G$  is bounded, it must be $\weight(\G,\pp)<\infty$.
  We prove by induction on $\weight(\G,\pp)$ that
$0 <\weight(\G,\pp)<\infty$
implies $\G\parN\Msg\stackred{\concat{\comseqA}\beta}\G'\parN\Msg'$ with
$\play\beta=\set{\pp}$.
%This will imply
 By Session Fidelity (\refToTheorem{sfA}) it will follow that
$\Nt\parN\Msg\stackred{\concat{\comseqA}\beta}\Nt'\parN\Msg'$.
%by Session Fidelity  (\refToTheorem{sfA}).
 Let $d=\weight(\G,\pp)$.

 \medskip\noindent
{\em Case $d=1$}.  Here  $\G=\agtI \q\pp \la {\G'}$.  Since
$\G\parN\Msg$ is  balanced,
%well formed, this implies
$\Msg\equiv\addMsg{\mq\q\la {\pp}}{\Msg'}$ by Rule
$\rulename{In}$ of \refToFigure{wfagtA}. Then
$\G\parN\Msg\stackred{\CommAsI\q\la\pp}\G'\parN\Msg'$
by Rule $\rulename{\AsIn}$ of \refToFigure{ltsgtAs}.

\smallskip\noindent
{\em Case $d>1$}. Here we have either $\G =
\agtO{\pr}{\ps}{i}{I}{\la}{\G}$ with $\pr \neq \pp$ or $\G =
\agtI{\pr}{\ps}{\la}{\G'''}$ with $\ps \neq \pp$.  By \refToLemma{ddA}
this implies $\weight(\G_i,\pp)<d$ for all $i\in I$ in the first case,
and $\weight(\G''',\pp)<d$ in the second case.  Hence in both cases,
by applying Rule $\rulename{\AsOut}$ or Rule $\rulename{\AsIn}$ of
\refToFigure{ltsgtAs}, we get
$\G\parN\Msg\stackred{\beta'}\G''\parN\Msg''$.  Since either $\G''=\G_k$ for some $k\in I$ or $\G''=\G'''$ we get
$\play{\beta'}\neq\set{\pp}$
and  $\weight(\G'',\pp)<d$.   In case $\G$ is a choice of outputs
we get $\pp \in \play{\G''}$ by projectability of $\G$ if $\pp\neq\ps$ and by balancing of $\G\parN\Msg$ if $\pp=\ps$. Thus  $0<\weight(\G'',\pp)<
d < \infty$.
We may then apply induction to get $\G''\parN\Msg''\stackred{\concat{\comseqA}\beta}\G'\parN\Msg'$ with
$\play\beta=\set{\pp}$. Therefore
$\G\parN\Msg\stackred{\concat{\beta'}{\concat{\comseqA}\beta}}\G'\parN\Msg'$
is the required transition sequence.

 \medskip (\ref{pr2})
 Let the {\em input depth} of the input $\CommAsI\pp\la\q$ in $\G$,
notation $\idepth(\G,\CommAsI\pp\la\q)$, be inductively defined
by:\medskip

{\small{
  \centerline{$\begin{array}{lll}
\idepth(\agtO{\pr}{\ps}i I{\la}{\G},\CommAsI\pp\la\q)&=&1+
max _{i\in I}\set{\idepth(\G_i,\CommAsI\pp\la\q)}\\
\idepth(\agtI \pr\ps {\la'} {\G'},\CommAsI\pp\la\q)&=&\begin{cases}
 1     & \text{if }\CommAsI\pp\la\q=\CommAsI\pr{\la'}\ps\\
 \infty& \text{if $\pp=\pr$ and $\q=\ps$ and $\la\neq\la'$}\\
 1+   \idepth(\G',\CommAsI\pp\la\q)  & \text{otherwise}
\end{cases}\\
\idepth(\End,\CommAsI\pp\la\q)&=&\infty
\end{array}
$} } }

\medskip
\noindent By hypothesis $\Msg\equiv\addMsg{\mq\pp\la\q}{\Msg_1}$.

\smallskip
We show that, for all $\Msg$ and $\G$, $\derPI{}{\confAs{\G}{\addMsg{\mq\pp\la\q}{\Msg}}}$ implies that $\idepth(\G,\CommAsI\pp\la\q)$ is finite. The proof is by induction on $\maxP(\G)$ defined as the maximum length of a path from the root of $\G$ to either a leaf (if the path is finite) or the first cyclic global type encountered (if the path is infinite). Since $\G$ is regular every infinite path starting from its root must contain a  cyclic global type, so $\maxP(\G)=n$ for some $n$. \sm
\\
If $n=0$ either $\G=\End$ or $\G$ is cyclic. Therefore $\derPI{}{\confAs{\G}{\addMsg{\mq\pp\la\q}{\Msg}}}$ cannot hold
and the statement is true (Rules $\rulename{\End}$ and $\rulename{Out}$ of \refToFigure{wfagtA} require an empty queue). \sm
\\
If $n>0$, we consider the two possible shapes of $\G$. \sm
\\
Let $\G=\agtI \pr\ps {\la'} {\G'}$. If $\pr=\pp$, $\ps=\q$ and $\la'\neq\la$, then $\derPI{}{\confAs{\G}{\addMsg{\mq\pp\la\q}{\Msg}}}$ cannot hold and the statement is true (Rule $\rulename{In}$ requires a queue starting with $\mq\pp{\la'}\q$). If $\pr=\pp$, $\ps=\q$ and $\la'=\la$, then  $\idepth(\G,\CommAsI\pp\la\q)=1$. Otherwise $\maxP(\G')=n-1$ and, since $\derPI{}{\confAs{\G}{\addMsg{\mq\pp\la\q}{\Msg}}}$,
$\addMsg{\mq\pp\la\q}{\Msg}\equiv\addMsg{\mq\pr{\la'}\ps}{\addMsg{\mq\pp\la\q}{\Msg'}}$ for some $\Msg'$ (Rule $\rulename{In}$).
Moreover, $\derPI{}{\confAs{\G'}{\addMsg{\mq\pp\la\q}{\Msg'}}}$. By induction hypotheses we get that $\idepth(\G',\CommAsI\pp\la\q)$ is finite. Therefore $\idepth(\G,\CommAsI\pp\la\q)$ is~finite. \sm
\\
If $\G=\agtO{\pr}{\ps}i I{\la}{\G}$, then for all $i\in I$, $\maxP(\G_i)<n$ and $\derPI{}{\confAs{\G_i}{\addMsg{\addMsg{\mq\pp\la\q}{\Msg}}{\mq\pr{\la_i}\ps}}}$ (Rule $\rulename{Out}$). By induction hypotheses we get that $\idepth(\G_i,\CommAsI\pp\la\q)$
is finite for all $i\in I$. Therefore $\idepth(\G,\CommAsI\pp\la\q)$ is finite.\sm

 More precisely $\idepth(\G,\CommAsI\pp\la\q)$
is the number of rule applications between the rule which introduces
$\mq\pp\la\q$ and the conclusion in the derivation of
$\derPI{}{\confAs{\G}{\addMsg{\mq\pp\la\q}{\Msg_1}}}$.\sm

We prove by induction on $\idepth(\G,\pp)$ that
$\G\parN\Msg\stackred{\concat{\comseqA}{\CommAsI\pp\la\q}}\G'\parN\Msg'$.
By Session Fidelity  (\refToTheorem{sfA}) it will follow that
$\Nt\parN\Msg\stackred{\concat{\comseqA}{\CommAsI\pp\la\q}}\Nt'\parN\Msg'$.
Let $id=\idepth(\G,\pp)$.

\medskip\noindent
{\em Case $id=1$}.  Here $\G=\agtI \pp\q \la {\G'}$, which implies
$\G\parN\Msg\stackred{\CommAsI\pp\la\q}\G'\parN\Msg_1$  by Rule $\rulename{\AsIn}$ of \refToFigure{ltsgtAs}.

\smallskip\noindent
{\em Case $id>1$}.  As in the proof of Statement (\ref{pr1}),  by applying Rule $\rulename{\AsOut}$ or Rule
$\rulename{\AsIn}$ of \refToFigure{ltsgtAs} we get

\centerline{$\G\parN\Msg\stackred{\beta}\G''\parN\Msg''$}

\vspace*{1.5mm}\noindent where
$\beta \neq \CommAsI{\pp}{\la}{\q}$ and thus  $
\idepth(\G'',{\CommAsI\pp\la\q})<id$.

\medskip By induction $\G''\parN\Msg''\stackred{\concat{\comseqA}{\CommAsI\pp\la\q}}\G'\parN\Msg'$.
We  conclude that
$\G\parN\Msg\stackred{\concat{\beta}{\concat{\comseqA}{\CommAsI\pp\la\q}}}\G'\parN\Msg'$
 is the required transition sequence.
\end{proof}

The proof of \refToTheorem{pr} shows that the desired transition sequences use only Rules $\rulename{\AsOut}$ and  $\rulename{\AsIn}$ and the output choice is arbitrary. Moreover the lengths of these   transition sequences  are bounded by $\weight(\G,\pp)$ and $\idepth(\G,\CommAsI\pp\la\q)$, respectively.

\section{Event structures}\label{sec:eventStr}

We recall the definitions of \emph{Prime Event Structure} (PES)
from~\cite{NPW81} and \emph{Flow Event Structure} (FES)
from~\cite{BC88a}. The class of FESs is more general than that of
PESs, in that it allows for disjunctive causality and does not require
causality to be a transitive relation. As we shall see in
Sections~\ref{sec:process-ES} and~\ref{sec:netA-ES}, PESs are
sufficient to interpret processes, while we need FESs to interpret
networks. The reader is referred to~\cite{BC91} for a comparison of
the various classes of event structures.

This section is borrowed from~\cite{CDG-LNCS19} and therefore it is
also shared with its extended version~\cite{CDG22}.
\begin{definition}[Prime Event Structure] \mylabel{pes} A {\em prime event structure} {\rm (PES)} is a
    tuple
$S=(E,\leq, \gr)$ where:
\begin{enumerate}
\item \mylabel{pes1} $E$ is a denumerable set of events;
\item \mylabel{pes2}    $\leq\,\subseteq (E\times E)$ is a partial order relation,
called the \emph{causality} relation;
\item \mylabel{pes3}  $\gr\subseteq (E\times E)$ is an irreflexive symmetric relation, called the
\emph{conflict} relation, satisfying the property: $\forall e, e', e''\in E: e \grr e'\leq
e''\Rightarrow e \grr e''$ (\emph{conflict hereditariness}).
\end{enumerate}
\end{definition}
% We say that two events are \emph{concurrent} if they are neither causally related nor in conflict.

\begin{definition}[Flow Event Structure]\mylabel{fes} A {\em flow event structure} {\rm (FES)}
    is a  tuple $S=(E,\prec,\gr)$ where:
    \begin{enumerate}
 \item\mylabel{fes1} $E$ is a denumerable set of events;
\item\mylabel{fes2} $\prec\,\subseteq (E\times E)$ is an irreflexive relation,
called the \emph{flow} relation;
\item\mylabel{fes3} $\gr\subseteq (E\times E)$ is a symmetric relation, called the
\emph{conflict} relation.
\end{enumerate}
\end{definition}
Note that the flow relation is not required to be transitive, nor
acyclic (its reflexive and transitive closure is just a preorder, not
necessarily a partial order).  Intuitively, the flow relation
represents a possible {\sl direct causality} between two events.
Observe also that in a FES the conflict relation is not required to be
irreflexive nor hereditary; indeed, FESs may exhibit self-conflicting
events, as well as disjunctive causality (an event may have
conflicting causes).

Any PES $S = (E , \leq, \gr)$ may be regarded as a FES, with $\prec$
given by $<$ (the strict ordering) % $<$, or by
or by the covering relation of $\leq$.

\vspace{1.6mm}
We now recall the definition of {\sl configuration}\/ for event
structures. Intuitively, a configuration is a set of events having
occurred at some stage of the computation.   Thus, the semantics
of an event structure $S$ is given by its poset of configurations ordered
by set inclusion, where $\ESet_1 \subset \ESet_2$ means that $S$ may evolve
from $\ESet_1$ to $\ESet_2$.
\begin{definition}[PES configuration] \mylabel{configP}
  Let $S=(E,\leq, \gr)$ be a prime event structure. A
   {\em  configuration} of $S$ is a finite
% set $\ESet \subseteq E$
subset $\ESet$ of $E$
such that:
\begin{enumerate}
\itemsep=0.9pt
    \item \mylabel{configP1} $\ESet$ is downward-closed: \ $e'\leq e \in \ESet \, \ \impl \ \ e'\in \ESet$;
  \item \mylabel{configP2} $\ESet$ is conflict-free: $\forall e, e' \in \ESet, \neg (e \gr
e')$.
\end{enumerate}
\end{definition}

The definition of configuration for FESs is slightly more elaborated.
For a subset $\ESet$ of $E$, let $\prec_\ESet$ be the restriction of
the flow relation to $\ESet$ and $ \prec_\ESet^*$ be its transitive
and reflexive closure.

\begin{definition}[FES configuration]
\mylabel{configF}
Let $S=(E,\prec, \gr)$ be a flow event
structure. A {\em configuration} of $S$ is a finite
% set $\ESet \subseteq E$
subset $\ESet$ of $E$
such that:
\begin{enumerate}
\itemsep=0.9pt
\item\mylabel{configF1} $\ESet$ is downward-closed up to conflicts:
\ $e'\prec e \in \ESet, \ e'\notin \ESet\, \  \impl \
\,\exists \, e''\in \ESet.\,\, e'\gr \,e''\prec e$;
\item\mylabel{configF2} $\ESet$ is conflict-free: $\forall e, e' \in \ESet, \neg (e \gr
e')$;
\item\mylabel{configF3} $\ESet$ has no causality cycles: the relation $ \prec_\ESet^*$  is a partial order.
\end{enumerate}
\end{definition}
Condition (\ref{configF2}) is the same as for prime event
structures. Condition (\ref{configF1}) is adapted to account for the more general
-- non-hereditary -- conflict relation.
It states that any event appears in a configuration with a ``complete
set of causes''.  Condition (\ref{configF3}) ensures that any event in a
configuration is actually reachable at some stage of the computation.

If $S$ is a prime or flow event structure, we denote by $\Conf{S}$ its
set of configurations.  Then, the domain of
  configurations of $S$ is defined as follows:
\begin{definition}[ES configuration domain]
\mylabel{configDom}
Let $S$ be a prime or flow event
structure with set of configurations $\Conf{S}
 $. The \emph{domain of configurations} of $S$ is the partially ordered set
$\CD{S} \eqdef (\Conf{S}, \subseteq)$.
\end{definition}
We recall from~\cite{BC91} a useful characterisation for
configurations of FESs, which is based on the notion of proving
sequence, defined as follows:

\begin{definition}[Proving sequence]
  \mylabel{provseq} Given a flow event structure $S=(E,\prec, \gr)$, a
  \emph{proving sequence} in $S$ is a sequence $\Seq{e_1;
    \cdots}{e_n}$ of distinct non-conflicting events (i.e. $i\not=j\
  \impl\ e_i\not=e_j$ and $\neg(e_i\gr e_j)$ for all $i,j$) satisfying: \medskip

  \centerline{$\forall i\leq n\,\forall e \in E \,: \quad e\prec e_i\
    \ \impl\ \ \exists k<i\,. \ \ \text{ either } \ e = e_k\ \text{ or
    } \ e\grr e_k \prec e_i $}
\end{definition}
Note that any prefix of a proving sequence is itself a proving sequence.
\eject

We have the following characterisation of configurations of FESs
in terms of proving sequences.
\begin{proposition}[Representation of configurations as proving
  sequences~\cite{BC91}]
  \mylabel{provseqchar} Given a flow event structure $S=(E,\prec,
  \gr)$, a subset $\ESet$ of $E$ is a configuration of $S$ if and only
  if it can be enumerated as a proving sequence $\Seq{e_1;
    \cdots}{e_n}$.
\end{proposition}
Since PESs may be viewed as particular FESs, we may use
\refToDef{provseq} and \refToProp{provseqchar} both for the FESs
associated with networks (see Section~\ref{sec:netA-ES}) and for the
PESs associated with asynchronous global types (see
\refToSection{sec:eventsA}).  Note that for a PES
the condition of \refToDef{provseq} simplifies to

\smallskip
\centerline{$\forall i\leq n\,\forall e \in E \,: \quad e < e_i\ \
  \impl\ \ \exists k<i\,. \ \ e = e_k $}

\section{Event structure semantics of processes}\mylabel{sec:process-ES}

In this section, we present our ES semantics for processes and show
that the obtained ESs are PESs.  This semantics coincides with the one
given for processes of synchronous networks in our previous
work~\cite{CDG-LNCS19}. Indeed, a design choice of our ES semantics is
to implement asynchrony at the level of networks, and not at the level
of processes. Thus, processes are agnostic with respect to the
communication mode of the network, and sequentiality between their
actions is always interpreted as causality.

\medskip
A \emph{process event}, namely an event in the ES of a process, is an
occurrence of a send or receive action preceded by its \emph{causal
  history}, i.e., by the sequence of actions that caused that
occurrence in the process. Therefore, different occurrences of the
same send or receive action in the process give rise to different
process events in the associated ES.  Indeed, process events
correspond to \emph{paths} in the syntactic tree of a process.

Our ES semantics for processes will be the basis for defining the ES
semantics for networks in~\refToSection{sec:netA-ES},
%whose definition
which will reflect, as expected, the  asynchronous nature of
communication.
% distinction between synchronous and
% asynchronous communications.

We start by introducing process events, which
are non-empty sequences of  atomic actions $\pi$ as defined
at the beginning of \refToSection{sec:calculus}.

\begin{definition}[Process event]
  \mylabel{proceventP}
{\em Process events} ({\em p-events} for short) $\procev,
  \procev'$ are defined by: \sm

\centerline{$\procev \ \quad ::= \pi ~~ \mid~~
    \pastpref{\pi}{\procev}
      $}
\noindent We denote by $\Procev$ the set of p-events.
\end{definition}
\noindent
 Note the difference with the sequences $\vv\pi$ used in
\refToFigure{fig:projAsP}, where actions are separated by ``;''.

\smallskip
\mylabel{actseq} Let $\actseq$ denote a (possibly empty) sequence of
actions, and $\sqsubseteq$ denote the prefix ordering on such
sequences.  Each p-event $\procev$ may be written either in the form
$\procev = \concat {\pi}{\actseq}$ or in the form $\procev =
\concat{\actseq}{\pi}$. We shall feel free to use any of these forms.
When a p-event is written as $\procev = \concat{\actseq}{\pi} $, then
$\actseq$ may be viewed as the \emph{causal history} of $\procev$,
namely the sequence of actions that must have been executed by the
process for $\procev$ to be able to happen.

\medskip
We define the \emph{action} of a p-event to be
its last atomic action: \sm

 \centerline{$
\act{\concat{\actseq}{\pi}} = \pi$}

\noindent A p-event $\procev$ is an  {\em output p-event} if
$\act{\procev}$ is an output  %={\sendL{\q}{\la}}$,
and an {\em input p-event}
if $\act{\procev}$ is an input.  %={\rcvL{\q}{\la}}$.  \\

\begin{definition}[Causality and conflict relations on p-events]
\mylabel{procevent-relations}
The \emph{causality} relation $\leq$ and the \emph{conflict} relation
$\gr$ on the set of p-events $\Procev$ are defined by:
\begin{enumerate}
\item \mylabel{ila--esp2}
%the $\precP$ relation on $\Procev$ is  given by:\\
$\procev \sqsubseteq \procev' \ \impl \ \procev\precP\procev'$;
\item \mylabel{ila--esp3}
%the $\gr$ relation on $\Procev$ is given by: \\
  $\pi\neq\pi' \impl
  \concat{\concat{\actseq}{\pi}}{\actseq'}\,\grr\,\,\concat{\concat{\actseq}{\pi'}}{\actseq''}$.
\end{enumerate}
\end{definition}

\begin{definition} [Event structure of a process] \mylabel{esp}
The {\em event structure of process} $\PP$ is the triple \vspace*{1.8mm}

 \centerline{$\ESP{\PP} = (\ES(\PP), \precP_\PP , \gr_\PP)$}
 \noindent where:
\begin{enumerate}
\itemsep=0.95pt
\item \mylabel{ila-esp1} $\ES(\PP) \subseteq \Procev $ is the set of
  sequences of decorations along the nodes and edges of a
  path from the root to an edge in the tree of $\PP$;
\item \mylabel{ila-esp2}
$\precP_\PP$ is the restriction of
  $\precP$ to the set
$\ES(\PP)$;
\item \mylabel{ila-esp3}
$\gr_\PP$ is the restriction of
  $\grr$ to the set
$\ES(\PP)$.
\end{enumerate}
\end{definition}

In the following we shall feel free to drop the subscript in
$\precP_\PP$ and $\gr_\PP$.

%%%%%%%%%%%%%%%%%%%%%%%
Note that the set $\ES(\PP)$ may be denumerable, as shown by the following example.
\begin{example}\mylabel{ex:rec1}
If  $\PP=\sendL{\q}\la;\PP \oplus\sendL{\q}{\la'},$
then\vspace*{2.2mm}

\centerline{$ \ES(\PP)  = \begin{array}[t]{l}
\set{\underbrace{\sendL{\q}\la\cdot\ldots\cdot\sendL{\q}\la}_n \mid
  n\geq 1}  ~~~~ \cup ~~~~
\set{\underbrace{\sendL{\q}\la\cdot\ldots\cdot\sendL{\q}\la}_n\cdot\sendL{\q}{\la'}\mid
  n\geq 0}
\end{array} $}
\end{example}

We conclude this section by  stating  %showing
that the ESs of processes are
PESs. The proof is easy and may be found in~\cite{CDG22}.

\begin{proposition}\mylabel{basta10}
 Let $\PP$ be a process. Then $\ESP{\PP}$ is a prime event
  structure with an empty concurrency relation.
\end{proposition}

\section{Event structure semantics of networks}\mylabel{sec:netA-ES}

We present now the ES semantics of networks, which is grounded
on that of processes.

\medskip
Network events are simply process events located at some participant
of the network. As we are considering asynchronous communication,
matching output and input events are not paired together to yield a
single network event, but instead they are kept separate, with the
first representing the enqueuing of a message in the queue, and the
second the dequeuing of a message from the queue.  The event structure
of a network has to take into account the fact that the queue may
already contain some messages. In an asynchronous setting, output
events can always happen, provided that all events that causally
precede them have already been executed. Input events, on the other
side, must wait for the expected message to have been enqueued by the
sender.

This asymmetry is reflected in the definition of the {\em narrowing
  function} (\refToDef{nr}), which we use - as in our companion
paper~\cite{CDG22} dealing with synchronous communication - to
restrict the set of potential network events by discarding those that
are not causally well-founded. In the present case, narrowing will
keep all output events (whose predecessors have not been discarded),
while it will keep only those input events that are ``justified'' by
an output event (\refToDef{qgne}), or whose expected message in
already on the queue.  To check that an output event (send of
participant $\pp$ towards $\q$) justifies an input event (receive of
participant $\q$ from $\pp$), since messages sent by $\pp$ must be
read by $\q$ in the order they were sent, we look at the histories of
the two events to check their ``duality'' (Clause~(\ref{c1AbP1})
\refToDef{netaevent-relations}).  As discussed before \refToDef{pswd},
this notion of duality is more flexible than the standard one (which
matches a send in one participant with a receive in the other),
due to the use of a preorder that extends the standard duality check
by allowing output actions to be anticipated over input actions as
in~\cite{Mostrous2009}.

We start by defining the {\em o-trace of a queue} $\Msg$, notation
$\osq\Msg$, which is the sequence of output communications
corresponding to the messages in the queue. We use $\os$ to range over
o-traces.
\begin{definition}[o-trace]\label{qos}
The {\em o-\traceS corresponding to a queue} is defined by \sm
\\
\centerline{
$\osq{\emptyset }= \ee\qquad\osq{\addMsg{\mq\pp{\la}\q}{\Msg}}=
\concat{\CommAs\pp{\la}\q}{\osq{\Msg}}$}
\end{definition}

\noindent O-\tracesS are considered modulo the following equivalence
$\cong$, which mimics the structural equivalence on queues.\sm

\begin{definition}[o-trace equivalence $\cong$]\mylabel{def:permEqA2}
The {\em equivalence} $\cong$ on o-\tracesS is the least equivalence such that \sm

\centerline{$\concat{\concat{\concat{\os}{\CommAs\pp\la\q}}{\CommAs\pr{\la'}\ps}}{\os'}\,\cong\,\,
  \concat{\concat{\concat{\os}{\CommAs\pr{\la'}\ps}}{\CommAs\pp\la\q}}{\os'}
  ~~~\text{if}~~~ \pp\not=\pr$ or $\q\not=\ps$%and $\comseqA'\not=\ee$
}
\end{definition}

\noindent
Network events are p-events associated with a participant.\vspace*{-3mm}

\begin{definition}[Network event]
\label{proceventNA} \hbox{}\hfill\vspace*{-7mm}

\begin{enumerate}
\item\label{proceventNA1} {\em Network events} $\netevA,
\netevA'$, also called {\em n-events}, are p-events located at some participant $\pp$,
written $\locevA\pp\os\procev$.
\item\label{proceventNA2}  We define
$\io\netevA=\begin{cases}
\CommAs\pp{\la}\q      & \text{if
}\netevA=\locevA{\pp}\os{\concat\actseq{\sendL\q\M}}
\\
\CommAsI\pp{\M}\q      & \text{if }\netevA=\locevA{\q}\os{\concat{\actseq}{\rcvL\pp\M}}\\
\end{cases}$ \sm

\noindent
and we say that $\netevA$   is an {\em output n-event}
representing the communication $\CommAs\pp{\la}\q$ or an {\em input n-event} representing  the
communication $\CommAsI\pp{\M}\q$,
respectively.
\item We denote by $\DEA$ the set  of n-events.
\end{enumerate}
 \end{definition}

\noindent In order to define the flow relation between an
output n-event $\locevA{\pp}\os{\concat\actseq{\sendL\q\M}}$ and the
matching input n-event $\locevA{\q}\os{\concat{\actseq}{\rcvL\pp\M}}$,
we introduce a duality relation on projections of action sequences, see \refToDef{pswd}.
%which are defined as follows.
We first define the projection of
traces on participants, producing action sequences
(\refToDef{pep}(\ref{pep1})), and then the projection of action
sequences on participants, producing sequences of \emph{undirected
  actions} of the form  $!\la$ and $?\la$  (\refToDef{pep}(\ref{pep2})).

\medskip
In the sequel, we will use the symbol $\dagger$ to stand for
either $!$ or $?$. Then
$\dagL{\pp}{\la}$ will stand for either $\sendL{\pp}{\la}$ or
$\rcvL{\pp}{\la}$. Similarly, $\dagL{}{\,\la}$ will stand for either
$\sendL{}{\la}$ or $\rcvL{}{\la}$.

\eject
\hbox{}
\vspace*{-14mm}

\begin{definition}[Projections]\label{pep}  \hbox{}\hfill\vspace*{-7mm}

\begin{enumerate}
\item\label{pep1}
The {\em projection of a trace on a participant} is defined by:\sm

\centerline{
$
\projAP{\emptyseq}\pr= \emptyseq \qquad\qquad
\projAP{(\concat{\beta}{\comseqA})}{\pr} =
\begin{cases}
\concat{\sendL{\q}{\la}}{\projAP{\comseqA}{\pr}}   & \text{if  } \beta
= \CommAs{\pr}{\la}{\q}  \\ %\ \text{ and } \ \pr=\pp\\
\concat{\rcvL{\pp}{\la}}{\projAP{\comseqA}{\pr}}   & \text{if  } \beta
=  \CommAsI{\pp}{\la}{\pr}\\ %\ \text{ and } \ \pr=\q\\
  \projAP{\comseqA}{\pr}  & \text{otherwise}
\end{cases}
$}
\item\label{pep2}The {\em projection of an action sequence on a participant} is defined by:\sm

\centerline{
$
\projs{\emptyseq}\pr= \emptyseq \qquad\qquad
\projs{(\concat{\pi}{\actseq})}{\pr} =
\begin{cases}
\concat{\dagL{}{\,\la}}{\projs{\actseq}{\pr}}   & \text{if  } \pi =
\dagL{\pr}{\la} \\
%\text{ and } \pr=\pp, \\
  \projs{\actseq}{\pr}  & \text{otherwise}
\end{cases}
$}
\end{enumerate}
\end{definition}
We use $\preEv$ to range over sequences of output actions and
 $\bms$ to range over sequences of undirected actions.

\medskip
 We now introduce a variant of the
 % partial order relation $\precapprox$
 standard duality relation  on sequences of undirected
 actions.  This relation is meant to compare the sequences of actions
 of two participants communicating with each other,  in order to
 check that they match with each other.  In a synchronous setting,
 these sequences would be  required to be  in the standard
 symmetric duality relation $\Join$,  as  defined
 in Clause~(\ref{def:bowtie})  of \refToDef{pswd},  with an
 input  of  a label matching an output  of  the same label and
 viceversa. In our asynchronous setting,  the duality relation
 needs to reflect the fact, first observed in
 \cite{Mostrous2009}, that it is  ``better'' to anticipate
 outputs. To this %extent,
 end,
 in Clause~(\ref{def:precsim}) of
 \refToDef{pswd}  we define a partial order on sequences of
 actions
% (of the same participant) in which
 whereby a sequence
is better  than (i.e., less than or equal to) another one if it
anticipates outputs over inputs.
Finally,
to compare the sequences of actions of two participants communicating
with each other,
 in
Clause~(\ref{def:weak-duality}) of
\refToDef{pswd}  we define our weak duality relation
$_{\precapprox}\!\Join_{\,\succapprox}$,
a symmetric relation %preorder
% in Clause~(\ref{def:weak-duality}) allows to
that relates two sequences which are both less than or equal to
% equal or smaller (according to $\precapprox$) of
two sequences related by $\Join$.
In this way, we  allow outputs to be anticipated in both sequences.
%, we match with the duality relation $\Join$ all possible
%sequences ....

 \begin{definition} [Partial order and duality relations on undirected action sequences]
 \mylabel{pswd}
The three relations  $\Join$, $\precapprox$ and $_{\precapprox}\!\Join_{\,\succapprox}$
on undirected action sequences
are defined as follows:
\begin{enumerate}
\item \mylabel{def:bowtie}
 The  {\em duality}  relation $\Join$  on
undirected action sequences
is defined
   by: \sm

\centerline{$
\dualA{\ee}{\ee} \qquad\qquad\qquad
\dualA\bms{\bms'}\quad \impl\quad  \dualA{{\sendL{}{\la}} . \bms}
{ { \rcvL{}{\la}} . \bms'} \text {and }
\dualA{{\rcvL{}{\la}} . \bms}
{ { \sendL{}{\la}} . \bms'}   $}

 \item\mylabel{def:precsim}   The {\em partial order} relation  $\precapprox$  on undirected action sequences
 is defined
   as the smallest partial order such that: \sm

\centerline{$
     \pastpref{\pastpref\bms{\sendL{}{\la}}}{\pastpref{\rcvL{}{\la'}}{\bms'}}\precapprox
     \pastpref{\pastpref\bms{\rcvL{}{\la'}}}{\pastpref{\sendL{}{\la}}{\bms'}}
     $}
 \item    \mylabel{def:weak-duality}
   The  {\em weak duality}  relation $_{\precapprox}\!\Join_{\,\succapprox}$  on
undirected action sequences
is defined    by: \sm

\centerline{$
\dualprecsim{\bms_1}{\bms_2} \qquad\text{if} \qquad
\dualA{\bms'_1}{\bms'_2}$  \ for
  some $\bms'_1, \bms'_2$ such that $\bms_1 \precapprox
\bms'_1$ and $\bms_2 \precapprox \bms'_2$}
\end{enumerate}
\end{definition}
For example
$\pastpref{\sendL{}{\la_1}}{\pastpref{\sendL{}{\la_3}}{\pastpref{\rcvL{}{\la_2}}{\rcvL{}{\la_4}}}}
\precapprox\pastpref{\sendL{}{\la_1}}{\pastpref{\rcvL{}{\la_2}}{\pastpref{\sendL{}{\la_3}}{\rcvL{}{\la_4}}}}$
and
$\pastpref{\sendL{}{\la_2}}{\pastpref{\sendL{}{\la_4}}{\pastpref{\rcvL{}{\la_1}}{\rcvL{}{\la_3}}}}
\precapprox\pastpref{\rcvL{}{\la_1}}{\pastpref{\sendL{}{\la_2}}{\pastpref{\rcvL{}{\la_3}}{\sendL{}{\la_4}}}}$
and
$\pastpref{\sendL{}{\la_1}}{\pastpref{\rcvL{}{\la_2}}{\pastpref{\sendL{}{\la_3}}{\rcvL{}{\la_4}}}}\Join\pastpref{\rcvL{}{\la_1}}{\pastpref{\sendL{}{\la_2}}{\pastpref{\rcvL{}{\la_3}}{\sendL{}{\la_4}}}}$
imply
$\dualprecsim{\pastpref{\sendL{}{\la_1}}{\pastpref{\sendL{}{\la_3}}{\pastpref{\rcvL{}{\la_2}}{\rcvL{}{\la_4}}}}}{\pastpref{\sendL{}{\la_2}}{\pastpref{\sendL{}{\la_4}}{\pastpref{\rcvL{}{\la_1}}{\rcvL{}{\la_3}}}}}$.

\smallskip
We may now define the flow and conflict relations on n-events.  Notably the flow relation is parametrised on  an  o-trace representing the queue.

\begin{definition}[$\os$-flow and conflict relations on n-events]
\mylabel{netaevent-relations}
The \emph{$\os$-flow} relation $\precN$ and the \emph{conflict} relation
$\gr$ on the set of n-events $\DEA$ are defined by:
\begin{enumerate}
\item\mylabel{c1A}
\begin{enumerate}
\item\mylabel{c1Aa} $\procev  <  \procev' \impl  \locevA{\pp}\os{\procev}\precN \locevA{\pp}{\os}{\procev'}$;
\item\mylabel{c1AbP1}
$\dualprecsim{\projs{(\concat{\projAP{\os}{\pp}}{{\actseq}})}{\q}}{\projs{(\concat{\projAP{\os}{\q}}{{\actseq''}})}{\pp}}$
and $\projs{(\concat{\actseq'}{\rcvL\pp\la})} \pp\precapprox  \projs{(\concat{\concat{\actseq''}{\rcvL\pp\la}}{\preEv})}\pp$ for some $\actseq''$ and $\preEv$
$\impl \locevA\pp\os{\concat\actseq{\sendL\q\la}}\precN\locevA\q{\os}{\concat{\actseq'}{\rcvL\pp\la}}$;
\end{enumerate}
\item\mylabel{c2A}
  $ \procev \grr \procev' \impl  \locevA{\pp}\os{\procev}\grr \locevA{\pp}{\os}{\procev'}$.
\end{enumerate}
 \end{definition}

 \noindent
 Clause~(\ref{c1Aa}) defines flows within the same ``locality'' $\pp$,
 which we call \emph{local flows}, while Clause~(\ref{c1AbP1})
 defines flows between different localities, which
 we call \emph{cross-flows}: these are flows between an output
of $\pp$ towards $\q$
 %from $\pp$ to $\q$
and the corresponding input of $\q$ from $\pp$.
 The condition
%  hypotheses
 in Clause~(\ref{c1AbP1}) expresses a sort of
 %  ``soft  %``weak
  relaxed duality  between the
 history of the output and the history of the input: the intuition is
 that if $\q$ has some outputs towards $\pp$
 occurring in $\actseq'$, namely before its input $\rcvL\pp\la$,
 then when checking for duality these outputs
% (collected in $\preEv''$)
can be moved after $\rcvL\pp\la$, namely in $\preEv$,
because $\q$ does not need to wait until $\pp$
 has consumed these outputs to perform its input $\rcvL\pp\la$.  This
 condition can be seen at work in Examples~\ref{sync-async-characteristic-example-double} and ~\ref{asynchronous-network-with-choice}.

\medskip
 The reason for parametrising the flow relation with an o-trace $\os$
 is that the cross-flow relation depends on $\os$, which in the FES of
 a network $\Nt\parN\Msg$ will be  the image
 through  the mapping $\sf otr$ (see Definition \ref{qos})  of the queue
 $\Msg$.

\medskip
 For example, we have a cross-flow $\netevA \precN \netevA' $
 between the following n-events \sm

 \centerline{$\netevA = {\locevA\pp{\os}{\pastpref{\pastpref{\rcvL{\pr}{\la_4}}{\rcvL{\q}{\la_3}}}{\sendL{\q}{\la}}}}{}
  \, \precN \,
   {\locevA\q{\os}{\pastpref{\sendL{\pp}{\la'}}{\pastpref{\pastpref{\rcvL{\pp}{\la_1}}{\rcvL{\pp}{\la_2}}}{\rcvL{\pp}{\la}}}}}
   = \netevA' $}
   where \sm

 \centerline{$\os=\concat{\concat{\CommAs\pp{\la_1}\q}{\CommAs\pp{\la_2}\q}}{\concat{\CommAs\q{\la_5}\ps}{\CommAs\q{\la_3}\pp}}$ }

\sm since in this case
 $\actseq=\pastpref{\rcvL{\pr}{\la_4}}{\rcvL{\q}{\la_3}}$ and
$\actseq'=\pastpref{\sendL{\pp}{\la'}}{\pastpref{\rcvL{\pp}{\la_1}}{\rcvL{\pp}{\la_2}}}$,
%and taking
and thus,  taking $\actseq'' =
 \concat{\rcvL{\pp}{\la_1}}{\rcvL{\pp}{\la_2}}$ and  $\preEv=
 \sendL{\pp}{\la'}$, we  obtain \sm

 \centerline{
$\projs{(\concat{\pro\os\pp}{\actseq})}{\q} =
\; \, \concat{\concat{!\la_1}{!\la_2}}{?\la_3} \precapprox \;
\,
{\concat{?\la_3}{\concat{!\la_1}{!\la_2}}} \ \Join \ {\concat{!\la_3}{\concat{?\la_1}{?\la_2}}}
= \projs{(\concat{\pro\os\q}{\actseq''})}{\pp}$
}
and \sm

\centerline{
 $\projs{(\concat{\actseq'}{\rcvL\pp\la})} \pp = \; \,
 \concat{!\la'}{\concat{\concat{?\la_1}{?\la_2}}{?\la}}
\precapprox \; \,
\concat{\concat{\concat{?\la_1}{?\la_2}}{?\la}}{!\la'} =
 \projs{(\concat{\concat{\actseq''}{\rcvL\pp\la}}{\preEv})}\pp$
}

\medskip\noindent
 When $\netevA=\locevA\pp\os\te \precN \locevA\q\os{\te'}=\netevA'$ and $\pp
\neq \q$, then by definition $\rho$ is an output and $\rho'$ is an
input. In this case we say that the output $\rho$ \emph{$\os$-justifies} the
input $\rho'$, or symmetrically that the input $\rho'$ is
  $\os$-justified by the output $\rho$.
An input   n-event may also be justified by a message in the queue.
Both justifications are formalised by the following definition.\vspace*{-2mm}

\begin{definition}[Justifications of n-events]\label{qgne} \hbox{}\hfill\vspace*{-7mm}

\begin{enumerate}
\item\label{qgne1}The input   n-event $\netevA$ is {\em $\os$-justified} by
  the output   n-event $\netevA'$ if  $\netevA'\precN \netevA$
   and they are located at different participants.
\item\label{qgne2}
 The input  n-event
$\netevA=\locevA\q{(\concat{\concat\os{\CommAs\pp\la\q}}{\os'})}{\concat{\actseq}{\rcvL\pp\la}}$
is {\em $\os$-queue-justified} if there exists $\os'$ such that
$\concat{\os'}{\CommAs\pp\la\q}$ is a prefix of $\os$ (modulo $\cong$) and
$\locevA{\pp}{}{\concat{(\projAP{\os'}{\pp})}{\sendL\q\la}}\precNL{\epsilon} \netevA$.
\end{enumerate}
\end{definition}
\noindent
The condition
$\locevA{\pp}{}{\concat{(\projAP{\os'}{\pp})}{\sendL\q\la}}\precNL{\epsilon}
\netevA$ ensures that the inputs from $\pp$ in $\actseq$ will consume exactly the
messages from $\pp$ to $\q$ in the queue $\os'$.  For example, if
$\os=\concat{\CommAs\pp{\la}\q}{\CommAs\pp{\la}\q}$, then both
$\locevA\q {\os}{\pastpref{\sendL{\pp}{\la'}}{\rcvL{\pp}{\la}}}$ and
$\locevA\q
{\os}{\pastpref{\pastpref{\sendL{\pp}{\la'}}{\rcvL{\pp}{\la}}}{\rcvL{\pp}{\la}}}$
are $\os$-queue-justified.  On the other hand, if
$\os=\CommAs\pp{\la}\q$, then $\locevA\q {\os}{\pastpref{\sendL{\pp}{\la'}}{\rcvL{\pp}{\la}}}$ is
$\os$-queue-justified, but
$\locevA\q
{\os}{\pastpref{\pastpref{\sendL{\pp}{\la'}}{\rcvL{\pp}{\la}}}{\rcvL{\pp}{\la}}}$  is not
$\os$-queue-justified.

\medskip
To define the set of n-events associated with a network, we filter the
set of all its potential n-events by keeping only
\begin{itemize}
\item those n-events whose constituent p-events have all their
  predecessors appearing in some other n-event of the network and
\item those input n-events that are either queue-justified or justified by
output n-events of the network.
 \end{itemize}

\begin{definition}[Narrowing]\label{nr}
Given a set $E$ of n-events and an o-trace $\os$, we define the {\em narrowing} of
 $E$ with respect to $\os$ (notation $\nr {E,\os}$)
as the greatest fixpoint of the function $f_{E,\os}$ on sets of
n-events defined by: \sm

\centerline{$\begin{array}{lll} f_{E,\os}(X) &= &\{\netevA\in E \mid \netevA=
     \locevA{\pp}\os{\concat{\procev}{\pi}}
     ~\Rightarrow
     \locevA{\pp}{\os}\procev\in X\text { and }\\
    &&~~ (\,\netevA \text{ is an input n-event~$\Rightarrow$~
      $\netevA$ is either $\os$-queue-justified \qquad\quad}\\&&\hfill  \text{or $\os$-justified by some }\netevA'\in
     X \,) \}
    \end{array} $}
    \end{definition}

    \noindent
Thus, $\nr{E,\os}$ is the greatest set $X\subseteq E$
such that $X = f_{E,\os}(X)$.

\medskip
Note that we could not have taken $\nr{E,\os}$ to be the least
fixpoint of $f_{E,\os}$ rather than its greatest fixpoint.
Indeed, the least fixpoint of $f_{E,\os}$
would  be the empty set.

%The following property of narrowing is easy to verify and handy.
 It is easy to verify that  the  n-events  which
are discarded
%discharged
by the narrowing while their local predecessors are not  discarded
%discharded
must be input events. More precisely:

\begin{fact}\label{f}
 If $\netevA\in E$ and $\netevA\not\in \nr{E,\os}$ and either
$\netevA=\locevA\pp{}{\pi}$ or $\netevA=\locevA\pp{}{\procev\cdot\pi}$
with $\locevA\pp{}\procev\in \nr{E,\os}$,
then $\netevA$ is an input event.\vspace*{-1mm}
\end{fact}

We have now enough machinery to define the ES of networks.

\begin{definition}[Event structure of a network]
  \mylabel{netev-relationsA}
The {\em event structure of the
    network} $\Nt\parallel\Msg$
  is the triple: \sm

  \centerline{$\ESNA{\Nt\parallel\Msg} = (\GEA(\Nt\parallel\Msg),
    \precN_{\Nt\parallel\Msg} , \grr_{\Nt\parallel\Msg})$}

  \noindent where   $\os=\osq{\Msg}$ and

\begin{enumerate}
\item\mylabel{netev-relations1A} $\GEA(\Nt\parallel\Msg)
  =\nr{\DE(\Nt),\os}$,  where
  $\DE(\Nt)=\set{\locevA{\pp}{\os}{\procev}\sep \procev \in
    \ESA(\PP)\text{ with } \pP\pp\PP\in\Nt}$;
\item\mylabel{c1AA}$\precN_{\Nt\parallel\Msg}$ is the restriction of $\precN$ to the set $\GEA(\Nt\parallel\Msg)$;
\item\mylabel{c2AA} $\gr_{\Nt\parallel\Msg}$ is the restriction of $\gr$ to the set $\GEA(\Nt\parallel\Msg)$.
\end{enumerate}
\end{definition}

The following example shows how the operation of narrowing
prunes the set of potential n-events of a network ES. It also
illustrates the interplay between the two conditions in the definition of
narrowing.

\begin{example}\label{n}
 Consider the network $\Nt\parN\emptyset$,
  where   $\Nt=\pP\pp{\rcvL\q\la;\sendL\pr{\la'}}\parN\pP\pr{\rcvL\pp{\la'}}$.
The set of potential  n-events of $\ESNA{\Nt\parallel\emptyset}$ is
$\set{\locevAS\pp\ee{\rcvL\q\la},\locevAS\pp\ee{\rcvL\q\la;\sendL\pr{\la'}},\locevAS\pr\ee{\rcvL\pp{\la'}}}$. The
 n-event $\locevAS\pp\ee{\rcvL\q\la}$ is cancelled, since it is neither
 $\emptyset$-queue-justified nor $\emptyset$-justified  by another  n-event of  the ES.
Then $\locevAS\pp\ee{\rcvL\q\la;\sendL\pr{\la'}}$ is cancelled since it
lacks  its  predecessor $\locevAS\pp\ee{\rcvL\q\la}$. Lastly
$\locevAS\pr\ee{\rcvL\pp{\la'}}$ is cancelled, since it is neither
$\emptyset$-queue-justified nor $\emptyset$-justified by another n-event of  the ES.
%the network.
Notice that  $\locevAS\pp\ee{\rcvL\q\la;\sendL\pr{\la'}}$ would have $\emptyset$-justified
$\locevAS\pr\ee{\rcvL\pp{\la'}}$, if it had not been cancelled.
We conclude that $\GEA(\Nt\parallel\emptyset)=\emptyset$.
\end{example}

The following two examples illustrate the definitions given in this section.

\begin{example}
\mylabel{sync-async-characteristic-example-double}
Consider  the ES associated with the network  $\Nt \parallel \emptyset$,
with
\[
\centerline{$\Nt=\pP{\pp}{\sendL{\q}\la;\rcvL{\q}{\la'};\sendL{\q}\la;\rcvL{\q}{\la'}
  } \parN \pP{\q}{\sendL{\pp}{\la'};\rcvL{\pp}\la ;
    \sendL{\pp}{\la'};\rcvL{\pp}\la } $}
\]
The  n-events  of $\ESNA{\Nt\parallel\emptyset}$ are:\sm

 \centerline{$\begin{array}{lcl}
\netevA_1 = \locev{\pp}{\sendL{\q}\la} & \qquad & \netevA'_1 = \locev{\q}{\sendL{\pp}\la'}\\[-0.5pt]
 \netevA_2 = \locev{\pp}{\concat{\sendL{\q}\la}{\rcvL{\q}\la'}}
 & \qquad & \netevA'_2 = \locev{\q}{ \concat{\sendL{\pp}\la'}{\rcvL{\pp}\la}}\\[-0.5pt]
\netevA_3 = \locev{\pp}{
  \concat{\sendL{\q}\la}{\concat{\rcvL{\q}\la'}{\sendL{\q}\la}}}
&\qquad & \netevA'_3 = \locev{\q}{ \concat{\sendL{\pp}\la'}{\concat{\rcvL{\pp}\la}{\sendL{\pp}\la'}}}\\[-0.5pt]
\netevA_4 = \locev{\pp}{
  \concat{\sendL{\q}\la}{\concat{\rcvL{\q}\la'}{\concat{\sendL{\q}\la}{\rcvL{\q}\la'}}}}
& \qquad &  \netevA'_4 = \locev{\q}{
  \concat{\sendL{\pp}\la'}{\concat{\rcvL{\pp}\la}{\concat{\sendL{\pp}\la'}{\rcvL{\pp}\la}}}}
\end{array}
$}\sm

\noindent The $\ee$-flow relation is given by  the cross-flows  $ \netevA_1
\precE \netevA'_2 , \netevA_3 \precE \netevA'_4 , \netevA'_1 \precE
\netevA_2 , \netevA'_3 \precE \netevA_4 $,  as well as by the
local flows
%: (i.e., flows within the same location):
$ \netevA_i\precE \netevA_j$ and $ \netevA'_i \precE \netevA'_j\,$ for all $i,j$
such that $i \in \set{1,2,3}$, $j \in \set{2,3,4}$ and $i < j$.  The
conflict relation is empty.

\medskip
The configurations of $\ESNA{\Nt\parallel\emptyset}$ are:\sm

 \centerline{$\begin{array}{l}
 \set{\netevA_1}\quad \set{\netevA'_1}\quad \set{\netevA_1,\netevA'_1}
\quad \set{\netevA_1,\netevA'_1,\netevA_2}\quad
  \set{\netevA_1,\netevA'_1,\netevA'_2} \quad
 \set{\netevA_1,\netevA'_1,\netevA_2, \netevA'_2} \\[-0.5pt]
\set{\netevA_1,\netevA'_1,\netevA_2,\netevA_3} \quad \set{\netevA_1,\netevA'_1,\netevA'_2,\netevA'_3}
  \quad \set{\netevA_1,\netevA'_1,\netevA_2, \netevA'_2, \netevA_3} \quad
  \set{\netevA_1,\netevA'_1, \netevA_2, \netevA'_2,\netevA'_3} \\[-0.5pt]
  \set{\netevA_1,\netevA'_1, \netevA_2, \netevA'_2, \netevA_3, \netevA'_3} \quad
  \set{\netevA_1,\netevA'_1, \netevA_2, \netevA'_2, \netevA_3,
    \netevA'_3, \netevA_4} \\ \set{\netevA_1,\netevA'_1, \netevA_2, \netevA'_2, \netevA_3,
    \netevA'_3, \netevA'_4} \quad \set{\netevA_1,\netevA'_1, \netevA_2, \netevA'_2, \netevA_3,
    \netevA'_3, \netevA_4, \netevA'_4}
 \end{array}
$}\vspace{1.8mm}

\noindent
The network  $\Nt \parallel \emptyset$  can evolve in two steps to the network: \vspace{1.5mm}

\centerline{
  $\Nt' \parallel \Msg' =\pP{\pp}{\rcvL{\q}{\la'};\sendL{\q}\la;\rcvL{\q}{\la'} } \parN
  \pP{\q}{\rcvL{\pp}\la ; \sendL{\pp}{\la'};\rcvL{\pp}\la } \parallel \addMsg{\mq\pp{\la}\q}{\mq\q{\la'}\pp}$}

\vspace*{1.6mm} The   n-events of $\ESNA{\Nt'\parallel\Msg'}$ are: \vspace*{1.2mm}

\centerline{$\begin{array}{lcl}
 \netevA_5 = \locev{\pp}{{\rcvL{\q}\la'}}
 & \qquad & \netevA'_5 = \locev{\q}{{\rcvL{\pp}\la}}\\[-0.5pt]
\netevA_6 = \locev{\pp}{{\concat{\rcvL{\q}\la'}{\sendL{\q}\la}}}
&\qquad & \netevA'_6 = \locev{\q}{{\concat{\rcvL{\pp}\la}{\sendL{\pp}\la'}}}\\[-0.5pt]
\netevA_7 = \locev{\pp}{{\concat{\rcvL{\q}\la'}{\concat{\sendL{\q}\la}{\rcvL{\q}\la'}}}}
& \qquad &  \netevA'_7 = \locev{\q}{{\concat{\rcvL{\pp}\la}{\concat{\sendL{\pp}\la'}{\rcvL{\pp}\la}}}}
\end{array}
$} \sm

\noindent  Let  $\os=\concat{\CommAs\pp\la\q}{\CommAs\q{\la'}\pp}$.
 The $\os$-flow relation is given by the cross-flows $ \netevA_6 \precN
\netevA'_7$ , $\netevA'_6 \precN \netevA_7$,  and by the
local flows
$ \netevA_i \precN \netevA_j$ and $ \netevA'_i
\precN \netevA'_j\,$ for all $i,j$ such that $i \in \set{5,6}$, $j \in
\set{6,7}$ and $i < j$.  The  input   n-events $\netevA_5$ and $\netevA'_5$,
 which are the only ones without causes,  are $\os$-queue-justified.  The conflict relation is empty.

\medskip
The network $\Nt' \parallel \Msg'$ can evolve in five steps to the network: \sm
\[\centerline{
  $\Nt'' \parallel \Msg'' =
  \pP{\q}{\rcvL{\pp}\la } \parallel \mq\pp{\la}\q$}
  \]
The only  n-event  of $\ESNA{\Nt''\parallel\Msg''}$ is $\locev{\q}{\rcvL{\pp}\la}$.
\end{example}

%%%%%%%%%%%%%%%%%%%%%%%%%%%%%%%%%%%%%%%%%

\begin{example}
\mylabel{asynchronous-network-with-choice}
Let  $\Nt=   \pP{\pp}{\sendL{\q}\la_1;\sendL{\pr}{\la} \oplus \sendL{\q}\la_2;\sendL{\pr}{\la}
  } \parN   \pP{\q}{\rcvL{\pp}{\la_1} +     \rcvL{\pp}{\la_2}} \parN   \pP{\pr}{\rcvL{\pp}{\la}} $.
 The  n-events  of $\ESNA{\Nt\parallel\emptyset}$ are: \vspace*{-1.6mm}

 \centerline{$\begin{array}{lcl}
\netevA_1 = \locev{\pp}{ \sendL{\q}\la_1} & \qquad & \netevA'_1 = \locev{\q}{ \rcvL{\pp}\la_1}\\[-0.5pt]
 \netevA_2 = \locev{\pp}{ \sendL{\q}\la_2}
 & \qquad & \netevA'_2 = \locev{\q}{\rcvL{\pp}\la_2}\\[-0.5pt]
\netevA_3 =  \locev{\pp}{ \concat{ \sendL{\q}\la_1}{\sendL{\pr}\la}}
&\qquad & \netevA''_1 = \locev{\pr}{ \rcvL{\pp}\la}\\[-0.5pt]
\netevA_4 = \locev{\pp}{\concat{ \sendL{\q}\la_2}{\sendL{\pr}\la}}
& \qquad &
\end{array}
$}

\medskip\noindent The $\ee$-flow relation is given by the local flows
$ \netevA_1 \precE \netevA_3$, $\netevA_2 \precE \netevA_4$, and by
the cross-flows $ \netevA_1 \precE \netevA'_1$,
$\netevA_2 \precE \netevA'_2$, $ \netevA_3 \precE \netevA''_1$,
$ \netevA_4 \precE \netevA''_1$.  The conflict relation is given by
$ \netevA_1 \grr \netevA_2$, $ \netevA_1 \grr \netevA_4$,
$ \netevA_2 \grr \netevA_3$, $ \netevA_3 \grr \netevA_4$ and
$\netevA'_1 \grr \netevA'_2$. Notice that $\netevA_3$ and $\netevA_4$
are conflicting causes of $\netevA''_1$.
\refToFigure{fig:network-FES}(a) in Section~\ref{sec:resultsA}
illustrates this event structure.  The configurations are \sm

\centerline{$\begin{array}{l} \set{\netevA_1}\qquad
    \set{\netevA_1, \netevA_3}  \qquad \set{\netevA_1,\netevA'_1}
    \qquad \set{\netevA_1,\netevA_3,\netevA'_1}\qquad
 \set{\netevA_1,\netevA_3,\netevA''_1} \qquad \set{\netevA_1,\netevA_3,\netevA'_1,\netevA''_1}\\
    \set{\netevA_2} \qquad  \set{\netevA_2,\netevA_4}  \qquad \set{\netevA_2,\netevA'_2} \qquad
    \set{\netevA_2,\netevA_4, \netevA'_2}\qquad   \set{\netevA_2,\netevA_4,\netevA''_1}
    \qquad \set{\netevA_2,\netevA_4,\netevA'_2,\netevA''_1}
 \end{array}
 $}

 \medskip  The network $\Nt \parallel \Msg$ can evolve in one step to the network: \medskip

\centerline{ $\Nt' \parallel \Msg' = \pP{\pp}{\sendL{\pr}{\la}} \parN
  \pP{\q}{\rcvL{\pp}{\la_1} + \rcvL{\pp}{\la_2}} \parN
  \pP{\pr}{\rcvL{\pp}{\la}} \parN{\mq\pp{\la_1}\q}$}

\medskip\noindent
 The  n-events of $\ESNA{\Nt'\parallel\Msg'}$ are $
\netevA_5 = \locev{\pp}{\sendL{\pr}{\la}}$, $  \netevA'_3
 = \locev{\q}{ \rcvL{\pp}{\la_1}}$ and $  \netevA''_2
 = \locev{\pr}{ \rcvL{\pp}{\la}}$. Let
$\os=\CommAs\pp{\la_1}\q$.  The $\os$-flow relation is given by the
cross-flow $\netevA_5 \precN  \netevA''_2 $.   Notice
that  the input n-event $ \netevA'_3 $ is
$\os$-queue-justified,  and that there is no n-event  corresponding to the
branch $ \rcvL{\pp}{\la_2}$ of $\q$, since such an
n-event would not be $\os$-queue-justified. Hence the conflict relation is
empty.  The configurations are

\centerline{$\begin{array}{l} \set{\netevA_5}\qquad \set{
      \netevA'_3 }\qquad  \set{\netevA_5, \netevA'_3
    }  \qquad \set{\netevA_5, \netevA''_2
      }\qquad \set{\netevA_5, \netevA'_3 ,
      \netevA''_2 }
  \end{array}
$}
\end{example}
It is easy to show that the ESs of networks are FESs.
\begin{proposition}\mylabel{nfA}
Let $\Nt\parN\Msg$ be a network. Then $\ESNA{\Nt\parN\Msg}$ is a flow event
 structure.
\end{proposition}

\begin{proof}
 Let $\os=\osq{\Msg}$.  The relation $\precN$ is irreflexive since:
\begin{enumerate}
\item $ \procev  <
\procev'$ implies $\locevA{\pp}\os{\procev}\not= \locevA{\pp}{\os}{\procev'}$; \vspace*{-1mm}
\item $\pp\not=\q$ implies $\locevA\pp\os{\concat\actseq{\sendL\q\la}}\not=\locevA\q{\os}{\concat{\actseq'}{\rcvL\pp\la}}$.
\end{enumerate}
Symmetry   of the conflict relation between n-events follows from the corresponding
property of conflict between p-events.
\end{proof}

In the remainder of this section we show that projections of n-event
configurations give p-event configurations.  We start by formalising
the projection function of n-events to p-events and showing that it is
downward surjective.

\begin{definition}[Projection of n-events to p-events]
\label{def:proj-network-n-event}The {\em projection} function $\projnet{\pp}{\cdot}$ is defined by: \sm
\\
\centerline{$
\projnet{\pp}{\netevA}=\begin{cases}
\procev     & \text{if }\netevA=\locevA\pp\os\procev\\
 \textit{undefined}   & \text{otherwise}
\end{cases}
$}
The projection function $\projnet{\pp}{\cdot}$ is extended to sets of
n-events in the obvious way:

\vspace*{1.8mm}
\centerline{$\projnet{\pp}{X}  = \set{ \procev \mid \exists\netevA \in X~. ~ \projnet{\pp}{\netevA} = \procev} $}
\end{definition}

\begin{proposition}[Downward surjectivity of projection]
\label{prop:down-ontoA}
Let \sm

\centerline{$\pP{\pp}{\PP}\in\Nt$  and   $\ESNA{\Nt\parallel\Msg}=(\GEA(\Nt\parallel\Msg),
    \precN , \grr)$ and $\ESP{\PP} = (\ES(\PP), \precP_\PP , \gr_\PP)$}

\medskip\noindent
    Then the partial function $\projnetfun{\pp}:\GEA(\Nt\parallel\Msg)
\rightarrow \ES(\PP)$ is downward surjective.
\end{proposition}

\begin{proof}
  Follows immediately from the fact that $\GEA(\Nt\parallel\Msg)$ is
  the narrowing of a set of n-events $\locevA\pp\os\procev$ with
  $\os=\osq\Msg$ and $\pP{\pp}{\PP}\in\Nt$ and $\procev\in\ES(\PP)$.
\end{proof}

The operation of narrowing on network events makes sure that
each configuration of the ES of a network projects down to
configurations of the ESs of the component processes.

\begin{proposition}[Projection preserves configurations]
\mylabel{p:pc}
Let $\pP{\pp}{\PP}\in\Nt$.  If $\ESet \in \Conf{\ESN{\Nt\parallel\Msg}}$,
then $\projnet{\pp}{\ESet} \in \Conf{\ESP{\PP}}$.
\end{proposition}

\begin{proof}
Let $\ESet \in \Conf{\ESNA{\Nt\parallel\Msg}}$ and $\GSet = \projnet{\pp}{\ESet}$. We want to show that $\GSet \in
\Conf{\ESP{\PP}}$, namely that $\GSet$ satisfies Conditions (\ref{configP1})
and (\ref{configP2}) of \refToDef{configP}.
\begin{enumerate}
\item[(\ref{configP1})] \emph{Downward-closure.} Let $\procev
  \in \GSet$. Since $\GSet = \projnet{\pp}{\ESet}$, there
  exists $\netevA\in \ESet$ such that
  $\netevA=\locevA\pp\os\procev$. Suppose $\procev' <
  \procev$.
  %Because of downward surjectivity of $\projnet{\pp}{\cdot}$,
   From \refToProp{prop:down-ontoA}\
  there exists $\netevA'\in \GEA(\Nt\parallel\Msg)$ such that
  $\netevA'=\locevA\pp\os{\procev'}$.   Let $\os=\osq{\Msg}$. By \refToDef{netaevent-relations}(\ref{c1Aa})   we have then $\netevA' \precN \netevA$. Since $\ESet$ is
  left-closed up to conflicts, we know that either $\netevA' \in \ESet$
  or there exists $\netevA''\in \ESet$ such that
  $\netevA'' \grr \netevA'$ and $\netevA''\prec \netevA$. We examine the
  two cases in turn:
\begin{itemize}
\itemsep=0.9pt
\item $\netevA' \in \ESet$. Then, since $\procev'=
  \projnet{\pp}{\netevA'}$, we have $\procev'\in
  \projnet{\pp}{\ESet} =\GSet$ and we are done.
\item $\exists\netevA''\in \ESet \, . \ \netevA'' \grr \netevA'$ and
  $\netevA'' \prec \netevA$.  From $\netevA'' \grr \netevA'$ we get $\netevA''=\locevA\pp\os{\procev''}$ and $\procev'' \grr \procev'$. This implies $\procev'' \grr \procev$. By  \refToDef{netaevent-relations}(\ref{c2A})  this implies $\netevA \grr \netevA'$, contradicting
  the hypothesis that $\ESet$ is conflict-free. So this case is impossible.
  \end{itemize}
    \item[(\ref{configP2})] \emph{Conflict-freeness.} Ad absurdum, suppose
  there exist $\procev, \procev' \in \GSet$ such that $\procev \grr
  \procev'$. Then, since $\GSet = \projnet{\pp}{\ESet}$, there must
  exist $\netevA, \netevA' \in \ESet$ such that
  $\netevA=\locevA\pp\os\procev$ and
  $\netevA'=\locevA\pp\os{\procev'}$. By   \refToDef{netaevent-relations}(\ref{c2A})   this implies $\netevA \grr \netevA'$, contradicting
  the hypothesis that $\ESet$ is conflict-free.
   \end{enumerate}\vspace*{-4.6mm}
    \end{proof}
 Notice that there are configurations of $\Conf{\ESP{\PP}}$ which cannot be obtained by projecting configurations of $\Conf{\ESN{\Nt\parallel\Msg}}$ in spite of the condition $\pP\pp\PP\in\Nt$. A simple example is $\pP\pp{\q?\la}\parN\emptyset$.

\section{Event structure semantics of  asynchronous types}
\mylabel{sec:eventsA}

We define now the ES semantics of \agts, which is based on particular
traces.  In the ES of an \agt, as in the ES of a network, an event
represents a particular occurrence of an input or output
communication, preceded by its causal history.  However, while in the
ES of a network an event is a located process event, and the
causal history of an input or output action is its \emph{local
  history} within the participant % it pertains to,
where it is located,
in the ES of an
\agt~the causal history of a communication is its \emph{global
  history}, which may include communications from other
participants.  Recall that an asynchronous type is a global type
coupled with a queue.  Then,
% the intuition is that
the global history of a communication is obtained by taking
the trace labelling
the path that leads to that communication in the tree of the
global type, and removing from it all the communications that do not cause the
last one. %in this trace.
%its labelling trace.
% where all the communications that do not cause it have
% been removed.
A trace of this kind, with the property that all its communications
cause a subsequent communication in the trace, will be called
\emph{pointed}.  Moreover, since a communication may have concurrent
causes, whose order in the computation is irrelevant, such pointed
traces will be considered up to permutation of % adjacent
concurrent communications. So far, the treatment is very similar to
the one proposed for the synchronous case in~\cite{CDG22}.
However, asynchrony introduces additional subtleties in the definition
of causality and concurrency among communications. Indeed, two
communications will be causally related not only when they have the
same player, but also when one of them is an output $\pp\q!\la$ and
the other is the matching input $\pp\q?\la$
(\refToDef{def:matching}). Moreover, this matching relation is
affected by the presence of the queue (since an input can match either
a message in the queue or an output in its trace within the global
type), so it has to be computed relatively to a prefixing output trace
that represents the queue. This gives rise to a notion of
well-formedness for traces (\refToDef{def:WFnew}) that reflects the
balancing condition for asynchronous types. Similarly, the concurrency
relation between adjacent communications (\refToDef{swap}) and the
resulting permutation equivalence (\refToDef{def:permEqA}), as well as
the notion of pointedness (\refToDef{pcsA}), will have to be defined
relatively to a prefixing output trace.

To sum up, the events in the ES of an asynchronous type will be pairs
made of an output trace representing the queue of the type (taken up
to an equivalence that reflects the structural congruence on queues,
see~\refToDef{def:permEqA2}), and of a trace (taken up to
permutation equivalence, see~\refToDef{def:permEqA}) that is pointed
with respect to the queue output trace (\refToDef{def:pm}) and which
is obtained from a trace of the global type by removing only the
communications that are not causes of the last one (\refToDef{egA}).

\medskip
Although the events of an asynchronous type ES have a more complex
and indirect definition than the events of a network ES,
% which are simply process events assigned to a participant,
they have two importants benefits
with respect to the latter:
\begin{itemize}
\itemsep=0.9pt
\item %they enjoy a very simple definition of causality and conflict;
 the relations of causality and conflict are very simple to define on
 them (\refToDef{ageo});
\item they do not raise well-foundedness issues, since they
  are extracted from paths in the tree of the global type by removing
  only the unnecessary communications (\refToDef{egA}).
\end{itemize}

For \tracesS $\comseqA$, as given in \refToDef{tra}, we use\ the
following notational conventions:
\begin{itemize}
\itemsep=0.9pt
\item We denote by
  $\at{\comseqA}{i}$ the $i$-th
  element of $\comseqA$, $i > 0$.
\item If $i \leq j$, %
  we define $\range{\comseqA}{i}{j} = \at{\comseqA}{i} \cdots
  \at{\comseqA}{j}$ to be the sub\traceS of $\comseqA$ consisting of
  the $(j-i+1)$ elements starting from the $i$-th one and ending with
  the $j$-th one.  If $i> j$, we define $\range{\comseqA}{i}{j}$ to be
  the empty \traceS $\ee$.
\end{itemize}
If not otherwise stated we assume that $\comseqA$ has $n$ elements, so
$\comseqA=\range{\comseqA}{1}{n}$.

\medskip
In the \tracesS  appearing in events,  we want to require
that every input matches a corresponding output. This is checked using
the \emph{multiplicity of $\pp\q\dagger$ in $\comseqA$},
defined by induction as follows: \vspace*{2mm}

\label{dagger-shorthand}
\centerline{$ \begin{array}{l}
\mult{\pp\q\dagger}{\ee}  = 0 \qquad \qquad
\mult{\pp\q\dagger}{\concat{\beta}{\comseqA}}  =
\begin{cases}
\mult{\pp\q\dagger}{\comseqA} + 1 & \text {if } \beta=\pp\q\!\dagger\!\la \\
\mult{\pp\q\dagger}{\comseqA}  & \text {otherwise }
\end{cases}
\end{array}
$}
\noindent where $\dagger\in\set{!,?}$ (as in \refToDef{pep}).
\eject

 An input of $\q$ from $\pp$ matches a
preceding output from $\pp$ to $\q$ in a \traceS if it has the same
label $\la$ and the number of inputs from $\pp$ to $\q$
in the subtrace before the given input is equal to the number of outputs from $\pp$ to
$\q$ in the subtrace before the given output.
This is formalised using the above multiplicity  notion  and the  positions
of communications in \traces.

\begin{definition} [Matching]\mylabel{def:j}\mylabel{def:matching}
The input
  $\at{\comseqA}{j}=\CommAsI\pp\la\q$ {\em matches} the output
  $\at{\comseqA}{i}=\CommAs\pp\la\q$ in $\comseqA$, dubbed $\ct i
  \comseqA j$, if  $ i < j $ and
 $\mult{\pp\q!}{\linebreak\range{\comseqA}{1}{i-1}} =
  \mult{\pp\q?}{\range{\comseqA}{1}{j-1}}$.
\end{definition}
For example, if
$\comseqA=\CommAs\pp\la\q;\CommAs\pp\la\q;\CommAs\pp\la\q;\CommAsI\pp\la\q;\CommAsI\pp\la\q$,
then $\ct 1 \comseqA {\, 4}$ and $\ct 2\comseqA {\,5}$, while no input matches
the output at position $3$, denoted by $\neg(\ct 3 \comseqA {\,4})$ and $\neg(\ct 3  \comseqA {\, 5})$.
 Similarly, if
$\comseqA=\CommAs\pp\la\q;\CommAs\pp{\la'}\q;\CommAsI\pp{\la''}\q;\CommAsI\pp{\la'}
\q$, then $\ct 2 \comseqA {\,4}$, while no output is matched by the input
at position $3$.

\medskip
As mentioned earlier, o-\tracesS will be used to represent queues and
general \tracesS are paths in \sgt\ trees.  We want to define an
equivalence relation on general traces, which allows us to exchange
the order of adjacent communications when this order is not
essential. This is the case if the communications have different
players and in addition they are not matching according to
\refToDef{def:matching}.  However, the matching relation must also
take into account the fact that some outputs are already on the queue.
So we will consider well-formedness with respect to a prefixing
o-\trace. We proceed as follows:
\begin{itemize}
\itemsep=0.85pt
\item we start with well-formed \tracesS (\refToDef{def:WFnew});
\item we define the swapping relation $\swap{}{}\os$ which
 allows two communications to be interchanged in a trace
$\comseqA$,
when these communications are independent in
the trace $\concat{\os}{\comseqA}$
%under given conditions
  (\refToDef{swap});
\item then we show that $\swap{}{}\os$ preserves $\os$-well-formedness
(\refToLemma{prop:congSim0});
\item finally we define the equivalence $\approx_\os$  on
  $\os$-well-formed traces
  (\refToDef{def:permEqA}).
\end{itemize}

In a well-formed \traceS each input must have a corresponding output.
A matching input/output pair corresponds to a communication in the
standard global types of \cite{CHY16}. So, if we find an input at some
position in a trace, the corresponding output must already occur at
some earlier position in the trace.  We also introduce a notion of
well-formedness with respect to a prefixing \trace, where the prefix
represents communications that have already occurred. \vspace*{-2.5mm}

\begin{definition} [Well-formedness]\mylabel{def:WFnew}\hbox{}\hfill\vspace*{-7mm}

\begin{enumerate}
\item \mylabel{def:WFnew11}
A \traceS $\comseqA$  is {\em well formed} if  every input
 matches an output in $\comseqA$. \vspace*{-1mm}
\item \mylabel{def:WFnew12} A \traceS $\comseqA$ is
$\comseqA'$-{\em well formed} if $\concat{\comseqA'}{\comseqA}$ is
well formed. \vspace*{-2mm}
    \end{enumerate}
\end{definition}
As an example, the \traceS
$\comseqA=\concat{\concat{\CommAs\pp\la\q}{\CommAs\pp{\la'}\q}}{\CommAsI\pp{\la'}\q}$
is not well formed since  the input $\CommAsI\pp{\la'}\q$ in the
third position  does not match
% is not matched by
the output $\CommAs\pp{\la'}\q$ in
the second position,  i.e.,  $\neg(\ct 2 \comseqA {\,3})$.
On the
other hand, $\comseqA$ is $\CommAs\pp{\la'}\q$-well formed, since
$\concat{\CommAs\pp{\la'}\q}{\comseqA}$ is well formed given that the
input $\CommAsI\pp{\la'}\q$ in the fourth position
% is matched by
 matches  the
output $\CommAs\pp{\la'}\q$ in the first position,  i.e.,
 $\ct 1 {\concat{\CommAs\pp{\la'}\q}\comseqA} 4$.\vspace{1.8mm}

 Notice that any o-\traceS is well formed and any well-formed \traceS
 of length 1 must be an output. A well-formed \traceS of length 2 can
 consist of either two outputs or an output followed by the matching
 input.
%
 % \bmc non capisco la frase seguente \emc \bcomila Ho provato a
 % svilupparla, se non vi piace la togliamo del tutto. \ecomila
 Note also that, if in \refToDef{def:WFnew}(\ref{def:WFnew12})
 the trace $\tau'$ is an o-trace, then it may be viewed as
 representing a queue, and therefore  in this case the
 matching between an input in $\tau$ and an output in $\tau'$
 is akin to the balancing condition of Rule~\rulename{In} in
 \refToFigure{wfagtA} for asynchronous types.

\begin{definition}[Swapping]\label{swap}
Let $\comseqA$ be $\os$-well formed. We say that  $\comseqA$ {\em $\os$-swaps to} % and let
$\comseqA'$, %be any \traceS,
notation $\swap{\comseqA}{\comseqA'}{\os}$,  if \sm

\centerline{
$
\begin{array}{l}
\comseqA=\concat{\concat{\concat{\range{\comseqA}{ 1}{i-1}}{\beta}}{\beta'}}{\comseqA''}\qquad
\comseqA'=\concat{\concat{\concat{\range{\comseqA}{ 1 }{i-1}}{\beta'}}{\beta}}{\comseqA''}\quad\text{and}\\
\play{\beta}\cap\play{\beta'} = \emptyset
%\play{\beta}\not=\play{\beta'}
\quad\text{and} \quad\neg(\ct {i+ \cardin{\os}}{\os\cdot\comseqA} {\,i+1+ \cardin{\os}})
\end{array}
$
 }
\end{definition}
\noindent
For instance, if $\os = \CommAs\pp\la\q$ and
$\comseqA=\concat{\CommAs\pp\la\q}{\CommAsI\pp{\la}\q}$, then
$\comseqA$ $\os$-swaps to $\comseqA'
=\concat{\CommAsI\pp{\la}\q}{\CommAs\pp\la\q}$ because the input in
$\comseqA$ matches the output in $\os$  and therefore it is independent from the output in  $\comseqA$.
\begin{lemma}
  \mylabel{prop:congSim0} If $\comseqA$ is $\os$-well formed and
  $\swap{\comseqA}{\comseqA'}{\os}$, then $\comseqA'$ is $\os$-well
  formed too.
\end{lemma}

\begin{proof}
  Let $\comseqA =\concat{\concat{\concat{\range{\comseqA}{ 1}{ i-1}}{\beta}}{\beta'}}{\comseqA_1}$
  and $\comseqA' =
  \concat{\concat{\concat{\range{\comseqA}{ 1 }{ i-1}}{\beta'}}{\beta}}{\comseqA_1}$.
  We want to prove that $\concat{\os}{\comseqA'}$ is well formed. To
  this end, we will show that if $\beta$ or $\beta'$ is an input, then
  it matches an output that occurs in the prefix
  $%\range{(\concat{\os}{\comseqA'})}{ 1 }{i + \cardin{\os}} =
  \range{(\concat{\os}{\comseqA})}{ 1 }{i-1 + \cardin{\os}}$ of
  $\concat{\os}{\comseqA'}$. Note that it must be $\beta \neq\beta'$,
  since by hypothesis $\play{\beta}\cap\play{\beta'} =
  \emptyset$.  \sm

  Suppose $\beta'$ is an input. Since $\comseqA$ is $\os$-well formed,
  $\beta'$ matches an output in
  $\concat{\range{(\concat{\os}{\comseqA})}{ 1 }{i-1 +
      \cardin{\os}}}{\beta}$.  This output cannot be $\beta$, since
by hypothesis  $\neg(\ct {i+ \cardin{\os}}{\os\cdot\comseqA} {\,i+1+
  \cardin{\os}})$.
Hence $\beta'$ matches an output which occurs in the prefix
$\range{(\concat{\os}{\comseqA})}{ 1 }{i-1 + \cardin{\os}}$ of
$\concat{\os}{\comseqA'}$. \sm
%, as required.

Suppose now  $\beta=\CommAsI\pp\la\q$ and
$\mult{\pp\q?}{\range{(\concat{\os}{\comseqA})}{ 1 }{i-1 + \cardin{\os}}} =
m$.  Since $\comseqA$ is $\os$-well formed, $\beta$ matches an output
$\at{(\concat{\os}{\comseqA})}{j} = \CommAs\pp\la\q$ in the prefix
$\range{(\concat{\os}{\comseqA})}{ 1 }{i-1 + \cardin{\os}}$ of
$\concat{\os}{\comseqA}$.
%Let $ \at{(\concat{\os}{\comseqA})}{j}$ be this output.
Then \mbox{$1\leq j < i + \cardin{\os}$}
and $\mult{\pp\q!}{\range{(\concat{\os}{\comseqA})}{ 1 }{j-1 +
    \cardin{\os}}} = m$.
Since $\beta \neq \beta'$,  we get  also
\mbox{$\mult{\pp\q?}{\range{(\concat{\os}{\comseqA'})}{ 1 }{i+
    \cardin{\os}}} = m$.}   Then $\beta$ matches
$\at{(\concat{\os}{\comseqA})}{j}$ also in $\concat{\os}{\comseqA'}$.
\end{proof}
From the previous lemma  and the observation that, if $\comseqA$
is $\os$-well formed and $\comseqA'$ is obtained by swapping the
$i$-th and $(i+1)$-th element of $\comseqA$,
then
$\play{\at{\comseqA'}{i}}\cap\play{\at{\comseqA'}{i+1}}=\emptyset$ and\linebreak
$\neg(\ct {i + \cardin{\os} }{\os\cdot\comseqA'} {i+1 + \cardin{\os}
})$,
we deduce that the swapping relation is symmetric. This allows us
to define $\approx_{\os}$ as the equivalence relation induced by  the swapping relation.
%$\swap{}{}{\os}$.

\begin{definition}[Equivalence  $\approx_{\os}$ on $\os$-well-formed traces]\mylabel{def:permEqA}
  The   {\em equivalence}  $\approx_{\os}$ on
  $\os$-well-formed \tracesS is the reflexive and transitive closure
  of $\swap{}{}{\os}$.%
\end{definition}
Observe that for o-\tracesS all the equivalences
$\approx_{\os}$ collapse to $\approx_{\ee}$ and $\approx_{\ee}
\,\subset\, \cong$, where $\cong$ is  the  o-trace  equivalence
 given in~\refToDef{def:permEqA2}. Indeed, it should be clear
%is clear
 that $\approx_{\ee} \,\subseteq\, \cong$. To show $\approx_{\ee}
 \,\neq \,\cong$, consider  $\os =
 \concat{\CommAs\pp{\la}\q}{\CommAs\pp{\la'}\pr}$ and $\os' =
 \concat{\CommAs\pp{\la'}\pr}{\CommAs\pp{\la}\q}$. Then $\os\cong
 \os'$ but $ \os \not\approx_{\ee} \os'$.   This agrees
 with the fact that o-traces represent messages in queues, while
 %traces of inputs and outputs
general traces
represent future communication actions.

\medskip
Another constraint that we want to impose on \tracesS in order to
build events is that each communication must be a cause of at least
one of those that follow it. This happens when:
\begin{itemize}
\item either the two communications have the same player,
%and in this case
 in which case  we say that the first communication is required  in the
  trace  (\refToDef{def:required});\sm
\item or  the first communication is an output and the second is the
  matching input.
\end{itemize}
We call pointedness the property of a \traceS in which each
communication, except the last one, satisfies one of the two
conditions above.  Like well-formedness, also pointedness is
parameterised on traces.

\medskip
We first define required communications.

\begin{definition}[Required communication]
\mylabel{def:required}
%Let $\comseqA =\range{\comseqA}{ 1 }{n}$.
We say that
$\at{\comseqA}{i}$ is {\em required in $\comseqA$},
notation $\oks i  \comseqA$,
  if  $\play{\at{\comseqA}{i}} \subseteq
  \play{\range{\comseqA}{(i+1)}{n}}$, where $n=\cardin\comseqA$.
\end{definition}\noindent
Note that by definition the last element $\at{\comseqA}{n}$ is
not required in $\comseqA$.

\begin{definition}[Pointedness]\mylabel{pcsA}
\mylabel{pcsA1}
The \traceS $\comseqA$ is {\em $\comseqA'$-pointed} if $\comseqA$
is $\comseqA'$-well formed and for all $i$,
$1\leq i<n$,  one of the following holds:
\begin{enumerate}
\item\mylabel{pcsA11} either $\oks i {\comseqA}$\vspace*{-1mm}
 \item\mylabel{pcsA12} or $\ct {i +
     \cardin{\comseqA'}}{\concat{\comseqA'}\comseqA}{j +
     \cardin{\comseqA'}}$  for some $j>i$.
  \end{enumerate} \vspace*{-4mm}
 \end{definition}

\noindent Observe that the two conditions of the above definition are reminiscent
of the two kinds of causality - local flow and cross-flow - discussed for
network events in \refToSection{sec:netA-ES} (\refToDef{netaevent-relations}).  Indeed, Condition (\ref{pcsA11})
holds if $\at{\comseqA}{i}$ is a local cause of some
$\at{\comseqA}{j}$, $j > i$, while Condition (\ref{pcsA12})
holds if $\at{\comseqA}{i}$ is a cross-cause of some
$\at{\comseqA}{j}$, $j > i$.

\smallskip
Note also that the conditions of \refToDef{pcsA} must be satisfied
only by every $\at{\comseqA}{i}$ with $i < n$, thus they hold
vacuously for any single communication and for the empty \trace.
 This does not imply that a single-communication trace $\comseqA$
is $\comseqA'$-pointed for any $\comseqA'$, since to this end
$\comseqA$ also needs to be $\comseqA'$-well formed. For instance, the
trace $\CommAsI{\q}{\la}{\pp}$ is not $\ee$-well formed nor
$\CommAs{\pp}{\la}{\q}$-well formed (beware not to confuse
$\CommAsI\q\la\pp$ with $\CommAsI\pp\la\q$).   If $\comseqA =
\concat{\concat{\comseqA_1}{\beta}}{\beta'}$ is $\comseqA'$-pointed,
then either $\play{\beta}=\play{\beta'}$ or $\beta'$ matches $\beta$
in $\concat{\comseqA'}{\comseqA}$, i.e.,
$\ct{\eh{\comseqA_1}+1+\eh{\comseqA'}}{\concat{\comseqA'}{\comseqA}}{\eh{\comseqA_1}+2+\eh{\comseqA'}}$.
Also, if a \traceS $\comseqA$ is $\comseqA'$-pointed for some
$\comseqA'$, we know that each communication in $\comseqA$ must be
executed before the last one.  Indeed, the reader familiar with
ESs will have noticed that pointed traces are very similar in spirit
to ES \emph{prime configurations}.\vspace*{-2mm}

\begin{example}
  Let $\os=\CommAs{\pp}{\la}{\q}\cdot\CommAs{\pr}{\la}{\q}$ and
  $\comseqA=\CommAs{\pp}{\la}{\q}\cdot\CommAsI{\pp}{\la}{\q}\cdot\CommAsI{\pr}{\la}{\q}$.
  The \traceS $\comseqA$ is not $\os$-pointed, since the output
  $\CommAs{\pp}{\la}{\q}$ in $\comseqA$ is not matched by any input in
  $\os\cdot\comseqA$ (the input $\CommAsI{\pp}{\la}{\q}$ in $\comseqA$
  matches the output $\CommAs{\pp}{\la}{\q}$ in $\os$) and it is not
  required in $\comseqA$ because its player $\pp$ is neither  the player
  of $\CommAsI{\pp}{\la}{\q}$ nor the  player of
  $\CommAsI{\pr}{\la}{\q}$. So the condition of \refToDef{pcsA} is not
  satisfied for the output $\CommAs{\pp}{\la}{\q}$ in $\comseqA$.
  Instead the \traceS
  $\comseqA'=\CommAsI{\pp}{\la}{\q}\cdot\CommAsI{\pr}{\la}{\q}$ is
  $\os$-pointed, as well as the \traceS
  $\comseqA''=\CommAsI{\pr}{\la}{\q}\cdot\CommAsI{\pp}{\la}{\q}$.\vspace*{-2mm}
\end{example}

Pointedness is preserved by suffixing.

\begin{lemma}\mylabel{suffix-pointedness}
If $\comseqA$ is $\comseqA'$-pointed and $\comseqA =
\concat{\comseqA_1}{\comseqA_2}$, then $\comseqA_2$ is $\concat{\comseqA'}{\comseqA_1}$-pointed. \vspace*{-2mm}
\end{lemma}

\begin{proof}
  Immediate, since
  $\concat{(\concat{\comseqA'}{\comseqA_1})}{\comseqA_2} =
  \concat{\comseqA'}{(\concat{\comseqA_1}{\comseqA_2})}$ and
  $\comseqA_2$ is a suffix of $\comseqA$ and therefore its elements
  are a subset of those of $\comseqA$.
\end{proof}

Note on the other hand that if $\comseqA$ is $\comseqA'$-pointed and $\comseqA' =
\concat{\comseqA'_1}{\comseqA'_2}$, then it is not true that
$\concat{\comseqA'_2}{\comseqA}$ is
$\comseqA'_1$-pointed, because in this case the set of elements
of  $\concat{\comseqA'_2}{\comseqA}$ is a superset of that of
$\comseqA$.
For instance, if $\comseqA'_1 = \ee$, $\comseqA'_2 = \CommAs\pp\la\q$
and $\comseqA = \concat{\CommAs\pr{\la'}\ps}{\CommAsI\pr{\la'}\ps}$,
then
%$\comseqA$ is $\concat{\comseqA'_1}{\comseqA'_2}$-pointed but
$\concat{\comseqA'_2}{\comseqA}$ is not $\comseqA'_1$-pointed.

%The utility of $\os$-pointedness
A useful property of  $\os$-pointedness
is that it is preserved by the
equivalence $\approx_\os$, which does not change\ the rightmost
communication in $\os$-pointed \traces. We use $\last{\comseqA}$ to
denote the last communication of $\comseqA$.

\begin{lemma}\mylabel{prop:congSimP}
  \mylabel{prop:congSimP1} Let $\comseqA$ be $\os$-pointed and
  $\comseqA\approx_{\os} \comseqA'$. Then $\comseqA'$ is $\os$-pointed
  and $\last{\comseqA'}=\last{\comseqA}$.
\end{lemma}

\begin{proof}
  Let $\comseqA\approx_{\os} \comseqA'$. By
  \refToDef{def:permEqA} $\comseqA'$ is obtained
  from $\comseqA$ by  $m$  swaps of adjacent communications. The proof
  is by induction on the number $m$ of swaps.

\medskip\noindent
 {\it Case $m=0$}.  The result is obvious.

\smallskip\noindent
 {\it Case $m>0$}. In this case there is   $\comseqA_1$ obtained from
  $\comseqA$
  by $m-1$ swaps of adjacent communications and
there are $\beta, \beta', \comseqA_2$ such that\sm

\centerline{$\begin{array}{c}\comseqA_1 =
      \concat{\concat{\concat{\range{\comseqA_1}{ 1 }{i-1}}{\beta}}{\beta'}}{\comseqA_2}
\,\approx_{\os}\,\,
     \concat{\concat{\concat{\range{\comseqA_1}{ 1 }{i-1}}{\beta'}}{\beta}}{\comseqA_2}
      = \comseqA'\\ \text{and }
      %\play{\beta}\neq\play{\beta'}
\play{\beta}\cap\play{\beta'} = \emptyset
\text{ and }\neg(\ct
      {i + \cardin{\os} }{\os\cdot\comseqA_1} {\, i+1 + \cardin{\os} })
\end{array}$
}

\medskip\noindent
By induction hypothesis $\comseqA_1$ is $\os$-pointed and
$\last{\comseqA_1}=\last{\comseqA}$.

\medskip
To show that $\comseqA'$ is $\os$-pointed, observe that
$\play{\beta}\cap\play{\beta'} = \emptyset$ implies: \vspace*{1.6mm}

\centerline{$
\begin{array}{l}
\play{\beta}\subseteq \play{\beta'}\cup\play{\comseqA_2} \ \Leftrightarrow\
 \play{\beta}\subseteq\play{\comseqA_2}\\
\play{\beta'}\subseteq\play{\comseqA_2}\ \Leftrightarrow\
 \play{\beta'}\subseteq\play{\beta}\cup\play{\comseqA_2}
\end{array}
$} \vspace*{1.6mm}

\noindent From this we deduce $\oks{i}{\comseqA_1} \iff \oks{i}{\comseqA'}$ and
$\oks{i+1}{\comseqA_1} \iff \oks{i+1}{\comseqA'}$, so if both
%$\beta$ and $\beta'$
$\at{\comseqA_1}{i}$ and $\at{\comseqA_1}{i+1}$
are required in $\comseqA_1$ we are done.

\medskip
Otherwise, suppose that $\,\ct {i + \cardin{\os} }{\os\cdot\comseqA_1}
{\, j + \cardin{\os} }$ where either $\oks{j}{\comseqA_1}$ or
$j=n$. If $\oks{j}{\comseqA_1}$ then also $\oks{j}{\comseqA'}$, as
we just saw. Now, $j$ cannot be $ i+1$ since by hypothesis $\neg(\ct
{i + \cardin{\os} }{\os\cdot\comseqA_1} {\, i+1 + \cardin{\os} })$.
This implies $\ct{i+1 + \cardin{\os}}{\os\cdot\comseqA'}{j +
  \cardin{\os}}$.  Similarly we can show that $\ct {i+1 + \cardin{\os}
} {\os\cdot\comseqA_1} {j + \cardin{\os} }$ implies $\ct {i +
  \cardin{\os}} {\os\cdot\comseqA'} {j + \cardin{\os} }$.
Therefore $\comseqA'$ is $\os$-pointed.

\medskip
To show that $\last{\comseqA}=\last{\comseqA'}$, assume ad absurdum
that $ \comseqA_2 =\epsilon$.  Then
$\concat{\concat{\range{\comseqA_1}{ 1 }{i-1}}{\beta}}{\beta'}$ is
$\os$-pointed and thus, as observed after \refToDef{pcsA}, we have
either $\play{\beta} \cap\play{\beta'} \neq \emptyset$ or $\ct {i +
  \cardin{\os}}{\os\cdot\comseqA_1} {\,i+1 + \cardin{\os}}$. In both
cases $\beta$ and $\beta'$ cannot be swapped. So it must be
$\comseqA_2\neq\epsilon$.
\end{proof}

 We now relate \agts\ with pairs of o-\tracesS and \traces.

\begin{lemma}
\mylabel{prop:wellForG}
If   $\derSI{}{ \G\parN\Msg }$   and $\os=\osq{\Msg}$ and
$\comseqA\in\FPaths{\G}$, then
 ${\comseqA}$ is
$\os$-well formed.
\end{lemma}

\begin{proof}
We prove  by induction on $\comseqA$ that $\derSI{}{ \G\parN\Msg }$ implies that
 $\concat\os{\comseqA}$  is well formed.

\vspace*{1.5mm}\noindent
 {\it Case $\comseqA=\beta$}.   If $\beta$ is an output the result is obvious. If $\beta=\CommAsI\pp{\la}\q$, by  Rule $\rulename{In}$
of \refToFigure{wfagtA}, we get $\Msg\equiv\addMsg{\mq\pp{\la}\q}{\Msg'}$. Therefore $\os=\concat{\CommAs\pp{\la}\q}{\os'}$
and $\concat{\os}{\beta}$ is well formed.

\vspace*{1.5mm}\noindent
 {\it Case $\comseqA=\beta\cdot\comseqA'$ with $\comseqA'\in\FPaths{\G'}$.}   If $\beta=\CommAs\pp{\la}\q$, then
$\G=\agtO{\pp}{\q}{i}{I}{\la}{\G}$ and $\la=\la_k$ and $\G'=\G_k$ for some $k\in I$.
From $\derSI{}{ \G\parN\Msg }$ and  Rule $\rulename{Out}$ of \refToFigure{wfagtA}, we get
  $\derSI{}{\G'\parN\addMsg\Msg{\mq\pp{\la}\q}}$. By induction hypothesis on $\comseqA'$, the trace
   $\concat{\osq{\addMsg\Msg{\mq\pp{\la}\q}}}{\comseqA'}$ is well formed. So since
 $\osq{\addMsg\Msg{\mq\pp{\la}\q}}=\os\cdot \CommAs\pp{\la}\q$ we get
that $\os\cdot\comseqA$ is well formed.

\medskip\noindent
  If  $\beta=\CommAsI\pp{\la}\q$, then
   $\G=\agtI \pp\q \la {\G'}$.  From $\derSI{}{ \G\parN\Msg }$ and Rule $\rulename{In}$ of
  \refToFigure{wfagtA}, we get $\Msg\equiv\addMsg{\mq\pp{\la}\q}\Msg'$ and
  $\derSI{ \G\parN\Msg }{\G'\parN\Msg'}$.  Let $\os'=\osq{\Msg'}$. Then
  $\os\cong\CommAs\pp{\la}\q\cdot\os'$. By induction hypothesis on
  $\comseqA'$ the \traceS $\concat{\os'}{\comseqA'}$ is well formed.  We
  want now to show that also the \traceS\
$\comseqA '' =
  \concat{\os}{\comseqA} =
  \concat{\concat{\CommAs{\pp}{\la}{\q}}{\os'}}
  {\concat{\CommAsI{\pp}{\la}{\q}}{\comseqA'}}$
is well formed, namely that in $\comseqA''$ every input  matches an output.
Note that the first input in $\comseqA''$ is  $\at{\comseqA''}{|\os|+1}
  = \CommAsI{\pp}{\la}{\q}$.   This input matches the output
  $\at{\comseqA''}{ 1 } = \CommAs{\pp}{\la}{\q}$.  For inputs
  $\at{\comseqA''}{i}$
with  $i>\cardin{\os}+1$, we know that $\at{\comseqA''}{i} =
\at{(\concat{\os'}{\comseqA'})}{i - 2}$, where
$\at{(\concat{\os'}{\comseqA'})}{i - 2}$ matches  some output
$\at{(\concat{\os'}{\comseqA'})}{j}$ in $\concat{\os'}{\comseqA'}$.
Then $\at{\comseqA''}{i}$ matches $\at{\comseqA''}{j+1}$ if $j \leq
\cardin{\os'}$ and $\at{\comseqA''}{j+2}$ otherwise. This proves that
$ \concat{\os}{\comseqA}$ is well \linebreak formed.\vspace{-2mm}
\end{proof}

We have now enough machinery to define events of \agts, which are
equivalence classes of pairs whose  first elements are  o-\trace s $\os$
(representing queues) and whose second elements are \trace s
$\comseqA$ (representing paths in  the \sgt\ components of the
\agts).   The traces $\os$ and $\comseqA$ are considered respectively
modulo $\cong$ and modulo $\approx_{\os}$.  The \traceS
 ${\comseqA}$  is $\os$-well formed,  reflecting the
 \balancing\ of asynchronous types.
%\oi\ matching of \sgts\ with respect to queues.
The
communication represented by an event is the last communication of
$\comseqA$.\vspace*{-3mm}

\begin{definition}[Type event]\label{def:glEventA}\hbox{}\hfill\vspace*{-7mm}

\begin{enumerate}
\item\label{def:glEventA1} The {\em equivalence} $\sim$ on
  pairs $\Commpair\os\comseqA$, where $\comseqA \neq\ee$ is
$\os$-pointed,  is the least equivalence such that  \vspace*{1.8mm}

  \centerline{$\Commpair\os\comseqA\sim\Commpair{\os'}{\comseqA'}$ if
    $\os\cong\os'$ and $\comseqA\approx_{\os}{\comseqA'}$}

\item\label{def:glEventA2}\label{def:glEventA3} A  \emph{type event (t-event)}  $\delta =
\eqA\os\comseqA$ is the equivalence class of the pair
$\Commpair\os\comseqA$.
The communication of $\delta$, notation  $\io\delta$, is
  defined to be  $\last{\comseqA }$.
\item\label{def:glEventA4} We denote by  $\GEa$  the set of t-events.\vspace*{-2mm}
\end{enumerate}
\end{definition}
Notice that the function $\sf i/o$ can be applied both to an n-event
(\refToDef{proceventNA}(\ref{proceventNA2})) and to a t-event
(\refToDef{def:glEventA}(\ref{def:glEventA3})). In all cases the
result is a communication.

\medskip
Given an o-trace $\os$ and an arbitrary trace $\comseqA$, we want to
build a t-event $\eqA\os{\comseqA'}$ (\refToDef{def:pf}).  To this aim
we scan $\comseqA$ from right to left and remove all and only the
communications $\at{\comseqA}{i}$ which  make $\comseqA$ violate
the $\os$-pointedness property.

\begin{definition}[Trace  filtering]\mylabel{def:trace-filtering}\mylabel{def:pm}
The {\em filtering of $\concat\comseqA{\comseqA'}$ by $\os$ with cursor at $\comseqA$,} denoted by $\filt\comseqA{\comseqA'}$,
is defined by induction on $\comseqA$ as follows:  \smallskip

\centerline{$\filt{\ee}{\comseqA'}=\comseqA'\qquad\qquad
\filt{(\concat{\comseqA''}{\beta})}{\comseqA'} =\begin{cases}
\filt{\comseqA''}{(\concat{\beta}{\comseqA'})}     & \text{if $\concat\beta{\comseqA'}$ is $(\concat\os{\comseqA''})$-pointed}  \\[2pt]
\filt{\comseqA''}{\comseqA'}   & \text{otherwise}
\end{cases}
$}
\end{definition}
For example
$\filtP{\concat{\CommAsI\pp\la\q}{\CommAsI\q\la\pp}}{\CommAs\pp\la\q}\ee=\filtP{\CommAsI\pp\la\q}{\CommAs\pp\la\q}\ee=\filtP\ee{\CommAs\pp\la\q}{\CommAsI\pp\la\q}={\CommAsI\pp\la\q}$.
The resulting trace can also be empty,
 in case the last communication is an input and
$\concat{\comseqA}{\comseqA'}$ is not $\os$-well formed. For
instance,
$\filtP{\CommAsI\q\la\pp}{\CommAs\pp\la\q}\ee=\filtP\ee{\CommAs\pp\la\q}{\ee}=\ee$
because $\CommAsI{\q}{\la}{\pp}$ is not $\CommAs{\pp}{\la}{\q}$-well formed.
% as in
% $\filtP{\CommAsI\q\la\pp}{\CommAs\pp\la\q}\ee=\filtP\ee{\CommAs\pp\la\q}{\ee}=\ee$.
It is easy to verify that $\filt\comseqA{\comseqA'}$ is a subtrace of
$\concat\comseqA{\comseqA'}$,  and that if $\comseqA$ is $\os$-pointed, then
$\filt{\comseqA}{\ee} = \comseqA$.

\begin{definition}[t-event of a pair]\mylabel{def:pf}
Let $\comseqA\not=\ee$ be $\os$-well formed.
%The {\em  t-event of} $\pairA\os\comseqA$, notation $\point(\os,\comseqA)$, is $\eqA{\os}{\comseqA\pent\pairA\os\ee}$
  The {\em  t-event} generated by $\os$ and $\comseqA$, notation
 $\point(\os,\comseqA)$, is defined to be $\point(\os,\comseqA)
 =\eqA{\os}{\filt{\comseqA}\ee}$.
\end{definition}

Hence  the trace of the event $\point(\os,\comseqA)$ is the
filtering of $\comseqA$ by $\os$ with cursor at the end of $\comseqA$.
This definition is sound  since $\os\cong\os'$ implies  $\filt{\comseqA}{\comseqA'}=\filtP{\comseqA}{\os'}{\comseqA'}$.
%Moreover $\point(\os,\comseqA)$ enjoys a useful property, as shown by the following lemma.
 Moreover the communication of $\point(\os,\comseqA)$ is the last communication of $\comseqA$.

\begin{lemma}\label{gl}
 If $\point(\os,\comseqA)$ is defined, then $\filt{\comseqA}\ee\not=\ee$ and $\io{\point(\os,\comseqA)}=\last{\filt\comseqA\ee}=\last\comseqA$.
\end{lemma}

\begin{proof}
  Let $\comseqA\not=\ee$ be $\os$-well formed and
  $\comseqA=\concat{\comseqA'}\beta$. Then $\beta$ is
  $(\concat\os{\comseqA'})$-well formed by~\refToDef{def:WFnew}. This
  implies that $\beta$ is $(\concat\os{\comseqA'})$-pointed
  by~\refToDef{pcsA1}, and thus
  $\filt\comseqA\ee=\filt{(\concat{\comseqA'}{\beta})}\ee=
  \filt{\comseqA'}\beta$. This gives
  $\io{\point(\os,\comseqA)}=\last{\filt\comseqA\ee}=\last\comseqA$.
\end{proof}

Since the o-\tracesS in the t-events of  an   \agt\   correspond to the %message
queue, we define the causality and conflict relations only between
t-events with the same o-\traces.  Causality is then simply prefixing
of \traces\  modulo $\approx_\os$,  while conflict is induced by the conflict relation
on the p-events obtained by projecting the traces on participants (\refToDef{pep}(\ref{pep1})).

\begin{definition}[Causality and conflict relations on t-events] \label{ageo}
The {\em causality} relation $\precP$ and the {\em conflict} relation $\gr$ on the set of t-events $\GEa$ are defined by:
\begin{enumerate}
\item\mylabel{ageo1} $\eqA\os{\comseqA}\precP\eqA{\os}{\comseqA'}$ if
$\comseqA'\approx_\os\concat{\comseqA}{\comseqA_1}$ for some  $\comseqA_1$;
\item\mylabel{ageo2}  $\eqA\os{\comseqA}\,\grr \,\eqA{\os}{\comseqA'}$
  if  $\projAP{\comseqA}\pp \grr \projAP{\comseqA'}\pp$ for some $\pp$.
 \end{enumerate}
\end{definition}
\noindent
Concerning Clause~(\ref{ageo1}), note that the relation $\precP$
is able to express cross-causality as well as local
causality, thanks to the hypothesis of $\os$-well formedness of $\tau$
in any t-event  $\eqA\os\comseqA$.
Indeed, this hypothesis implies that, whenever $\comseqA$  ends by an
input $\CommAsI{\pp}{\la}{\q}$, then the matched
%justifying
output
$\CommAs{\pp}{\la}{\q}$ must appear either in $\os$, in which case
the output has already occurred, or at some position $i\,$ in
$\comseqA$.
In the latter case, the t-event
$\point(\os, \range{\comseqA}{1}{i})$,
which represents the output $\CommAs{\pp}{\la}{\q}$,  is such that
  $\point(\os, \range{\comseqA}{1}{i}) \precP \eqA\os\comseqA$.

\medskip
As regards Clause~(\ref{ageo2}), note that if
  $\comseqA\approx_{\os}{\comseqA'}$, then
  $\projAP{\comseqA}\pp=\projAP{\comseqA'}\pp$ for all $\pp$, because
  $\approx_{\os}$ does not swap communications with the same player.
  Hence, conflict is well defined, since it does not depend on the
  trace chosen in the equivalence class.  The condition
  $\projAP{\comseqA}\pp \grr \projAP{\comseqA'}\pp$ states that
  participant $\pp$ does the same actions in both traces up to some
  point, after which it performs two different actions in $\comseqA$
  and $\comseqA'$.

\medskip
We get the events of  an   \agt\ $\G\parN\Msg$ by applying the
function $\point$ to the pairs made of the o-trace  representing   the queue $\Msg$
and a trace in the tree of $\G$.  \refToLemma{prop:wellForG} and
\refToDef{def:pf} ensure that $\point$ is defined.  We then build the
ES associated with  an  \agt\ $\G\parN\Msg$ as follows.

\begin{definition}[Event structure of an  \agt]
  \mylabel{egA}
The {\em
    event structure of the   \agt} $\G\parN\Msg$ is the triple  \vspace*{1.6mm}

 \centerline{$\ESGA{\G\parN\Msg} = (\EGGA(\G\parN\Msg), \precP_{\G\parN\Msg} , \grr_{\G\parN\Msg})$}

 \vspace*{1.6mm} \noindent where:
\begin{enumerate}
\itemsep=0.95pt
\item \mylabel{eg1AP}
    $\EGGA(\G\parN\Msg) =
  \{\point(\os,\comseqA) \mid \os=\osq\Msg ~\&~
  \comseqA\in\FPaths{\G}\}$;
\item\mylabel{eg2A} $\precP_{\G\parN\Msg}$ is the restriction of
  $\precP$ to the set $\EGGA(\G\parN\Msg)$;
\item\mylabel{eg3A} $\gr_{\G\parN\Msg}$ is the restriction of $\gr$
  to the set $\EGGA(\G\parN\Msg)$.
\end{enumerate}
\end{definition}

\begin{example}\label{last}
The network of \refToExample{asynchronous-network-with-choice} can be typed by the asynchronous type $\G\parN\emptyset$ with $\G = \agtOS \pp\q {\la_1}{\agtIS \pp\q {\la_1}; \agtOS{\pp}{\pr}{\la}{\agtIS
  \pp\pr{\la}}}~\GlSyB~\agtOS \pp\q {\la_2}{\Seq{\agtIS \pp\q
    {\la_2}}{\agtOS{\pp}{\pr}{\la}{\agtIS \pp\pr{\la}}}}$. The
t-events of $\ESGA{\G\parN\emptyset}$ are:
 \[\centerline{$\begin{array}{lcl}
\globevA_1 = \eqA{\ee}{\pp\q!\la_1} & \qquad & \globevA'_1 = \eqA{\ee}{\pp\q!\la_1\cdot \pp\q?\la_1}
\\
 \globevA_2 = \eqA{\ee}{\pp\q!\la_2}
 & \qquad & \globevA'_2 = \eqA{\ee}{\pp\q!\la_2\cdot \pp\q?\la_2}\\
\globevA_3 =  \eqA{\ee}{\pp\q!\la_1\cdot \pp\pr!\la}
&\qquad & \globevA''_1 = \eqA{\ee}{\pp\q!\la_1\cdot \pp\pr!\la \cdot\pp\pr?\la}\\
\globevA_4 = \eqA{\ee}{\pp\q!\la_2\cdot \pp\pr!\la}
%%%%%
& \qquad & \globevA''_2 =  \eqA{\ee}{\pp\q!\la_2\cdot \pp\pr!\la \cdot\pp\pr?\la}
\end{array}$}\]
\noindent The causality relation is given by $\globevA_1\precP\globevA_3$,
$\globevA_1\precP\globevA'_1$, $\globevA_2\precP\globevA_4$, $\globevA_2\precP\globevA'_2$,
$\globevA_3\precP\globevA''_1$, $\globevA_4\precP\globevA''_2$,
$\globevA_1\precP\globevA''_1$, $\globevA_2\precP\globevA''_2$.
The conflict relation
is given by $\globevA_1\grr\globevA_2$ and all the conflicts inherited
from it. \refToFigure{fig:type-PES}(b) in Section~\ref{sec:resultsA} illustrates this event structure.
\end{example}

\medskip
The following example shows that, due to  the fact that \sgts\ are not able to represent concurrency explicitly,
two forking traces in the tree representation of $\G$ do not necessarily give rise to two conflicting events
in  $\ESGA{\G\parN\Msg}$.

\begin{example}\mylabel{ex:cila}
Let $\G=\CommAs\pp{\la}\q; (\pr \ps !\la_1 ;\CommAsI{\pp}{\la}{\q} ;\CommAsI{\pr}{\la_1}{\ps}
 \GlSyB\pr \ps !\la_2;  {\CommAsI{\pp}{\la}{\q}};\CommAsI{\pr}{\la_2}{\ps})$.
Then
$\ESGA{\G\parN\emptyset}$ contains the t-event
$\eqA{\ee}{\concat{\CommAs\pp{\la}\q}{\CommAsI\pp{\la}\q}}$
generated by the two forking
%diverging
traces in $\FPaths{\G}$:
\[
 \concat{\CommAs{\pp}{\la}{\q}}{\concat{\CommAs{\pr}{\la_1}{\ps}}{\CommAsI{\pp}{\la}{\q}}}
\qquad\qquad\qquad
  \concat{\CommAs{\pp}{\la}{\q}}{\concat{\CommAs{\pr}{\la_2}{\ps}}{\CommAsI{\pp}{\la}{\q}}}
\]
\noindent Note on the other hand that if we replace $\pr$ by $\q$ in $\G$, namely if we consider the
\sgt\ $\G'=\CommAs\pp{\la}\q; (\q \ps ! \la_1 ;\CommAsI{\pp}{\la}{\q} ;\CommAsI{\q}{\la_1}{\ps}
\GlSyB \q \ps ! \la_2;  {\CommAsI{\pp}{\la}{\q}};\CommAsI{\q}{\la_2}{\ps})$, then
$\ESGA{\G'\parN\emptyset}$ contains $\comoccA=\eqA{\ee}{\concat{\concat{\CommAs{\pp}{\la}{\q}}{\CommAs{\q}{\la_1}{\ps}}}{\CommAsI{\pp}{\la}{\q}}}$ and
$\comoccA'=\eqA{\ee}{\concat{\concat{\CommAs{\pp}{\la}{\q}}{\CommAs{\q}{\la_2}{\ps}}}{\CommAsI{\pp}{\la}{\q}}}$.
Here $\comoccA\grr \comoccA'$ because
\[
\projAP{(\concat{\concat{\CommAs{\pp}{\la}{\q}}{\CommAs{\q}{\la_1}{\ps}}}{\CommAsI{\pp}{\la}{\q}}
  )}{\q}= \concat{\sendL{\ps}{\la_1}}{\rcvL{\pp}{\la}} \,\grr\, \concat{\sendL{\ps}{\la_2}}{\rcvL{\pp}{\la}}
=
\projAP{(\concat{\concat{\CommAs{\pp}{\la}{\q}}{\CommAs{\q}{\la_2}{\ps}}}{\CommAsI{\pp}{\la}{\q}}
  )}{\q}
\]
\noindent  So, here the two occurrences of
$\CommAsI{\pp}{\la}{\q}$ in the type are represented by two distinct events that are in conflict.\vspace*{-1mm}
\end{example}

We end this section by showing that the obtained ES is a PES.

\begin{proposition}\mylabel{basta12A}
%Let $\G$ be a gt.
Let  $\G\parN\Msg$  be an    \agt.
Then  $\ESGA{\G\parN\Msg}$  is a prime event structure.
\end{proposition}

\begin{proof}
  We show that $\precP$ and $\gr$ satisfy Properties (\ref{pes2}) and
  (\ref{pes3}) of \refToDef{pes}.  Reflexivity and transitivity of
  $\precP$ follow easily from the properties of concatenation and the
  properties of the two equivalences in
  Definitions~\ref{def:permEqA2} and \ref{def:permEqA}.  As for
  antisymmetry note that, by Clause (\ref{ageo1}) of \refToDef{ageo},
  if $\eqA\os{\comseqA}\,\leq \,\eqA{\os}{\comseqA'}$ and
  $\eqA{\os}{\comseqA'}\,\leq \,\eqA{\os}{\comseqA}$, then
  $\concat{\comseqA}{\comseqA_1}\approx_{\os}{\comseqA'}$ and
  $\concat{\comseqA'}{ \comseqA_2}\approx_{\os}\comseqA$ for some
  $\comseqA_1$ and $ \comseqA_2$.
  Hence $\concat{\concat{\comseqA}{\comseqA_1}}{ \comseqA_2}
  \approx_\os \comseqA$,
  which implies $\comseqA_1= \comseqA_2=\emptyseq$, i.e. ${\comseqA}\approx_\os{\comseqA'}$.

\medskip
  The conflict between t-events inherits irreflexivity, symmetry and
  hereditariness from the conflict between p-events. In particular,
  for hereditariness, suppose that $\eqA{\os}{\comseqA}\grr
  \eqA{\os}{\comseqA'}\precP \eqA{\os}{\comseqA''}$. Then
  $\comseqA''\approx_{\os}\concat{\comseqA'}{\comseqA_1}$ for some
  ${\comseqA_1}$ and
  $\projAP{\comseqA''}\pp=\projAP{(\concat{\comseqA'}{\comseqA_1})}\pp=\concat{(\projAP{\comseqA'}\pp)}{(\projAP{\comseqA_1}\pp)}\grr\projAP{\comseqA}\pp$
  since $\projAP{\comseqA'}\pp\grr\projAP{\comseqA}\pp$. \vspace*{-9mm}
\end{proof}

\begin{figure}[!ht]
%\vspace*{-2mm}
\begin{center}
$\Nt = \pP\pp{\q!\la_1 ; \pr!\la\oplus \q!\la_2 ;
  \pr!\la} \parN\pP\q{\pp?\la_1+
  \pp?\la_2}\parN\pP\pr{\pp?\la}\vspace*{-2mm}$
\setlength{\unitlength}{1mm}
\scalebox{0.9}{
\begin{picture}(100,70)
\put(24,60){{\small $\locevAS{\pp}{\ee}{\,\q!\la_1}$}}
\put(20,35){{{\small $\locevAS{\pp}{\ee}{\,\q!\la_1\cdot\pr!\la}$}}}
\put(47,10){{{\small $\locevAS{\pr}{\ee}{\pp?\la}$}}}
\put(50,1){(a)}
\put(67,60){{\small $\locevAS{\pp}{\ee}{\,\q!\la_2}$}}
\put(65,35){{{\small $\locevAS{\pp}{\ee}{\,\q!\la_2\cdot\pr!\la}$}}}
\put(7,47){{\small $\locevAS{\q}{\ee}{\pp?\la_1}$}}
\put(85,47){{\small $\locevAS{\q}{\ee}{\pp?\la_2}$}}
\put(50,60){{$\gr$}}
\put(50,35){{$\gr$}}
\put(50,47.5){{$\gr$}}
\thicklines
\linethickness{0.3mm}
\put(32,57){\vector(0,-1){17}}
\put(72,57){\vector(0,-1){17}}
\multiput(21,47)(1,0){9}{\bf{$\cdot$}}
\multiput(34,47)(1,0){15}{\bf{$\cdot$}}
\multiput(55,47)(1,0){15}{\bf{$\cdot$}}
\multiput(38,60)(1,0){10}{\bf{$\cdot$}}
\multiput(55,60)(1,0){10}{\bf{$\cdot$}}
\multiput(40,35)(1,0){8}{\bf{$\cdot$}}
\multiput(55,35)(1,0){8}{\bf{$\cdot$}}
\multiput(74,47)(1,0){9}{\bf{$\cdot$}}
\multiput(55,50)(1.5,1){7}{\bf{$\cdot$}}
\multiput(38.5,38.5)(1.5,1){7}{\bf{$\cdot$}}
\multiput(39,56)(1.5,-1){7}{\bf{$\cdot$}}
\multiput(54,45)(1.5,-1){7}{\bf{$\cdot$}}
\put(33,31){\vector(1,-1){15}}
\put(71,31){\vector(-1,-1){15}}
\put(24,57){\vector(-1,-1){5.5}}
\put(80,57){\vector(1,-1){5.5}}
\end{picture} }

\vspace{0.4cm}

$\G = \agtOS \pp\q {\la_1}{\agtIS \pp\q {\la_1}; \agtOS{\pp}{\pr}{\la}{\agtIS
  \pp\pr{\la}}}~\GlSyB~\agtOS \pp\q {\la_2}{\Seq{\agtIS \pp\q
    {\la_2}}{\agtOS{\pp}{\pr}{\la}{\agtIS \pp\pr{\la}}}}\vspace*{-2mm}$
\setlength{\unitlength}{1mm}
\scalebox{0.9}{
\begin{picture}(100,70)
\put(24,60){{\small $\eqA{\ee}{\pp\q!\la_1}$}}
\put(20,35){{{\small $\eqA{\ee}{\pp\q!\la_1; \pp\pr!\la}$}}}
\put(16,10){{{\small $\eqA{\ee}{\pp\q!\la_1; \pp\pr!\la ;
      \pp\pr?\la}$}}}
\put(60,10){{{\small $\eqA{\ee}{\pp\q!\la_2; \pp\pr!\la ; \pp\pr?\la}$}}}
\put(65,60){{\small $ \eqA{\ee}{\pp\q!\la_2}$}}
\put(50,1){(b)}
\put(61,35){{{\small $\eqA{\ee}{\pp\q!\la_2; \pp\pr!\la}$}}}
\put(-5,47){{\small $\eqA{\ee}{\pp\q!\la_1; \pp\q?\la_1}$}}
\put(85,47){{\small $\eqA{\ee}{\pp\q!\la_2; \pp\q?\la_2}$}}
\put(50,60){{$\gr$}}
\thicklines
\linethickness{0.3mm}
\put(32,57){\vector(0,-1){17}}
\put(72,57){\vector(0,-1){17}}
\put(32,32){\vector(0,-1){17}}
\put(72,32){\vector(0,-1){17}}
\multiput(40,60)(1,0){9}{\bf{$\cdot$}}
\multiput(54,60)(1,0){9}{\bf{$\cdot$}}
\put(24,57){\vector(-1,-1){5.5}}
\put(80,57){\vector(1,-1){5.5}}
\end{picture} }
\end{center} \vspace*{-7mm}
\caption{
(a) ~FES of %the network
$\Nt\parN\emptyset$ in
  \refToExample{asynchronous-network-with-choice}. \qquad
(b) ~PES of %the asynchronous type
$\G\parN\emptyset$ in
\refToExample{last}.
}\vspace*{-2mm}
\mylabel{fig:network-FES} \mylabel{fig:type-PES}
\end{figure}

\section{Equivalence of the two  event structure semantics}\label{sec:resultsA}

In the previous two sections, we defined the ES semantics of networks
and types, respectively.  As expected, the FES of a network is not
isomorphic to the PES of its type, unless the former is itself a
PES. As an example, consider the network FES pictured in
\refToFigure{fig:network-FES}(a) (where the arrows represent the flow
relation) and its type PES pictured in
\refToFigure{fig:type-PES}(b) (where the arrows represent the covering
relation of causality and inherited conflicts are not shown).
The rationale is that events in the network FES record the \emph{local
  history} of a communication, while events in the type FES record its
\emph{global causal history}, which contains more information. Indeed,
while the network FES may be obtained from the type PES simply by
projecting each t-event on the player of its last communication, the
inverse construction is not as direct: essentially, one needs to
construct the configuration domain of the network FES, and from this,
by selecting the complete prime configurations according to the
classic construction of \cite{NPW81}, retrieve the type PES.  To show
that this is indeed the type PES, however, we would need to rely on
well-formedness properties of the network FES, namely on semantic
counterparts of the well-formedness properties of types. We will not
follow this approach here. Instead, we will compare the FESs of networks
and the PESs of their types at a more operational level, by looking at
the configuration domains they generate.

\medskip
In the rest of this section we establish our main theorem for typed
networks, namely the isomorphism between the configuration domain of
the FES of the network and the configuration domain of the PES of its
\agt.  To prove the various results leading to this theorem, we will
largely use the characterisation of configurations as proving
sequences, given in \refToProp{provseqchar}. Let us briefly sketch how
these results are articulated.

The proof of the isomorphism is grounded on the Subject Reduction
Theorem (\refToTheorem{srA}) and the Session Fidelity Theorem
(\refToTheorem{sfA}). These theorems state that if
$\derN{\Nt\parN\Msg}{\G\parN\Msg}$, then
$\Nt\parN\Msg\stackred\comseqA\Nt'\parN\Msg'$ if and only if
$\G\parN\Msg\stackred\comseqA\G'\parN\Msg'$, and in both directions
\mbox{$\derN{\Nt'\parN\Msg'}{\G'\parN\Msg'}$.} We can then relate the ESs of
networks and \agts\ by connecting them through the \tracesS of
their transition sequences, and by taking into account the queues by
means of the mapping ${\sf otr}$ given by \refToDef{qos}. This is
achieved as follows.

\medskip
If $\Nt\parN\Msg\stackred\comseqA$ and $\osq\Msg=\os$, then the
function ${\sf nec}$ (\refToDef{necA}) applied to $\os$ and $\comseqA$
gives a proving sequence in $\ESNA{\Nt\parallel \Msg}$
(\refToTheorem{uf10A}).  Vice versa, if
$\Seq{\netevA_1;\cdots}{\netevA_n}$ is a proving sequence in
$\ESNA{\Nt\parallel\Msg}$, then
$\Nt\parallel\Msg\stackred\comseqA\Nt'\parallel\Msg'$, where
$\comseqA=\io{\netevA_1}\cdots\io{\netevA_n}$ and ${\sf i/o}$ is the
mapping given in \refToDef{proceventNA}(\ref{proceventNA2})
(\refToTheorem{uf12A}).

Similarly, if $\G\parN\Msg\stackred\comseqA\G'\parN\Msg'$ and
$\osq\Msg=\os$, then the function  ${\sf tec}$  (\refToDef{gecA})
applied to $\os$ and $\comseqA$ gives a proving sequence in
$\ESGA{\G\parallel \Msg}$ (\refToTheorem{lemma:keyA1}).  Lastly, if
$\Seq{\comoccA_1;\ldots}{\comoccA_n}$ is a proving sequence in
$\ESGA{\G\parN\Msg}$, then
$\G\parG\Msg\stackred\comseqA\G'\parG\Msg'$, where
$\comseqA=\concat{\concat{\io{\comoccA_1}}\ldots}{\io{\comoccA_n}}$
and ${\sf i/o}$ is the mapping given in
\refToDef{def:glEventA}(\ref{def:glEventA3})
(\refToTheorem{lemma:keyA2}).

\medskip
 It is then natural to split this section in three subsections: the
 first establishing the relationship between network transition
 sequences and proving sequences of their event structure, the second
 doing the same for \agts\ and finally a third subsection in
 which the isomorphism between the two configuration domains is proved
 relying on these relationships.

 \subsection{Transition sequences of networks and proving
   sequences of their ESs}

 We start by showing how network communications affect n-events in the
 associated ES.  To this aim we define two partial operators
 $\blacklozenge$ and $\lozenge$, which applied to a communication
 $\beta$ and an n-event $\netevA$ yield another n-event $\netevA'$
 (when defined).  The intuition is that $\netevA'$ represents the
 event $\netevA$ as it  will be  after the communication
 $\beta$, or as it  was  before the communication $\beta$,
 respectively. So, in particular, if $\set\pp=\play{\beta}$ and $\netevA$ is not located at
 $\pp$,
 it will  remain unchanged under both mappings  $\blacklozenge$ and
 $\lozenge$.   We shall now explain in more detail
 how these operators work.

\medskip
The operator $\blacklozenge$, when applied to $\beta$ and $\netevA$,
yields the n-event $\netevA'$ obtained from $\netevA$ after executing
the communication $\beta$, if this event exists. We call
$\postA{\netevA}{\beta}$ the \emph{residual} of $\netevA$ after
$\beta$.  So, if $\beta=\CommAs\pp{\M}\q$ and $\netevA$ is located at $\pp$ and its p-event starts with
the action $\sendL\q\la$, then the p-event of $\netevA'$ is obtained
by removing this action, provided the result is still a p-event (this
will  not be the case
%not happen
if the p-event of $\netevA$ is a simple action);
% ,because in this case it will be consumed by the communication
% $\beta$);
otherwise, the operation is not defined.
% for $\netevA$ located at $\pp$.
If $\beta=\CommAsI\pp{\M}\q$ and  $\netevA$ is
located at $\q$ and its p-event starts with the action $\rcvL\pp\la$,
the p-event of $\netevA'$ is obtained by removing $\rcvL\pp\la$, if
possible; otherwise, the operation is not defined.

\medskip
The operator $\lozenge$, when applied to $\beta$ and $\netevA$, yields
the n-event $\netevA'$ obtained from $\netevA$ before executing the
communication $\beta$. We call $\preA{\netevA}{\beta}$ the
\emph{retrieval} of $\netevA$ before $\beta$.  So, if
$\beta=\CommAs\pp{\M}\q$ and $\netevA$ is located at $\pp$, the
p-event of $\netevA'$ is obtained by adding $\sendL\q\la$ in front of
the p-event of $\netevA$.  If $\beta=\CommAsI\pp{\M}\q$ and $\netevA$
is located at $\q$, the p-event of $\netevA'$ is obtained by adding
$\rcvL\pp\la$ in front of the p-event of $\netevA$.  We use the
projection $ \projAP{\comseqA}{\pr}$ of a trace on a
participant given in \refToDef{pep}(\ref{pep1}).\vspace*{-2mm}

\begin{definition}[Residual and retrieval of an n-event with respect to a communication]\mylabel{def:PostPre}\hbox{}\hfill\vspace*{-7mm}

\begin{enumerate}
\item\mylabel{def:Post}
The {\em residual of an n-event $\locevA{\pr}\os{\procev}$ after a communication $\beta$} is defined by\sm

 \centerline{
$\postPA{\locevA{\pr}\os{\procev}}{\beta}= \locevA{\pr}{ (\concat{\os}{\CommAs\pp{\M}\q})}{{\procev'}}\quad$ if
 $\procev=\concat{(\projAP\beta{\pr})}{\procev'}$
 }\sm
 \item\mylabel{def:Pre}
 The {\em retrieval of an  n-event $\locevA{\pr}\os{\procev}$ before a communication $\beta$} is defined by\sm

 \centerline{$\prePA{\locevA{\pr}\os{\procev}}{\beta}=
      \locevA{\pr}{} \concat{(\projAP\beta{\pr})}{\procev} $}
\end{enumerate}
\end{definition}
\noindent  Notice that in Clause (\ref{def:Post}) of the above definition
$\procev'\not=\ee$, see \refToDef{proceventP}.   So
$\postPA{\locevA{\pr}\os{\procev}}{\beta}$ is not defined if $\set\pr=\play{\beta}$ and
 either  $\procev$ is just an atomic action or $\projAP\beta{\pr}$ is not the first action of $\procev$.

\medskip Observe also that the
operators $\blacklozenge$ and $\lozenge$ preserve the communication of
n-events, namely\sm

\centerline{$
\io{\postA{\netevA}{\beta}} = \io{\preA{\netevA}{\beta}} =\io{\netevA}$}

\noindent  Residual and retrieval are inverse of each other.
\begin{lemma}\mylabel{prop:prePostNetA}
\begin{enumerate}
\itemsep=0.9pt
\item \mylabel{ppn1A} If $\postA{\netevA}{\beta}$ is defined, then $\prePA{\postA{\netevA}{\beta}}{\beta}=\netevA$.
\item \mylabel{ppn1bA}   $\postPA{\preA{\netevA}{\beta}}{\beta}=\netevA$.
\end{enumerate}
\end{lemma}
 The residual and retrieval operators on n-events  are mirrored by
 (partial) mappings on o-traces, which it is handy to
 define explicitly.

 \begin{definition}\label{def:pQueue} The {\em partial mappings}
   $\mapBl{\beta}{\os}$ and $\mapWh{\beta}{\os}$ are defined by:
%We define:
 \begin{enumerate}
 \itemsep=0.9pt
  \item \label{def:pQueue2}$\mapBl{\CommAs\pp\la\q}\os=\concat{\os}{\CommAs\pp\la\q}$ and  $\mapBl{\CommAsI\pp\la\q}\os=\os'$ if $\os\cong\concat{\CommAs\pp\la\q}{\os'}$;
%$\mapBl{\CommAsI\pp\la\q}{\concat{\CommAs\pp\la\q}{\os}}=\os$;
\item \label{def:pQueue1}  $\mapWh{\CommAs\pp\la\q}\os=\os'$ if $\os\cong\concat{\os'}{\CommAs\pp\la\q}$
%$\mapWh{\CommAs\pp\la\q}{\concat{\os}{\CommAs\pp\la\q}}=\os$
and
 $\mapWh{\CommAsI\pp\la\q}\os=\concat{\CommAs\pp\la\q}{\os}$.
\end{enumerate}
  \end{definition}
  \noindent
   It is easy to verify that if $\mapBl{\beta}{\os}$ is defined, then  $\mapWh{\beta}{\mapBl{\beta}{\os}}\cong\os$, and  if $\mapWh{\beta}{\os}$ is defined, then $\mapBl{\beta}{\mapWh{\beta}{\os}}\cong\os$.

\medskip
    We can show that the operators $\blacktriangleright$ and $\triangleright$ applied to a communication $\beta$ modify the queues in the same way as the (forward or backward) execution of $\beta$ would do in the underlying network.

 \begin{lemma}\mylabel{defpost}
 If $\Nt\parallel \Msg\stackred\beta\Nt'\parallel \Msg'$, then $\mapBl{\beta}\osq\Msg\cong\osq{\Msg'}$ and $\mapWh{\beta}\osq{\Msg'}\cong\osq{\Msg}$.
\end{lemma}

\begin{proof}
From $\Nt\parallel
  \Msg\stackred\beta\Nt'\parallel \Msg'$ we get $\Msg'\equiv
  \addMsg{\Msg}{\mq\pp{\la}\q}$  if $\beta=\CommAs{\pp}{\la}{\q}$
   and $\Msg\equiv
  \addMsg{\mq\pp{\la}\q}{\Msg'}$  if
  $\beta=\CommAsI{\pp}{\la}{\q}$.
 In the first case, we have $\osq{\Msg'} \cong
\concat{\osq{\Msg}}{\CommAs{\pp}{\la}{\q}} \cong \mapBl{\beta}\osq\Msg$,
whence also $\mapWh{\beta}{\osq{\Msg'}} \cong
\mapWh{\beta}{\mapBl{\beta}\osq\Msg} \cong \osq\Msg$.
In the second case, we have $\osq{\Msg} \cong
\concat {\CommAs{\pp}{\la}{\q}}{\osq{\Msg'}} \cong
\mapWh{\beta}\osq{\Msg'}$,
whence also $\mapBl{\beta}{\osq{\Msg}} \cong
\mapBl{\beta}{\mapWh{\beta}\osq{\Msg'}} \cong \osq{\Msg'}$.
\end{proof}

  The residual and retrieval operators preserve the  $\os$-flow
   and conflict relations.  For the flow relation the  parametrising  o-traces
  are obtained by the previously defined mappings.

\begin{lemma}\mylabel{prop:prePostNetArel}\hbox{}\vspace*{-5mm}

\begin{enumerate}
\itemsep=0.9pt
\item \mylabel{ppn2A}  If  $\netevA\precN \netevA'$ and $\postA{\netevA}{\beta}$ and
  $\postA{\netevA'}{\beta}$ and $\mapBl{\beta}\os$ are defined,\\ then
  $\postA{\netevA}{\beta}\precNL{\mapBl{\beta}\os} \postA{\netevA'}{\beta}$.
\item  \mylabel{ppn2bA}  If  $\netevA\precNL{\os} \netevA'$ and $\mapWh{\beta}\os$ is defined,
then $\preA{\netevA}{\beta}\precNL{\mapWh{\beta}\os} \preA{\netevA'}{\beta}$.
\item \mylabel{ppn3A} If  $\netevA\grr
  \netevA'$ and both $\postA{\netevA}{\beta}$ and $\postA{\netevA'}{\beta}$
  are defined, then $\postA{\netevA}{\beta}\grr\postA{\netevA'}{\beta}$.
\item \mylabel{ppn3bA}  If  $\netevA\grr \netevA'$,  then $\preA{\netevA}{\beta}\grr
  \preA{\netevA'}{\beta}$.
 \end{enumerate}
\end{lemma}

We now define the total function $\sf nec$, which yields sequences of
n-events starting from
%pairs of o-traces and traces.
 a trace.
%We
 The definition makes use of  the projection given in \refToDef{pep}(\ref{pep1}).\sm

\begin{definition}[n-events from \traces]
  \mylabel{necA} %Let $\comseqA\not=\ee$ be $\os$-well formed.
  We define
  the  {\em sequence of n-events corresponding  to
%$\os$ and
the trace $\comseqA$}
   by\sm

  \centerline{$\necA{\comseqA}=\Seq{\netevA_1;\cdots}{\netevA_n}$}
 \noindent  where

\centerline{
$\netevA_i=
\locevA{\pp_i}{\os}{\procev_i}~~~~\text {  if  } ~~~~ \set{\pp_i} = \play{\at\comseqA{i}}$ and $\procev_i =
\projtau{\range\comseqA1{i}}{\pp_i} $
}
\end{definition}

\noindent    It is immediate to see that, if $\comseqA=\CommAs{\pp}{\la}{\q}$ or
  $\comseqA=\CommAsI{\pp}{\la}{\q}$, then
%$\necA{\os,\comseqA}$
   $\necA{\comseqA}$  consists only of the n-event
  $\locevA{\pp}{\os}{\sendL{\q}{\la}}$ or of the n-event
  $\locevA{\q}{\os}{\rcvL{\pp}{\la}}$, respectively, because
  $\range\comseqA1{1} = \at{\comseqA}{1}$.

\medskip
We show now that two n-events
%from the same pair of o-\traceS and \traceS are not in conflict.
appearing in the sequence generated from a given  trace $\comseqA$
%pair $(\os$, $\comseqA)$
cannot be in conflict. Moreover, from
%$\necA{\os,\comseqA}$
 $\necA{\comseqA}$
we can recover $\comseqA$ by means of the
function $\sf i/o$ of \refToDef{proceventNA}(\ref{proceventNA2}).

\begin{lemma}\mylabel{eion} Let $ \necA{\comseqA}
  =\Seq{\netevA_1;\cdots}{\netevA_n}$.
\begin{enumerate}
\itemsep=0.9pt
\item\mylabel{eion1}  If $1\leq k,l\leq n$, then $\neg (\netevA_k \grr \netevA_l)$;\vspace{-0.5mm}
\item\mylabel{eion2} $\at\comseqA{i} = \io{\netevA_i}\,$  for all $i$,  $1\leq i\leq n$.
\end{enumerate}
  \end{lemma}

  \begin{proof} (\ref{eion1})  Let $\netevA_i=\locevA{\pp_i}{\os}{\procev_i}$ for all $i$,  $1\leq i\leq n$. If  $\pp_k\not=\pp_l$,  then  $\netevA_k$ and  $\netevA_l$ cannot be  in conflict. If  $\pp_k=\pp_l$,  then by~\refToDef{necA} either  $\procev_k  < \procev_l$ or $\procev_k  <  \procev_l$.  So in all cases  we have $\neg (\netevA_k \grr \netevA_l)$.\sm

  (\ref{eion2}) Immediate from~\refToDef{necA}.
  \end{proof}
  The following lemma relates the operators $\blacklozenge$ and  $\lozenge$ with the mapping $\sf
  nec$. This will be handy for the proof of
  \refToTheorem{uf10A}.

    \begin{lemma}\label{knec-ila}
  \begin{enumerate}
   \item\label{knec-ila2}
   Let $\comseqA=\concat\beta{\comseqA'}$. If
   $\necA{\comseqA}=\Seq{\netevA_1;\cdots}{\netevA_n}$ and
   $\necA{\comseqA'}=\Seq{\netevA'_2;\cdots}{\netevA'_n}$, then
   $\postA{\netevA_i}\beta={\netevA'_i}\,$ for  all $i$,   $2\leq i\leq n$.
  \item\label{knec-ila1}   Let $\comseqA=\concat\beta{\comseqA'}$. If
   $\necA{\comseqA}=\Seq{\netevA_1;\cdots}{\netevA_n}$ and
   $\necA{\comseqA'}=\Seq{\netevA'_2;\cdots}{\netevA'_n}$, then
   $\preA{\netevA'_i}\beta={\netevA_i}\,$ for  all $i$, $2\leq i\leq n$.
   \end{enumerate}
    \end{lemma}
\eject
    \begin{proof}
    (\ref{knec-ila2})   Note that $\at\comseqA{i} = \at{\comseqA'}{i-1}$ for all $i$, $2\leq
  i\leq n$.   Then we can assume  $\netevA_i=\locevA{\pp_i}{\os}{\procev_i}$ for all $i$, $1\leq i\leq n$
  and $\netevA'_i= \locevA{\pp_i}{\os'}{\procev'_i}$ for all $i$, $2\leq i\leq
  n$. By~\refToDef{necA} $\procev_i =
  \projtau{\range{\comseqA}1{i}}{\pp_i} =
  \projtau{(\concat\beta{\range{\comseqA'}1{i -1}})}{\pp_i} $ for  all $i$,
  $1\leq i\leq n$ and $\procev'_i =
  \projtau{\range{\comseqA'}1{i-1}}{\pp_i} $ for all $i$, $2\leq i\leq n$. Then
   we get $\postA{\netevA_i}\beta  =
  \postA{(\locevA{\pp_i}{(\mapWh{\beta}{\os'})}{\concat{\projtau{\beta}{\pp_i}}{\procev'_i}})}\beta
  = \locevA{\pp_i}{\os}{\procev'_i} = {\netevA'_i}$    for all $i$, $2\leq
  i\leq n$.

  \medskip (\ref{knec-ila1}) From Point %Fact
  (\ref{knec-ila2}) and
\refToLemma{prop:prePostNetA}(\ref{ppn2A}).
\end{proof}

\noindent We end this subsection with the two theorems for networks discussed at
the beginning of the whole section.  We first show that two n-events
which $\os$-justify the same n-event of the same network must be in
conflict.
\begin{lemma}\label{gr}
  Let $\netevA, \netevA_1,\netevA_2\in\GEA(\Nt\parN\Msg)$,
   $\os=\osq\Msg$  and
  $\netevA_i\precN \netevA\,$  for $i\in\set{1,2}$,  where
   each  $\netevA_i\precN \netevA\,$ is
  derived  by Clause (\ref{c1AbP1}) of
  \refToDef{netaevent-relations}. Then $\netevA_1\gr\netevA_2$.
\end{lemma}

\begin{proof}
  Clause (\ref{c1AbP1}) of \refToDef{netaevent-relations} prescribes
  $\netevA=\locevA\q{\os}{\concat{\actseq}{\rcvL\pp\la}}$,
  $\netevA_i=\locevA\pp{\os}{\concat{\actseq_i}{\sendL\q\la}}$ and \vspace*{1.6mm}

  \centerline{(*)
    $\dualprecsim{\projs{(\concat{\projAP{\os}{\pp}}{{\actseq_i}})}{\q}}{\projs{(\concat{\projAP{\os}{\q}}{{\actseq'_i}})}{\pp}}$}
  \centerline{(**) $\projs{(\concat{\actseq}{\rcvL\pp\la})}
    \pp\precapprox
    \projs{(\concat{\concat{\actseq'_i}{\rcvL\pp\la}}{\preEv_i})}\pp$} \vspace*{1.6mm}

  \noindent for some $\actseq'_i$ and $\preEv_i$, where $i\in\set{1,2}$. Let $n$
  be the number of occurrences of $\rcvL\pp\la$ in $\actseq$, $n_i$ be
  the number of occurrences of $\sendL\q\la$ in $\actseq_i$ and $n_i'$
  be the number of occurrences of $\rcvL\pp\la$ in $\actseq_i'$. From
  (*) we get $n_i=n'_i$ and from (**) we get $n=n_i'$ for
  $i\in\set{1,2}$. Then $n_i=n_j$ for $\set{i,j}=\set{1,2}$. Assume
  ad absurdum $\netevA_i\precN \netevA_j$  for some $i, j \in
  \set{1,2}, \, i \neq j$. Then $\netevA_i\precN \netevA_j$ is
  derived  by Clause (\ref{c1Aa}) of
  \refToDef{netaevent-relations}, thus $\concat{\actseq_i}{\sendL\q\la} \sqsubset
\concat{\actseq_j}{\sendL\q\la}$, that is
$\concat{\actseq_i}{\sendL\q\la} \sqsubseteq \actseq_j$.
This means that $\actseq_j$ contains at least one more occurrence of $\sendL\q\la$
than $\actseq_i$, namely $n_i < n_j$,
%$n_i+1\geq n_j$ for $\set{i,j}=\set{1,2}$
which is
%impossible.
 a contradiction.
\end{proof}

\begin{theorem}\mylabel{uf10A}
If $\Nt\parallel \Msg\stackred\comseqA\Nt'\parallel \Msg'$,
then  $\necA{\osq \Msg,\comseqA}$ is a proving sequence in  the event structure   $\ESNA{\Nt\parallel \Msg}$.
\end{theorem}

\begin{proof} %Let $\os=\osq \Msg$. The proof is by  induction on ${\comseqA}$.\\
  The proof is by  induction on ${\comseqA}$. Let $\os=\osq \Msg$.

 \vspace*{1.4mm} \noindent {\it Case $\comseqA=\beta$.}  Assume first that
  $\beta=\CommAs{\pp}{\la}{\q}$.  From
  $\Nt\parallel\Msg\stackred\beta\Nt'\parallel\Msg'$ we get
  $\pP{\pp}{\oup\q{i}{I}{\la}{\PP}}\in\Nt$ with $\la=\la_k$ for some
  $k\in I$. Thus $\pP{\pp}{\PP_k}\in\Nt'$ and $\Msg'\equiv
  \addMsg{\Msg}{\mq\pp{\la}\q}$.  By \refToDef{esp}(\ref{ila-esp1})
  $\sendL{\q}{\la}\in \ES(\oup\q{i}{I}{\la}{\PP})$. By
  \refToDef{netev-relationsA}(\ref{netev-relations1A})
  $\locevA{\pp}{\os}{\sendL{\q}{\la}}\in\GEA(\Nt\parallel\Msg)$.  By
  Definition \ref{necA} $\,\necA{\beta}= \netevA_1 =
  \locevA{\pp}{\os}{\sendL{\q}{\la}}$. Clearly, $\netevA_1$ is a
  proving sequence in $\ESNA{\Nt\parallel\Msg}$, since $\netevA \precN
  \netevA_1$ would imply $\netevA = \locevA{\pp}{\os}{\procev}$ for
  some $\procev$ such that $\procev < \sendL{\q}{\la}$, which is not
  possible.

  \vspace*{1.6mm}
  Assume now that $\beta=\CommAsI{\pp}{\la}{\q}$. In this case we get
  $\pP{\q}{\inp\pp{i}{I}{\la}{\Q}}\in\Nt$ with $\la=\la_k$ for some
  $k\in I$. Thus $\pP{\q}{\Q_k}\in\Nt'$ and $\Msg \equiv
  \prefMsg{\Msg'}{\mq\pp{\la}\q}$. With a similar reasoning as in the
  previous case, we obtain $\,\necA{\beta}= \netevA_1 =
  \locevA{\q}{\os}{\rcvL{\pp}{\la}}$.  Since $\os \cong
  \concat{\CommAs{\pp}{\la}{\q}}{\os'}$, where $\os' = \osq{\Msg'}$,
  it is immediate to see that $\netevA_1$ is $\os$-queue-justified.
  As in the previous case,
  there is no event $\netevA$ in $\GEA(\Nt\parallel\Msg)$ such that $\netevA\precN\netevA_1$, and thus $\netevA_1$ is a proving sequence in $\ESNA{\Nt\parallel\Msg}$.

  \vspace{1.5mm}\noindent
  {\it Case $\comseqA=\concat{\beta}{\comseqA'}$ with
    $\comseqA'\not=\ee$.}
  In this case, from $\Nt\parallel\Msg\stackred{\comseqA}\Nt'\parallel\Msg'$ we get  \vspace*{1.6mm}

  \centerline{$\Nt\parallel\Msg\stackred{\beta} \Nt''\parallel
    \Msg''\stackred{\comseqA'}\Nt'\parallel\Msg'$}

  \medskip \noindent for some   $\Nt'',\Msg''$. Let $\os'\!=\osq{\Msg''}$.  By \refToLemma{defpost}
  $\os=\mapWh\beta{\os'}$.  Let $\necA{\comseqA}
  \!=\Seq{\netevA_1;\cdots}{\netevA_{n}}$  and $\;\necA{\comseqA'}
  \!=\!\Seq{\netevA'_2;\cdots }{\netevA'_{n}}$.$\,$By   induction
  $\necA{\comseqA'}$ is a proving sequence in $\ESNA{\Nt''\parallel
    \Msg''}$.  By$\,$\refToLemma{knec-ila}(\ref{knec-ila1})
    \eject

 \noindent $\,\preA{\netevA'_j}{\beta}=\netevA_j$ for all $j$, $2\leq j \leq
  n$.  We show that $\netevA_j\in\GEA(\Nt\parallel\Msg)$ for all $j$,
  $2\leq j\leq n$.  Let ad absurdum $k$ $(2\leq k\leq n)$ be the
  minimum index such that $\netevA_k\not\in\GEA(\Nt\parallel\Msg)$. By
  Fact \ref{f}, $\netevA_k$ should be an input which is not
  $\os$-queue justified and $\preA{\netevA'}{\beta}$ should be
  undefined for all $\netevA'$ $\os'$-justifying $\netevA'_k$. Since
  $\netevA_k'\in\GEA(\Nt''\parallel\Msg'')$, either $\netevA'_k$ is
  $\os'$- queue-justified or $\netevA'_k$ is $\os'$-justified by some
  output, which must be an event $\netevA'_l$ for some $l<k$, $2\leq
  l\leq n$ given that $\Seq{\netevA'_2;\cdots }{\netevA'_{n}}$ is a
  proving sequence.  In the first case $\netevA_k$ is $\os$-
  queue-justified. In the second case, we get
  $\preA{\netevA'_l}\beta\in\GEA(\Nt\parallel\Msg)$ since $l<k$.  So,
  in both cases we reach a contradiction.  Finally, from the proof of
  the base case we know that $\netevA_1 =
  \locevA{\pp}{\os}{\projtau{\beta}{\pp}} \in\GEA(\Nt\parallel\Msg)$
  where $\set\pp=\play\beta$.

 \medskip    What is left to show is that $\Seq{\netevA_1;\cdots}{\netevA_n}$ is
  a proving sequence in $\ESNA{\Nt\parallel\Msg}$.  By
  \refToLemma{eion}(\ref{eion1}) no two events in this sequence can be
  in conflict.  Let $\netevA\in\GEA(\Nt\parallel\Msg)$ and
  $\netevA\precN \netevA_h$ for some $h$, $1\leq h\leq n$.  As argued
  in the base case, this implies $h > 1$.  We distinguish two cases,
  depending on
  whether $\postA\netevA\beta$ is defined or not.

\vspace*{1.8mm}
   If $\postA\netevA\beta$ is defined, let $\netevA' =
  \postA{\netevA}{\beta}$.   \sm
  \\
  If  $\netevA' \in\GEA(\Nt''\parallel \Msg'')$, then  by
  \refToLemma{prop:prePostNetArel}(\ref{ppn2A}) we have $\netevA'
  \precNL {\os'}{\postA{\netevA_h}{\beta}}$. By
  \refToLemma{knec-ila}(\ref{knec-ila2})
  $\,\postA{\netevA_j}{\beta}=\netevA'_j$ for all $j$, $2\leq j \leq
  n$.  Thus we have $\netevA'\precNL{\os'}{\netevA'_h}$.
  Since $\necA{\comseqA'}$ is a proving sequence in
  $\ESNA{\Nt''\parallel \Msg''}$, by \refToDef{provseq} there is $l<h$
  such that either $\netevA' =\netevA'_l$ or $\netevA'
  \grr\netevA'_l\prec\netevA'_h$.  In the first case we have $\netevA=
  \preA{\netevA'}{\beta} =\preA{\netevA'_l}{\beta}=\netevA_l$.  In the
  second case, from $\netevA' \grr\netevA'_l$ we deduce $\netevA
  \grr\netevA_l$ by \refToLemma{prop:prePostNetArel}(\ref{ppn3bA}),
  and from $\netevA'_l\prec^{\os'}\netevA_h'$ we deduce $\netevA_l
  \precN
  \netevA_h$ by \refToLemma{prop:prePostNetArel}(\ref{ppn2bA}).

  \sm
   If  instead  % $\netevA' $ is defined but
  $\netevA' \not\in\GEA(\Nt''\parallel \Msg'')$, we distinguish two
  cases according to  whether  $\netevA \precN \netevA_h$
   is deduced  by Clause (\ref{c1Aa}) or by Clause
  (\ref{c1AbP1}) of \refToDef{netaevent-relations}.  If $\netevA
  \precN\netevA_h$ by Clause (\ref{c1Aa}) of
  \refToDef{netaevent-relations}, then $\netevA'
  \prec^{\os'}\netevA'_h$ again by Clause (\ref{c1Aa}) of
  \refToDef{netaevent-relations} as proved in
  \refToLemma{prop:prePostNetArel}(\ref{ppn2A}). Then $\netevA'
  \not\in\GEA(\Nt''\parallel \Msg'')$ implies
  $\netevA'_h\not\in\GEA(\Nt''\parallel \Msg'')$ by narrowing, so this
  case is impossible.  If $\netevA \precN\netevA_h$ by Clause
  (\ref{c1AbP1}) of \refToDef{netaevent-relations}, then $\netevA_h$
  is an input and $\netevA$ $\os$-justifies $\netevA_h$. Then also
  $\netevA_h'$ is an input and by definition of proving sequence there
  is $\netevA'_k$ for some $k, 2\leq k\leq n$ which $\os'$-justifies
  $\netevA_h'$. Then $\netevA_k$ $\os$-justifies $\netevA_h$ by
  \refToLemma{prop:prePostNetArel}(\ref{ppn2bA}). Since $\netevA$ and
  $\netevA_k$ both $\os$-justify $\netevA_h$ we get  $\netevA
  \grr\netevA_k$ by \refToLemma{gr}.

  \vspace*{1.8mm}
  If $\postA\netevA\beta$ is undefined, then by
  \refToDef{def:PostPre}(\ref{def:Post}) either $\netevA=\netevA_1$ or
  $\netevA=\locevA{\pp}{\os}{\pi\cdot \actseq}$ with
  $\pi\not=\projtau{\beta}{\pp}$, which implies
  $\netevA\grr\netevA_1$.  In the first case we are done. So, suppose
  $\netevA\grr\netevA_1$. Let $\pi' =\projtau{\beta}{\pp}$.  Since
  $\rho$ and $\rho_1$ are  n-events  in $\GEA(\Nt\parallel\Msg)$, we may
  assume $\pi=\sendL{\q}{\la}$ and $\pi'=\sendL{\q}{\la'}$ and
  therefore $\beta=\CommAs{\pp}{\la'}{\q}$. Indeed, we know that
  $\play{\beta} = \set\pp$, and $\beta$ cannot be an input
  $\CommAsI{\q}{\la'}{\pp}$ since in this case there should be
  $\netevA_0=\locevA{\pp}{\os}{\rcvL{\q}{\la}} \in
  \GEA(\Nt\parallel\Msg)$ by narrowing, and the two input n-events
  $\netevA_0$ and $\netevA_1=\locevA{\pp}{\os}{\rcvL{\q}{\la'}}$ could
  not be both
   $\os$-queue-justified.   %justified by the same queue $\os$.
  Note that $\netevA$
  cannot be a local cause of $ \netevA_h$,  i.e
  $\netevA\precN\netevA_h$ cannot hold by Clause (\ref{c1Aa}) of
  \refToDef{netaevent-relations},  because
  $\netevA_h=\locevA{\pp}{\os}{\pi\cdot\actseq\cdot\procev}$ would
  imply $\netevA_h\grr\netevA_1$, contradicting what said above.
  Therefore $\netevA$ is a cross-cause of $\netevA_h$,  i.e
  $\netevA\precN\netevA_h$ holds by Clause (\ref{c1AbP1}) of
  \refToDef{netaevent-relations},  so $\netevA=
  \locevA{\pp}{\os}{\pi\cdot\actseq'\cdot\sendL{\pr}{\la''}}$ and
  $\netevA_h=\locevA{\pr}{\os}{\actseq''\cdot\rcvL{\pp}{\la''}}$.  We
  know that $\netevA_h=\preA{\netevA'_h}{\beta}$.  By
  \refToDef{def:PostPre}(\ref{def:Pre}) we have
  $\netevA'_h=\locevA{\pr}{(\os\cdot\beta)}{\actseq''\cdot\rcvL{\pp}{\la''}}$,
  because $\pr$ is the receiver of a message sent by $\pp$ and thus by
  construction $\pr \neq \pp$. %$\pr \neq \pp = \play{\beta}$.
  Since $\netevA'_h$ is an input n-event in
  $\GEA(\Nt''\parallel\Msg'')$, it must either be justified by the
  queue $\os\cdot\beta$ or have a cross-cause in
  $\GEA(\Nt''\parallel\Msg'')$. Since $\netevA_h$ is not  $\os$-queue-justified
  (because $\netevA\precN \netevA_h$), the only way
  for $\netevA'_h$ to be $\os\cdot\beta$-queue-justified   would be that
  $\CommAs{\pp}{\la''}{\pr} = \beta$, that is $\pr = \q$ and $\la'' =
  \la'$, and that
  $(*)~~\dualprecsim{(\projs{\projAP{\os}\pp)}\q}{\projs{\actseq_0}\pp}$
  and
  $\projs{(\actseq''\cdot\rcvL{\pp}{\la'})}{\pp}\precapprox\projs{(\actseq_0\cdot\rcvL{\pp}{\la'}\cdot\preEv)}{\pp}$
  for some $\actseq_0$ and $\preEv$, see \refToDef{qgne}.  This means
  that $\projs{\actseq_0}{\pp}$ is the subsequence of
  $\projs{\actseq''}{\pp}$ obtained by keeping all and only its
  inputs.
%contains exactly all the inputs in $\projs{\actseq''}{\pp}$.
Now, if
$\netevA'_h=\locevA{\q}{(\os\cdot\beta)}{\actseq''\cdot\rcvL{\pp}{\la'}}$,
then $\netevA_h=\locevA{\q}{\os}{\actseq''\cdot\rcvL{\pp}{\la'}}$.
Since
$\netevA=\locevA{\pp}{\os}{\sendL{\q}{\la}\cdot\actseq'\cdot\sendL{\q}{\la'}}$
is a cross-cause of $\netevA_h$, we have
$(**)~\dualprecsim{\projs{(\concat{\projAP{\os}{\pp}}{\sendL{\q}{\la}\cdot{\actseq'}})}{\q}}{\projs{(\concat{\projAP{\os}{\q}}{{\actseq_1\cdot\preEv'}})}{\pp}}$
and $\projs{(\actseq''\cdot\rcvL{\pp}{\la'})}{\pp}\precapprox
\projs{(\actseq_1\cdot\rcvL{\pp}{\la'}\cdot\preEv')}{\pp}$ for some
$\actseq_1$ and $\preEv'$, see Clause (\ref{c1AbP1}) of
\refToDef{netaevent-relations}.
It follows that the inputs in $\projs{\actseq_1}{\pp}$
coincide with the inputs in $\projs{\actseq''}{\pp}$ and thus with
those in $\projs{\actseq_0}{\pp}$. From (*) we know that all inputs in
$\projs{\actseq_0}{\pp}$ match some output in
$\projs{(\projAP{\os}\pp)}\q$. Therefore no input in
$\projs{(\concat{\projAP{\os}{\q}}{{\actseq_1\cdot\preEv'}})}{\pp}$
can match the output $\sendL{\q}{\la}$ in
$\projs{(\concat{\projAP{\os}{\pp}}{\sendL{\q}{\la}\cdot{\actseq'}})}{\q}$,
contradicting (**).
Hence $\netevA'_h$ must have a cross-cause in
$\GEA(\Nt''\parallel\Msg'')$.  Let $\netevA'$ be such a cross-cause.
Then $\netevA'=
\locevA{\pp}{(\os\cdot\beta)}{\actseq_2\cdot\sendL{\pr}{ \la''}}$ for some $\actseq_2$.
Since $\necA{\comseqA'}$ is a proving sequence in
$\ESNA{\Nt''\parallel \Msg''}$, by \refToDef{provseq} there is $l<h$
such that either $\netevA' =\netevA'_l$ or $\netevA'
\grr\netevA'_l\prec\netevA'_h$. In the first case
$\preA{\netevA'}{\beta}
=\preA{\netevA'_l}{\beta}=\netevA_l\prec\netevA_h$,  and
$(\preA{\netevA'}{\beta})\grr\netevA$ because $\preA{\netevA'}{\beta}
= \locevA{\pp}{\os}{\pi' \cdot
  \actseq_2\cdot\sendL{\pr}{ \la''}}$.
In the second case,  let $\netevA'_l=\locevA{\pp}{(\os\cdot\beta)}{\procev}$ for some $\procev$. From $\netevA'
\grr\netevA'_l\prec\netevA'_h$ we derive
$\preA{\netevA'}{\beta} \grr \preA{\netevA'_l}{\beta} \prec \preA{\netevA'_h}{\beta}$
by \refToLemma{prop:prePostNetArel}(\ref{ppn3bA}) and (\ref{ppn2bA}).
This implies $\netevA_l = \preA{\netevA'_l}{\beta} = \locevA{\pp}{\os}{\pi'
  \cdot \procev}$. Hence
$\netevA \grr \netevA_l \prec \netevA_h$.
  \end{proof}

  \begin{theorem}\mylabel{uf12A}
If  $\Seq{\netevA_1;\cdots}{\netevA_n}$ is a proving sequence in $\ESNA{\Nt\parallel\Msg}$, then $\Nt\parallel\Msg\stackred\comseqA\Nt'\parallel\Msg'$ where $\comseqA=\io{\netevA_1}\cdots\io{\netevA_n}$.
\end{theorem}

\begin{proof}
The proof is by induction on the length $n$ of the proving sequence. Let $\os=\osq\Msg$.

\medskip\noindent  {\it Case $n=1$.}  Let $\io{\netevA_1}=\beta$ where
 $\beta=\CommAsI{\pp}{\la}{\q}$. The proof for
 $\beta=\CommAs{\pp}{\la}{\q}$ is similar and simpler.  By
 \refToDef{netev-relationsA}(\ref{netev-relations1A})
 $\netevA_1=\locevA\q\os{\concat\actseq{\rcvL\pp\la}}$.
% where $\os=\osq\Msg$.
Note that it must be
 $\actseq=\ee$, since otherwise we would have
$\locevA\q\os{\actseq} \in \GEA(\Nt\parallel\Msg)$
by narrowing,  where
%and thus
$\locevA\q\os{\actseq} \precN \netevA_1$ by
\refToDef{netaevent-relations}(\ref{c1A})(a), contradicting the
hypothesis that $\netevA_1$ is minimal.  Moreover, $\netevA_1$ cannot
be $\os$-justified by an  output n-event  $\netevA \in
\GEA(\Nt\parallel\Msg)$, because  this would imply  $\netevA
\precN \netevA_1$,
%would again
contradicting again  the minimality of $\netevA_1$. Hence,
by \refToDef{netev-relationsA}(\ref{netev-relations1A}) $\netevA_1
 = \locevA\q\os{\rcvL\pp\la}$ must be $\os$-queue-justified, which
means that $\os \cong \concat{\CommAs\pp\la\q}{\os'}$. Thus $\Msg
\equiv \addMsg{\mq\pp{\la}\q}{\Msg'}$,  where $\osq{\Msg'}=
\os'$.  By  \refToDef{esp}(\ref{ila-esp1})  and
\refToDef{netev-relationsA}(\ref{netev-relations1A}) we have
$\Nt\equiv\pP{\q}{\inp\pp{i}{I}{\la}{\Q}}\parN\Nt_0$ where $\la_k =
\la$ for some $k\in I$.  We may then conclude that
$\Nt\parallel\Msg\stackred\beta \pP{\q}{\Q_k} \parN
\Nt_0\parallel\Msg' = \Nt'\parallel\Msg' $.

\smallskip\noindent
{\it Case $n>1$.}
 Let $\io{\netevA_1}=\beta$ and $\Nt\parallel\Msg\stackred\beta \Nt''\parallel\Msg''$
be the corresponding transition as obtained from the base case. Let  $\os'=\osq{\Msg''}$. By \refToLemma{defpost} $\os'=\mapBl{\beta}\os$.
We show that $\postA{\netevA_j}{\beta}$ is defined for all $j$,  $2\leq
j\leq n$. If $\postA{\netevA_k}\beta$ were undefined
%would be undefined
for some $k$, $2\leq k\leq n$, then by
\refToDef{def:PostPre}(\ref{def:Post}) either $\netevA_k=\netevA_1$ or
$\netevA_k=\locevA{\pp}{\os}{\pi\cdot\zeta}$ where
$\set\pp=\play\beta$ and $\pi\not=\projtau{\beta}{\pp}$, which implies
$\netevA_k\grr\netevA_1$. So both cases are impossible.   Thus we
may define $\netevA'_j = \postA{\netevA_j}{\beta}$ for all $j$, $2\leq
j\leq n$. We show that $\netevA'_j\in\GEA(\Nt''\parallel\Msg'')$ for
all $j$, $2\leq j\leq n$.  Let ad absurdum $k$ $(2\leq k\leq n)$ be
the minimum index such that
$\netevA'_k\not\in\GEA(\Nt''\parallel\Msg'')$. By Fact~\ref{f},
$\netevA'_k$ should be an input  which is  not $\os'$-queue
justified and $\postA{\netevA'}{\beta}$ should be undefined  for
 all $\netevA'$ $\os$-justifying  $\netevA_k'$.
 Since $\netevA_k\in\GEA(\Nt\parallel\Msg)$, either $\netevA_k$ is
 $\os$-queue justified or $\netevA_k$ is $\os$-justified by some
 output, which must be an event $\netevA_l$ for some $l<k$, $2\leq
 l\leq n$ given that $\netevA_1, \ldots, \netevA_n$ is a proving
 sequence.  In the first case $\netevA_k'$ is $\os'$-queue
 justified. In the second case we get
 $\postA{\netevA_l}\beta\in\GEA(\Nt''\parallel\Msg'')$ since $l<k$.
 So in both cases we reach a contradiction.

 \medskip
 We show that $\netevA'_2 ;\cdots ;\netevA'_n$ is a proving sequence
 in $\ESNA{\Nt''\parallel\Msg''}$. By
 \refToLemma{prop:prePostNetA}(\ref{ppn1A}) $\netevA_j =
 \preA{\netevA_j'}{\beta}$ for all $j$, $2\leq j\leq n$. Then by
 \refToLemma{prop:prePostNetArel}(\ref{ppn3bA}) no two n-events in the
 sequence $\netevA'_2 ;\cdots ;\netevA'_n$ can be in
 conflict.

  \vspace*{1.6mm}
 Let $\netevA\in\GEA(\Nt''\parallel\Msg'')$ and $\netevA \prec^{\os'}
 \netevA'_h$ for some $h$, $2\leq h\leq n$.  Let $\netevA'=
 \preA\netevA\beta$.  By
 \refToLemma{prop:prePostNetArel}(\ref{ppn2bA})
 $\preA\netevA\beta\precN \preA{\netevA'_h}{\beta} =\netevA_h$.
 Therefore $\netevA' \precN\netevA_h$. If
 $\netevA'\in\GEA(\Nt\parallel\Msg)$,
since $\Seq{\netevA_1;\cdots}{\netevA_n}$ is a proving sequence in
$\ESNA{\Nt\parallel\Msg}$, by \refToDef{provseq} there is $l<h$ such
that either $\netevA'=\netevA_l$ or
$\netevA'\grr\netevA_l\prec\netevA_h$. In the first case,
%since  $\netevA'= \preA\netevA\beta$
% and  $\netevA_l = \preA{\netevA'_l}{\beta}$,
by \refToLemma{prop:prePostNetA}(\ref{ppn1bA}) we get
$\netevA=\postA{\netevA'}{\beta}=\postA{\netevA_l}{\beta} =
\netevA'_l$. In the second case, by
\refToLemma{prop:prePostNetArel}(\ref{ppn2A}) and (\ref{ppn3A}) we get
$\netevA \grr \netevA'_l \precN \netevA'_h$.  If
$\netevA'\not\in\GEA(\Nt\parallel\Msg)$ we distinguish two cases
according to  whether  $\netevA \prec^{\os'} \netevA_h'$
 is deduced  by
Clause (\ref{c1Aa}) or Clause (\ref{c1AbP1}) of
\refToDef{netaevent-relations}.  If $\netevA \prec^{\os'}\netevA_h'$
by Clause (\ref{c1Aa}) of \refToDef{netaevent-relations}, then
$\netevA'\prec^{\os}\netevA_h$ again by Clause (\ref{c1Aa}) of
\refToDef{netaevent-relations} as proved in
\refToLemma{prop:prePostNetArel}(\ref{ppn2A}). Then
$\netevA\not\in\GEA(\Nt\parallel \Msg)$ implies
$\netevA_h\not\in\GEA(\Nt\parallel \Msg)$ by narrowing, so this case
is impossible.  If $\netevA \prec^{\os'}\netevA_h'$ by Clause
(\ref{c1AbP1}) of \refToDef{netaevent-relations}, then $\netevA_h'$ is
an input and $\netevA$ $\os'$-justifies $\netevA_h'$. Then also
$\netevA_h$ is an input and by definition of proving sequence there is
$\netevA_k$ for some $k< h, 2\leq k\leq n$ which $\os$-justifies
$\netevA_h$. Then $\netevA_k'$ $\os'$-justifies $\netevA_h'$ by
\refToLemma{prop:prePostNetArel}(\ref{ppn2bA}). Since $\netevA$
and $\netevA_k'$ both $\os'$-justify $\netevA_h'$ we get
$\netevA \grr\netevA_k'$ by \refToLemma{gr}.

\medskip
We have shown that $\netevA'_2 ;\cdots ;\netevA'_n$ is a proving
sequence in the event structure   $\ESNA{\Nt''\parallel\Msg''}$. By induction
$\Nt''\parallel\Msg''\stackred{\comseqA'}\Nt'\parallel\Msg'$ where
$\comseqA'=\io{\netevA'_2}\cdots\io{\netevA'_n}$.  Since
$\io{\netevA'_j} = \io{\netevA_j}$ for all $j, 2\leq j\leq n$, we have
$\comseqA = \concat{\beta}{\comseqA'}$.
%We may then conclude that
Hence  $\Nt\parallel\Msg\stackred{\beta}
\Nt''\parallel\Msg''\stackred{\comseqA'}\Nt'\parallel\Msg'$ is the
required  transition sequence.
\end{proof}

 \begin{remark}\label{r} We can show that if $\Nt\parallel
  \Msg\stackred\beta\Nt'\parallel \Msg'$ and
  $\netevA\in\GEA(\Nt'\parallel \Msg')$, then  we get
  $\preA\netevA\beta\in\GEA(\Nt\parallel \Msg)$.
%The addition of this lemma would
 The use of this property  would simplify the proof of \refToTheorem{uf12A}, since we
  would avoid to consider the case
  $\preA\netevA\beta\not\in\GEA(\Nt\parallel \Msg)$. Instead,
  the fact that
  $\Nt\parallel \Msg\stackred\beta\Nt'\parallel \Msg'$ and
  $\netevA\in\GEA(\Nt\parallel \Msg)$ and $\postA\netevA\beta$
  is  defined
   does  not imply $\postA\netevA\beta\in\GEA(\Nt'\parallel \Msg')$. An
  example is \vspace*{1.6mm}

  \centerline{$\begin{array}{l}
\pP\pp{\q?\la}\parN\pP\q{\pr!\la_1;\pp!\la\oplus\pr!\la_2}\parN\pP\pr{\q?\la_1+\q?\la_2}
\parN\emptyset \stackred{\q\pr!\la_2}\\
\pP\pp{\q?\la}\parN\pP\q\inact\parN\pP\pr{\q?\la_1+\q?\la_2}
\parN\mq{\q}{\la_2}\pr
\end{array}
$} \vspace*{1.6mm}

\noindent with  $\beta = \CommAs{\q}{\la_2}{\pr}$ and $\netevA
=\locevA\pp{}{\q?\la}$.  Our choice is justified both by the
shortening of the whole proofs and by the uniformity between the
proofs of Theorems \ref{uf10A} and \ref{uf12A}.
\end{remark}

 \subsection{Transition sequences of asynchronous types and proving
   sequences of their ESs}

We introduce two operators  $\bullet$ and  $\circ$ for t-events,
which play the same role as the operators $\blacklozenge$ and $\lozenge$
 for n-events.  In defining these operators we
must make sure that, in the resulting t-event $\eqA{\os'}{\comseqA'}$,
the  trace  $\comseqA'$ is $\os'$-pointed, see
 \refToDef{def:glEventA}(\ref{def:glEventA1}) and (\ref{def:glEventA2}).

Let us start with the formal definition, and then we shall explain it in detail.\vspace*{-2mm}

 \begin{definition}[Residual  and retrieval of  a t-event with respect to a
 communication]\mylabel{def:PostPreGlA}\hbox{}\hfill \vspace*{-6mm}
   \begin{enumerate}
  \item\mylabel{def:PostPreGl1A} The {\em residual of a t-event $\eqA\os\comseqA$ after a communication $\beta$}  is defined by:

\centerline{$
\postGA{\eqA\os\comseqA}\beta=\begin{cases}
\eqA{\mapBl{\beta}{\os}}{\comseqA'}      & \text{if }  \comseqA\approx_\os\concat{\beta}{\comseqA'}
\text { with }\comseqA'\not=\ee\\
\eqA{\mapBl{\beta}{\os}}{\comseqA}      & \text{if }  \play\beta\not\subseteq\play{\comseqA}
\end{cases}$}\vspace{-1mm}

\item\mylabel{def:PreGl1A} The {\em retrieval of a t-event $\eqA\os\comseqA$ before a
  communication $\beta$} is defined by:

 \centerline{$
\cauA{\beta}{\eqA{\os}\comseqA}=\begin{cases}
\eqA{\mapWh{\beta}{\os}}{\beta\cdot\comseqA}     & \text{if }  \beta\cdot\comseqA
\text { is $\mapWh{\beta}{\os}$-pointed }\\
\eqA{\mapWh{\beta}{\os}}{\comseqA}      &\text{if } \play\beta\not\subseteq\play{\comseqA}
\end{cases}$
}
  \end{enumerate}
\end{definition}
Note that the operators $\bullet$ and $\circ$ preserve the
communication of t-events, namely $
\io{\postGA{\delta}{\beta}} =\linebreak \io{\preGA{\delta}{\beta}} = \io{\delta}$,
and transform the o-trace using the operators $\blacktriangleright$ and $\triangleright$,
see
\refToDef{def:pQueue}.  %(\ref{def:pQueue1}) and (\ref{def:pQueue2}).
We now explain the transformation of the trace $\comseqA$.
%\oi\ trace.

\noindent Consider first the case of $\postGA{\eqA\os\comseqA}{\beta}$. If  the communication $\beta$ can be brought to the head of the
 trace $\comseqA$ using the equivalence $\approx_{\os}$, we obtain the residual of $\eqA\os\comseqA$ after
$\beta$ by
removing the message $\beta$ from the head of the  trace, provided this does
not result in the empty trace (otherwise, the residual is
undefined).  Then, letting $\os' =\mapBl\beta\os$, it is easy to see that the  trace
 $\comseqA'$ is $\os'$-pointed,  since it is a suffix of
$\comseqA = \concat{\beta}{\comseqA'}$ which is $\os$-pointed (see
\refToLemma{suffix-pointedness}). On the other hand, if $\play\beta\not\subseteq\play{\comseqA}$, then
the residual of $\eqA\os\comseqA$ after
$\beta$ is simply obtained by %adding the message $\CommAs\pp\la\q$ at the end of the queue,
leaving the trace unchanged. In this case, letting again $\os' =
\mapBl\beta\os$, the $\os'$-pointedness of $\comseqA$
follows immediately from its $\os$-pointedness.
For instance, consider the t-event
$\eqA{\CommAs\pp{\la'}\pr}{\CommAsI\pp{\la'}\pr}$ where
$\os = \CommAs\pp{\la'}\pr$ and $\comseqA =
\CommAsI\pp{\la'}\pr$. Observe that $\pp$
 occurs in $\comseqA$,  but $\pp
\notin \play{\comseqA}$. Then we have
$\postGA{\eqA{\CommAs\pp{\la'}\pr}{\CommAsI\pp{\la'}\pr}}{\CommAs\pp\la\q}
= \eqA{\concat{\CommAs\pp{\la'}\pr}{\CommAs\pp\la\q}}{\CommAsI\pp{\la'}\pr}$.

\medskip
Next, consider the definition of $\cauA{\beta}{\eqA{\os}\comseqA}$. The resulting  trace
 will be the prefixing of $\comseqA$ by $\beta$ if it is $\mapWh{\beta}{\os}$-pointed. Otherwise the resulting  trace is
 $\comseqA$ if $\play\beta$ is not a player of $\comseqA$.
   For instance, for the t-event
$\eqA{\CommAs\pp\la\q}{\CommAsI\pp{\la}\q}$, where  $\os =\CommAs\pp\la\q$  and
$\comseqA = \CommAsI\pp\la\q$, we have $\pp\notin\play{\comseqA}$, but
$\ct {1}{\concat{\CommAs\pp\la\q}{\CommAsI\pp\la\q }}{\,2}$, thus
$\cauA{\CommAs\pp\la\q}{\eqA{\CommAs\pp\la\q}{\CommAsI\pp{\la}\q}} =
\eqA{\ee}{\concat{\CommAs\pp\la\q}{\CommAsI\pp{\la}\q}}$. On the other
hand, for the t-event
$\eqA{\CommAs\pp\la\q}{\concat{\CommAs\pr{\la'}\ps}{\CommAsI\pr{\la'}\ps}}$,
where  $\os\! =\!\CommAs\pp\la\q$  and $\comseqA\! =\!
\concat{\CommAs\pr{\la'}\ps}{\CommAsI\pr{\la'}\ps}$, we have
$\pp\!\notin\!\play{\comseqA}$ and \mbox{$\neg (\ct
{1}{\concat{\CommAs\pp\la\q}{\concat{\CommAs\pr{\la'}\ps}{\CommAsI\pr{\la'}\ps}}}{2})$}
and \mbox{$\neg (\ct
{1}{\concat{\CommAs\pp\la\q}{\concat{\CommAs\pr{\la'}\ps}{\CommAsI\pr{\la'}\ps}}}{3})$},
so
$\cauA{\CommAs\pp\la\q}{\eqA{\CommAs\pp\la\q}{\concat{\CommAs\pr{\la'}\ps}{\CommAsI\pr{\la'}\ps}}}
= \eqA{\ee}{\concat{\CommAs\pr{\la'}\ps}{\CommAsI\pr{\la'}\ps}}$.

\medskip
\refToLemma{prop:prePostGlA}
%(proved in the Appendix)
is the analogous of \refToLemma{prop:prePostNetA} as regards the first
two statements.  The remaining two statements establish some
commutativity properties of the mappings $\bullet$ and $\circ$ when
applied to two communications with different players. These properties
rely on the corresponding commutativity properties for the mappings
$\mapBl{}{}$ and $\mapWh{}{}$ on o-traces, given in ~\refToLemma{tr}.
% (proved in the Appendix).
Note that
these properties are needed for $\bullet$ and $\circ$ whereas they
were not needed for $\blacklozenge$ and $\lozenge$, because the Rules
\rulename{IComm-Out} and \rulename{IComm-In} of \refToFigure{ltsgtAs}
allow transitions to occur inside \agts, whereas the LTS for networks
only allows transitions for top-level communications.  In fact
Statements (\ref{ppg5}) and (\ref{ppg6A}) of
\refToLemma{prop:prePostGlA} are used in the proof of
\refToLemma{keybeta-bis}.
% We define $\overline{\CommAsI\pp\la\q}=\CommAs\pp\la\q$.

\begin{lemma}\mylabel{tr}
Let $\play{\beta_1}\cap\play{\beta_2}=\emptyset$.
\begin{enumerate}
\item \mylabel{tr2}
If both
  $\mapBl{\beta_2}{\os}$ and $\mapBl{\beta_2}{({\mapWh{\beta_1}{\os}})}$ are defined,
then
$\mapWh{\beta_1}{{(\mapBl{\beta_2}{\os})}}\cong
 \mapBl{\beta_2}{({\mapWh{\beta_1}{\os}})}$.
 \item  \mylabel{tr1} If both $\mapWh{\beta_1}{\os}$ and
  $\mapWh{\beta_2}{\os}$ are defined,
then
$\mapWh{\beta_1}{{(\mapWh{\beta_2}{\os})}}$ is defined and
$\mapWh{\beta_1}{{(\mapWh{\beta_2}{\os})}}\cong
 \mapWh{\beta_2}{({\mapWh{\beta_1}{\os}})}$.
\end{enumerate}
\end{lemma}

\begin{lemma}\mylabel{prop:prePostGlA}
\begin{enumerate}
\item \mylabel{ppg0b}
\mylabel{ppg2A}\mylabel{ppg1}If  $\postGA{\comoccA}{\beta}$ is defined,
then
$\preGA{(\postGA{\comoccA}{\beta})}{\beta}=\comoccA$.
\item \mylabel{ppg0a} \mylabel{ppg1bA} \mylabel{ppg1aA}   If  $\preGA{\comoccA}{\beta}$ is defined,
then
$\postGA{(\preGA{\comoccA}{\beta})}{\beta}=\comoccA$.
\item \mylabel{ppg5} If both
  $\postGA{\comoccA}{\beta_2}$,
  $\postGA{(\preGA{\comoccA}{\beta_1})}{\beta_2}$ are defined, and
  $\play{\beta_1}\cap\play{\beta_2}=\emptyset$, then
$\preGA{(\postGA{\comoccA}{\beta_2})}{\beta_1}=\postGA{(\preGA{\comoccA}{\beta_1})}{\beta_2}$.
\item  \mylabel{ppg6A} If both $\preGA{\comoccA}{\beta_1}$,  $\preGA{\comoccA}{\beta_2}$ are defined,  and  $\play{\beta_1}\cap\play{\beta_2}=\emptyset$,  then $\preGA{(\preGA{\comoccA}{\beta_2})}{\beta_1}$ is defined and
$\preGA{(\preGA{\comoccA}{\beta_2})}{\beta_1}=\preGA{(\preGA{\comoccA}{\beta_1})}{\beta_2}$.
\end{enumerate}
\end{lemma}

\noindent
The next lemma shows that the residual and retrieval operators on
t-events preserve causality and that the retrieval operator preserves
conflict.  It is the analogous of \refToLemma{prop:prePostNetArel},
but without the statement corresponding to
\refToLemma{prop:prePostNetArel}(\ref{ppn3A}),
% Fact (\ref{ppn3A}),
which is true but not required for later results.  The difference is
due to the fact that ESs of networks are FESs, while those of
\agts\ are PESs. This appears clearly when looking at the proof of
\refToTheorem{uf12A} which uses
\refToLemma{prop:prePostNetArel}(\ref{ppn3A}), while %the proof
that of
\refToTheorem{lemma:keyA2} does not need the corresponding property.
% Aggiungere provato in Appendix?

\begin{lemma}\mylabel{prop:prePostGlArel}
\begin{enumerate}
\item \mylabel{ppg4bA} If $\comoccA_1  < \comoccA_2$
 and  both $\postGA{\comoccA_1}{\beta}$, $\postGA{\comoccA_2}{\beta}$ are defined,
 then
  $\postGA{\comoccA_1}{\beta}  <   \postGA{\comoccA_2}{\beta}$.
  \item  \mylabel{ppg4aA} If $\comoccA_1 < \comoccA_2$ and
$\preGA{\comoccA_1}{\beta}$ is defined, then
$\preGA{\comoccA_1}{\beta}  <  \preGA{\comoccA_2}{\beta}$.
 \item  \mylabel{ppg-conflictA-a} If $\comoccA_1\grr\comoccA_2$ and   both $\preGA{\comoccA_1}{\beta}$,  $\preGA{\comoccA_2}{\beta}$ are defined,  then
 $\preGA{\comoccA_1}{\beta}\grr \preGA{\comoccA_2}{\beta}$.
\end{enumerate}
\end{lemma}

We show now  that the operator $\bullet$ starting from t-events of $\G \parN \Msg$ builds t-events of \agts\ whose \sgts\
are subtypes of $\G$ composed in parallel with the  queues
given by the
%\oi\ matching
\balancing\
of \refToFigure{wfagtA}.
Symmetrically,
$\circ$ builds t-events of an
\agt\ $\G \parN \Msg$ from t-events of the immediate subtypes of
$\G$ composed in parallel with the
 queues
given by the
%\oi\ matching
\balancing\
of \refToFigure{wfagtA}.

%  Aggiungere provato in Appendix?

\begin{lemma}\mylabel{lemma:subev}\mylabel{keybeta}
\begin{enumerate}
 \item\mylabel{lemma:subev3}  If $\comoccA\in \EGGA( \agtO{\pp}{\q}i  I{\la}{\G}\parN\Msg)$ and
$\postGA\comoccA{\CommAs\pp{\la_k}\q}$ is defined,
then

 \centerline{$\postGA\comoccA{\CommAs\pp{\la_k}\q}\in
   \EGGA(\G_k\parN\addMsg\Msg{\mq\pp{\la_k}\q})$ where $k\in I$.}
 \item\mylabel{lemma:subev4} If $\comoccA\in \EGGA( \agtI \pp\q \la \G \parN\addMsg{\mq\pp\la\q}\Msg)$  and $\postGA\comoccA{\CommAsI\pp\la\q}$ is defined, then
 $\postGA\comoccA{\CommAsI\pp\la\q}\in\EGGA(\G\parN\Msg)$.
 \item\mylabel{lemma:subev1}
 If $\comoccA\in\EGGA(\G\parN\addMsg\Msg{\mq\pp{\la}\q})$,
then

\centerline{$\cauA{\CommAs\pp{\la}\q}\comoccA\in \EGGA(
  \agtO{\pp}{\q}i I{\la}{\G}\parN\Msg)$  where $\la = \la_k$ and $\G =
  \G_k$ for some $k\in I$.}
 \item\mylabel{lemma:subev2} If $\comoccA\in\EGGA(\G\parN\Msg)$, then $\cauA{\CommAsI\pp\la\q}\comoccA\in \EGGA( \agtI \pp\q \la \G \parN\addMsg{\mq\pp\la\q}\Msg)$.
\end{enumerate}
\end{lemma}
The operators
$\bullet$ and $\circ$ modify t-events in the same way as the transitions in the
LTS would do.  This is formalised and proved in the following lemma.
Notice
that $\lozenge$  enjoys  this property, while
$\blacklozenge$  does not,  see Remark \ref{r}.
%  Aggiungere provato in Appendix?

\begin{lemma}\mylabel{lemma:subev-bis}\mylabel{keybeta-bis}
  Let $\G\parG\Msg\stackred\beta \G'\parG\Msg'$. Then $\osq\Msg
  \cong\mapWh\beta{\osq{\Msg'}}$ and
\begin{enumerate}
\item\mylabel{keybeta5} if $\comoccA\in\EGGA(\G\parN\Msg)$ and
  $\postGA{\comoccA}{\beta}$ is defined, then
  $\postGA{\comoccA}{\beta}\in\EGGA(\G'\parN\Msg')$;
\item\mylabel{keybeta4} if $\comoccA\in\EGGA(\G'\parN\Msg')$, then
$\preGA{\comoccA}{\beta}\in\EGGA(\G\parN\Msg)$.
\end{enumerate}
\end{lemma}
The function $\sf tec$, which builds a sequence of t-events
corresponding to a pair $\pair\os\comseqA$, is simply defined applying
the function $\sf ev$ to $\os$ and to the prefixes of $\comseqA$.

\begin{definition}[t-events from  pairs of o-traces and traces]
  \mylabel{gecA}
Let
 $\comseqA \neq \ee$ be $\os$-well formed.
We define the {\em sequence of global events
corresponding to $\os$ and $\comseqA$} by \vspace*{1.6mm}

  \centerline{$\gecA{\os,\comseqA}=\Seq{\comoccA_1;\cdots}{\comoccA_n}$}  \vspace*{1.6mm}

  \noindent where    $\comoccA_i=    \point(\os,\range\comseqA 1 i )$
   for all $i$, $1\leq i\leq n$.
  \end{definition}

 The following lemma establishes the soundness
  of the above definition.

 \begin{lemma}
\mylabel{owf-trace-properties}
If  $\comseqA \neq \ee$ is $\omega$-well formed,  then:
\begin{enumerate}
\item \mylabel{otp1} $\range\comseqA{1}{i}$
is $\omega$-well formed for all $i$, $1\leq i\leq n$;
\item \mylabel{otp2}
    $\point(\os,\range\comseqA{1}{i})$ is defined and $\io{\point(\os,\range\comseqA{1}{i})}
 = \at{\comseqA}{i} $ for all $i$,  $1\leq i\leq n$.
\end{enumerate}
\end{lemma}

\begin{proof}
%By definition $\comseqA$ is $\omega$-well formed if $\concat{\os}{\comseqA}$ is well formed.
  The proof of (\ref{otp1}) is immediate since by
  Definitions~\ref{def:matching} and~\ref{def:WFnew} every prefix of
  an $\omega$-well formed trace is $\omega$-well formed.  Fact
  (\ref{otp2}) follows from Fact (\ref{otp1}), \refToDef{def:pf}
  and \refToLemma{gl}.
\end{proof}

As for the function $\sf nec$ (\refToLemma{eion}), the t-events in a
sequence generated by  the function $\sf tec$  are not in conflict, and we can
retrieve $\comseqA$ from $\gecA{\os,\comseqA}$ by using the function
$\sf i/o$ given in \refToDef{def:glEventA}(\ref{def:glEventA3}).

 \begin{lemma}\mylabel{eiog}
 Let   $\comseqA \neq \ee$ be
$\os$-well formed and  $\gecA{\os,\comseqA}=\Seq{\comoccA_1;\cdots}{\comoccA_n}$.
\begin{enumerate}
\itemsep=0.9pt
\item\mylabel{eiog1}  If $1\leq k,l\leq n$, then $\neg (\comoccA_k \grr \comoccA_l)$;
\item\mylabel{eiog2}  $\at\comseqA{i} = \io{\comoccA_i} \,$ for all $i$,  $1\leq i\leq n$.
\end{enumerate}
  \end{lemma}
   The following lemma,  together
  with \refToLemma{owf-trace-properties}, ensures  that
  $\gecA{\os,\comseqA}$ is defined when $\os =
  \osq\Msg$ and
  $\G\parG\Msg\stackred\comseqA\G'\parG\Msg'$.
%  Aggiungere provato in Appendix?

\begin{lemma}\mylabel{lemma:Gred-trace-owf}
If $\G\parG\Msg\stackred\comseqA\G'\parG\Msg'$ and
 $\os = \osq\Msg$, then $\comseqA$ is $\omega$-well formed.
\end{lemma}

The following lemma mirrors \refToLemma{knec-ila}.

\begin{lemma}\label{kgec}
\begin{enumerate}
\item\label{kgec2}
   Let $\comseqA=\concat\beta{\comseqA'}$ and
   $\os'=\mapBl{\beta}{\os}$. If
   $\gecA{\os,{\comseqA}}=\Seq{\comoccA_1;\cdots}{\comoccA_n}$ and
   $\gecA{\os',{\comseqA'}}=\Seq{\comoccA'_2;\cdots}{\comoccA'_n}$, then
   $\postGA{\comoccA_i}\beta={\comoccA'_i}\,$  for all $i$,   $2\leq i\leq n$.
   \item\label{kgec1}
   Let $\comseqA=\concat\beta{\comseqA'}$ and
   $\os=\mapWh{\beta}{\os'}$. If
   $\gecA{\os,{\comseqA}}=\Seq{\comoccA_1;\cdots}{\comoccA_n}$ and
   $\gecA{\os',{\comseqA'}}=\Seq{\comoccA'_2;\cdots}{\comoccA'_n}$, then
   $\preGA{\comoccA'_i}\beta={\comoccA_i}\,$  for all $i$,   $2\leq i\leq n$.
   \end{enumerate}
    \end{lemma}

\noindent We end this subsection with the two theorems for \agts\ discussed
at the beginning of the whole section, which relate the transition
sequences of an \agt\ with the proving sequences of the associated
PES.

\begin{theorem}\mylabel{lemma:keyA1}
If $\G\parG\Msg\stackred\comseqA\G'\parG\Msg'$, then
 $\gecA{\osq\Msg,\comseqA}$ is a proving sequence in  the event structure   $\ESGA{\G\parN\Msg}$.
\end{theorem}

\begin{proof}
  Let $\os=\osq\Msg$. By \refToLemma{lemma:Gred-trace-owf}
  $\,\comseqA$ is $\os$-well formed. Then
  by \refToLemma{owf-trace-properties} $\gecA{\os,\comseqA}$ is defined
  and by \refToDef{gecA}
  $\gecA{\os,\comseqA}=\Seq{\comoccA_1;\cdots}{\comoccA_n}$, where
  $\comoccA_i= \point(\os,\range{\comseqA}{1}{i})$ for all $i$,  $1\leq i\leq
  n$.  We proceed by induction on $\comseqA$.

\medskip\noindent
 {\it Case $\comseqA =\beta$.}
%By \refToDef{gecANew} we have
 In this case,  $\gecA{\os,\beta}= \comoccA_1 =
\point(\os,\beta)$.
 By \refToDef{def:pf}
  we have
$ %\gecA{\os,\beta} =
 \point(\os,\beta) =
 \eqA{\os}{\filt{\beta}{\ee}}$.
By \refToDef{def:pm}
  $\eqA{\os}{\filt{\beta}{\ee}} = \eqA{\os}\beta$ since $\beta$ is $\os$-well formed.\sm

 We use now a further induction on the inference of the transition  $\G\parG\Msg\stackred{\beta}\G'\parG\Msg'$, see Figure~\ref{ltsgtAs}.

\vspace{1.6mm}\noindent
{\it Base  Subcases.} The rule applied is  \rulename{Ext-Out} or \rulename{Ext-In}.  Therefore
$\beta\in\FPaths{\G}$.
 By \refToDef{egA}(\ref{eg1AP}) this implies $\point(\os,\beta)\in \EGGA( \G\parG\Msg)$.

\smallskip\noindent
{\it Inductive  Subcases.}  If the last applied Rule is \rulename{IComm-Out}, then $\G=
\agtO{\pp}{\q}i I{\la}{\G}$ and $\G'=\agtO{\pp}{\q}i I{\la}{\G'}$ and
$\G_i\parG\addMsg{\Msg}{\mq\pp{\la_i}\q}\stackred\asCom\G_i'
\parG\addMsg{\Msg'}{\mq\pp{\la_i}\q}$
for all $i \in I $ and $\pp\not\in\play\beta$. \linebreak
We have
$\osq{\addMsg\Msg{\mq\pp{\la_i}\q}} =
\concat\os{\CommAs\pp{\la_i}\q}$.  By induction we get
$\gecA{\concat\os{\CommAs\pp{\la_i}\q},\beta} = \comoccA'_i =
\eqA{\concat{\os}{\CommAs\pp{\la_i}\q}}{\beta} \in\EGGA(
\G_i\parG\addMsg{\Msg}{\mq\pp{\la_i}\q})$.  By
Lemma~\ref{lemma:subev}(\ref{lemma:subev1})
$\preGA{\comoccA'_i}{\CommAs\pp{\la_i}\q}\in\EGGA( \G\parG\Msg)$.
Now, from $\pp  \notin  \play{\beta}$ it follows that $\CommAs{\pp}{\la_i}{\q}$
is not a local cause of $\beta$, namely $\neg(\oks{1}{\concat{\CommAs{\pp}{\la_i}{\q}}{\beta}})$.
From \refToLemma{lemma:Gred-trace-owf}   $\beta$ is
\mbox{$\os$-well}-formed. So, if $\beta$ is an input,
 its matched output
 must be in $\os$.
Hence $\CommAs{\pp}{\la_i}{\q}$
is not a cross-cause of $\beta$,
namely $\neg(\ct{1 + \cardin{\os}}{\concat{\os}{\concat{\CommAs{\pp}{\la_i}{\q}}{\beta}}}{\,2 + \cardin{\os}})$.
Therefore $\concat{\CommAs\pp{\la_i}\q}\beta$ is not $\os$-pointed.
By \refToDef{def:PostPreGlA}(\ref{def:PreGl1A}) we get
$\preGA{\comoccA'_i}{\CommAs\pp{\la_i}\q} =\eqA{\os}\beta =
\comoccA_1$.  We conclude again that $\comoccA_1\in\EGGA(
\G\parG\Msg)$ and clearly $\comoccA_1$ is a proving sequence in
$\ESGA{\G\parN\Msg}$ since
%$\point(\os,\beta)$ has no predecessor given that
$\beta$ has no
proper prefix.

 \medskip
If the last applied Rule is \rulename{IComm-In} the proof is
similar.

\medskip\noindent
{\it Case $\comseqA=\concat{\beta}{\comseqA'}$ with
  $\comseqA'\not=\ee$. }  From
$\G\parG\Msg\stackred{\comseqA}\G'\parG\Msg'$ we get
$\G\parG\Msg\stackred{\beta}\G''\parG\Msg''\stackred{\comseqA'}\G'\parG\Msg'$
for some $\G''$, $\Msg''$.  Let $\os'=\osq{\Msg''}$.
By~\refToLemma{lemma:Gred-trace-owf}
$\comseqA'$ is $\os'$-well formed.
Thus $\gecA{\os',\comseqA'}$ is defined
by \refToLemma{owf-trace-properties}.
Let
$\gecA{\os',\comseqA'}=\Seq{\comoccA'_2;\cdots}{\comoccA'_{n}}$.
By induction $\gecA{\os',\comseqA'}$ is a proving sequence in
$\ESGA{\G''\parG\Msg''}$.
By \refToLemma{kgec}(\ref{kgec1}) $\comoccA_j=\preGA{\comoccA_j'}\beta$ for all $j$, $2\leq j\leq n$.
By \refToLemma{keybeta-bis}(\ref{keybeta4}) this implies $\comoccA_j\in\EGGA(\G\parG\Msg)$ for all $j$, $2\leq j\leq n$.  From the proof of the base case we know that $\comoccA_1 =\eqA{\os}\beta \in\EGGA(\G\parG\Msg)$.
     What is left to show  is  that $\gecA{\os,\comseqA}$ is a proving sequence
in $\ESGA{\G\parG\Msg}$.  By \refToLemma{eiog}(\ref{eiog1}) no two
events in this sequence can be in
conflict.

\medskip
Let $\comoccA\in \EGGA( \G\parG\Msg)$ and $\comoccA< \comoccA_k$ for
some $k$, $1\leq k\leq n$.  Note that this implies $j>1$. If
$\postGA{\comoccA}{\beta}$ is undefined, then by \refToDef{def:PostPreGlA}(\ref{def:PostPreGl1A}) either $\comoccA=\comoccA_1$ or
$\comoccA=\eqA\os{\comseqA}$ with
$\comseqA\not\approx_\os\concat\beta{\comseqA'}$ and
$\play\beta\subseteq\play\comseqA$.
%So in both cases we are done.
In the first case we are done. In the second case
$\projAP\comseqA{\play\beta}\grr \projAP\beta{\play\beta}$,
which implies $\comoccA_1\grr\comoccA$.
Since $\comoccA < \comoccA_k$ and conflict is hereditary, it follows
that $\comoccA_1\grr\comoccA_k$, which contradicts what said above.
Hence this second case is not possible.
If $\postGA{\comoccA}{\beta}$ is defined,
 by \refToLemma{keybeta-bis}(\ref{keybeta5})
$\postGA{\comoccA}{\beta}\in\EGGA(\G''\parG\Msg'')$ and by
\refToLemma{prop:prePostGlArel}(\ref{ppg4bA})
$\postGA{\comoccA}{\beta}  <   \postGA{\comoccA_k}{\beta}$.  Let
$\comoccA' = \postGA{\comoccA}{\beta}$. By
\refToLemma{kgec}(\ref{kgec2})
$\postGA{\comoccA_j}{\beta}=\comoccA'_j$ for all $j$, $2\leq j\leq n$.
Thus we have $\comoccA' < \comoccA'_k$.  Since $\gecA{\os',\comseqA'}$
is a proving sequence in $\ESGA{\G''\parG\Msg''}$, by
\refToDef{provseq} there is $h<k$ such that $%\postGA{\comoccA}{\beta}
\comoccA' =\comoccA_h'$.  By
\refToLemma{prop:prePostGlA}(\ref{ppg0b}) we derive $\comoccA=
\preGA{\comoccA'}\beta=\preGA{\comoccA_h'}{\beta} =\comoccA_h$.
\end{proof}

\begin{theorem}\mylabel{lemma:keyA2}
  If $\Seq{\comoccA_1;\ldots}{\comoccA_n}$ is a proving sequence in $\ESGA{\G\parN\Msg}$, then $\G\parG\Msg\stackred\comseqA\G'\parG\Msg'$ where
  $\comseqA=\concat{\concat{\io{\comoccA_1}}\ldots}{\io{\comoccA_n}}$.
\end{theorem}

\begin{proof}
The proof is by induction on the length $n$ of the proving sequence.  Let $\os=\osq\Msg$.

\medskip\noindent
 {\it Case $n=1$.}   Let $\io{\comoccA_1}=\beta$.  Since $\comoccA_1$
is the first event of a proving sequence, it can have no causes,
so it must be $\comoccA_1 = \eqclass{\os,\beta}$.  We show this case by induction on $d=\weight(\G,\play\beta)$.

\vspace*{1.6mm}\noindent
 {\it  Subcase  $d=1$.}   If $\beta = \CommAs\pp\la\q$ we
have $\G=\agtO\pp\q i I \la \G$ with $ \la_k = \la$ for some
$k \in I$.  We deduce
$\G\parG\Msg\stackred\beta\G_k\parG\addMsg{\Msg}{\mq\pp{\la}\q}$ by
applying Rule \rulename{Ext-Out}.  If $\beta = \CommAsI\pp\la\q$ we
have $\G = \agtI \pp\q {\la} {\G'}$.
Since $\G\parG\Msg$ is well formed, by Rule \rulename{In} of
\refToFigure{wfagtA} we get $\Msg\equiv\addMsg{\mq\pp{\la}\q}{\Msg'}$.
We deduce
$\G\parG\Msg\stackred\beta\G'\parG\Msg'$ by applying
Rule \rulename{Ext-In}.

\vspace*{1.6mm}\noindent
{\it  Subcase  $d>1$.}
We are in one of the two situations:
\begin{enumerate}
\itemsep=0.9pt
\item\label{imp1} $\G= \agtO\pr\ps i I \la \G$ with $\pr  \notin   \play{\beta}$;
\item\label{imp2} $\G=\agtI \pr\ps {\la'} {\G''}$ with
 $\ps  \notin  \play{\beta}$.
\end{enumerate}
In  situation (\ref{imp1}),  %the first case,
$\pr \notin \play{\beta}$ implies  that
$\postGA{\comoccA_1}{\CommAs\pr{\la_i}\ps}$  is  defined for
all $i\in I$ by \refToDef{def:PostPreGlA}(\ref{def:PostPreGl1A}).  By
\refToLemma{keybeta}(\ref{lemma:subev3})
$\postGA{\comoccA_1}{\CommAs\pr{\la_i}\ps}\in\EGGA(\G_i\parallel\addMsg{\Msg}{\mq\pp{\la_i}\q})$
for all $i\in I$.  \refToLemma{ddA}(\ref{ddA1}) implies
$\weight(\G,\play\beta)>\weight(\G_i,\play\beta)$ for all $i\in I$.
By induction hypothesis we have \linebreak
$\G_i\parG\addMsg{\Msg}{\mq\pp{\la_i}\q}\stackred\asCom\G_i' \parG\addMsg{\Msg'}{\mq\pp{\la_i}\q}$
for all $i\in I$.  Then we may apply Rule \rulename{IComm-Out} to
deduce

\vspace*{1.6mm}
\centerline{$\agtO\pr\ps i I \la \G\parG\Msg\stackred\beta
  \agtO\pr\ps i I \la {\G'}\parG\Msg'$}

\vspace*{1.6mm}\noindent
In  situation  (\ref{imp2}),  %the second case,
since $\G\parG\Msg$ is well formed we get
$\Msg\equiv\addMsg{\mq\pr{\la'}\ps}{\Msg''}$ by Rule \rulename{In} of
\refToFigure{wfagtA}. Hence
$\os\cong\concat{\CommAs\pr{\la'}\ps}{\os'}$. This and $\ps \notin
\play{\beta}$ imply  that
$\postGA{\comoccA_1}{\CommAsI\pr{\la'}\ps}$  is  defined by
\refToDef{def:PostPreGlA}(\ref{def:PostPreGl1A}).  By
\refToLemma{keybeta}(\ref{lemma:subev4})
$\postGA{\comoccA_1}{\CommAsI\pr{\la'}\ps}\in\EGGA(\G''\parallel\Msg'')$.
\refToLemma{ddA}(\ref{ddA2}) gives
$\weight(\G,\play\beta)>\weight(\G'',\play\beta)$. By induction
hypothesis $\G''\parG\Msg''\stackred\asCom\G'''\parG\Msg'''$.
Then we may apply Rule
\rulename{IComm-In} to deduce

\vspace*{1.6mm}
\centerline{$\agtI \pr\ps {\la'}
  {\G''}\parG\addMsg{\mq\pr{\la'}\ps}{\Msg''}\stackred\beta\agtI
  \pr\ps {\la'} {\G'''}\parG\addMsg{\mq\pr{\la'}\ps}{\Msg'''}$}

\vspace*{1.8mm}\noindent
 {\it Case $n>1$.}
Let $\io{\comoccA_1} = \beta$, and $\G\parG\Msg\stackred\beta\G''\parG\Msg''$
be the corresponding transition as obtained from the base case.
We show that $\postGA{\comoccA_j}{\beta}$ is defined for all $j$, $2\leq j\leq n$.
If $\postGA{\comoccA_k}{\beta}$ were undefined
%would be undefined
for some $k$, $2\leq k\leq n$, then by
\refToDef{def:PostPreGlA}(\ref{def:PostPreGl1A}) either
$\comoccA_k=\comoccA_1$ or $\comoccA_k=\eqA\os{\comseqA}$ with
$\comseqA\not\approx_\os\concat\beta{\comseqA'}$ and
$\play\beta\subseteq\play\comseqA$. In the second case
$\projAP\beta{\play\beta}\grr\projAP\comseqA{\play\beta}$, which
implies $\comoccA_k\grr\comoccA_1$.  So both cases are impossible.  If
$\postGA{\comoccA_j}{\beta}$ is defined, by
\refToLemma{keybeta-bis}(\ref{keybeta5}) we may define $\comoccA'_j =
\postGA{\comoccA_j}{\beta}\in\EGGA(\G''\parG\Msg'')$ for all $j$,
$2\leq j\leq n$. We show that $\comoccA'_2 ;\cdots ;\comoccA'_n$ is a
proving sequence in $\ESGA{\G''\parallel\Msg''}$.  By
\refToLemma{prop:prePostGlA}(\ref{ppg0b}) $\comoccA_j =
\preGA{\comoccA'_j}{\beta}$ for all $j$, $2\leq j\leq n$. Then by
\refToLemma{prop:prePostGlArel}(\ref{ppg-conflictA-a}) no two
events in this sequence can be in conflict.

\medskip
Let $\comoccA\in\EGGA(\G''\parG\Msg'')$ and $\comoccA< \comoccA'_h$
for some $h$, $2\leq h\leq n$.  By
\refToLemma{keybeta-bis}(\ref{keybeta4}) $\preGA{\comoccA}\beta$ and
$\preGA{\comoccA'_h}{\beta}$ belong to $\EGGA(\G\parG\Msg)$. By
\refToLemma{prop:prePostGlArel}(\ref{ppg4aA})
%and (\ref{ppg1aA})
$\preGA{\comoccA}\beta<\preGA{\comoccA'_h}{\beta}$.
By \refToLemma{prop:prePostGlA}(\ref{ppg2A})
$\preGA{\comoccA'_h}{\beta}=\comoccA_h$.
Let $\comoccA' = \preGA{\comoccA}\beta$.
Then
$\comoccA' < \comoccA_h$ implies, by \refToDef{provseq}
and the fact that $\ESGA{\G\parG\Msg}$ is a PES, that there is $l<h$
such that $\comoccA'
=\comoccA_l$.  By \refToLemma{prop:prePostGlA}(\ref{ppg1bA})
we get $\comoccA=\postGA{\comoccA'}{\beta} = \postGA{\comoccA_l}\beta
= \comoccA'_l$. \sm

We have shown that $\comoccA'_2 ;\cdots ;\comoccA'_n$ is a proving
sequence in
% the event structure
$\ESGA{\G''\parallel\Msg''}$. By induction we get \linebreak
$\G''\parG\Msg''\stackred{\comseqA'}\G'\parG\Msg'$ where $\comseqA' =
\concat{\concat{\io{\comoccA'_2}}\ldots}{\io{\comoccA'_n}}$.  Let
$\comseqA=\concat{\concat{\io{\comoccA_1}}\ldots}{\io{\comoccA_n}}$. Since
$\io{\comoccA'_j} = \io{\comoccA_j}$ for all $j, 2\leq j\leq n$, we
have $\comseqA\! =\! \concat{\beta}{\comseqA'}$.  Therefore
\mbox{$\G\!\parG\Msg\stackred\beta\G''\parG\Msg'' \stackred{\comseqA'}\G'\!\parG\Msg'$}
 is the required transition sequence.
\end{proof}

\subsection{Isomorphism}

 We are finally able to
%prove our main theorem, namely to
show that the ES interpretation of a network is equivalent, when the
session is typable, to the ES interpretation of its \agt.

\medskip
To prove our main theorem, we will also use the following separation
result from~\cite{BC91} (Lemma 2.8 p. 12). Recall from
\refToSection{sec:eventStr} that $\Conf{S}$ denotes the set of
configurations of $S$.
\begin{lemma}[Separation~\cite{BC91}]
\mylabel{separation}
Let $S=(E,\prec, \gr)$ be a flow event structure and $\ESet, \ESet' \in \Conf{S}$
be such that $\ESet \subset \ESet'$.
Then there exist $e \in \ESet'\backslash \ESet$ such that $\ESet \cup \set{e} \in \Conf{S}$.
\end{lemma}
We may now establish the isomorphism between the domain of
configurations of the FES of a typable network and the domain of
configurations of the PES of its \agt. In the proof of
this result, we will use the characterisation of configurations as
proving sequences, %as
 given  in \refToProp{provseqchar}.  We will also
take the freedom of writing $\Seq{{\netevA_1};\cdots}{\netevA_n} \in
\Conf{\ESNA{\Nt\parG\Msg}}$ to mean that
$\Seq{{\netevA_1};\cdots}{\netevA_n}$ is a proving sequence such that
$\set{\netevA_1, \ldots, \netevA_n} \in \Conf{\ESNA{\Nt\parG\Msg}}$,
and similarly for $\Seq{{\comoccA_1};\cdots}{\comoccA_n} \in
\Conf{\ESGA{\RG\parG\Msg}}$.
\begin{theorem}\mylabel{isoA}
 If $\derN{\Nt\parG\Msg}{\RG\parG\Msg}$, then $\CD{\ESNA{\Nt\parG\Msg}} \simeq \CD{\ESGA{\RG\parG\Msg}}$.
\end{theorem}

\begin{proof}
  Let $\os=\osq\Msg$.  We start by constructing a bijection between
  the proving sequences of  the event structure   $\ESNA{\Nt\parG\Msg}$ and the proving
  sequences of   the event structure   $\ESGA{\RG\parG\Msg}$.  By \refToTheorem{uf12A}, if
  $\Seq{{\netevA_1};\cdots}{\netevA_n}\in \Conf{\ESNA{\Nt\parG\Msg}}$,
%is a proving sequence of $\ESNA{\Nt\parG\Msg}$.
   then $\Nt\parG\Msg\stackred\comseqA\Nt'\parG\Msg'$
  where
  $\comseqA=\concat{\concat{\io{\netevA_1}\cdots}}{\io{\netevA_n}}$.
By  applying iteratively  Subject
  Reduction (\refToTheorem{srA}), we obtain
  \[\G\parG\Msg\stackred\comseqA\G'\parG\Msg' \mbox{ and }
   \derN{\Nt'\parG\Msg'}{\RG'\parG\Msg'}
  \]
  By \refToTheorem{lemma:keyA1},  we get $\gecA{\os,\comseqA} \in
  \Conf{\ESGA{\RG\parG\Msg}}$.

  \medskip
  By \refToTheorem{lemma:keyA2}, if
   $\Seq{{\comoccA_1};\cdots}{\comoccA_n} \in
  \Conf{\ESGA{\RG\parG\Msg}}$,
%is a proving sequence of  $\ESGA{\RG\parG\Msg}$,
then $\G\parG\Msg\stackred\comseqA\G'\parG\Msg'$, where
  $\comseqA=\io{\comoccA_1}\cdots\io{\comoccA_n}$. By applying
  iteratively Session Fidelity (\refToTheorem{sfA}), we obtain
\[\Nt\parG\Msg\stackred\comseqA\Nt'\parG\Msg' \mbox{ and }
    \derN{\Nt'\parG\Msg'}{\RG'\parG\Msg'}
\]
By \refToTheorem{uf10A}, %(\ref{uf101A})
   we get $\necA{\comseqA} \in
  \Conf{\ESNA{\Nt\parG\Msg}}$.

\medskip
Therefore we have a bijection between $\CD{\ESNA{\Nt\parG\Msg}}$ and
$\CD{\ESGA{\RG\parG\Msg}}$, given by $\necA{\comseqA} \leftrightarrow \gecA{\os,\comseqA}$
for any $\comseqA$ generated by the (bisimilar) LTSs of $\Nt\parG\Msg$ and $\RG\parG\Msg$.

\medskip
We now show that this bijection preserves inclusion of configurations. \vspace*{1.8mm}

By \refToLemma{separation} it is enough to prove that if
$\Seq{{\netevA_1};\cdots}{\netevA_n} \in \Conf{\ESNA{\Nt\parG\Msg}}$
is mapped to $\Seq{{\comoccA_1};\cdots}{\comoccA_n}\! \in
\Conf{\ESGA{\RG\!\parG\Msg}}$,  then
$\Seq{\Seq{{\netevA_1};\cdots}{\netevA_n}}\netevA\! \in
\Conf{\ESNA{\Nt\!\parG\Msg}}$ iff
\mbox{$\Seq{\Seq{{\comoccA_1};\cdots}{\comoccA_n}}\comoccA\! \in
\Conf{\ESGA{\RG\!\parG\Msg}}$}, where
$\Seq{\Seq{{\comoccA_1};\cdots}{\comoccA_n}}\comoccA$ is the image of
$\Seq{\Seq{{\netevA_1};\cdots}{\netevA_n}}\netevA$ under the
bijection.
So, suppose  $\Seq{{\netevA_1};\cdots}{\netevA_n} \in
  \Conf{\ESNA{\Nt\parG\Msg}}$  and
$\Seq{{\globevA_1};\cdots}{\globevA_n} \in
  \Conf{\ESGA{\RG\parG\Msg}}$ are  such that
\[\Seq{{\netevA_1};\cdots}{\netevA_n} = \necA{\comseqA}
\leftrightarrow \gecA{\os,\comseqA} =
\Seq{{\globevA_1};\cdots}{\globevA_n}
\]
 Then $\io{\netevA_1}\cdots\io{\netevA_n}= \comseqA=
\io{\globevA_1}\cdots\io{\globevA_n}$.

\medskip
By \refToTheorem{uf12A}, if
 $\Seq{{\netevA_1};\cdots}{\netevA_n;\netevA} \in
  \Conf{\ESNA{\Nt\parG\Msg}}$
%is a proving sequence of $\ESNA{\Nt\parG\Msg}$,
with $\io{\netevA} = \beta$, then
\[ \Nt\parG\Msg\stackred{\concat\comseqA%\Nt_0\parG\Msg_0\stackred
    \beta}\Nt'\parG\Msg'
  \]

 By applying iteratively Subject Reduction (\refToTheorem{srA}) we get
\[ \G\parG\Msg\stackred{\concat\comseqA%\G_0\parG\Msg_0\stackred
    \beta}\G'\parG\Msg'  \mbox{ and } \derN{\Nt'\parG\Msg'}{\RG'\parG\Msg'}
 \]
 We conclude that  $\gecA{\os,\concat\comseqA\beta} \in
  \Conf{\ESGA{\RG\parG\Msg}}$
%is a proving sequence of $\ESGA{\RG\parG\Msg}$
by \refToTheorem{lemma:keyA1}.

\medskip
By \refToTheorem{lemma:keyA2}, if
 $\Seq{{\comoccA_1};\cdots}{\comoccA_n;\comoccA} \in
  \Conf{\ESGA{\RG\parG\Msg}}$
%is a proving sequence of $\ESGA{\RG\parG\Msg}$,
with $\io{\comoccA} = \beta$, then
\[
 \G\parG\Msg\stackred{\concat\comseqA
    \beta}\G'\parG\Msg'
 \]
 \eject
  By applying iteratively Session Fidelity (\refToTheorem{sfA}) we get
\[
\Nt\parG\Msg\stackred{\concat\comseqA%\Nt_0\parG\Msg_0\stackred
    \beta}\Nt'\parG\Msg'  \mbox{ and } \derN{\Nt'\parG\Msg'}{\RG'\parG\Msg'}
\]
 We conclude that
 $\necA{\concat\comseqA\beta} \in
  \Conf{\ESNA{\Nt\parG\Msg}}$
%is a proving sequence of $\ESNA{\Nt\parG\Msg}$
by \refToTheorem{uf10A}.
   \end{proof}

\section{Related work and conclusions}
\mylabel{sec:relatedA}

Session types, as originally proposed in~\cite{HVK98,CHY16} for binary
sessions, are grounded on types for the $\pi$-calculus.  Early
proposals for typing channels in the $\pi$-calculus include simple
sorts~\cite{M92}, \oi\ types~\cite{PS96}
%Mi sembra che loro li chiamino I/O tags !
and usage types~\cite{K05}. In particular, the notion of progress for multiparty
sessions~\cite{DY11,Coppo2016}
is inspired by the notion of lock-freedom developed for the
$\pi$-calculus in~\cite{K02,K06}. The more recent work~\cite{DGS12}
provides further evidence of the strong relationship between binary
session types and channel types in the linear $\pi$-calculus.  The
notion of lock-freedom for the linear $\pi$-calculus was also
revisited in~\cite{P15}.

\medskip
Multiparty sessions disciplined by global types were introduced in the
keystone papers~\cite{CHY08,CHY16}. These papers, as well as most
subsequent work on multiparty session types (for a survey
see~\cite{H2016}), were based on more expressive session calculi than
the one we use here, where sessions may be interleaved and
participants exchange pairs of labels and values.  In that more
general setting, global types are projected onto session types and in
turn session types are assigned to processes. Here, instead, we
consider only single sessions and pure label exchange: this allows us
to project global types directly to processes, as in~\cite{DS19},
where the considered global types are those of~\cite{CHY16}.
% We chose to concentrate on this very simple calculus, as our working
% plan was already quite challenging.
Possible extensions of our work to more expressive calculi are
discussed at the end of this section.
% A discussion on possible extensions of our work to more expressive
% calculi may be found at the end of this section.

Standard global types are too restrictive for typing processes which
communicate asynchronously. A powerful typability extension is
obtained by the use of the subtyping relation given
in~\cite{Mostrous2009}. This subtyping allows inputs and outputs to be
exchanged, stating that anticipating outputs is better. The rationale
is that outputs are not blocking, while inputs are blocking in
asynchronous communication. Unfortunately, this subtyping is
undecidable~\cite{BCZ17,LY17}, and thus type systems equipped with
this subtyping are not effective.  Decidable restrictions of this
subtyping relation have been proposed~\cite{BCZ17,LY17,BCZ18}.  In
particular, subtyping is decidable when both internal and external
choices are forbidden in one of the two compared
processes~\cite{BCZ17}.  This result is improved in~\cite{BCZ18},
where both the subtype and the supertype can contain either internal
or external choices.  More interestingly, the work ~\cite{BCLYZ21}
presents a sound (though not complete) algorithm for checking
asynchronous subtyping.  A very elegant formulation of asynchronous
subtyping is given in~\cite{GPPSY21}: it allows the authors to show
that any extension of this subtyping would be unsound.  In the present
paper we achieve a gain in typability for asynchronous networks by
using a more fine-grained syntax for global types.  Our type system is
decidable, since projection is computable and the preorder on
processes is decidable. Notice that there are networks that can be
typed using the algorithm in~\cite{BCLYZ21} but cannot be typed in our
system, like the  running example of
%video streaming example discussed in
that paper.

We claim that our asynchronous types are more ``prescribing'' than the
global types of~\cite{CHY08,CHY16} equipped with asynchronous
subtyping, since asynchronous types %give
 specify  the order in which participants must do inputs and
outputs in a more precise way. %For example
 For instance,  the asynchronous type
$\agtSOS \pp\q {\la};\agtSOS \q\pp{\la'};\agtIS \pp\q {\la};\agtIS
\q\pp {\la'}\parG\emptyset$ of \refToExample{ta} can type the network
$\pP{\pp}{\sendL{\q}\la;\rcvL{\q}{\la'}} \parN
\pP{\q}{\sendL{\pp}{\la'};\rcvL{\pp}\la}$ of
\refToExample{sync-async-characteristic-example}, but cannot type the
network
$\pP{\pp}{\sendL{\q}\la;\rcvL{\q}{\la'}} \parN
\pP{\q}{\rcvL{\pp}\la;\sendL{\pp}{\la'}}$. Instead, in the system
of~\cite{CHY08,CHY16} one needs the global type
$\pp\to\q:\la;\q\to\pp:\la'$ and the subtyping
$\sendL{\pp}{\la'};\rcvL{\pp}\la\leq\rcvL{\pp}\la;\sendL{\pp}{\la'}$
to type the network
$\pP{\pp}{\sendL{\q}\la;\rcvL{\q}{\la'}} \parN
\pP{\q}{\sendL{\pp}{\la'};\rcvL{\pp}\la}$. The drawback is that the
same global type can also type the network
$\pP{\pp}{\sendL{\q}\la;\rcvL{\q}{\la'}} \parN
\pP{\q}{\rcvL{\pp}\la;\sendL{\pp}{\la'}}$.

 More permissive variants of our asynchronous types, called
``deconfined global types'',  were subsequently  considered in~\cite{DGD21a}
and~\cite{DaGiDe23}.
In~\cite{DGD21a} both
projection and balancing are refined, the first one allowing the
participants which are not involved in a choice to have different
behaviours in the branches of the choice, and the second one allowing
unbounded queues.    The
type  system of~\cite{DaGiDe23} has global types %with
 that allow  %also
 choices of inputs, and types the
 running  example of~\cite{BCLYZ21} and also %, but it types
a network for which the algorithm of~\cite{BCLYZ21} fails.
With the type systems
of~\cite{DGD21a,DaGiDe23}, %these papers,
due to a more complex definition of balancing,  one can type networks
with  queues  that may grow unboundedly, which is not
possible with the balancing of \refToFigure{wfagtA}.  Since the focus
of the present paper was on the event structure semantics, we decided
to go for this simpler definition.

Since their introduction in~\cite{Win80,NPW81}, Event
Structures have been widely used to give semantics to process
calculi. Several ES interpretations of Milner's calculus CCS have been
proposed, using various classes of ESs: Stable ESs~\cite{Win82},
Prime ESs or variations of them~\cite{BC87, DDM88,DDM90}, and Flow
ESs~\cite{BC88a,GG04}. Other calculi such as TCSP~\cite{BHR84,Old86}
%(Theoretical CSP~\cite{BHR84,Old86})
and LOTOS have been provided respectively with a PES
semantics~\cite{LG91,BM94} and with a Bundle ES
semantics~\cite{Lan93,Kat96}.  More recently, ES semantics have been
investigated also for the
$\pi$-calculus~\cite{CVY07,VY10,CVY12,Cri15,CKV15,CKV16}.  A more
extensive discussion on ES semantics for process calculi may be found
in our companion paper~\cite{CDG22}.
% We refer the reader to our companion paper~\cite{CDG22} for a more
% extensive discussion on ES semantics for process calculi.

It is noteworthy that all the above-mentioned ES semantics were given
for calculi with synchronous communication. This is perhaps not
surprising since ESs are generally equipped with a
\emph{synchronisation algebra} when modelling process calculi, and a
communication is represented by a single event resulting from the
synchronisation of two events.
This is also the reason why, in our previous paper~\cite{CDG22},
we started by considering an ES semantics for a synchronous session
calculus with standard global types.

An asynchronous PES semantics for finite \emph{synchronous}
choreographies was recently proposed in~\cite{LMT20}, where,
like in the present paper,
%. In that paper like in the present one,
a communication is represented by two
distinct events,
% in the ES,
one for the output and the other for the
matching input. However, in our work the output and the matching input
are already decoupled in the types, and their matching relation needs
to be reconstructed in order to obtain the cross-causality relation in
the PES. Instead, in ~\cite{LMT20} the definition of cross-causality
is immediate, since the standard synchronous type construct gives
rises to a pair of events which are by construction in the
cross-causality relation. Moreover, only types are interpreted as ESs
in~\cite{LMT20}. To sum up, while asynchrony is an essential feature
of sessions in our calculus, and therefore it is modelled also in
their abstract specifications (asynchronous types), asynchrony is
rather viewed as an implementation feature of sessions
in~\cite{LMT20}, and therefore it is not modelled in their abstract
specifications (choreographies), which remain synchronous.

A denotational semantics based on concurrent games \cite{RW11} has been
proposed for the asynchronous $\pi$-calculus in \cite{ST17}.  Notice,
however, that in the asynchronous $\pi$-calculus an output can never
be a local cause of any other event, since the output construct has no
continuation.  Therefore the asynchrony of the asynchronous
$\pi$-calculus is more liberal than that of our calculus and of
session calculi in general, which adopt the definition of asynchrony
of standard protocols such as TCP/IP, where the order of messages
between any given pair of participants is preserved.

This work builds on the companion paper~\cite{CDG22}, where
synchronous rather than asynchronous communication was considered. In
that paper too, networks were interpreted as FESs, and global types,
which were the standard ones, were interpreted as PESs.  The key
result was again an isomorphism between the configuration domain of
the FES of a typed network and that of the PES of its global type.
Thus, the present paper completes the picture set up
in~\cite{CDG22} by exploring the ``asynchronous side'' of the
same construction.

%%%%%%%%%%%%% PEZZO AGGIUNTO %%%%%%%%%%%%%%%%%%%%%%%

An important feature of a denotational model such as Event Structures
is abstraction.  Clearly, our PES semantics for asynchronous types
abstracts away from their syntax, by making explicit the concurrency
relation between independent communications that is left implicit in
the types: for instance, it maps to the same PES all the types given
in \refToExample{ta} for the characteristic network of
\refToExample{sync-async-characteristic-example}. Indeed, it can be
shown that all well-formed asynchronous types that type the same
network give rise to the same PES. Our FES semantics for networks also
abstracts away from their syntax to some extent, via the narrowing
operation which prunes off all the input events that are not justified
by an output event or by a message in the queue, as well as all their
successors. As a consequence, the (non typable) network
$\pP{\pp}{\rcvL{\q}{ \la}; \sendL{\pr}{ \la'}} \parN
\pP{\pr}{\rcvL{\pp}{ \la'}} \parN \emptyset$ is interpreted as the FES
with an empty set of events, and so are other deadlocked networks of
the same kind.

As future work, we shall try to devise semantic counterparts for our
well-formedness conditions on asynchronous types, namely structural
conditions characterising both the PESs of well-formed asynchronous
types and the FESs of well-typed networks, along the lines of a
previous proposal for binary sessions as Linear Logic proofs based on
causal nets \cite{CY19}. This would allow us to reason entirely on the
semantic side, and in particular to establish the isomorphism of the
configuration domains of a well-typed network FES and the PES of one
its types in a more direct way.
%, without recourse to the Subject Reduction and Session Fidelity
%results.
Such semantic characterisations of well-formedness would also help
us to address the following \emph{synthesis problem}: starting from an
arbitrary Prime ES, is it possible (1) to verify that it represents a
well-formed asynchronous type, and, if this is the case, (2) to
reconstruct a network that behaves according to that asynchronous
type?
A step in this direction was recently made in \cite{CG24}, in the synchronous
setting of \cite{CDG22}.

%%%%%%%% FINE PEZZO AGGIUNTO %%%%%%%%%%%%%%%%%%

Other possible directions for future work have already
been sketched in~\cite{CDG22}: they include the investigation of reversibility, which
would benefit from previous work on reversible session
calculi~\cite{TY14,TY16,MP17,MP17a,NY17,CDG18} and Reversible Event
Structures~\cite{PU2016,CKV16,GPY2017,GPY2018, Gra21}.  We also plan
to investigate the extension of our asynchronous calculus with
delegation.  In the literature, delegation is usually modelled using
the channel passing mechanism of the $\pi$-calculus, which requires
interleaved sessions. %We feel that
Now, the extension of our event structure semantics to interleaved
sessions would require a deep rethinking, especially for the
definition of narrowing.  Hence we plan to use the alternative notion
of delegation proposed in~\cite{CDGH19} for a session calculus without
channels, called ``internal delegation''.  Note that delegation
remains essentially a synchronous mechanism, even in the asynchronous
setting: indeed, unlike ordinary outputs that become non-blocking,
delegation remains blocking for the principal, who has to wait until
the deputy returns the delegation to be able to proceed. As a matter
of fact, this is quite reasonable: not only does it prevent the issue
of ``power vacancy'' that would arise if the role of the principal
disappeared from the network for some time, but it also seems natural
to assume that the principal delegates a task only when it has the
guarantee that the deputy will accept it.

\paragraph{Acknowledgments} We are indebted to Francesco Dagnino for
suggesting a simplification in the definition of \balancing\ for
asynchronous types.  We also wish to thank the anonymous
referees for their helpful comments.  In particular,
they helped us to  clarify several definitions,
 % have been clarified,
 and to expand both the comparison with the literature and the
 discussion on future work.
 %have been expanded.

\bibliographystyle{fundam}
\bibliography{session}

\newpage

\appendix
\section{Appendix}

This Appendix contains the proofs of Lemmas~\ref{pb},
{\ref{wftr}},  \ref{keysrA}, \ref{keysrA34}, \ref{srgA}, \ref{prop:prePostNetArel}, \ref{tr}, \ref{prop:prePostGlA}, \ref{prop:prePostGlArel}, \ref{keybeta}, \ref{lemma:subev-bis}, \ref{eiog}, \ref{lemma:Gred-trace-owf}, \ref{kgec}  and the auxiliary Lemmas~\ref{scb}, \ref{nested-filtering} , \ref{prop:relPrePost}.

\begin{lemmaa}{\ref{pb}}{If $\GP$ is bounded, then $\proj\GP\pr$ is a partial function for all $\pr$. }
\end{lemmaa}

\begin{proof}
We redefine the projection   $\downarrow_\pr$    as the
  largest relation between \sgts\ and processes such that $\prR\G\PP\pr$
  implies:
\begin{enumerate}[i)]
\itemsep=0.8pt
\item if $\pr\not\in  \play{\G}$, then $\PP=\inact$;
\item if $\G=\agtO{\pr}{\q}i I{\la}{\G}$, then $\PP= \oupP\q{i}{I}{\M}{\PP_i}$ and $\prR{\G_i}{\PP_i}\pr$ for all $i\in I$;
\item if $\G=\CommAs \pp\la\pr ; {\G'}$, then $\prR{\G'}{\PP}\pr$;
\item if $\G=\agtO{\pp}{\pr}i I{\la}{\G}$ and $|I|>1$, then $\PP= \procActs   \,\inpP\pp{i}{I}{\M}{\PP_i}$ and $\prR{\G_i}{  \procActs   \,\Seq{\rcvL\pp{\M_i}}{\PP_i}}\pr$ for all $i\in I$;
\item if $\G=\agtO{\pp}{\q}i I{\la}{\G}$ and $\pr\not\in\set{\pp,\q}$ and $\pr\in  \play{\G_i}$, then $\prR{\G_i}{\PP}\pr$ for all $i\in I$;
\item if $\G=\agtI \pp\pr \la {\G'}$, then $\PP=\Seq{\rcvL{\pp}{\la}}{\PP'}$ and $\prR{\G'}{\PP'}\pr$;
\item if $\G=\agtI \pp\q \la {\G'}$ and $\pr\not=\q$ and $\pr\in  \play{\G'}$, then $\prR{\G'}{\PP}\pr$.
\end{enumerate}
%A bisimulation is a
%%%%%%%%%%%%%%%%%%%%
We define equality $\mathcal E$ of processes to be
   the largest symmetric binary relation   $\RR$
%$\mathcal E$
on processes such that
%$\iR\PP\Q{\mathcal E}$ implies:
$\iR\PP\Q{  \RR  }$ implies:

\begin{enumerate}[(a)]
\itemsep=0.8pt
\item\label{ca} if $\PP=\oupP\pp{i}{I}{\M}{\PP_i}$ , then $\Q=\oupP\pp{i}{I}{\M}{\Q_i}$ and $\iR{\PP_i}{\Q_i}{  \RR  }$ for all $i\in I$;
\item\label{cb}  if $\PP=\inpP\pp{i}{I}{\M}{\PP_i}$ , then $\Q=\inpP\pp{i}{I}{\M}{\Q_i}$ and $\iR{\PP_i}{\Q_i}{  \RR  }$ for all $i\in I$.
\end{enumerate}
It is then enough to show that the relation
$\RR_\pr  =\set{(\PP,\Q)\mid  \exists\, \G
  \, . \ \prR\G\PP\pr\text { and } \prR\G\Q\pr}$
%implies $\iR\PP\Q\mathcal E$,  %is a bisimulation,
%i.e. $\RR$
satisfies Clauses~(\ref{ca}) and~(\ref{cb})   (with $\RR$ replaced by
$\RR_\pr$),   since this will imply $ \RR_\pr \subseteq
\mathcal{E}$.  Note first that $(\inact, \inact) \in \RR_\pr$ because
$(\End, \inact) \in \downarrow_\pr$, and that $(\inact, \inact) \in
\mathcal E$ because Clauses~(\ref{ca}) and~(\ref{cb}) are vacuously
satisfied by the pair   $(\inact, \inact)$,   which must
therefore belong to $\mathcal E$.

\medskip
The proof is by induction on $d=\weight(\G,\pr)$. We only consider
Clause~(\ref{cb}), the proof   for   Clause~(\ref{ca})
being similar and simpler.  So, assume
  $(\PP,\Q)\in \RR_\pr$ and
$\PP=\inpP\pp{i}{I}{\M}{\PP_i}$.

\medskip\noindent
{\it Case $d=1$.}  In this case $\G=\agtI \pp\pr \la {\G'}$
%By definition of $\downarrow_\pr$,
and $\PP=\Seq{\rcvL{\pp}{\la}}{\PP'}$ and $\prR{\G'}{\PP'}\pr$.
From $\prR\G\Q\pr$ we get
$\Q=\Seq{\rcvL{\pp}{\la}}{\Q'}$ and $\prR{\G'}{\Q'}\pr$.
  Hence  $\Q$ has the required form and
   $\iR{\PP'}{\Q'}\RR_\pr$.

\vspace*{1.8mm}\noindent
{\it Case $d>1$.} % Let $\PP=\inpP\pp{i}{I}{\M}{\PP_i}$.
By definition of   $\downarrow_\pr$,
%projection
there are five possible subcases.
\begin{enumerate}
\item Case $\G\!=\CommAs \pp\la\pr ; {\G'}$ and $\prR{\G'}\PP\pr$. From
  $\prR\G\Q\pr$ we get $\prR{\G'}\Q\pr$.   Then  \mbox{$\iR{\PP}{\Q}{\mathcal{\RR_\pr }}$}.
\item Case $\G=\agtO{\pp}{\pr}i I{\la}{\G}$ and
$\prR{\G_i}{\Seq{\rcvL\pp{\M_i}}{\PP_i}}\pr$ for all  $i\in I$ and $|I|>1$. From
$\prR\G\Q\pr$ we get $\Q=
\procActs  \,\inpP\pp{i}{I}{\M}{\Q_i}$ and $\prR{\G_i}{
  \procActs  \,\Seq{\rcvL\pp{\M_i}}{\Q_i}}\pr$ for all $i\in I$.\\
  Since $(\Seq{\rcvL\pp{\M_i}}{\PP_i},
\procActs\,\Seq{\rcvL{\pp}{\M_i}}{\Q_i})\in   \RR_\pr $ for
all $i\in I$, by induction Clause ~(\ref{cb}) is satisfied.   Thus
$\procActsS=\ee$ and $\iR{\PP_i}{\Q_i}{  \RR_\pr }$ for all
$i\in I$.
\item Case  $\G=\agtO{\q}{\pr}j J{\la'}{\G}$   with $\q \neq \pp$
     and $\PP=\pp?\la;
\procActs  \,\inpP\q{j}{J}{\M'}{\PP'_j}$ and\\
$\prR{\G_j}{\pp?\la;
  \procActs  \,\Seq{\rcvL\q{\M'_j}}{\PP'_j}}\pr$ for all $j\in J$. From
$\prR\G\Q\pr$ we get $\Q=\procActsPS;\,\inpP\q{j}{J}{\M'}{\Q'_j}$ and
$\prR{\G_j}{\procActsPS;\,\Seq{\rcvL\q{\M'_j}}{\Q'_j}}\pr$ for all $j\in
J$.
%Hence $\iR{\pp?\la; \procActs\,\Seq{\rcvL\q{\M'_j}}{\PP'_j}\,}{\procActsPS;\,\Seq{\rcvL\q{\M'_j}}{\Q'_j}}{\RR_\pr}$.
  Since $\iR{\pp?\la;
  \procActs\,\Seq{\rcvL\q{\M'_j}}{\PP'_j}\,}{\procActsPS;\,\Seq{\rcvL\q{\M'_j}}{\Q'_j}}{{
    \RR_\pr }}$ for all $j\in J$, by induction Clause ~(\ref{cb}) is
satisfied.\sm
\\
Thus $\procActsPS =\pp?\la;\procActsS$ and $\iR{
  \procActs\,\Seq{\rcvL\q{\M'_j}}{\PP'_j}\,}{\procActs\,\Seq{\rcvL\q{\M'_j}}{\Q'_j}}{{
    \RR_\pr }}$ for all $j\in J$.
\item Case $\G=\agtO{\q}{\ps}j J{\la'}{\G}$ and
    $\pr\not=\ps$
and
$\pr\in  \play{\G_j}$ and $\prR{\G_j}{\PP}\pr$ for $j\in J$. From
$\prR\G\Q\pr$ we get $\prR{\G_j}{\Q}\pr$ for all $j\in J$.   Then   $\iR{\PP}{\Q}{{  \RR_\pr }}$.
\item Case $\G=\agtI \q\ps \la {\G'}$
and $\pr\in
\play{\G'}$. Then $\prR{\G'}{\PP}\pr$. From $\prR\G\Q\pr$ we get
$\prR{\G'}{\Q}\pr$.   Then   $\iR{\PP}{\Q}{{  \RR_\pr }}$.
\end{enumerate}
\end{proof}

%%%%%%%%%%%%%%%%%%%%%%%%%%%%

\begin{lemmaa}{\ref{wftr}}
If $\G\parN\Msg\stackred{\beta}\G'\parN\Msg'$ is a top transition and $\G\parG\Msg$ is well formed, then $\G'\parG\Msg'$
 is well formed too.
\end{lemmaa}

\begin{proof}
If  the transition is derived using Rule $\rulename{\AsOut}$,
then   $\G=\agtO{\pp}{\q}i I{\la}{\G}$
 and for some $k\in I$ we have   $\G'=\G_k$ and
$\Msg'\equiv\addMsg\Msg{\mq\pp{\la_k}\q}$.
% for some $k\in I$.
We show that $\G_k\parN\addMsg\Msg{\mq\pp{\la_k}\q}$ is well
formed. Since $\proj{\G}{\pp}$ is defined for all $\pp$, by definition
of projection also $\proj{\G_k}{\pp}$ is defined for all $\pp$. Since
$\G$ is bounded and $\G_k$ is a subtree of $\G$, also $\G_k$ is
bounded.  Finally, $\derSI{}{\G\parN\Msg}$ implies
$\derSI{}{\G_k\parN\addMsg\Msg{\mq\pp{\la_k}\q}}$ by inversion on Rule
$\rulename{Out}$ of \refToFigure{wfagtA}.

\sm
 If  the transition is derived using Rule $\rulename{\AsIn}$,
 then  $\G=\agtI \pp\q \la {\G'}$ and the proof is similar and simpler.
\end{proof}

%%%%%%%%%%%%%%%%%%%%%%%%%%%%

\begin{lemmaa}{\ref{keysrA}}{ Let $\G\parN\Msg$ be well formed.

\vspace*{-4mm}
\begin{enumerate}
\item If $\proj\G\pp=\oup\q{i}{I}{\M}{\PP}$, then
  $\G\parN\Msg\stackred{\CommAs\pp{\la_i}\q}\G_i\parN\addMsg\Msg{\mq\pp{\la_i}\q}$
  and $\proj{\G_i}\pp=\PP_i$ for all $i\in I$.
\item
If $\proj\G\q=\inp\pp{i}{I}{\M}{\PP}$
and  $\Msg\equiv\addMsg{\mq\pp{\la}\q}{\Msg'}$   for some $\la$,
then  $I=\set{k}\,$ and $\la=\la_k$  and
$\G\parN\Msg\stackred{\CommAsI\pp{\la_k}\q}\G'\parN\Msg'$
and $\proj{\G'}\q=\PP_k$.
\end{enumerate}
}\end{lemmaa}

\begin{proof}
(\ref{keysrA1})
The proof is by induction on $d=\weight(\G,\pp)$.

\vspace{1.6mm}\noindent
{\it Case $d=1$.}    By definition of projection (see
Figure~\ref{fig:projAsP}), $\proj\G\pp=\oup\q{i}{I}{\M}{\PP}$ implies
$\G = \agtO{\pp}{\q}i I{\la}{\G}$ with $\proj{\G_i}\pp=\PP_i$ for all $i\in I$.
Then by Rule $\rulename{\AsOut}$ we may conclude \linebreak
$\G\parG\Msg\stackred{\CommAs\pp{\la_i}\q}\G_i\parG\addMsg\Msg{\mq\pp{\la_i}\q}$
for all $i\in I$.

\vspace*{1.6mm}\noindent
{\it Case $d>1$.}  In this case either   i)  $\G=\agtO{\pr}{\ps}j
J{\la'}{\G}$ with
%$\pp\not=\pr$
$\pr\not=\pp$  or  ii)    $\G=\agtI \pr\ps \la \G$ with
%$\pp\not=\ps$.
$\ps\not=\pp$.

 \vspace*{1.6mm}\noindent
i)  There are three subcases. \sm
 \\
 If $\ps=\pp$ and $\eh{J}=1$, say $J=\set{1}$, then
$\G=\Seq{\CommAs\pr{\la'_1}\pp}{\G_1}$.  By definition of projection
and by assumption $\proj\G\pp=\proj{\G_1}\pp=\oup\q{i}{I}{\M}{\PP}$.
By  \refToLemma{ddA}(\ref{ddA1})
$\weight(\G,\pp)>\weight(\G_1,\pp)$.  By \refToLemma{wftr}
$\G_1\parG\addMsg\Msg{\mq\pr{\la'_1}\pp}$ is well formed.  Then by
induction

\vspace*{1.6mm}
\centerline{$\G_1\parG\addMsg\Msg{\mq\pr{\la'_1}\pp}
\stackred{\CommAs{\pp}{\la_i}{\q}}
\G'_i\parG\addMsg{\addMsg\Msg{\mq\pr{\la'_1}{\pp}}}{\mq\pp{\la_i}\q}$}

\eject

\noindent and $\proj{\G'_i}\pp= \PP_i$ for all $i\in I$.
Since
$\addMsg{\addMsg\Msg{\mq\pr{\la'_1}\pp}}{\mq\pp{\la_i}\q}\equiv\addMsg{\addMsg\Msg{\mq\pp{\la_i}\q}}{\mq\pr{\la'_1}\pp}$,
by Rule $\rulename{IComm-Out}$ we get
$\G\parG\Msg\stackred{\CommAs\pp{\la_i}\q}\Seq{\CommAs\pr{\la'_1}\pp}{\G'_i}\parG\addMsg\Msg{\mq\pp{\la_i}\q}$
for all $i\in I$. By definition of projection
$\proj{(\Seq{\CommAs\pr{\la'_1}\pp}{\G'_i})}{\pp}=\proj{\G'_i}{\pp}$
and so $\proj{(\Seq{\CommAs\pr{\la'_1}\pp}{\G'_i})}{\pp}=\PP_i$ for all $i\in I$.

\vspace*{1.8mm}\noindent
If $\ps=\pp$ and $\eh{J}>1$, by definition of projection and the
assumption that $\proj\G\pp$ is a choice of output actions on $\q$ we
have that $\proj\G\pp=\sendL\q\la;\PP$ with $\PP= \procActs
\,\inpP\pr{j}{J}{\M'}{\Q_j}$ and $\proj{\G_j}\pp=\sendL\q\la;
\procActs \rcvL\pr{\la'_j};\Q_j$ for all $j\in J$.  By
\refToLemma{ddA}(\ref{ddA1}) $\weight(\G,\pp)>\weight(\G_j,\pp)$ for
all $j\in J$.  By \refToLemma{wftr}
$\G_j\parG\addMsg\Msg{\mq\pr{\la'_j}\ps}$ is well formed.  This
implies %by induction
$\G_j\parG\addMsg\Msg{\mq\pr{\la'_j}\ps}
\stackred{\CommAs{\pp}{\la}{\q}}
\G'_j\parG\addMsg{\addMsg\Msg{\mq\pr{\la'_j}\ps}}{\mq\pp{\la}\q}$ and
$\proj{\G'_j}\pp= \procActs \rcvL\pr{\la'_j};\Q_j$ for all $j\in J$ by
induction.  Since
$\addMsg{\addMsg\Msg{\mq\pr{\la'_j}\ps}}{\mq\pp{\la}\q}\equiv\addMsg{\addMsg\Msg{\mq\pp{\la}\q}}{\mq\pr{\la'_j}\ps}$,
by Rule $\rulename{IComm-Out}$ we get
$\G\parG\Msg\stackred{\CommAs\pp{\la}\q}\agtO{\pr}{\pp}j
J{\la'}{\G'}\parG\addMsg\Msg{\mq\pp{\la}\q}$.
 Lastly
$\proj{(\agtO{\pr}{\pp}j J{\la'}{\G'})}\pp= \procActs
\,\inpP\pr{j}{J}{\M'}{\Q_j}$ since $\proj{\G'_j}\pp= \procActs
\rcvL\pr{\la'_j};\Q_j$.
%Since $ \procActs\,\inpP\pr{j}{J}{\M'}{\Q_j} = \PP$,
We may then conclude that $\proj{(\agtO{\pr}{\pp}j J{\la'}{\G'})}\pp=
\PP$.

\vspace*{1.8mm}\noindent
If $\ps\not=\pp$, then by definition of projection
$\proj\G\pp=\proj{\G_j}\pp$ for all $j\in J$.  By
\refToLemma{ddA}(\ref{ddA1}) $\weight(\G,\pp)>\weight(\G_j,\pp)$ for
all $j\in J$.  Then by induction $\G_j\parG\Msg
\stackred{\CommAs{\pp}{\la_i}{\q}} \G_{i,j}\parG\addMsg\Msg{\mq\pp{\la_i}\q}$ and $\proj{\G_{i,j}}\pp=\PP_i$ for all $i\in I$ and all $j\in J$. By Rule $\rulename{IComm-Out}$

\vspace*{1.8mm}
\centerline{$\G\parG\Msg\stackred{\CommAs\pp{\la_i}\q}\agtOP{\pr}{\ps}j J{\la'}{\G_{i,j}}\parG\addMsg\Msg{\mq\pp{\la_i}\q}$}

\vspace*{1.8mm}\noindent
for all $i\in I$. By definition of projection $\proj{(\agtOP{\pr}{\ps}j J{\la'}{\G_{i,j}})}\pp=\proj{\G_{i,j}}\pp=\PP_i$
for all $i\in I$.

\vspace*{1.6mm}\noindent
ii) The proof of this case is similar and simpler than the proof of
 Case  i).  It uses Lemmas~\ref{ddA}(\ref{ddA2}) and~\ref{wftr} and Rule
$\rulename{IComm-In}$, instead of Lemmas~\ref{ddA}(\ref{ddA1})
and~\ref{wftr} and Rule $\rulename{IComm-Out}$. Note that, in order to
apply Rule $\rulename{IComm-In}$, we need
$\Msg\equiv\addMsg{\mq\pr{\la}\ps}\Msg'$. This derives from
 \balancing\ of $\confAs{\agtI \pr\ps \la \G'}{\Msg}$
%\oi\ matching of $\agtI \pr\ps \la \G' $ for the queue $\Msg$
using Rule $\rulename{In}$ of Figure~\ref{wfagtA}.

\vspace*{1.8mm}
(\ref{keysrA2})
The proof is by induction on $d=\weight(\G,\q)$.

\medskip\noindent
{\it Case $d=1$.}
%In this case $\G= \agtI{\pp}{\q}{\la}{\G'}$.
By definition of projection and the hypothesis
$\proj\G\q=\inp\pp{i}{I}{\M}{\PP}\,$, it must be $\G=
\agtI{\pp}{\q}{\la}{\G'}$ and $\eh I= 1$, say $I=\set k$, and $\la =
\la_k$ and $\proj{\G'}{\q}= \PP_k$.
Then by Rule $\rulename{\AsIn}$ we deduce $\G\parG\addMsg{\mq\pp{\la_k}\q}\Msg'\stackred{\CommAsI\pp{\la_k}\q}\G'\parG\Msg'$.
%$\G\parG\addMsg{\mq\pp{\la_k}\q}\Msg'\stackred{\CommAsI\pp{\la_k}\q}\G'\parG\Msg'$.

\vspace*{1.6mm}\noindent
 $\G=\agtO{\pr}{\ps}j J{\la'}{\G}$ with  $\pr\not=\q$
or  ii) $\G=\agtI \pr\ps \la
 \G'$ with  $\ps\not=\q$.  \vspace*{1.8mm}
 \\
 i) %Again
There  are two subcases, depending
on whether $\ps = \q$ or $\ps \neq \q$.  The  most
interesting case is  the first one, namely
$\G=\agtO{\pr}{\q}j J{\la'}{\G}$. By definition of projection $\proj{\G}{\q}= \procActs
\,\inpP\pr{j}{J}{\M'}{\Q_j}$, where $\proj{\G_j}\q= \procActs
\,\Seq{\rcvL\pr{\M'_j}}{\Q_j}$.  By assumption
$\proj\G\q=\inp\pp{i}{I}{\M}{\PP}$, thus it must be either $\procActsS=\ee$
or $\eh I= 1$, say $I=\set k$, and $\procActsS=\pp?\la_k;\procActsPS$.

\vspace*{1.8mm}\noindent
 If $\procActsS=\ee$,  we have that $\pr=\pp$ and $J=I$ and
$\la'_i=\la_i$ and $\Q_i=\PP_i$ for all $i\in I$. This means
that $\G=\agtO{\pp}{\q}i I{\la}{\G}$ and $\proj{\G_i}\q= \Seq{\rcvL\pp{\la_i}}{\PP_i}$.
Let $\Msg_i  \equiv \addMsg{\mq\pp{\la}\q}{\Msg'_i}$, where
$\Msg'_i = \addMsg {\Msg'}{\mq\pp{\la_i}\q}$. By \refToLemma{wftr}
$\G_i\parG\Msg_i$
is well formed for all $i\in
I$. By  \refToLemma{ddA}(\ref{ddA1})
$\weight(\G,\q)>\weight(\G_i,\q)$ for all $i\in I$.
By induction hypothesis,
$\G_i\parN\Msg_i\stackred{\CommAsI\pp{\la}\q}\G_{i}'\parN \Msg'_i$ and
$\la = \la_i$ and
$\proj{\G'_{i}}\q=\PP_{i}$ for all $i\in I$.
This implies that  $\eh I=1$, say $I=\set k$.    %$I$ is a singleton.
Then $\G=\Seq{\CommAs{\pp}{\la}{\q}}{\G_k}$
and
by Rule $\rulename{IComm-Out}$ we deduce
$\G\parN\Msg\stackred{\CommAsI\pp{\la}\q}\G'\parN{\Msg'}$, where $\G'
= \Seq{\CommAs{\pp}{\la}{\q}}{\G'_k}$.
Whence by definition of projection  $\proj{\G}\q=\proj{\G_k}\q=\Seq{\rcvL{\pp}{\la_k}}{\PP_k}$ and %$\proj{\G'}\q=\proj{\G_k'}\q$ and $\proj{\G_k'}\q=\PP_k$.
 $\proj{\G'}\q=\proj{\G_k'}\q=\PP_k$.

 \eject

 \noindent If $\procActsS=\pp?\la_k;\procActsPS$, then
$\proj{\G}{\q}= \Seq{\rcvL{\pp}{\la_k}}{\PP_k}$,
%$\proj{\G}{\q}= \Seq{\rcvL{\pp}{\la}}{\PP_k}$
where $\PP_k = \procActsP \,\inpP\pr{j}{J}{\M'}{\Q_j}$.  Let $\Msg_j
\equiv\addMsg{\mq\pp{\la}\q}{\Msg'_j}$, where $\Msg'_j = \addMsg
{\Msg'}{\mq\pr{\la'_j}\q}$.  For all $j\in J$, $\G_j\parG\Msg_j$ is
well formed by \refToLemma{wftr} and $\weight(\G,\q)>\weight(\G_j,\q)$
by \refToLemma{ddA}(\ref{ddA1}).  By induction hypothesis we get $\la=
\la_k$ and
$\G_j\parN\Msg_j\stackred{\CommAsI\pp{\la}\q}\G_j'\parN\Msg'_j$ for
all $j\in J$. Let $\G'=\agtO{\pr}{\q}j J{\la'}{\G'}$.  Then
$\G\parN\Msg\stackred{\CommAsI\pp{\la}\q}\G'\parN{\Msg'}$ by Rule
$\rulename{IComm-Out}$ and
$\proj{\G'}\q=\procActsPS;\,\inpP\pr{j}{J}{\M'}{\Q_j} = P_k$.

\vspace*{1.8mm}\noindent
ii) The proof of this case is similar and simpler than the proof of
 Case  i).  It uses Lemmas~\ref{ddA}(\ref{ddA2}) and~\ref{wftr} and Rule
$\rulename{IComm-In}$, instead of Lemmas~\ref{ddA}(\ref{ddA1})
and~\ref{wftr} and Rule $\rulename{IComm-Out}$. Note that, in order to
apply Rule $\rulename{IComm-In}$, we need
$\Msg\equiv\addMsg{\mq\pr{\la}\ps}\Msg'$. This derives from
 \balancing\ of $\confAs{\agtI \pr\ps \la \G'}{\Msg}$
%\oi\ matching of $\agtI \pr\ps \la \G' $ for the queue $\Msg$
using Rule $\rulename{In}$ of Figure~\ref{wfagtA}.
\end{proof}

\begin{lemmaa}{\ref{keysrA34}}{Let $\G\parN\Msg$ be well formed.
\begin{enumerate}
\item If $\G\parN\Msg\stackred{\CommAs\pp{\la}\q}\G'\parN\Msg'$,  then
  $\Msg'\equiv\addMsg{\Msg}{\mq\pp\la\q}$ and
  $\proj\G\pp=\oup\q{i}{I}{\M}{\PP}$
 and $\M=\M_k$
  and $\proj{\G'}\pp=\PP_k$ for some $k\in I$
  and  $\proj{\G}\pr\subt\proj{\G'}\pr$ for all   $\pr\not=\pp$.
\item If $\G\parN\Msg\stackred{\CommAsI\pp{\la}\q}\G'\parN\Msg'$,  then  $\Msg\equiv\addMsg{\mq\pp{\la}\q}\Msg'$ and $\proj\G\q=\agtI{\pp}{\q}{\la}{\proj{\G'}\q}$
  and  $\proj\G\pr\subt\proj{\G'}\pr$ for all   $\pr\not=\q$.
  \end{enumerate}}
  \end{lemmaa}

\begin{proof}
  (\ref{keysrA3}) By induction on the inference of the transition
$\G\parG\Msg  \stackred{\CommAs{\pp}{\la}{\q}}\G'\parG\Msg'$.

\medskip\noindent
{\it Base Case.} The applied rule must be Rule $\rulename{\AsOut}$,  so $\G=  \agtO{\pp}{\q}i I{\la}{\G}$   and $\la=\la_k$ and $\G'=\G_k$  for some $k\in I$, and \vspace*{1.8mm}

\centerline{$ \agtO{\pp}{\q}i I{\la}{\G}\parG\Msg \stackred{\CommAs\pp{\la_k}\q}\G_k\parG\addMsg\Msg{\mq\pp{\la_k}\q}$}

\vspace*{2mm}\noindent
By definition of projection  $\proj{\G}\pp=\oupp\q{i}{I}{\M}{\proj{{\G_i}}\pp}$  and $\proj{\G'}\pp=\proj{\G_k}\pp$.
Again
by definition of projection, if  $\pr\not\in\{\pp,\q\}$ or $\pr=\q$ and $\eh{I}=1$,
we have  $\proj{\G}\pr=\proj{\G_1}\pr$  and so $\proj{\G}\pr=\proj{\G'}\pr$.
If $\pr=\q$ and $\eh{I}>1$, then $\proj{\G}{\q}= \procActs
\,\inpP\pp{i}{I}{\M}{\Q_i}$, where $\proj{\G_i}\q= \procActs
\,\Seq{\rcvL\pp{\M_i}}{\Q_i}$  for all $i\in I$  and so $\proj{\G}{\q}\subt\proj{\G_k}{\q}$.

\vspace*{1.8mm}\noindent
{\it Inductive Cases.} If the applied rule is $\rulename{IComm-Out}$, then $\G= \agtO{\ps}{\pt}j J{\la'}{\G} $ and $\G'= \agtO{\ps}{\pt}j J{\la'}{\G'} $
and \vspace*{1.8mm}

\centerline{$\prooftree
 \G_j\parN\addMsg\Msg{\mq\ps{\la'_j}\pt}\stackred{\CommAs\pp{\la}\q}\G_j' \parN\addMsg{\Msg'}{\mq\ps{\la'_j}\pt}\quad
j \in J \quad\pp\not=\ps
 \justifies
 \agtO{\ps}{\pt}j J{\la'}{\G}\parN\Msg\stackred{\CommAs\pp{\la}\q}\agtO{\ps}{\pt}j J{\la'}{\G'}\parN\Msg'
% \using ~~~\rulename{IComm-Out}
 \endprooftree$}

 \vspace*{1.6mm}\noindent
  By \refToLemma{wftr}
$\G_j\parG\addMsg\Msg{\mq\ps{\la'_j}\pt}$ is well formed.  By
induction hypothesis $\addMsg{\Msg'}{\mq\ps{\la'_j}\pt} \equiv
\addMsg{\addMsg\Msg{\mq\ps{\la'_j}\pt}}{\mq\pp{\la}\q}$, which implies
$\Msg' \equiv\addMsg{ \Msg}{\mq\pp{\la}\q}$.  If $\pp\not=\pt$, by
definition of projection $\proj{\G}\pp=\proj{\G_1}\pp$ and
$\proj{\G_j}\pp=\proj{\G_1}\pp$ for all $j\in J$.  Similarly
$\proj{\G'}\pp=\proj{\G'_1}\pp$ and $\proj{\G'_j}\pp=\proj{\G'_1}\pp$
for all $j\in J$. By induction hypothesis
$\proj{\G_1}\pp=\oup\q{i}{I}{\M}{\PP}$ and $\M=\M_k$ and
$\proj{\G'_1}\pp=\PP_k$ for some $k\in I$. This implies
$\proj{\G}\pp=\oup\q{i}{I}{\M}{\PP}$ and $\proj{\G'}\pp=\PP_k$.

\vspace*{1.6mm}\noindent
If $\pp=\pt$ and $\eh{J}=1$ the proof is as in the previous case by definition of projection.

\vspace*{1.6mm}\noindent
If $\pp=\pt$ and $\eh{J}>1$, then the definition of projection gives
$\proj{\G}\pp= \procActs \,\inpP\ps{j}{J}{\M'}{\Q_j} $ and
$\proj{\G_j}\pp= \procActs \,\Seq{\rcvL\ps{\M'_j}}{\Q_j}$ and
$\proj{\G'}\pp= \procActsP \,\inpP\ps{j}{J}{\M'}{\Q'_j} $ and
$\proj{\G'_j}\pp=
\procActsP \,\Seq{\rcvL\ps{\M'_j}}{\Q'_j}$ for all $j\in J$. By induction hypothesis $\procActsS=\Seq{\sendL\q\la}\procActsPS$, which implies $\proj{\G}\pp=\Seq{\sendL\q\la}{\proj{\G'}\pp}$.

%%\vspace*{1.6mm}\noindent
\eject

\noindent For $\pr\not\in\set{\pp,\ps,\pt}$ by definition of projection
$\proj{\G}\pr=\proj{\G_1}\pr$ and $\proj{\G_j}\pr=\proj{\G_1}\pr$ for
all $j\in J$. Similarly
$\proj{\G'}\pr=\proj{\G'_1}\pr$ and $\proj{\G'_j}\pr=\proj{\G'_1}\pr$ for all $j\in J$. By induction hypothesis \mbox{$\proj{\G_1}\pr\leq\proj{\G'_1}\pr$,}  which implies $\proj{\G}\pr\leq\proj{\G'}\pr$.

\vspace*{1.6mm}
For participant $\ps$ we have $ \proj{\G}\ps=
\oupp\pt{j}{J}{\M'}{\proj{\G_j}\ps}\leq\oupp\pt{j}{J}{\M'}{\proj{\G'_j}\ps}=\proj{\G'}\ps$.

\vspace*{1.6mm}
For participant $\pt\not=\pp$ if $\eh{J}=1$ %we use the same proof
the proof is the same as for $\pr\not\in\set{\pp,\ps,\pt}$.  If
$\eh{J}>1$, then we have $\proj{\G}\pt= \procActs
\,\inpP\ps{j}{J}{\M'}{\R_j} $, where $\proj{\G_j}\pt= \procActs
\,\Seq{\rcvL\ps{\M'_j}}{\R_j}$ and $\proj{\G'}\pt= \procActsP
\,\inpP\ps{j}{J}{\M'}{\R'_j} $, where $\proj{\G'_j}\pt= \procActsP
\,\Seq{\rcvL\ps{\M'_j}}{\R'_j}$. From
$\proj{\G_j}\pt\subt\proj{\G'_j}\pt$ for all $j\in J$ we get
$\procActsPS =\procActsS$ and $\R_j\subt\R'_j$ for all $j\in
J$.  This implies $\proj{\G}\pt\subt\proj{\G'}\pt$.

\vspace*{1.6mm}
If the applied rule is $\rulename{IComm-In}$ the proof is similar and
simpler.

\vspace*{1.8mm}
(\ref{keysrA4}) The proof is similar to the proof of (\ref{keysrA3}).
The most interesting case is the application of Rule $\rulename{IComm-Out}$

\centerline{$\prooftree
 \G_j\parN\addMsg\Msg{\mq\ps{\la'_j}\pt}\stackred{\CommAsI\pp{\la}\q}\G_j' \parN\addMsg{\Msg'}{\mq\ps{\la'_j}\pt}\quad
j \in J \quad\q\not=\ps
 \justifies
 \agtO{\ps}{\pt}j J{\la'}{\G}\parN\Msg\stackred{\CommAsI\pp{\la}\q}\agtO{\ps}{\pt}j J{\la'}{\G'}\parN\Msg'
 \endprooftree$}

 \vspace*{2mm}\noindent
 By \refToLemma{wftr}
$\G_j\parG\addMsg\Msg{\mq\ps{\la'_j}\pt}$ is well formed.  By
induction hypothesis $\addMsg{\Msg}{\mq\ps{\la'_j}\pt} \equiv
\addMsg{\mq\pp{\la}\q}{\addMsg{\Msg'}{\mq\ps{\la'_j}\pt}}$, which
implies $\Msg \equiv\addMsg{\mq\pp{\la}\q}{\Msg'}$.  If $\q\not=\pt$,
by definition of projection $\proj{\G}\q=\proj{\G_1}\q$ and
$\proj{\G_j}\q=\proj{\G_1}\q$ for all $j\in J$. Similarly
$\proj{\G'}\q=\proj{\G'_1}\q$ and $\proj{\G'_j}\q=\proj{\G'_1}\q$ for all $j\in J$. By induction hypothesis $\proj{\G_1}\q=\agtI{\pp}{\q}{\la}{\proj{\G'_1}\q}$. This implies $\proj{\G}\q=\agtI{\pp}{\q}{\la}{\proj{\G'}\pp}$.

\vspace*{1.8mm}\noindent
If $\q=\pt$ and $\eh{J}=1$ the proof is as in the previous case by definition of projection.

\vspace*{1.8mm}\noindent
If $\q=\pt$ and $\eh{J}>1$, then the definition of projection gives
$\proj{\G}\q= \procActs \,\inpP\ps{j}{J}{\M'}{\Q_j} $ and
$\proj{\G_j}\q= \procActs \,\Seq{\rcvL\ps{\M'_j}}{\Q_j}$ and
$\proj{\G'}\q= \procActsP \,\inpP\ps{j}{J}{\M'}{\Q'_j} $ and
$\proj{\G'_j}\q=
\procActsP \,\Seq{\rcvL\ps{\M'_j}}{\Q'_j}$ for all $j\in J$. By induction hypothesis $\procActsS=\Seq{\rcvL\pp\la}\procActsPS$, which implies $\proj{\G}\q=\Seq{\rcvL\pp\q\la}{\proj{\G'}\pp}$.

\medskip
The proof of $\proj\G\pr\subt\proj{\G'}\pr$ for all $\pr\not=\q$ is as
in  Case  (\ref{keysrA3}).
\end{proof}

\begin{lemmaa}{\ref{srgA}}
If  $\derPI{}{\confAs{\G}{\Msg}}$ and $\G\parG\Msg\stackred{\beta}\G'\parG\Msg'$, then $\derPI{}{\confAs{\G'}{\Msg'}}$.
 \end{lemmaa}

 \begin{proof}
  By induction on the inference of the transition  $\G\parG\Msg\stackred{\beta}\G'\parG\Msg'$ of Figure \ref{ltsgtAs}.

  \vspace*{1.8mm}\noindent
 {\it Base Cases.} Immediate from  \refToLemma{wftr}.

\vspace*{1.8mm}\noindent
 {\it Inductive Cases.} Let  $\G\parG\Msg\stackred{\beta}\G'\parG\Msg'$ with Rule $\rulename{IComm-Out}$.
Then we get
$\G= \agtO{\pp}{\q}i I{\la}{\G}$ and $\G'=\agtO{\pp}{\q}i I{\la}{\G'}$ and
 $\G_i\parG\addMsg{\Msg}{\mq\pp{\la_i}\q}\stackred\asCom\G_i' \parG\addMsg{\Msg'}{\mq\pp{\la_i}\q}$ for all
 $i \in I$.
From Rule $\rulename{Out}$
of \refToFigure{wfagtA}, we get $\derSI{}{\G_i\parN\addMsg\Msg{\mq\pp{\la_i}\q}}$ for all
 $i \in I$.
By induction hypotheses for all $i\in I$ we can derive $\derSI{}{\G'_i\parN\addMsg\Msg{\mq\pp{\la_i}\q}}$.
Therefore using Rule $\rulename{Out}$ we conclude $\derSI{}{\G'\parN\Msg'}$.

Similarly for Rule $\rulename{IComm-In}$.
 \end{proof}

\begin{lemmaa}{\ref{prop:prePostNetArel}}{\begin{enumerate}
\item  If  $\netevA\precN \netevA'$ and $\postA{\netevA}{\beta}$ and
  $\postA{\netevA'}{\beta}$ and $\mapBl{\beta}\os$ are defined,\\ then
  $\postA{\netevA}{\beta}\precNL{\mapBl{\beta}\os} \postA{\netevA'}{\beta}$.
\item   If  $\netevA\precNL{\os} \netevA'$ and $\mapWh{\beta}\os$ is defined,
then $\preA{\netevA}{\beta}\precNL{\mapWh{\beta}\os} \preA{\netevA'}{\beta}$.
\item  If  $\netevA\grr
  \netevA'$ and both $\postA{\netevA}{\beta}$ and $\postA{\netevA'}{\beta}$
  are defined, then $\postA{\netevA}{\beta}\grr\postA{\netevA'}{\beta}$.
\item   If  $\netevA\grr \netevA'$,  then $\preA{\netevA}{\beta}\grr
  \preA{\netevA'}{\beta}$.
 \end{enumerate}}
 \end{lemmaa}

\begin{proof}   (\ref{ppn2A})
If $\netevA\precNL{\os} \netevA'$, then %either

\vspace*{-3mm}
\begin{itemize}
\itemsep=0.9pt
\item either $\netevA=\locevA{\pp}\os{\procev}$ and
  $\netevA'=\locevA{\pp}{\os}{\procev'}$ and $\procev < \procev'$, % or
\item or $\netevA=\locevA\pp\os{\concat\actseq{\sendL\q\la}}$ and
  $\netevA'= \locevA\q{\os}{\concat{\actseq'}{\rcvL\pp\la}}$ and
  $\dualprecsim{\projs{(\concat{\pro\os\pp}{\actseq})}{\q}}{\projs{(\concat{\pro\os\q}{\actseq''})}{\pp}}$ \sm
  \\
  for some $\actseq'' $ and $\preEv$ such that $
  \projs{(\concat{\actseq'}{\rcvL\pp\la})} \pp\precsim
  \projs{(\concat{\concat{\actseq''}{\rcvL\pp\la}}{\preEv})}\pp $.
\end{itemize}
 In the first case, from  the fact that
$\postA{\netevA}{\beta}$ and $\postA{\netevA'}{\beta}$  are
defined and \refToDef{def:PostPre}(\ref{def:Post}) we get
$\postA{\netevA}{\beta}= \locevA\pp{}{\procev_1}$ and
$\postA{\netevA'}{\beta}= \locevA\pp{}{\procev'_1}$, where
$\procev=\concat{\projAP\beta{\pp}}{\procev_1}$ and
$\procev'=\concat{\projAP\beta{\pp}}{\procev'_1}$.  Since
$\procev_1< \procev'_1$ we conclude
$\postA{\netevA}{\beta}\precNL{\mapBl{\beta}\os} \postA{\netevA'}{\beta}$.

\medskip\noindent
In the second case, let $\os'=\mapBl{\beta}\os$.

\smallskip
If $\play{\beta}\not\subset\set{\pp,\q}$, then
$\projAP\beta{\pp}=\projAP\beta{\q}=\epsilon$ and
$\postA{\netevA}{\beta}=\netevA$ and $\postA{\netevA'}{\beta}=
\netevA'$. Moreover, by \refToDef{def:pQueue}(\ref{def:pQueue2})
$\projs{(\projAP{\os'}{\pp})}\q=\projs{(\projAP{\os}{\pp})}\q$ and
$\projs{(\projAP{\os'}{\q})}\pp=\projs{(\projAP{\os}{\q})}\pp$. Therefore
$\dualprecsim{\projs{(\concat{\pro{\os}\pp}{\actseq})}{\q}}{\projs{(\concat{\pro{\os}\q}{\actseq''})}{\pp}}$
implies
$\dualprecsim{\projs{(\concat{\pro{\os'}\pp}{\actseq})}{\q}}{\projs{(\concat{\pro{\os'}\q}{\actseq''})}{\pp}}$
which proves that $\postA{\netevA}{\beta}\precNL{\os'}
\postA{\netevA'}{\beta}$.

\smallskip
If $\play{\beta}=\set\pp$,  then either
$\beta=\CommAs\pp{\la'}\pr$ or $\beta=\CommAsI\pr{\la'}\pp$.\\
If $\beta=\CommAs\pp{\la'}{\pr}$, then
$\projAP\beta{\pp}=\sendL{\pr}{\la'}$ and $\projAP\beta{\q}=\epsilon$.
By \refToDef{def:PostPre}(\ref{def:Post}), since
$\postA{\netevA}{\beta}$ is defined we have
$\actseq=\concat{\sendL{\pr}{\la'}}{\actseq_1}$. Then
$\postA{\netevA}{\beta}=\locevA\pp{}{\concat{\actseq_1}{\sendL{\q}\la}}$
and $\postA{\netevA'}{\beta}= \netevA'$.  Moreover, by
\refToDef{def:pQueue}(\ref{def:pQueue2})
$\pro{\os'}\pp=\concat{(\pro{\os}\pp)}{\sendL{\pr}{\la'}}$ and
$\pro{\os'}\q=\pro{\os}\q$. Therefore
$\concat{\pro\os\pp}{\actseq =
  \concat{\pro\os\pp}{\concat{\sendL{\pr}{\la'}}{\actseq_1}}=
  \pro{\os'}\pp\cdot\actseq_1}$  and
$\concat{\pro\os\q}{\actseq''} =
\concat{\pro{\os'}\q}{\actseq''}$.  Then,  from the fact
that
$\dualprecsim{\projs{(\concat{\pro\os\pp}{\actseq})}{\q}}{\projs{(\concat{\pro\os\q}{\actseq''})}{\pp}}$
it follows that $\dualprecsim{
  \projs{(\pro{\os'}\pp\cdot\actseq_1)}{\q} }{
  \projs{(\concat{\pro{\os'}\q}{\actseq''})}{\pp}}$.

\smallskip
If  $\beta=\CommAsI\pr{\la'}\pp$,  then
$\projAP\beta{\pp}=\rcvL{\pr}{\la'}$ and $\projAP\beta{\q}=\epsilon$.
By \refToDef{def:PostPre}(\ref{def:Post}),  since   $\postA{\netevA}{\beta}$
 is  defined  we have   $\actseq=\concat{\rcvL\pr{\la'}}{\actseq_1}$. Then
$\postA{\netevA}{\beta}=\locevA\pp{}{\concat{\actseq_1}{\sendL\q\la}}$
and $\postA{\netevA'}{\beta}= \netevA'$.
We now distinguish two subcases,
according to whether $\pr = \q$ or $\pr \neq \q$.

\vspace*{1.8mm}%\noindent
 If $\pr = \q$,  then by
\refToDef{def:pQueue}(\ref{def:pQueue2}) $\pro\os\pp=\pro{\os'}\pp$
and $\pro\os\q=\concat{\sendL{\pp}{\la'}}{(\pro{\os'}\q)}$.
Therefore \linebreak  we get
$\concat{\pro\os\pp}{\actseq}=
\concat{(\pro{\os'}\pp)}{\rcvL\q{\la'}\cdot\actseq_1}$
and
$\concat{\pro\os\q}{\actseq''}=\sendL{\pp}{\la'}\cdot\concat{\pro{\os'}\q}{\actseq''}$.
 Then, from the fact that
$\dualprecsim{\projs{(\concat{\pro\os\pp}{\actseq})}{\q}}{\projs{(\concat{\pro\os\q}{\actseq''})}{\pp}}$
and
%the fact that
$\projs{(\pro{\os'}\pp)}{\q}$ cannot contain inputs,
% imply
 it follows that \linebreak %have that
$\dualprecsim{
  \concat{\projs{(\pro{\os'}\pp)}{\q}}{\projs{\actseq_1}{\q}} }{
  \projs{(\concat{\pro{\os'}\q}{\actseq''})}{\pp} }$.

\vspace*{1.8mm}%\noindent
 If $\pr \neq \q$, then by
\refToDef{def:pQueue}(\ref{def:pQueue2}) $\pro\os\pp=\pro{\os'}\pp$
and $\pro\os\q=\pro{\os'}\q$.  In this case we get
$\projs{(\concat{\pro\os\pp}{\actseq})}{\q}=
\projs{(\concat{\pro{\os'}\pp}{\concat{\rcvL{\pr}{\la'}}{\actseq_1}})}{\q}
= \projs{(\concat{\pro{\os'}\pp}{\actseq_1})}{\q}$ and
 $\concat{\pro\os\q}{\actseq''}=\concat{\pro{\os'}\q}{\actseq''}$.
Then,  from $\dualprecsim{\projs{(\concat{\pro\os\pp}{\actseq})}{\q}}{\projs{(\concat{\pro\os\q}{\actseq''})}
  {\pp}}$ it follows that $\dualprecsim{ \projs{(\pro{\os'}\pp\cdot\actseq_1)}{\q} }{
  \projs{(\concat{\pro{\os'}\q}{\actseq''})}{\pp}}$.

 \medskip
If $\play{\beta}=\set\q$ the proof is similar.

\vspace*{1.8mm}
(\ref{ppn2bA}) The proof is similar to that of Fact   (\ref{ppn2A}).

\vspace*{1.8mm}
(\ref{ppn3A})
Let $\netevA=\locevA{\pp}\os{\procev}$ and $\netevA'=\locevA{\pp}{\os}{\procev'}$ and $\procev\grr\procev'$.
From $\postA{\netevA}{\beta}$ and
  $\postA{\netevA'}{\beta}$ defined we get $\procev=\concat{\projAP\beta{\pp}}{\procev_1}$ and $\procev'=\concat{\projAP\beta{\pp}}{\procev'_1}$ and $\postA{\netevA}{\beta}=
 \locevA\pp{}{\procev_1}$ and
 $\postA{\netevA'}{\beta}=
 \locevA\pp{}{\procev'_1}$ by \refToDef{def:PostPre}(\ref{def:Post}).

 \medskip\noindent
Since $\procev\grr\procev'$ implies $\procev_1\grr\procev'_1$
 we conclude $\postA{\netevA}{\beta}\grr \postA{\netevA'}{\beta}$.
\end{proof}

%%%%%%%%%%%%%%%%%%%%

In the following we use the notation $\overline{\CommAsI\pp\la\q}$ defined by
$\overline{\CommAsI\pp\la\q}=\CommAs\pp\la\q$.

\begin{lemmaa}{\ref{tr}}{Let $\play{\beta_1}\cap\play{\beta_2}=\emptyset$.
\begin{enumerate}
\item
If both
  $\mapBl{\beta_2}{\os}$ and $\mapBl{\beta_2}{({\mapWh{\beta_1}{\os}})}$ are defined,
then
$\mapWh{\beta_1}{{(\mapBl{\beta_2}{\os})}}\cong
 \mapBl{\beta_2}{({\mapWh{\beta_1}{\os}})}$.
 \item  If both $\mapWh{\beta_1}{\os}$ and
  $\mapWh{\beta_2}{\os}$ are defined,
then
$\mapWh{\beta_1}{{(\mapWh{\beta_2}{\os})}}$ is defined and
$\mapWh{\beta_1}{{(\mapWh{\beta_2}{\os})}}\cong
 \mapWh{\beta_2}{({\mapWh{\beta_1}{\os}})}$.
\end{enumerate}}
\end{lemmaa}

\begin{proof}
 (\ref{tr2}) Since
  $\os_2 = \mapBl{\beta_2}{\os}$ is defined, by \refToDef{def:pQueue}(\ref{def:pQueue2})
$\os\cong\overline{\beta_2}\cdot \os_2$ when
$\beta_2$ is an input. Since
$\mapBl{\beta_2}{({\mapWh{\beta_1}{\os}})}$ is defined,
$\os_1 = \mapWh{\beta_1}{\os}$ is defined and
by \refToDef{def:pQueue} $\os\cong
\os_1\cdot\beta_1$ when $\beta_1$ is an output and
$\os\cong\overline{\beta_2}\cdot \os_0\cdot\beta_1$ for some $\os_0$
 such that $\os_1 \cong \concat{\overline{\beta_2}}{\os_0}$ and
$\os_2 \cong \concat {\os_0}{\beta_1}$,
when $\beta_1$ is an output and $\beta_2$ is an input. Using \refToDef{def:pQueue} we compute:

\vspace*{3mm}
\centerline{$
\mapWh{\beta_1}{{(\mapBl{\beta_2}{\os})}}  \cong
\mapBl{\beta_2}{({\mapWh{\beta_1}{\os}})}  \cong \begin{cases}
   \os_1\cdot\beta_2   & \text{if both $\beta_1$ and $\beta_2$ are outputs}\\
   \os_0  & \text{if $\beta_1$ is an output and $\beta_2$ is an input}\\
    \overline{\beta_1}\cdot\os\cdot\beta_2 & \text{if $\beta_1$ is an input and $\beta_2$ is an output}\\
 \overline{\beta_1}\cdot\os_2   & \text{if both $\beta_1$ and $\beta_2$ are inputs}
\end{cases}
$}

\medskip (\ref{tr1})  Since
  $\os_i = \mapWh{\beta_i}{\os}$ is defined for $i\in\set{1,2}$, by
  \refToDef{def:pQueue}(\ref{def:pQueue1})  $\os\cong
  \os_i\cdot\beta_i$ when $\beta_i$ is an output. Then from
$\play{\beta_1}\cap\play{\beta_2}=\emptyset$ we get
$\os\cong \os'\cdot\beta_1\cdot\beta_2 \cong
\os'\cdot\beta_2\cdot\beta_1$ for some $\os'$ when both
$\beta_1$ and $\beta_2$ are outputs.
Using \refToDef{def:pQueue}(\ref{def:pQueue1}) we compute:

\vspace{2mm}
\centerline{$ \mapWh{\beta_1}{{(\mapWh{\beta_2}{\os})}}
  \cong  \mapWh{\beta_2}{{(\mapWh{\beta_1}{\os})}}
  \cong  \begin{cases}
    \os'   & \text{if both $\beta_1$ and $\beta_2$ are outputs}\\
    \overline{\beta_i}\cdot\os_j  & \text{if $\beta_i$ is an input and $\beta_j$ is an output}\\
    \overline{\beta_1}\cdot\overline{\beta_2}\cdot\os & \text{if both
      $\beta_1$ and $\beta_2$ are inputs}%\qquad\qquad\hfill\qedhere
\end{cases}
$}\end{proof}

%%%%%%%%%%%%%%%%%%%%%%%%%%%%%%%%%%%%%%%%%%%%

\begin{lemma}\label{scb}
\begin{enumerate}
\item\label{scb5}\label{scb3}\label{scb4} If $\postGA{\eqA{\os}{\filtP{\beta}{{\os}}\comseqA}}{\beta}$ is defined, then $\os\cdot\beta\cdot\comseqA$ is well formed and\sm

 \centerline{
$\postGA{\eqA{\os}{\filtP{\beta}{{\os}}\comseqA}}{\beta}=\eqA{\mapBl{\beta}{\os}}{\comseqA}$}\sm

\item\label{scb0}\label{scb1}\label{scb2}If $ \cauA{\beta}{\eqA{\os}\comseqA}$ is defined, then $(\mapWh{\beta}{\os})\cdot\beta\cdot\comseqA$ is well formed and\sm

\centerline {$\cauA{\beta}{\eqA{\os}\comseqA}=\eqA{\mapWh{\beta}{\os}}{\filtP{\beta}{(\mapWh{\beta}{\os})}{\comseqA}}$}
\end{enumerate}
\end{lemma}

\begin{lemmaa}{\ref{prop:prePostGlA}}
{\begin{enumerate}
\item If  $\postGA{\comoccA}{\beta}$ is defined,
then
$\preGA{(\postGA{\comoccA}{\beta})}{\beta}=\comoccA$.
\item    If  $\preGA{\comoccA}{\beta}$ is defined,
then
$\postGA{(\preGA{\comoccA}{\beta})}{\beta}=\comoccA$.
\item If both
  $\postGA{\comoccA}{\beta_2}$,
  $\postGA{(\preGA{\comoccA}{\beta_1})}{\beta_2}$ are defined, and
  $\play{\beta_1}\cap\play{\beta_2}=\emptyset$, then
$\preGA{(\postGA{\comoccA}{\beta_2})}{\beta_1}=\postGA{(\preGA{\comoccA}{\beta_1})}{\beta_2}$.
\item  If both $\preGA{\comoccA}{\beta_1}$,  $\preGA{\comoccA}{\beta_2}$ are defined,  and  $\play{\beta_1}\cap\play{\beta_2}=\emptyset$,  then $\preGA{(\preGA{\comoccA}{\beta_2})}{\beta_1}$ is defined and
$\preGA{(\preGA{\comoccA}{\beta_2})}{\beta_1}=\preGA{(\preGA{\comoccA}{\beta_1})}{\beta_2}$.
\end{enumerate}}
\end{lemmaa}

\begin{proof}
  Statements (\ref{ppg1}) and (\ref{ppg0a}) immediately follow from
  \refToLemma{scb}.  In the proofs of the remaining statements we
  convene that ``$\beta$ is required in
  $\comseqA_1\cdot\beta\cdot\comseqA_2$'' is short for ``the shown
  occurrence of $\beta$ is required in
  $\comseqA_1\cdot\beta\cdot\comseqA_2$'' and similarly for ``$\beta$
  matches an  output  in $\comseqA_1\cdot\beta\cdot\comseqA_2$''.

  \medskip (\ref{ppg5})
 Let $\comoccA=\eqA{\os}{\comseqA}$.  Since both
 $\postGA{\comoccA}{\beta_2}$ and
 $\postGA{(\preGA{\comoccA}{\beta_1})}{\beta_2}$ are defined, by
 \refToLemma{scb} both
 $\mapBl{\beta_2}{\os}$ and
 $\mapBl{\beta_2}{({\mapWh{\beta_1}{\os}})}$ must be defined.  Then, by \refToLemma{tr}(\ref{tr2})
 $\mapBl{\beta_2}{{(\mapWh{\beta_1}{\os})}}\cong
 \mapWh{\beta_1}{({\mapBl{\beta_2}{\os}})}$.  So we set
$\os'=\mapWh{\beta_1}{({\mapBl{\beta_2}{\os}})}$.
Let $\os_1=\mapWh{\beta_1}{\os}$ .
By %\refToLemma{scb}(\ref{scb0})
 \refToDef{def:PostPreGlA}(\ref{def:PreGl1A})
we get

\vspace*{1.8mm}
\centerline{ $\comoccA_1 =
  \cauA{\beta_1}{\comoccA}=\begin{cases}
  \eqA{\os_1}{ \beta_1\cdot\comseqA}    & \text{if  $ \beta_1\cdot\comseqA$ is $\os_1$-pointed } \\
      \eqA{\os_1}{\comseqA}     & \text{otherwise}
\end{cases}$ }

\vspace*{1.6mm}
  Let $\os_2=\mapBl{\beta_2}{\os}$.
By \refToDef{def:PostPreGlA}(\ref{def:PostPreGl1A})
we get

\vspace*{1.8mm}
\centerline{ $\comoccA_2 =
\postGA{\comoccA}{\beta_2}=\begin{cases}
\eqA{\os_2}{\comseqA'}      & \text{if }\comseqA\approx_\os\concat{\beta_2}{\comseqA'} \\
  \eqA{\os_2}\comseqA    & \text{if }
\play{\beta_2}\cap\play{\comseqA} =\emptyset
\end{cases}$}

\vspace*{1.8mm}\noindent
The remainder of this proof is split  into  two cases, according to
the shape of $\comoccA_2$.

\noindent\medskip
{\em Case $\comoccA_2 = \eqA{\os_2}\comseqA$.}  Then
$\play{\beta_2}\cap\play\comseqA =\emptyset$.
By \refToDef{def:PostPreGlA}(\ref{def:PreGl1A})  we get

\vspace*{1.5mm}
\centerline{$\begin{array}{lll}\preGA{\comoccA_2}{\beta_1}&=&\begin{cases}
     \eqA{\os'}{ \beta_1\cdot\comseqA }   & \text{if  $ \beta_1\cdot\comseqA$ is $\os'$-pointed} \\
     \eqA{\os'}{ \comseqA}& \text{otherwise}
\end{cases}
 \end{array}$}

\medskip
Since $\play{\beta_2}\cap\play{\beta_1\cdot\comseqA}=\emptyset$,
by \refToDef{def:PostPreGlA}(\ref{def:PostPreGl1A}) we get

\vspace*{1.5mm}
\centerline{$\begin{array}{lll}\postGA{\comoccA_1}{\beta_2}&=&\begin{cases}
    \eqA{\os'}{ \beta_1\cdot\comseqA}    & \text{if  $ \beta_1\cdot\comseqA$ is $\os_1$-pointed} \\
    \eqA{\os'}{ \comseqA} & \text{otherwise}
\end{cases}
 \end{array}$}
We have to show that\sm

\centerline{$ (**)\ \ \beta_1\cdot\comseqA$ is $\os'$-pointed iff $
  \beta_1\cdot\comseqA$ is $\os_1$-pointed}

  \medskip\noindent If $\beta_1$ is an input,
it must be required in $\comseqA$  for both $\os'$-pointedness
and $\os_1$-pointedness,  so this case is obvious.

\medskip
Let $\beta_1=\CommAs{\pp}{\la}{\q}$.

\vspace*{1.8mm}\noindent
If $\beta_2$ is an output, then $\os'\cong\os_1\cdot\beta_2$ by
\refToDef{def:pQueue}(\ref{def:pQueue2}).  Since
$\beta_2\neq\CommAs{\pp}{\la'}{\q}$ for all $\la'$,  an  input  %$\beta_0$
in
$\comseqA$ matches $\beta_1$ in
$\os_1\cdot\beta_2\cdot\beta_1\cdot\comseqA$ iff
it matches $\beta_1$ in $\os_1\cdot\beta_1\cdot\comseqA$.

\vspace*{1.8mm}\noindent
If $\beta_2$ is  an input, then $\os_1\cong\overline{\beta_2}\cdot\os'$
by \refToDef{def:pQueue}(\ref{def:pQueue2}). If
$\beta_2\not=\CommAsI{\pp}{\la'}{\q}$ for all $\la'$, then an  input  %$\beta_0$
in $\comseqA$ matches $\beta_1$ in $\os'\cdot\beta_1\cdot\comseqA$ iff
it matches $\beta_1$ in
$\overline{\beta_2}\cdot\os'\cdot\beta_1\cdot\comseqA$.  Let
$\beta_2=\CommAsI{\pp}{\la'}{\q}$ for some $\la'$. Since
$\play{\beta_2}\cap\play{\comseqA} \neq \emptyset$, there is no  input  $\beta_0$ in
$\comseqA$ such that $\beta_0$ matches $\beta_1$ in
$\os'\cdot\beta_1\cdot\comseqA$ or in
$\overline{\beta_2}\cdot\os'\cdot\beta_1\cdot\comseqA$.
This concludes the proof of $(**)$.

\medskip\noindent
{\em Case $\comoccA_2 = \eqA{\os_2}{\comseqA'}$. } Then
$\comseqA\approx_\os\concat{\beta_2}{\comseqA'}$.
By \refToDef{def:PostPreGlA}(\ref{def:PreGl1A})  we get \sm

\centerline{$\begin{array}{lll}\preGA{\comoccA_2}{\beta_1}&=&\begin{cases}
     \eqA{\os'}{ \beta_1\cdot\comseqA' }   & \text{if  $ \beta_1\cdot\comseqA' $ is $\os'$-pointed} \\
     \eqA{\os'}{ \comseqA' }& \text{otherwise}
\end{cases}
\end{array}$}

\vspace*{1.6mm}\noindent
and, since $\comoccA = \eqA{\os}{\beta_2\cdot\comseqA'}$, by the
same definition we get \vspace*{1.6mm}

\centerline{$\begin{array}{lll}\preGA{\comoccA}{\beta_1}&=&\begin{cases}
      \eqA{\os_1}{\beta_1\cdot\beta_2\cdot\comseqA'}    & \text{if  $ \beta_1\cdot\beta_2\cdot\comseqA'$ is $\os_1$-pointed} \\
      \eqA{\os_1}{\beta_2\cdot\comseqA'} & \text{otherwise}
\end{cases}
 \end{array}$}

\vspace*{1.8mm}\noindent
We first show that $\beta_2\cdot\beta_1\cdot\comseqA' \approx_{\os_1}
\beta_1\cdot\beta_2\cdot\comseqA'$. Since
$\preGA{\comoccA}{\beta_1}$ is defined,  the trace
$\os_1\cdot\beta_1\cdot\beta_2\cdot\comseqA'$ is well formed by
\refToLemma{scb}(\ref{scb1}).  So $\beta_1$ cannot be a matching input
for $\beta_2$.  To show that $\beta_2$ cannot be a matching input
for $\beta_1$ observe that, if it were, then $\beta_1=\overline{\beta_2}$.
Since $\postGA{(\preGA{\comoccA}{\beta_1})}{\beta_2}$ is
defined we have that
$\os_1\equiv\overline{\beta_2}\cdot\os'$ by \refToDef{def:pQueue}(\ref{def:pQueue2}). Therefore $\beta_2$ cannot be a matching input
for $\beta_1$ in
$\overline{\beta_2}\cdot\os'\cdot\beta_1\cdot\beta_2\cdot\comseqA'$, since
it is the matching input of the first $\overline{\beta_2}$.
From this and
$\play{\beta_1}\cap\play{\beta_2}=\emptyset$ we get that $\beta_2\cdot\beta_1\cdot\comseqA' \approx_{\os_1}
\beta_1\cdot\beta_2\cdot\comseqA'$. Therefore \vspace*{1.6mm}

\centerline{$\begin{array}{lll}\preGA{\comoccA}{\beta_1}&=&\begin{cases}
      \eqA{\os_1}{\beta_2\cdot\beta_1\cdot\comseqA'}    & \text{if  $ \beta_1\cdot\beta_2\cdot\comseqA'$ is $\os_1$-pointed} \\
     \eqA{\os_1}{\beta_2\cdot\comseqA'} & \text{otherwise}
\end{cases}
 \end{array}$}

 \noindent \vspace*{1.6mm} and by \refToDef{def:PostPreGlA}(\ref{def:PostPreGl1A})

\vspace*{1.6mm}
\centerline{$\begin{array}{lll}\postGA{\comoccA_1}{\beta_2}&=&\begin{cases}
    \eqA{\os'}{ \beta_1\cdot\comseqA'}    & \text{if  $ \beta_1\cdot\beta_2\cdot\comseqA'$ is $\os_1$-pointed} \\
    \eqA{\os'}{ \comseqA' } & \text{otherwise}
\end{cases}
\end{array}$}

\vspace*{1.8mm}\noindent We have to show that\sm

\centerline{$(***)\  \ \beta_1\cdot\comseqA'$ is $\os'$-pointed iff $ \beta_1\cdot\beta_2\cdot\comseqA'$ is $\os_1$-pointed}

\vspace*{1.8mm}
Note that $\beta_1$ is required in $\comseqA'$ iff  it
is required in $\beta_2\cdot\comseqA'$ since $\play{\beta_2}\cap\play{\beta_1}=\emptyset$. Therefore the result is immediate when $\beta_1$ is an input.

\vspace{1.8mm}
Let $\beta_1$ be an output. \sm
\\
If $\beta_2$ is an output, then $\os'\cong\os_1\cdot\beta_2$ by
\refToDef{def:pQueue}(\ref{def:pQueue2}). Suppose that $\beta_1\cdot\comseqA'$ is $\os'$-pointed,
where $\comseqA' = \comseqA'_0\cdot\beta_0\cdot\comseqA''_0$ and %that
%the shown occurrence of
 $\beta_0$ matches $\beta_1$ in
$\os_1\cdot\beta_2\cdot\beta_1\cdot\comseqA'_0\cdot\beta_0\cdot\comseqA''_0$.  Then, since
$\beta_2\cdot\beta_1\cdot\comseqA' \approx_{\os_1}
\beta_1\cdot\beta_2\cdot\comseqA'$, we have that $\beta_0$ matches $\beta_1$ in
$\os_1\cdot\beta_1\cdot\beta_2\cdot\comseqA'_0\cdot\beta_0\cdot\comseqA''_0$.   In a similar way we
can prove that, if an input
$\beta_0$ matches $\beta_1$ in
$\os_1\cdot\beta_1\cdot\beta_2\cdot\comseqA'_0\cdot\beta_0\cdot\comseqA''_0$, then $\beta_0$ matches $\beta_1$ in
$\os_1\cdot\beta_2\cdot\beta_1\cdot\comseqA'_0\cdot\beta_0\cdot\comseqA''_0$.

\noindent \vspace*{1.6mm}
If $\beta_2$ is an input, then $\os_1\cong\overline{\beta_2}\cdot\os'$ by
\refToDef{def:pQueue}(\ref{def:pQueue2}). Suppose that
$\beta_1\cdot\comseqA'$ is $\os'$-pointed,
where $\comseqA' = \comseqA'_0\cdot\beta_0\cdot\comseqA''_0$ and
%the shown occurrence of
 $\beta_0$ matches $\beta_1$ in
$\os'\cdot\beta_1\cdot\comseqA'_0\cdot\beta_0\cdot\comseqA''_0$. Then $\beta_0$ matches $\beta_1$ in
$\overline{\beta_2}\cdot\os'\cdot\beta_1\cdot\beta_2\cdot\comseqA'_0\cdot\beta_0\cdot\comseqA''_0$,
since $\beta_2$ is the first input in the trace and it matches %the shown occurrence of
$\overline{\beta_2}$.   In a similar way we can prove that, if an input
$\beta_0$ matches $\beta_1$ in
$\overline{\beta_2}\cdot\os'\cdot\beta_1\cdot\beta_2\cdot\comseqA'_0\cdot\beta_0\cdot\comseqA''_0$, then
$\beta_0$ matches $\beta_1$ in
$\os'\cdot\beta_1\cdot\comseqA'_0\cdot\beta_0\cdot\comseqA''_0$.
Therefore $(***)$ holds.

\medskip (\ref{ppg6A}) Let $\comoccA=\eqA{\os}{\comseqA}$.  Since
$\preGA{\comoccA}{\beta_i}$ is defined  for $i\in\set{1,2}$,
 by \refToLemma{scb}(\ref{scb1}) $\os_i=\mapWh{\beta_i}{\os}$ is
defined for $i\in\set{1,2}$. Then by \refToLemma{tr}(\ref{tr1})
$\mapWh{\beta_1}{(\mapWh{\beta_2}{\os})}\cong
\mapWh{\beta_2}{(\mapWh{\beta_1}{\os})}$.   Let
$\os'=\mapWh{\beta_1}{(\mapWh{\beta_2}{\os})}$.

\vspace*{1.8mm}
  Using  \refToLemma{scb}(\ref{scb0}) we
get  for $i\in\set{1,2}$ \sm

\centerline{ $\comoccA_i =
  \cauA{\beta_i}{\eqA{\os}\comseqA}=\eqA{\os_i}{\filtP{\,\beta_i}{\os_i}{\comseqA}}$
}

\vspace*{1.8mm}
  Using again \refToLemma{scb}(\ref{scb0}) we get \sm

\centerline{ $\preGA{\comoccA_1}{\beta_2} =
  \preGA{\eqA{\os_1}{\filtP{\,\beta_1}{\os_1}{\comseqA}}}{\beta_2} =
  \eqA{\os'}{\filtP{\beta_2}{\os'}{(\filtP{\,\beta_1}{\os_1}{\comseqA})}}$
}
Similarly \sm

\centerline{ $\preGA{\comoccA_2}{\beta_1} =
  \preGA{\eqA{\os_2}{\filtP{\,\beta_2}{\os_2}{\comseqA}}}{\beta_1} =
  \eqA{\os'}{\filtP{\beta_1}{\os'}{(\filtP{\,\beta_2}{\os_2}{\comseqA})}}$
}\sm

\noindent  We want to prove that \sm

\centerline{ $ (*)\ \ \
  \filtP{\beta_1}{\os'}{(\filtP{\,\beta_2}{\os_2}{\comseqA})}\approx_{\os'}
  \filtP{\beta_2}{\os'}{(\filtP{\,\beta_1}{\os_1}{\comseqA})} $
}

\medskip
 In the proof of (*) we will use the following facts, where  $h,k=1,2$  %$h,k\in\set{1,2}$
 and $h\neq k$:
\begin{enumerate}[(a)]
\itemsep=0.8pt
\item \label{c0i}
$\concat{\beta_h}{\concat{\beta_k}{\comseqA}}\approx_{\os'}\concat{\beta_k}{\concat{\beta_h}{\comseqA}}$;
\item\label{c1i} if $\concat{\beta_h}{\comseqA}$ is $\os'$-pointed and
$\concat{\beta_k}{\comseqA}$ is
not $\os_k$-pointed, then $\concat{\beta_h}{\comseqA}$ is
$\os_h$-pointed;
\item\label{c2i} if $\concat{\beta_h}{\comseqA}$ is $\os_h$-pointed and $\concat{\beta_k}{\comseqA}$ is not $\os_k$-pointed, then $\concat{\beta_h}{\comseqA}$ is $\os'$-pointed;
\item \label{c3i}
  ${\concat{\beta_h}{\concat{\beta_k}{\comseqA } } } $ is
 $\os'$-pointed
iff
$\concat{\beta_h}{\comseqA}$ is
$\os_h$-pointed
and
$\concat{\beta_k}{\comseqA}$ is
$\os_k$-pointed.
\end{enumerate}
{\em Fact \ref{c0i}}. We show that
$\concat{\beta_h}{\concat{\beta_k}{\comseqA}}\,$  $\,\os'$-swaps to
 $\concat{\beta_k}{\concat{\beta_h}{\comseqA}}$.
By hypothesis $\play{\beta_h}\cap\play{\beta_k}=\emptyset$, so it is
enough to show that  $\beta_k$ does
not match    $\beta_h$ in the trace
$ \os'\cdot\beta_h\cdot\beta_k\cdot\comseqA  =
 (\mapWh{\beta_h}{(\mapWh{\beta_k}{\os})})\cdot\beta_h\cdot\beta_k\cdot\comseqA$.
\\
Suppose that $\beta_h$ is an output and $\beta_k$ is an input such
 that $\overline{\beta_k} = \beta_h$.  Since $\comoccA_h =
 \preGA{\comoccA}{\beta_h}$ is defined and $\beta_h$ is an output, it
 must be $\os \cong \concat{\os_h}{\beta_h}$.
 Then,
 since $\comoccA_k = \preGA{\comoccA}{\beta_k}$ is defined and
 $\beta_k$ is an input and $\overline{\beta_k} = \beta_h$, we get
 $\mapWh{\beta_k}{\os} = \concat{\overline{\beta}_k}{\os} \cong
 \concat{\overline{\beta}_k}{\concat{\os_h}{\beta_h}} \cong
 \concat{\beta_h}{\concat{\os_h}{\beta_h}} $.  Then $\os' =
 \mapWh{\beta_h}{(\mapWh{\beta_k}{\os})} \cong \concat{\beta_h}{\os_h}$.
 Clearly, $\beta_k$  matches the initial
 output $\beta_h$ in the trace
 $\os'\cdot\beta_h\cdot\beta_k\cdot\comseqA$,
% = \concat{\concat{\beta_h}{\os_h}}{\beta_h\cdot\concat{\beta_k}{\comseqA}}$,
 since $\beta_k$ is the first input in the trace and the initial
 $\beta_h$ is the first complementary output in the trace.  Therefore
 $\beta_k$ does not match its adjacent output
 $\beta_h$.

 \vspace{1.5mm}
 {\em Fact \ref{c1i}}. If $\beta_h$ is required in
 $\concat{\beta_h}{\comseqA}$ - a condition that is always true when
 $\beta_h$ is an input and $\concat{\beta_h}{\comseqA}$ is
 $\os'$-pointed - then
 $\concat{\beta_h}{\comseqA}$ is $\os_0$-pointed for all $\os_0$.\\
We may then assume that $\beta_h$ is an output that is not required in
$\beta_h\cdot\comseqA$. \\
%Consider now the shape of $\beta_k$.
If $\beta_k$ is an output, then $\os_h\cong\os'\cdot\beta_k$.  If
 an input matches
$\beta_h$ in  $\os'\cdot\concat{\beta_h}{\comseqA}$,
then   the same input matches $\beta_h$  in
$\os_h\cdot\concat{\beta_h}{\comseqA}$, since
$\play{\beta_h}\cap\play{\beta_k}=\emptyset$. \\
If $\beta_k$ is an input, then
$\os'\cong\overline{\beta_k}\cdot\os_h$. Suppose
$\beta_h=\CommAs\pp\la\q$ and $\beta_k=\CommAsI\pr{\la'}\ps$.  Observe
that it must be %$\pp\neq\pr$ %or
  $\q\neq \ps$,
because otherwise no
input $\CommAsI\pp\la\q$ could occur in $\comseqA$, since
$\concat{\beta_k}{\comseqA}$ is not $\os_k$-pointed, contradicting the
hypotheses that $\concat{\beta_h}{\comseqA}$ is $\os'$-pointed and
$\beta_h$ is not required in $\beta_h\cdot\comseqA$.  Then the
presence of $\overline{\beta_k} = \CommAs\pr\la\ps$ cannot affect the
multiplicity of $\pp\q!$ or $\pp\q?$ in any trace.
Therefore,  if an input matches
$\beta_h$  in $\os'\cdot\concat{\beta_h}{\comseqA}$,
then  the same input matches  $\beta_h$ in
$\os_h\cdot\concat{\beta_h}{\comseqA}$.

\vspace{1.5mm}
  {\em Fact \ref{c2i}}.
   Again, we may assume that $\beta_h$ is an output that is not
  required in
  $\beta_h\cdot\comseqA$.  \sm
  \\
 % If both $\beta_h$ and $\beta_k$ are outputs,
  If $\beta_k$ is an output, then $\os_h\cong\os'\cdot\beta_k$. If
   an input matches $\beta_h$  in
  $\os_h\cdot\concat{\beta_h}{\comseqA}$, then
   the same input  matches  $\beta_h$ in
$\os'\cdot\concat{\beta_h}{\comseqA}$, since $\play{\beta_h}\cap\play{\beta_k}=\emptyset$.\\
  If $\beta_k$ is an input, then
  $\os'\cong\overline{\beta_k}\cdot\os_h$.
 %and $\os_k\cong\os_h\cdot\beta_h$.
 Let $\beta_h=\CommAs\pp\la\q$
  and $\beta_k=\CommAsI\pr{\la'}\ps$.
 Again, it must be %$\pp\neq\pr$ or
  $\q\neq \ps$,  because
otherwise no input $\CommAsI\pp\la\q$ could occur in $\comseqA$, since
$\concat{\beta_k}{\comseqA}$ is not $\os_k$-pointed, contradicting the
hypotheses that $\concat{\beta_h}{\comseqA}$ is $\os_h$-pointed and
$\beta_h$ is not required in $\beta_h\cdot\comseqA$.
Therefore, if  an input  matches  $\beta_h$ in
  $\os_h\cdot\concat{\beta_h}{\comseqA}$, then  the
  same input matches  $\beta_h$ in $\os'\cdot\concat{\beta_h}{\comseqA}$.

  \vspace{1.5mm}
 {\em Fact \ref{c3i}}.  From
 $\play{\beta_h}\cap\play{\beta_k}=\emptyset$ it follows that
 $\beta_h$ is required in
 $\concat{\beta_h}{\beta_k\cdot\comseqA}$
    iff
%- condition always true when $\beta_h$ is an input - then
$\beta_h$ is required in $\concat{\beta_h}{\comseqA}$, and similarly
for $\beta_k$.
Let us then assume that $\beta_h$ and $\beta_k$ are not both
required in $\concat{\beta_h}{\beta_k\cdot\comseqA}$, i.e., that at
least one of them is an output not required in $\concat{\beta_h}{\beta_k\cdot\comseqA}$. \\
If both $\beta_h$ and $\beta_k$ are outputs, then
$\os_h\cong\os'\cdot\beta_k$. Then  an input matches  $\beta_h$ in
$\os'\cdot\concat{\beta_h\cdot\beta_k}{\comseqA}$ iff  the same
input matches  $\beta_h$
in $\os_h\cdot\concat{\beta_h}{\comseqA}$,
since
$\concat{\beta_h}{\concat{\beta_k}{\comseqA}}\approx_{\os'}
\concat{\beta_k}{\concat{\beta_h}{\comseqA}}$ by \refToFact{c0i}.
  \\
  Let $\beta_h=\CommAs\pp\la\q$ and $\beta_k=\CommAsI\pr{\la'}\ps$,
   where $\beta_h$ is not required in $\concat{\beta_h}{\beta_k\cdot\comseqA}$.
  Then $\os'\cong\overline{\beta_k}\cdot\os_h$. Therefore  an
input matches  $\beta_h$ in
  $\os'\cdot\concat{\beta_h\cdot\beta_k}{\comseqA}$ iff  the same
input matches  $\beta_h$
 in $\os_h\cdot\concat{\beta_h}{\comseqA}$,
  since
$\concat{\beta_h}{\concat{\beta_k}{\comseqA}}\approx_{\os'}\concat{\beta_k}{\concat{\beta_h}{\comseqA}}$
by \refToFact{c0i}.

\vspace{1.5mm}
 We proceed now to prove (*).  We distinguish three cases, according
  to whether:
\begin{enumerate}[i)]
\item \mylabel{o1}
each $\concat{\beta_i}{\comseqA}$ is $\os_i$-pointed, for $i
  =1,2$;
\item \mylabel{o2} no $\concat{\beta_i}{\comseqA}$ is $\os_i$-pointed, for $i
  =1,2$;
\item \mylabel{o3} $\concat{\beta_h}{\comseqA}$ is $\os_h$-pointed and
  $\concat{\beta_k}{\comseqA}$ is not $\os_k$-pointed, for $h,k = 1,2$
  and $h\neq k$.
\end{enumerate}
\emph{Case} \ref{o1}.  Suppose each $\concat{\beta_i}{\comseqA}$ is
$\os_i$-pointed, for $i = 1,2$.  Then
$\filtP{\beta_1}{\os'}{(\filtP{\,\beta_2}{\os_2}{\comseqA})}
\approx_{\os'} \filtP{\beta_1}{\os'}{\beta_2\cdot\comseqA}$ and
$\filtP{\beta_2}{\os'}{(\filtP{\,\beta_1}{\os_1}{\comseqA})}
\approx_{\os'} \filtP{\beta_2}{\os'}{\beta_1\cdot\comseqA}$.  %Then, by
By \refToFact{c3i} %it follows that
both
$\concat{\beta_1}{\concat{\beta_2}{\comseqA}}$ and
$\concat{\beta_2}{\concat{\beta_1}{\comseqA}}$ are
$\os'$-pointed. Then $\filtP{\beta_1}{\os'}{\beta_2\cdot\comseqA}
\approx_{\os'} \concat{\beta_1}{\concat{\beta_2}{\comseqA}}$ and
$\filtP{\beta_2}{\os'}{\beta_1\cdot\comseqA} \approx_{\os'}
\concat{\beta_2}{\concat{\beta_1}{\comseqA}}$.  By \refToFact{c0i}
%we have
$\concat{\beta_1}{\concat{\beta_2}{\comseqA}}\approx_{\os'}
\concat{\beta_2}{\concat{\beta_1}{\comseqA}}$.

\medskip\noindent
\emph{Case} \ref{o2}.  Suppose no $\concat{\beta_i}{\comseqA}$ is
$\os_i$-pointed, for $i = 1,2$.  Then
$\filtP{\beta_1}{\os'}{(\filtP{\,\beta_2}{\os_2}{\comseqA})}
\approx_{\os'} \filtP{\beta_1}{\os'}{\comseqA}$
and
$\filtP{\beta_2}{\os'}{(\filtP{\,\beta_1}{\os_1}{\comseqA})}
\approx_{\os'} \filtP{\beta_2}{\os'}{\comseqA}$.
 By \refToFact{c1i}, no $\concat{\beta_i}{\comseqA}$ can be
$\os'$-pointed, for $i\in\set{1,2}$.
Hence
$\filtP{\beta_1}{\os'}{\comseqA} \approx_{\os'}  \comseqA
\approx_{\os'} \filtP{\beta_2}{\os'}{\comseqA}$.

\medskip\noindent
\emph{Case} \ref{o3}. Suppose
$\concat{\beta_h}{\comseqA}$ is $\os_h$-pointed and
  $\concat{\beta_k}{\comseqA}$ is not $\os_k$-pointed,
for $h,k = 1,2$ and $h\neq k$.
Then $\filtP{\beta_h}{\os'}{(\filtP{\,\beta_k}{\os_k}{\comseqA})}
\approx_{\os'} \filtP{\beta_h}{\os'}{\comseqA}$ and
$\filtP{\beta_k}{\os'}{(\filtP{\,\beta_h}{\os_h}{\comseqA})}
\approx_{\os'} \filtP{\beta_k}{\os'}{\beta_h\cdot\comseqA}$.
% Since $\concat{\beta_h}{\comseqA}$ is $\os_h$-pointed and
% $\concat{\beta_k}{\comseqA}$ is not $\os_k$-pointed,
By \refToFact{c2i} $\concat{\beta_h}{\comseqA}$ is
$\os'$-pointed. Hence $\filtP{\beta_h}{\os'}{\comseqA}
\approx_{\os'}\beta_h\cdot\comseqA$.
By \refToFact{c3i}
$\concat{\beta_k}{\concat{\beta_h}{\comseqA}}$ is not
$\os'$-pointed.
Therefore $\filtP{\beta_k}{\os'}{\beta_h\cdot\comseqA}
\approx_{\os'}\beta_h\cdot\comseqA$.
\end{proof}

%%%%%%%%%%%%%%%%%%%%

 \begin{lemma}
\mylabel{nested-filtering}
\begin{enumerate}
\item\mylabel{nested-filtering2}
Let $\mapBl{\beta}{\os}$ be defined and $\os'=\mapBl{\beta}{\os}$.
 Let $\comseqA,\comseqA'$ be such that $\comseqA'$ is
$(\concat{\os}{\concat{\beta}{\comseqA}})$-pointed. Then

\centerline{
 $\filtP{(\concat{\beta}{\comseqA})}{\os}{\comseqA'} =
\filtP{\beta}{\os}{(\filtP{\comseqA}{\os'}{\comseqA'})}$\sm
}

\item\mylabel{nested-filtering1}  Let $\mapWh{\beta}{\os}$ be defined
  and $\os'=\mapWh{\beta}{\os}$.
   Let $\comseqA,\comseqA'$ be such that $\comseqA'$ is
$(\concat{\os'}{\concat{\beta}{\comseqA}})$-pointed.  Then\sm

\centerline{
$\filtP{(\concat{\beta}{\comseqA})}{\os'}{\comseqA'} =
\filtP{\beta}{\os'}{(\filtP{\comseqA}{\os}{\comseqA'})}$
}
\end{enumerate}
\end{lemma}

\begin{proof}
(\ref{nested-filtering2}) We show
$\filtP{(\concat{\beta}{\comseqA})}{\os}{\comseqA'} =
\filtP{\beta}{\os}{(\filtP{\comseqA}{\os'}{\comseqA'})}$
by induction on $\comseqA$.

\medskip\noindent
 \emph{Case $\comseqA = \ee$.}  In this case both the LHS
and RHS reduce to $\filtP{\beta}{\os}{\comseqA'}$, for whatever $\os$.

\vspace*{1.6mm}\noindent
 \emph{Case $\comseqA = \concat{\comseqA''}{\beta'}$.}     By
\refToDef{def:trace-filtering} we obtain for the LHS:

\vspace*{1.6mm}
\centerline{$
\filtP{(\concat{\beta}{\concat{\comseqA''}{\beta'}})}{\os}{\comseqA'} =\begin{cases}
\filtP{(\concat{\beta}{\comseqA''})}{\os}{(\concat{\beta'}{\comseqA'})}     & \text{if $\concat{\beta'}{\comseqA'}$ is $(\concat{\os}{\concat{\beta}{\comseqA''}})$-pointed}  \\
\filtP{(\concat{\beta}{\comseqA''})}{\os}{\comseqA'}   & \text{otherwise}
\end{cases}
$}

\vspace*{1.6mm}\noindent
By \refToDef{def:trace-filtering} (applied to the internal
filtering) we obtain for the RHS:

\vspace*{1.6mm}
\centerline{$
\filtP{\beta}{\os}{(\filtP{(\concat{\comseqA''}{\beta'})}{\os'}{\comseqA'})} =
\begin{cases}
  \filtP{\beta}{\os}{(\filtP{\comseqA''}{\os'}{(\concat{\beta'}\comseqA'}))}
& \text{if $\concat{\beta'}{\comseqA'}$ is $(\concat{\os'}{\comseqA''})$-pointed}  \\
 \filtP{\beta}{\os}{(\filtP{\comseqA''}{\os'}{\comseqA'})}
& \text{otherwise}
\end{cases}
$}

\vspace*{1.6mm}\noindent
We distinguish two cases, according to whether $\beta$ is an input
or an output.

\medskip
Suppose first that $\beta$ is an output. Then $\os'
=\concat{\os}\beta$.  The side condition, i.e. the requirement that
$\concat{\beta'}{\comseqA'}$ be $(\concat{\os'}{\comseqA''})$-pointed,
is the same in both cases. We may then immediately conclude that LHS =
RHS using the induction hypothesis.

\medskip
Suppose now that $\beta$ is an input. Then
$\os=\concat{\overline\beta}{\os'}$. Observe that, since
$\,(\concat{\os'}{\comseqA"})$ is obtained from
$(\concat{\os}{\concat{\beta}{\comseqA''}}) =
(\concat{\concat{\overline{\beta}}{\os'}}{\concat{\beta}{\comseqA''}})$
by erasing a pair of matching communications,
$(\concat{\beta'}{\comseqA'})$ is $(\concat{\os'}{\comseqA''})$-pointed
if and only if $(\concat{\beta'}{\comseqA'})$ is
$(\concat{\os}{\concat{\beta}{\comseqA''}})$-pointed.  Then we may
again conclude by induction.

\medskip
(\ref{nested-filtering1})  follows from (\ref{nested-filtering2})
since $\mapBl{\beta}{(\mapWh\beta\os)}=\os$.
\end{proof}

%%%%%%%%%%%%%%%%%%%%

\begin{lemma}\mylabel{prop:relPrePost}
\begin{enumerate}
\item  \mylabel{prop:relPrePost12} \mylabel{prop:relPrePost3}\mylabel{prop:relPrePost4}
 If $\comseqA\not=\ee$ and  $\mapBl{\beta}{\os}$ is defined,  then $\postGA{\point(\os,\concat{\beta}\comseqA)}{\beta}=\point(\mapBl{\beta}{\os},\comseqA)$.
\item  \mylabel{prop:relPrePost11}\mylabel{prop:relPrePost1}\mylabel{prop:relPrePost2}If $\mapWh{\beta}{\os}$ is defined, then
$\cauA{\beta}{\point(\os,\comseqA)}=\point(\mapWh{\beta}{\os},\beta\cdot\comseqA)$.
\end{enumerate}
\end{lemma}

\begin{proof}  \refToDef{def:pf} and Lemmas~\ref{scb} and~\ref{nested-filtering} with $\comseqA'=\ee$ imply (\ref{prop:relPrePost4}) and (\ref{prop:relPrePost1})  since:

\vspace*{1.8mm}

   \centerline{$\begin{array}{llllll}
(\ref{prop:relPrePost3})
&\postGA{\point(\os,\beta\cdot\comseqA)}{\beta}
&=&\postGA{\eqA{\os}{\filt{(\beta\cdot\comseqA)}\ee}}{\beta}
& \mbox{by \refToDef{def:pf}} &\\
&&=&\postGA{\eqA{\os}{\filt{\beta}{(\filtP\comseqA{\os'}\ee)}}}{\beta}&
\mbox{by \refToLemma{nested-filtering}(\ref{nested-filtering2})} &\\
&&=&\eqA{\os'}{\filtP\comseqA{\os'}\ee}
& \mbox{by \refToLemma{scb}(\ref{scb3})} &\\
&\point(\os',\comseqA)&=&\eqA{\os'}{\filtP\comseqA{\os'}\ee}& \mbox{by \refToDef{def:pf}}\\&&&&\text{where }\os'=\mapBl\beta\os \\
  (\ref{prop:relPrePost1}) &\cauA{\beta}{\point(\os,\comseqA)}&=&\cauA{\beta}{\eqA{\os}{\filtP\comseqA{\os}\ee}}&
    \mbox{by \refToDef{def:pf}} &\\
  &&=&
    \eqA{\os'}{\filtP{\beta}{\os'}{(\filtP\comseqA{\os}\ee)}}
    & \mbox{by \refToLemma{scb}(\ref{scb1})} &\\
&&=&\eqA{\os'}{\filtP{(\beta\cdot\comseqA)}{\os'}\ee}
 &
 \mbox{by \refToLemma{nested-filtering}(\ref{nested-filtering1})}&\\
 &\point(\os',\beta\cdot\comseqA)&=&\eqA{\os'}{\filtP{(\beta\cdot\comseqA)}{\os'}\ee}
 &
 \mbox{by \refToDef{def:pf}}&\\
 &&&&\text{where }\os'=\mapWh\beta\os%\quad\hfill  \qedhere
 \end{array}$ }
 \end{proof}

 %%%%%%%%%%%%%%%%%%%%

 \begin{lemmaa}{\ref{prop:prePostGlArel}}{\begin{enumerate}
\item  If $\comoccA_1  < \comoccA_2$
 and  both $\postGA{\comoccA_1}{\beta}$, $\postGA{\comoccA_2}{\beta}$ are defined,
 then
  $\postGA{\comoccA_1}{\beta}  <   \postGA{\comoccA_2}{\beta}$.
  \item If $\comoccA_1 < \comoccA_2$ and
$\preGA{\comoccA_1}{\beta}$ is defined, then
$\preGA{\comoccA_1}{\beta}  <  \preGA{\comoccA_2}{\beta}$.
 \item  If $\comoccA_1\grr\comoccA_2$ and   both $\preGA{\comoccA_1}{\beta}$,  $\preGA{\comoccA_2}{\beta}$ are defined,  then
 $\preGA{\comoccA_1}{\beta}\grr \preGA{\comoccA_2}{\beta}$.
\end{enumerate}}
\end{lemmaa}

\begin{proof}
(\ref{ppg4bA}) Let
$\comoccA_1=\eqA\os{\comseqA}$ and
$\comoccA_2=\eqA\os{\concat\comseqA{\comseqA'}}$.
  If $\postGA{\comoccA_1}{\beta}=\eqA{\os'}{\comseqA}$ and
  $\postGA{\comoccA_2}{\beta}=\eqA{\os'}{\concat\comseqA{\comseqA'}}$
  for some $\os'$, then $\postGA{\comoccA_1}{\beta}  <  \postGA{\comoccA_2}{\beta}$. \sm

Let $\beta$ be an output.  If $\comseqA\approx_\os\concat\beta{\comseqA_1}$ with $\comseqA_1\not=\ee$, then $\postGA{\comoccA_1}{\beta}=\eqA{\os\cdot\beta}{{\comseqA_1}}$ and $\postGA{\comoccA_2}{\beta}=\eqA{\os\cdot\beta}{\concat{\comseqA_1}{\comseqA'}}$. Therefore
$\postGA{\comoccA_1}{\beta}  <   \postGA{\comoccA_2}{\beta}$.  Let
$\play\beta\not\subseteq\play{\comseqA}$ and
$\concat\comseqA{\comseqA'}\approx_\os\concat\beta{\comseqA_2}$ with
$\comseqA_2\not=\ee$.
This implies $\concat\beta{\comseqA_2} \approx_\os
\concat{\beta}{\concat{\comseqA}{\comseqA'_2}}$ for some
$\comseqA'_2$. It follows that $\comseqA_2 \approx_{\os\cdot\beta}
\concat{\comseqA}{\comseqA'_2}$. Then we get
$\postGA{\comoccA_1}{\beta}=\eqA{\os\cdot\beta}{{\comseqA}}$ and
$\postGA{\comoccA_2}{\beta}=\eqA{\os\cdot\beta}{\comseqA_2}=\eqA{\os\cdot\beta}{\concat\comseqA{\comseqA'_2}}$,
which imply
$\postGA{\comoccA_1}{\beta}  < \postGA{\comoccA_2}{\beta}$. \sm

Let $\beta$ be an input. The proof is similar.

\medskip  (\ref{ppg4aA}) Since $\comoccA_1 < \comoccA_2$ and
  $\preGA{\comoccA_1}{\beta}$  is  %are
  defined, then also
  $\preGA{\comoccA_2}{\beta}$ is defined.
  Let %$\comoccA_1\precP\comoccA_2$ and
  $\comoccA_1=\eqA\os{\comseqA}$ and
  $\comoccA_2=\eqA\os{\concat\comseqA{\comseqA'}}$.  If
  $\preGA{\comoccA_1}{\beta}=\eqA{\os'}{\comseqA}$ and
  $\preGA{\comoccA_2}{\beta}=\eqA{\os'}{\concat\comseqA{\comseqA'}}$
  for some $\os'$, then $\preGA{\comoccA_1}{\beta}  <  \preGA{\comoccA_2}{\beta}$.\sm

   Let $\beta$ be an output. Then $\os\cong\os'\cdot\beta$.  If
  $\preGA{\comoccA_1}{\beta}=\eqA{\os'}{\beta\cdot{\comseqA}}$, then
  it must be
  $\preGA{\comoccA_2}{\beta}=\eqA{\os'}{\beta\cdot\concat\comseqA{\comseqA'}}$. Thus
  $\preGA{\comoccA_1}{\beta} <
  \preGA{\comoccA_2}{\beta}$. %Otherwise we can only have that
  The only other case is
  $\preGA{\comoccA_1}{\beta}=\eqA{\os'}{{\comseqA}}$ and
  $\preGA{\comoccA_2}{\beta}=\eqA{\os'}{\beta\cdot\concat\comseqA{\comseqA'}}$.
  Since $\preGA{\comoccA_1}{\beta}=\eqA{\os'}{{\comseqA}}$,  the trace
  $\beta\cdot\comseqA$ is not $\os'$-pointed, so
  $\play\beta\not\subseteq\play{\comseqA}$ and $\comseqA$ does not
  contain the matching input of $\beta$. Therefore
  $\beta\cdot\concat\comseqA{\comseqA'}\approx_{\os'}\concat\comseqA{\beta\cdot\comseqA'}$
  and
  $\preGA{\comoccA_2}{\beta}=\eqA{\os'}{\beta\cdot\concat\comseqA{\comseqA'}}=\eqA{\os'}{
    \concat\comseqA{\beta\cdot\comseqA'}}$, so
  $\preGA{\comoccA_1}{\beta}  <  \preGA{\comoccA_2}{\beta}$. \sm

  Let $\beta$ be an input.  If
  $\preGA{\comoccA_1}{\beta}=\eqA{\overline\beta\cdot\os}{\beta\cdot{\comseqA}}$,
  then it must be
  $\preGA{\comoccA_2}{\beta}=\eqA{\overline\beta\cdot\os}{\beta\cdot\concat\comseqA{\comseqA'}}$.
  We get $\preGA{\comoccA_1}{\beta} < \preGA{\comoccA_2}{\beta}$.  The
  only other case is
  $\preGA{\comoccA_1}{\beta}=\eqA{\overline\beta\cdot\os}{{\comseqA}}$
  and
  $\preGA{\comoccA_2}{\beta}=\eqA{\overline\beta\cdot\os}{\beta\cdot\concat\comseqA{\comseqA'}}$.
  If
  $\preGA{\comoccA_1}{\beta}=\eqA{\overline\beta\cdot\os}{{\comseqA}}$,
  then $\play\beta\not\subseteq\play{\comseqA}$. Therefore
  $\beta\cdot\concat\comseqA{\comseqA'}\approx_{\overline\beta\cdot\os}\concat\comseqA{\beta\cdot\comseqA'}$
  and
  $\preGA{\comoccA_2}{\beta}=\eqA{\overline\beta\cdot\os}{\beta\cdot\concat\comseqA{\comseqA'}}=
  \eqA{\overline\beta\cdot\os}{\concat\comseqA{\beta\cdot\comseqA'}}$,
  so $\preGA{\comoccA_1}{\beta} < \preGA{\comoccA_2}{\beta}$.

\medskip  (\ref{ppg-conflictA-a}) Let $\comoccA_1=\eqA\os{\comseqA}$ and
$\comoccA_2=\eqA\os{\comseqA'}$
 and $\projAP{\comseqA}\pp\grr\projAP{\comseqA'}\pp$.
 % We only consider interesting cases.\\
 We  select some  interesting cases. \sm

  Note first that $\projAP{\comseqA}\pp\grr\projAP{\comseqA'}\pp$
 implies
 $\pp \in \play{\comseqA} \cap \play{\comseqA'}$. \sm

 If $\beta$ is an output , then $\os\cong \concat{\os'}{\beta}$. If
 both $\concat\beta\comseqA$ and $\concat\beta{\comseqA'}$ are
 $\os'$-pointed or not $\os'$-pointed, then the result is
 immediate.  If $\concat\beta\comseqA$ is  $\os'$-pointed
  while $\concat\beta{\comseqA'}$ is not  $\os'$-pointed,
  then  $\play{\beta} \not \subseteq \play{\comseqA'}$.
 This implies  $\pp\not\in\play\beta$.
 Similarly,  if $\beta$ is an input and  $\play\beta\subseteq\play\comseqA$  while
  $\play\beta \not\subseteq\play{\comseqA'}$,  then $\pp\not\in\play\beta$. In
 both cases we get $\projAP{(\concat\beta{\comseqA})}\pp=
 \projAP{\comseqA}\pp$  and $\projAP{(\concat\beta{\comseqA'})}\pp=
 \projAP{\comseqA'}\pp$,  %\grr\projAP{\comseqA'}\pp$,
 so we conclude
 $\preGA{\comoccA_1}{\beta}\grr \preGA{\comoccA_2}{\beta}$.
 \end{proof}

 %%%%%%%%%%%%%%%%%%%%

\begin{lemmaa}{\ref{keybeta}}{\begin{enumerate}
 \item  If $\comoccA\in \EGGA( \agtO{\pp}{\q}i  I{\la}{\G}\parN\Msg)$ and
$\postGA\comoccA{\CommAs\pp{\la_k}\q}$ is defined,
then

 \centerline{$\postGA\comoccA{\CommAs\pp{\la_k}\q}\in
   \EGGA(\G_k\parN\addMsg\Msg{\mq\pp{\la_k}\q})$, where $k\in I$.}
 \item If $\comoccA\in \EGGA( \agtI \pp\q \la \G \parN\addMsg{\mq\pp\la\q}\Msg)$  and $\postGA\comoccA{\CommAsI\pp\la\q}$ is defined, then
 $\postGA\comoccA{\CommAsI\pp\la\q}\in\EGGA(\G\parN\Msg)$.

 \item  If $\comoccA\in\EGGA(\G\parN\addMsg\Msg{\mq\pp{\la}\q})$,  then

\centerline{$\cauA{\CommAs\pp{\la}\q}\comoccA\in \EGGA(
  \agtO{\pp}{\q}i I{\la}{\G}\parN\Msg)$,  where $\la = \la_k$ and $\G =
  \G_k$ for some $k\in I$.}
 \item If $\comoccA\in\EGGA(\G\parN\Msg)$, then $\cauA{\CommAsI\pp\la\q}\comoccA\in \EGGA( \agtI \pp\q \la \G \parN\addMsg{\mq\pp\la\q}\Msg)$.
\end{enumerate}}
\end{lemmaa}

\begin{proof}
 (\ref{lemma:subev3})
By \refToDef{egA}(\ref{eg1AP}), if $\comoccA\in \EGGA( \agtO{\pp}{\q}i
I{\la}{\G}\parN\Msg)$, then $\comoccA=\point(\os,\comseqA)$, where
$\os=\osq{\Msg}$ and $\comseqA\in\FPaths{\agtO{\pp}{\q}i I{\la}{\G}}$,
which %implies
 gives $\comseqA\approx_\os\concat{\CommAs\pp{\la_h}\q}{\comseqA_h}$ with
$\comseqA_h\in \FPaths{ \G_h}$ for some $h\in I$.   By hypothesis
$\postGA{\comoccA}{\CommAs\pp{\la_k}\q}$ is defined, which implies
% then $\comoccA=\eqA\os{\comseqA'}$ with either $\comseqA'\approx_\os\concat{\CommAs\pp{\la_k}\q}{\comseqA''}$ and
%$\comseqA''\not=\ee$ or $\pp\not\in\play{\comseqA'}$. In both cases we get
$\comseqA\approx_\os\concat{\CommAs\pp{\la_k}\q}{\comseqA_k}$ and
$\comseqA_k\not=\ee$. Then
\refToLemma{prop:relPrePost}(\ref{prop:relPrePost3}) gives
$\postGA{\comoccA}{\CommAs\pp{\la_k}\q}=\point(\concat\os{\CommAs\pp{\la_k}\q},\comseqA_k)$.
We conclude that

\vspace{1.5mm}
\centerline{$\postGA{\comoccA}{\CommAs\pp{\la_k}\q}\in \EGGA(
\G_k\parN\addMsg\Msg{\mq\pp{\la_k}\q})$}

\medskip (\ref{lemma:subev4})  Similar to the proof of (\ref{lemma:subev3}).

\medskip (\ref{lemma:subev1}) By
  \refToDef{egA}(\ref{eg1AP}), if
  $\comoccA\in\EGGA(\G\parN\addMsg\Msg{\mq\pp{\la}\q})$,
  then $\comoccA=\point(\concat\os{\CommAs\pp{\la}\q},\comseqA)$,
  where $\os=\osq\Msg$ and $\comseqA\in\FPaths{\G}$.  By
  \refToLemma{prop:relPrePost}(\ref{prop:relPrePost1})
  $\preGA{\comoccA}{\CommAs\pp{\la}\q}=
\point(\os,\concat{\CommAs\pp{\la}\q}\comseqA)$.
  Then, again by \refToDef{egA}(\ref{eg1AP}),
  $\preGA{\comoccA}{\CommAs\pp{\la}\q}\in \EGGA( \agtO{\pp}{\q}i
  I{\la}{\G}\parN\Msg)$,  where $\la = \la_k$ and $\G =
  \G_k$ for some $k\in I$,  since
  $\concat{\CommAs\pp{\la_k}\q}{\comseqA}\in \FPaths{ \agtO{\pp}{\q}i
    I{\la}{\G}}$.

\medskip (\ref{lemma:subev2})  Similar to the proof of (\ref{lemma:subev1}).
\end{proof}

%%%%%%%%%%%%%%%%%%%%

\begin{lemmaa}{\ref{lemma:subev-bis}}{ Let $\G\parG\Msg\stackred\beta \G'\parG\Msg'$. Then $\osq\Msg
  \cong\mapWh\beta{\osq{\Msg'}}$ and
\begin{enumerate}
\item if $\comoccA\in\EGGA(\G\parN\Msg)$ and
  $\postGA{\comoccA}{\beta}$ is defined, then
  $\postGA{\comoccA}{\beta}\in\EGGA(\G'\parN\Msg')$;
\item if $\comoccA\in\EGGA(\G'\parN\Msg')$, then
$\preGA{\comoccA}{\beta}\in\EGGA(\G\parN\Msg)$.
\end{enumerate} }
\end{lemmaa}

\begin{proof}
\refToLemma{defpost} and Session Fidelity (\refToTheorem{sfA}) imply $\osq\Msg\cong\mapWh\beta{\osq{\Msg'}}$.

(\ref{keybeta5}) By induction on the inference of the transition
$\G\parG\Msg\stackred{\beta}\G'\parG\Msg'$ , see Figure \ref{ltsgtAs}.

\sm \noindent
{\it Base Cases.}
If the  applied rule is \rulename{Ext-Out}, then  $\G=\agtO\pp\q i I \la
\G$ and $\beta = \CommAs{\pp}{\la_k}{\q}$ and $\G'=\G_k$ and
$\Msg'\equiv\addMsg{\Msg}{\mq\pp{\la_k}\q}$ for some $k\in I$.
 By assumption $\postGA{ \comoccA}\beta$ is defined.
By \refToLemma{lemma:subev}(\ref{lemma:subev3})
 $\postGA\comoccA\beta\in\EGGA(\G'\parG\Msg')$. \sm

If the  applied rule   is \rulename{Ext-In}, then
 $\G= \agtI\pp\q \la
{\G'}$ and $\beta = \CommAsI{\pp}{\la}{\q}$ and
$\Msg\equiv\addMsg{\mq\pp{\la}\q}\Msg'$.
 By assumption $\postGA{ \comoccA}\beta$ is defined.
By \refToLemma{lemma:subev}(\ref{lemma:subev4})
% if $\postGA\comoccA{\beta}$ is defined then
$\postGA\comoccA\beta\in\EGGA(\G'\parG\Msg')$. \sm

\vspace*{1.6mm}\noindent
{\it Inductive Cases.}  If the last applied  rule  is
\rulename{IComm-Out}, then $\G=\agtO\pp\q i I \la \G$ and
$\G'=\agtO\pp\q i I \la {\G'}$ and $
\G_i\parG\addMsg{\Msg}{\mq\pp{\la_i}\q}\stackred\asCom\G_i' \parG\addMsg{\Msg'}{\mq\pp{\la_i}\q}$
for all $i \in I $ and $\pp\not\in\play\beta$. \sm
\\
By \refToDef{egA}(\ref{eg1AP}) $\comoccA\in\EGGA(\G\parN\Msg)$ implies
$\comoccA=\point(\os,\comseqA)$, where $\os=\osq{\Msg}$ and
$\comseqA\in\FPaths{\G}$. Then
$\comseqA=\concat{\CommAs\pp{\la_k}\q}{\comseqA'}$ and
$\comoccA=\eqA\os{\comseqA_0}$ with
$\comseqA_0=\filt{(\concat{\CommAs\pp{\la_k}\q}{\comseqA'})}\ee$ for
some $k\in I$ by \refToDef{def:pf}. We get either
$\comseqA_0\approx_\os\concat{\CommAs\pp{\la_k}\q}{\comseqA_0'}$ or
$\pp\not\in\play{\comseqA_0}$ by \refToDef{def:pm}. Then
$\postGA\comoccA{\CommAs\pp{\la_k}\q}$ is defined unless  $\comseqA_0\approx_\os\concat{\CommAs\pp{\la_k}\q}{\comseqA_0'}$ and
$\comseqA_0'=\ee$  by
\refToDef{def:PostPreGlA}(\ref{def:PostPreGl1A}). We consider the two cases. \sm

\vspace*{1.6mm}\noindent
 {\it  Case  $\comseqA_0\approx_\os\concat{\CommAs\pp{\la_k}\q}{\comseqA_0'}$ and  $\comseqA_0'=\ee$.}
We get $\postGA\comoccA\beta=\eqA{\mapBl\beta\os}{\CommAs\pp{\la_k}\q}$ since $\play\beta\cap\play{\CommAs\pp{\la_k}\q}=\emptyset$, which implies $\postGA\comoccA\beta\in\EGGA(\G'\parN\Msg')$ by \refToDef{egA}(\ref{eg1AP}).

\vspace*{1.6mm}\noindent
 {\it  Case  $\comseqA_0\approx_\os\concat{\CommAs\pp{\la_k}\q}{\comseqA_0'}$ and  $\comseqA_0'\not=\ee$ or  $\pp\not\in\play{\comseqA_0}$.}
Let $\comoccA'=\postGA\comoccA{\CommAs\pp{\la_k}\q}$.
 By  \refToLemma{lemma:subev}(\ref{lemma:subev3})
$\comoccA' \in\EGGA(\G_k\parG\addMsg{\Msg}{\mq\pp{\la_k}\q})$.
By assumption $\postGA{ \comoccA}\beta$ is defined.
We  first  show that $\postGA{ \comoccA' }\beta$ is defined.
Since $\postGA{ \comoccA}\beta$ and
$\postGA\comoccA{\CommAs\pp{\la_k}\q}$ are defined,
by  \refToDef{def:PostPreGlA}(\ref{def:PostPreGl1A}) we have four cases:
\begin{enumerate}[(a)]
\itemsep=0.9pt
\item \label{keybetaE1} $\comseqA_0\approx_\os\concat{\beta}{\comseqA_1}$ for some $\comseqA_1$ and $\comseqA_0\approx_\os\concat{\CommAs\pp{\la_k}\q}{\comseqA_0'}$;
\item \label{keybetaE2}$\comseqA_0\approx_\os\concat{\beta}{\comseqA_1}$ and $\pp\not\in\play{\comseqA_0}$;
\item \label{keybetaE3}$\play\beta\cap\play{\comseqA_0}=\emptyset$ and $\comseqA_0\approx_\os\concat{\CommAs\pp{\la_k}\q}{\comseqA_0'}$;
\item \label{keybetaE4} $\play\beta\cap\play{\comseqA_0}=\emptyset$ and $\pp\not\in\play{\comseqA_0}$.
\end{enumerate}
Let
$\os'=\mapBl{\CommAs\pp{\la_k}\q}\os=\os\cdot{\CommAs\pp{\la_k}\q}$
and $\os''=\mapBl{\beta}{\os'}$.

\vspace*{1mm}
In Case \ref{keybetaE1} we have
$\comseqA_0\approx_\os\concat{\concat{\beta}{\CommAs\pp{\la_k}\q}}{\comseqA'_1}\approx_\os\concat{\concat{\CommAs\pp{\la_k}\q}{\beta}}{\comseqA'_1}$
for some $\comseqA'_1$. Let $\comseqA_2 =
\concat{\beta}{\comseqA'_1}$. Then
$\comoccA=\eqA\os{\concat{\CommAs\pp{\la_k}\q}{\comseqA_2}}$ and
therefore $\comoccA' =\eqA{\os'}{\comseqA_2} =
\eqA{\os'}{\concat{\beta}{\comseqA'_1}}$.
Hence $\postGA{ \comoccA' }\beta=\eqA{\os''}{\comseqA_1'}$.

\vspace*{1mm}
In Case \ref{keybetaE2} we have
$\comoccA=\eqA\os{\concat{\beta}{\comseqA_1}}$ and
$\pp\not\in\play{\concat{\beta}{\comseqA_1}}$.
Therefore $\comoccA' =\eqA{\os'}{\concat{\beta}{\comseqA_1}}$. Hence   $\postGA{
 \comoccA' }\beta=\eqA{\os''}{\comseqA_1}$.

\vspace*{1mm}
%%%%%%%%%%% FINE VARIANTE ILARIA %%%%%%%%%%%%%%%%%%%%%%%%%%%%
In  Case  \ref{keybetaE3} we have $\comoccA'
=\eqA{\os'}{\comseqA_0'}$ and $\postGA{ \comoccA'
}\beta=\eqA{\os''}{\comseqA_0'}$ since
$\play\beta\cap\play{\comseqA_0}=\emptyset$ implies
$\play\beta\cap\play{\comseqA'_0}=\emptyset$.

\vspace*{1mm}
In  Case  \ref{keybetaE4} we have $\comoccA'
=\eqA{\os'}{\comseqA_0}$ and
$\postGA{ \comoccA' }\beta=\eqA{\os''}{\comseqA_0}$.

\medskip\noindent
So in all cases we conclude that $\postGA{ \comoccA' }\beta$ is defined.

\sm
By induction $\postGA{ \comoccA' }\beta\in\EGGA(\G'_k\parG\addMsg{\Msg'}{\mq\pp{\la_k}\q})$.
By  \refToLemma{lemma:subev}(\ref{lemma:subev1})
$\cauA{\CommAs\pp{\la_k}\q}{(\postGA{ \comoccA' }\beta)}\in\EGGA(\G'\parG\Msg')$.
Since $\comoccA'$ is defined,
\refToLemma{prop:prePostGlA}(\ref{ppg1}) implies
$\cauA{\CommAs\pp{\la_k}\q}{ \comoccA' }=\comoccA$.  Since
$\postGA{ \comoccA' }\beta$ and
$\postGA{(\cauA{\CommAs\pp{\la_k}\q}{ \comoccA' })}\beta$ are defined
 and $\pp\not\in\play\beta$, by
\refToLemma{prop:prePostGlA}(\ref{ppg5}) we get  $\cauA{\CommAs\pp{\la_k}\q}{(\postGA{ \comoccA' }\beta)}=\postGA{(\cauA{\CommAs\pp{\la_k}\q}{ \comoccA' })}\beta=\postGA\comoccA\beta$. We
conclude that
$\postGA\comoccA\beta\in\EGGA(\G'\parG\Msg')$. \sm

If the last applied  rule  is \rulename{IComm-In} the proof is similar.

\medskip  (\ref{keybeta4}) By induction on the inference of the transition
  $\G\parG\Msg\stackred{\beta}\G'\parG\Msg'$, see Figure
  \ref{ltsgtAs}.

\medskip \noindent
{\it Base Cases.}
If the applied rule is \rulename{Ext-Out}, then $\G=\agtO\pp\q i I \la
\G$ and $\beta = \CommAs{\pp}{\la_k}{\q}$ and $\G'=\G_k$ and
$\Msg'\equiv\addMsg{\Msg}{\mq\pp{\la_k}\q}$ for some $k\in I$.  By
\refToLemma{lemma:subev}(\ref{lemma:subev1})
$\preGA\comoccA\beta\in\EGGA(\G\parG\Msg)$.\\
If the applied rule is \rulename{Ext-In}, then $\G= \agtI\pp\q \la
{\G'}$ and $\beta = \CommAsI{\pp}{\la}{\q}$ and
$\Msg\equiv\addMsg{\mq\pp{\la}\q}\Msg'$. By
\refToLemma{lemma:subev}(\ref{lemma:subev2})
$\preGA\comoccA\beta\in\EGGA(\G\parG\Msg)$.

\medskip \noindent
{\it Inductive Cases.}  If the last applied  rule  is
\rulename{IComm-Out}, then $\G=\agtO\pp\q i I \la \G$ and
$\G'=\agtO\pp\q i I \la {\G'}$ and $
\G_i\parG\addMsg{\Msg}{\mq\pp{\la_i}\q}\stackred\asCom\G_i'
\parG\addMsg{\Msg'}{\mq\pp{\la_i}\q}$
for all $i \in I $ and
$\pp\not\in\play\beta$.

\medskip
By \refToDef{egA}(\ref{eg1AP}) $\comoccA\in\EGGA(\G'\parN\Msg')$
implies $\comoccA=\point(\os,\comseqA)$, where $\os=\osq{\Msg'}$ and
$\comseqA\in\FPaths{\G'}$. Then
$\comseqA=\concat{\CommAs\pp{\la_k}\q}{\comseqA'}$ and
$\comoccA=\eqA\os{\comseqA_0}$ with
$\comseqA_0=\filt{(\concat{\CommAs\pp{\la_k}\q}{\comseqA'})}\ee$ for
some $k\in I$ by \refToDef{def:pf}. We get either
$\comseqA_0\approx_\os\concat{\CommAs\pp{\la_k}\q}{\comseqA_0'}$ or
$\pp\not\in\play{\comseqA_0}$ by \refToDef{def:pm}.
Then $\postGA{\comoccA}{\CommAs\pp{\la_k}\q}$ is defined unless  $\comseqA_0\approx_\os\concat{\CommAs\pp{\la_k}\q}{\comseqA_0'}$ and
$\comseqA_0'=\ee$
by \refToDef{def:PostPreGlA}(\ref{def:PostPreGl1A}).  We consider
 the  two cases.

 \medskip\noindent
{\it Case  $\comseqA_0\approx_\os\concat{\CommAs\pp{\la_k}\q}{\comseqA_0'}$ and  $\comseqA_0'=\ee$.} We get
$\preGA\comoccA\beta=\eqA{\mapWh\beta\os}{\CommAs\pp{\la_k}\q}$ since
$\pp\not\in\play\beta$, which implies
$\preGA\comoccA\beta\in\EGGA(\G\parN\Msg)$ by
\refToDef{egA}(\ref{eg1AP}).

\vspace*{1.6mm}\noindent
{\it Case  $\comseqA_0\approx_\os\concat{\CommAs\pp{\la_k}\q}{\comseqA_0'}$ and   $\comseqA_0'\not=\ee$ or $\pp\not\in\play{\comseqA_0}$.}
Let $\comoccA'=\postGA\comoccA{\CommAs\pp{\la_k}\q}$.  By
\refToLemma{lemma:subev}(\ref{lemma:subev3}) $\comoccA'
\in\EGGA(\G'_k\parG\addMsg{\Msg'}{\mq\pp{\la_k}\q})$.  By induction
$\cauA\beta{ \comoccA'
}\in\EGGA(\G_k\parG\addMsg\Msg{\mq\pp{\la_k}\q})$.  Since $\comoccA'$
is defined, \refToLemma{prop:prePostGlA}(\ref{ppg1}) implies
$\cauA{\CommAs\pp{\la_k}\q}{ \comoccA' }=\comoccA$.  Since
$\cauA\beta{ \comoccA' }$ and $\cauA{\CommAs\pp{\la_k}\q}{ \comoccA'
}$ are defined, by \refToLemma{prop:prePostGlA}(\ref{ppg6A}) and
$\pp\not\in\play\beta$ we get $\cauA{\CommAs\pp{\la_k}\q}{(\cauA\beta{
    \comoccA' })}=\cauA\beta{(\cauA{\CommAs\pp{\la_k}\q}{ \comoccA'
  })}=\preGA\comoccA\beta$.  By
\refToLemma{lemma:subev}(\ref{lemma:subev1})
$\cauA{\CommAs\pp{\la_k}\q}{(\cauA\beta{ \comoccA'
  })}\in\EGGA(\G\parG\Msg)$.  We conclude that
$\preGA\comoccA\beta\in\EGGA(\G\parG\Msg)$.

\medskip
If the last applied  rule  is \rulename{IComm-In} the proof
is similar.
\end{proof}

%%%%%%%%%%%%%%%%%%%%

\begin{lemmaa}{\ref{eiog}}{ Let   $\comseqA \neq \ee$ be
$\os$-well formed and  $\gecA{\os,\comseqA}=\Seq{\comoccA_1;\cdots}{\comoccA_n}$.
\begin{enumerate}
\itemsep=0.9pt
\item  If $1\leq k,l\leq n$, then $\neg (\comoccA_k \grr \comoccA_l)$;
\item  $\at\comseqA{i} = \io{\comoccA_i} \,$ for all $i$,  $1\leq i\leq n$.
\end{enumerate} }
\end{lemmaa}

\begin{proof}
(\ref{eiog1})
Let $\comoccA_i=\eqA{\os}{\comseqA_i}$ for all $i$,  $1\leq
i\leq n$. By Definitions~\ref{gecA} and~\ref{def:pf} $\,\comseqA_i =
\filt{\range{\comseqA}{1}{i}}{\ee}$.
By \refToDef{def:trace-filtering} if $\projAP{\comseqA_i}\pp\not=\ee$, then there are $k_i\leq i$ and $\comseqA_i'$ such that $\play{\at\comseqA {k_i}}=\set\pp$, $\pp\not \in\play{\comseqA_i'}$ and $\comseqA_i=\filt{\range\comseqA 1 {k_i-1}}{\concat{\at\comseqA {k_i}}{\comseqA_i'}}$. By the same definition all  $\at\comseqA {j}$ with $j\leq k_i$ and
$\play{\at\comseqA {j}}=\set\pp$  occur in $\comseqA_i$, in the same order as in $\comseqA$. Theferore $\projAP{\comseqA_i}\pp$ is a prefix of $\projAP{\comseqA}\pp$ for all $\pp$ and all $i$, $1\leq i\leq n$. This implies that $\projAP{\comseqA_h}\pp$ cannot be in conflict with $\projAP{\comseqA_l}\pp$ for any $\pp$ and any  $h,l$, $1\leq h,l\leq n$.

\medskip (\ref{eiog2}) Immediate from Definitions~\ref{gecA}, \ref{def:pf}
and  \refToLemma{gl}.
 \end{proof}

%%%%%%%%%%%%%%%%%%%%

\begin{lemmaa}{\ref{lemma:Gred-trace-owf}}{If $\G\parG\Msg\stackred\comseqA\G'\parG\Msg'$ and
 $\os = \osq\Msg$, then $\comseqA$ is $\omega$-well formed.}
 \end{lemmaa}

\begin{proof}
The proof is by induction on $\comseqA$.

\vspace*{1.6mm}\noindent
 {\it Case $\comseqA =\beta$.}
 If $\beta=\CommAs{\pp}{\la}{\q}$, then the result is immediate. \\
If $\beta =\CommAsI\pp\la\q$, then from
$\G\parG\Msg\stackred\beta\G'\parG\Msg'$ we get
$\Msg\equiv\addMsg{\mq\pp{\la}\q}{\Msg'}$  by
Lemma~\ref{keysrA34}(\ref{keysrA4}), which implies
$\os\cong\concat{\CommAs\pp\la\q}{\os'}$.
 Then the trace $\concat{\os}{\beta} =
\concat{\concat{\CommAs\pp\la\q}{\os'}}{\CommAsI\pp\la\q}$ is well formed,
%$\comseqA$ is $\omega$-well formed,
 since $\CommAsI{\pp}{\la}{\q}$ is the first %and unique
 input of $\q$ from $\pp$ and $\CommAs{\pp}{\la}{\q}$ is the first output of  %communication from
$\pp$ to $\q$,  and therefore
$\ct 1 {\concat\os\beta} {\cardin{\os}+1}$. Hence $\beta$ is
$\omega$-well formed.

\vspace*{1.6mm}\noindent
 {\it Case $\comseqA =\concat\beta{\comseqA'}$ with $\comseqA'\neq\epsilon$.}  Let
$\G\parG\Msg\stackred\beta\G''\parG\Msg''\stackred{\comseqA'}\G'\parG\Msg'$
and $\os' = \osq{\Msg''}$.  By induction $\comseqA'$ is $\omega'$-well
formed.   If $\beta=\CommAs{\pp}{\la}{\q}$, then from
$\G\parG\Msg\stackred\beta\G''\parG\Msg''$ we get
$\Msg''=\addMsg{\Msg}{\mq{\pp}{\la}{\q}}$ by
Lemma~\ref{keysrA34}(\ref{keysrA3}).  Therefore
$\osq{\Msg''}=\omega\cdot\beta = \os'$. Since $\comseqA'$ is
$(\concat{\omega}{\beta})$-well formed,
i.e. $\concat{\concat{\omega}{\beta}}{\comseqA'}$ is well formed, we may conclude that
$\comseqA = \concat\beta{\comseqA'}$ is ${\omega}$-well formed.
If $\beta=\CommAsI{\pp}{\la}{\q}$,  as in the base case we get
$\Msg\equiv\addMsg{\mq{\pp}{\la}{\q}}{\Msg''}$ by Lemma~\ref{keysrA34}(\ref{keysrA4}),
and thus $\omega=\concat{\CommAs{\pp}{\la}{\q}}{\os'}$.
 We know that $\comseqA'$ is $\os'$-well formed,
i.e. ${\concat{\os'}{\comseqA'}}$ is well formed.
Therefore we have that
$\concat{\CommAs{\pp}{\la}{\q}}{\concat{\os'}{\concat{\CommAsI{\pp}{\la}{\q}}{\comseqA'}}}$
is well formed,  since  $\ct 1
{\concat\os\comseqA} {\cardin{\os}+1}$, and
we may conclude that $\comseqA$ is $\os$-well formed.
\end{proof}

%%%%%%%%%%%%%%%%%%%%

\begin{lemmaa}{\ref{kgec}}{\begin{enumerate}
\item
   Let $\comseqA=\concat\beta{\comseqA'}$ and
   $\os'=\mapBl{\beta}{\os}$. If
   $\gecA{\os,{\comseqA}}=\Seq{\comoccA_1;\cdots}{\comoccA_n}$ and
   $\gecA{\os',{\comseqA'}}=\Seq{\comoccA'_2;\cdots}{\comoccA'_n}$, then
   $\postGA{\comoccA_i}\beta={\comoccA'_i}\,$  for all $i$,   $2\leq i\leq n$.
   \item
   Let $\comseqA=\concat\beta{\comseqA'}$ and
   $\os=\mapWh{\beta}{\os'}$. If
   $\gecA{\os,{\comseqA}}=\Seq{\comoccA_1;\cdots}{\comoccA_n}$ and
   $\gecA{\os',{\comseqA'}}=\Seq{\comoccA'_2;\cdots}{\comoccA'_n}$, then
   $\preGA{\comoccA'_i}\beta={\comoccA_i}\,$  for all $i$,   $2\leq i\leq n$.
   \end{enumerate}
}\end{lemmaa}

\begin{proof}
The proof is by induction on $\comseqA$.

\medskip (\ref{kgec2})
By \refToDef{gecA}
$\comoccA_i=\point(\os,\concat{\beta}{\range{\comseqA'}{1}{i}})$ and
$\comoccA'_i=\point(\os',{\range{\comseqA'}{1}{i}})$ for all $i$,
$2\leq i\leq n$. Then by
\refToLemma{prop:relPrePost}(\ref{prop:relPrePost3}) $\preGA{\comoccA'_i}\beta={\comoccA_i}\,$  for all $i$,   $2\leq i\leq n$.

\medskip (\ref{kgec1}) By  Point (\ref{kgec2}) and \refToLemma{prop:relPrePost}(\ref{prop:relPrePost1}).
 \end{proof}

\subsection*{Glossary of symbols}

\vspace*{-4mm}
$\;$\\
$\begin{array}{ll}\small
symbol & meaning \\[6pt]
\beta & \text{input/output communication} ~\CommAs{\pp}{\M}{\q}, \CommAsI{\pp}{\M}{\q} \\[-1pt]
\delta & \text{type event} \\[-1pt]
\epsilon  & \text{empty trace} \\[-1pt]
\zeta & \text{sequence of input/output actions} \\[-1pt]
\eta & \text{process event (nonempty sequence of input/output actions)} \\[-1pt]
\vartheta & \text{\text{sequence of undirected actions $!\la$,
    $?\la$ } } \\[-1pt]
\pi & \text{input/output action} \ \sendL{\pp}{\la}, \rcvL{\pp}{\la}\\[-1pt]
\rho & \text{network event} \\[-1pt]
\tau & \text{trace (sequence of input/output communications)} \\[-1pt]
\chi & \text{sequence of output actions} \\[-1pt]
\omega & \text{sequence of output communications} \\[-1pt]
\end{array}
$

\end{document}